\documentclass[a4paper,11pt]{article}
\pdfoutput=1

\usepackage{verbatim}
\usepackage{jheppub}
\usepackage[T1]{fontenc}
\usepackage{slashed}
\usepackage[utf8]{inputenc}
\usepackage{array}
\usepackage{wrapfig}
\usepackage{multirow}
\usepackage{tabu}
\usepackage{tabularx}
\usepackage{tabularray}
\usepackage{amsmath,amssymb,amsthm}
\usepackage{enumitem}
\usepackage{graphicx}
\usepackage{caption}
\usepackage{placeins}
\usepackage{longtable}
\usepackage{float}
\usepackage{amsfonts}
\usepackage{caption}
\usepackage{subcaption}
\usepackage[dvipsnames]{xcolor}
\usepackage{tikz,tikz-3dplot}
\usepackage[compat=1.0.0]{tikz-feynman}
\usepackage{axodraw2}
\usetikzlibrary{calc,arrows.meta,shapes.geometric}
\usepackage{soul}


\def\R{\mathcal{R}}
\def\I{\mathcal{I}}
\def\T{\mathcal{T}}
\def\x{\boldsymbol{x}}
\def\v{\boldsymbol{v}}
\def\w{\boldsymbol{w}}

\def\s{\boldsymbol{s}}
\def\t{\boldsymbol{t}}
\def\eps{\epsilon}

\newtheorem{theorem}{Theorem}
\newtheorem{corollary}{Corollary}[theorem]

\newtheorem{proposition}{Proposition}

\title{\boldmath The on-shell expansion: from Landau equations to the Newton polytope}

\author[a]{Einan Gardi,} 
\author[a]{Franz Herzog,}
\author[b]{Stephen Jones,}
\author[a]{Yao Ma,}
\author[c]{Johannes Schlenk}

\affiliation[a]{Higgs Centre for Theoretical Physics, School of Physics and Astronomy,\\
The University of Edinburgh, Edinburgh EH9 3FD, Scotland, UK}
\affiliation[b]{Institute for Particle Physics Phenomenology, Durham University, Durham DH1 3LE, UK}
\affiliation[c]{ICS, University of Zurich, Winterthurerstrasse 190, Zurich, Switzerland}

\emailAdd{einan.gardi@ed.ac.uk}
\emailAdd{fherzog@ed.ac.uk}
\emailAdd{stephen.jones@durham.ac.uk}
\emailAdd{yao.ma@ed.ac.uk}
\emailAdd{johanneskaspar.schlenk@uzh.ch}

\abstract{ 
We study the application of the method of regions to Feynman integrals with massless propagators contributing to off-shell Green's functions in Minkowski spacetime (with non-exceptional momenta) around vanishing external masses, $p_i^2\to 0$.
This \emph{on-shell expansion} allows us to identify all infrared-sensitive regions at any power, in terms of infrared subgraphs in which a subset of the propagators become collinear to external lightlike momenta and others become soft.
We show that each such region can be viewed as a solution to the Landau equations, or equivalently, as a facet in the Newton polytope constructed from the Symanzik graph polynomials. This identification allows us to study the properties of the graph polynomials associated with infrared regions, as well as to construct a graph-finding algorithm for the on-shell expansion, which identifies all regions using exclusively graph-theoretical conditions.
We also use the results to investigate the analytic structure of integrals associated with regions in which every connected soft subgraph connects to just two jets. For such regions we prove that multiple on-shell expansions commute. This applies in particular to all regions in Sudakov form-factor diagrams as well as in any planar diagram.}

\begin{document}

\maketitle
\flushbottom

\section{Introduction}
\label{section-introduction}

The infrared structure of on-shell amplitudes in gauge theory has been a topic of continued interest for several decades, as discussed from different perspectives in several textbooks and reviews~\cite{EdenLdshfOlvPkhn02book,Stm95book,Cls11book,Stmg18,Hrch21,AgwMgnSgnrlTrpth21,BurStw13lectures, BchBrgFrl15book}.
The origin of these singularities has long been understood at the level of individual Feynman integrals with on-shell external momenta, using the Landau equations~\cite{Lnd59,Stm78I,LbyStm78,Stm96lectures,ClsSprStm04,Cls20}.
The solutions of these equations identify manifolds in the space of loop momentum integration where the Feynman integrand diverges, and in addition, the integration contour is pinched, generating potential singularities for the integral. These manifolds are called \emph{pinch surfaces}, and they completely capture the infrared singularity structure of Feynman integrals.

Insight regarding the infrared singularity structure of Feynman integrals may also be gained using the method of regions (MoR). This technique provides a systematic way to compute Feynman integrals involving multiple kinematic scales.
The main statement is that a Feynman integral $\mathcal{I}$ can be approximated, and even reproduced, by summing over integrals that are expanded in certain regions~$\left\{R_i\right\}$, i.e.
\begin{eqnarray}
\mathcal{I}= \mathcal{I}^{(R_1)}+\mathcal{I}^{(R_2)}+\dots+\mathcal{I}^{(R_n)}.
\label{MoR_definition}
\end{eqnarray}

Pioneered by Smirnov in momentum space, the MoR was first established for the large mass and momentum expansions in Euclidean space~\cite{Smn90,Smn94}. In particular, Smirnov showed that each of the regions in this expansion corresponds to a specific assignment of large loop momenta in a certain subgraph.
These subgraphs have been referred to as \emph{asymptotically irreducible} subgraphs.\footnote{The asymptotically irreducible subgraphs are the same as motic graphs~\cite{Brn15}, a notion used in more recent literature. We will revisit this concept in more detail in section~\ref{section-graph_finding_algorithm_regions}.} Prior to Smirnov's work, it was known since the 80s~\cite{CtkGrshnTch82,Ctk83,GrshnLrnTch83,Grshn86,GrshnLrn87,Ctk88I,Ctk88II,SmthDVr88,Grshn89} that an asymptotic expansion may be performed by identifying subgraphs whose loop momenta are of the same order of magnitude as the large masses or external momenta one expands in. This identification also facilitated uncovering a deep connection with forest formulas such as Bogoliubov's $R$-operation~\cite{BglPrs57,Hepp66,Zmm69} and its infrared generalisation in Euclidean space, the so-called $R^*$-operation~\cite{CtkTch82,CtkSmn84}. This connection also opened the way to proving the convergence of the large-momentum and large-mass expansions~\cite{Smn90}, putting these expansions on a more rigorous footing.

While the expansion-by-subgraph interpretation was limited to Euclidean space, the MoR is applicable also in Minkowski space. First examples included the threshold limit $q^2\to 4m^2$ for two-loop self-energy and vertex graphs~\cite{BnkSmn97}, the Sudakov limits for two-loop vertex graphs~\cite{SmnRkmt99}, etc. The departure from Euclidean space, and, in particular, the presence of lightlike external momenta, introduces new regions such as soft and collinear regions. Furthermore, in certain kinematic configurations, additional regions such as potential and Glauber show up \cite{Smn99,Smn02book}.
The identification of these regions has been made on a case-by-case basis, often using heuristic methods based on examples and experience. In recent years, detailed explanations of how the MoR would work in general has been provided by Jantzen~\cite{Jtz11}, and the MoR has been implemented in the formalism of loop-tree duality by Plenter and Rodrigo \cite{PltRdrg21}. Meanwhile, effective field theories, notably Heavy Quark Effective Theory (HQET) and Soft-Collinear Effective Theory (SCET), have been developed based on the complete characterisation of the regions appearing in particular kinematic situations.

Another way to understand the regions that appear in kinematic expansions, is to interpret them as certain scaling vectors in the parameter representation of a given Feynman graph~$G$~\cite{Plp08,PakSmn11,JtzSmnSmn12,SmnvSmnSmv19,AnthnrySkrRmn19,HrchJnsSlk22}. In particular, for the Lee-Pomeransky representation~\cite{LeePmrsk13}, which will be central for our analysis, each region is considered as the vector normal to a so-called \emph{lower facet} of the Newton polytope,\footnote{Besides the MoR, approaches based on the geometry of parametric representations of Feynman integrals, e.g. the Newton polytopes corresponding to the ${\cal U}$ and ${\cal F}$ polynomials or their Minkowski sum, have also been applied in the context of sector decomposition, tropical geometry, UV/IR divergences, maximal cut of Feynman graphs, etc.~\cite{KnkUeda10,AkHmHlmMzr22,AnthnryDasSkr20}.} which is defined as the convex hull of the exponent vectors of the sum, $\mathcal{P}(\x;\s)$, of the Symanzik polynomials, $\mathcal{P}\equiv \mathcal{U}+\mathcal{F}$. This approach can be applied directly when the monomials of $\mathcal{P}(\x;\s)$ are all of the same sign~\cite{PakSmn11, HrchJnsSlk22}. Instead, when the monomials of $\mathcal{P}(\x;\s)$ are of indefinite signs (possibly leading to a potential region or a Glauber region), a change of integration variables may be required before constructing the Newton polytope~\cite{JtzSmnSmn12,SmnvSmnSmv19,AnthnrySkrRmn19}. Based on these observations, computer codes such as Asy2~\cite{JtzSmnSmn12} (as part of the program FIESTA~\cite{SmnTtyk09FIESTA,SmnSmnTtyk11FIESTA2,Smn14FIESTA3,Smn16FIESTA4,Smn22FIESTA5}), ASPIRE~\cite{AnthnrySkrRmn19} and pySecDec~\cite{HrchJnsSlk22} have been developed to identify the complete set of regions.

A key observation we make here is that the regions of the MoR have a clear physical interpretation in terms of the infrared structure. Specifically, we expect that each of the aforementioned lower facets of the Newton polytope realises a particular solution of the Landau equations, namely a pinch surface.

In order to make a precise connection between the regions of the MoR and the solutions of the Landau equations, we focus in this paper on the special case of \emph{on-shell expansion} of \emph{wide-angle scattering}. 
To this end we consider Feynman integrals contributing to off-shell Green’s functions with massless fields in Minkowski spacetime. Starting with $(K+L)$ non-exceptional external momenta, $p_1^\mu,\dots,p_K^\mu,q_1^\mu,\dots,q_L^\mu$, we define the on-shell expansion by considering the limit where all $p_i^2$ become small while all other Lorentz invariants remain large. More precisely, introducing a scaling variable $\lambda\ll 1$ and a hard scale $Q^2$, we have
\begin{equation}
\begin{aligned}
&\text{on-shell:}
&p_i^2\sim \lambda Q^2\ \ (i=1,\dots,K),\\
&\text{off-shell:}&
q_j^2\sim Q^2\ \ (j=1,\dots,L),&\\
&\text{wide-angle:}&
p_k\cdot p_l\sim Q^2\ \ (k\neq l).&
\label{wideangle_scattering_kinematics}
\end{aligned}
\end{equation}
As we shall see, in this case, the MoR expresses any Feynman integral as a sum of a single hard region and a set of infrared regions. Each of the regions gives rise to an infinite series in powers of $p_i^2/Q^2\sim \lambda$. The infrared regions all correspond to nontrivial solutions of the Landau equations, which are characterised by infrared subgraphs, having propagators that become collinear with the external $p_i$, as well as ones that become soft when $\lambda\to 0$. Any such infrared region can be described by figure~\ref{classical_picture}.

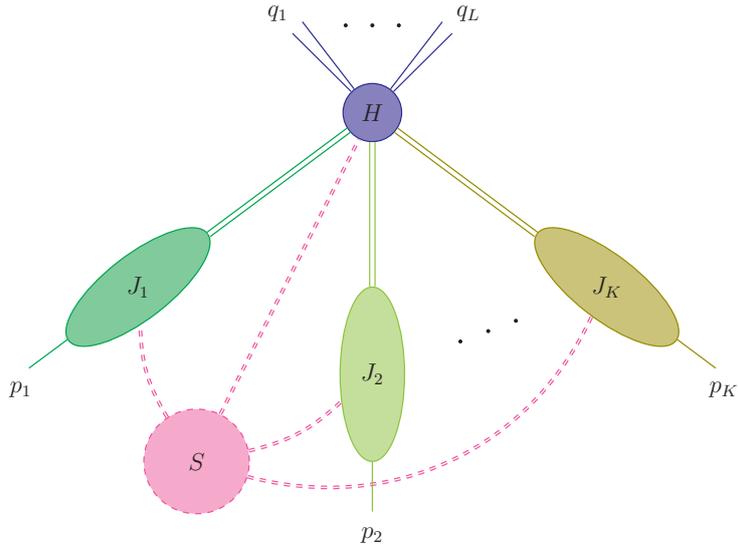
\begin{figure}[t]
\centering
  \resizebox{0.66\textwidth}{!}{
\begin{tikzpicture}[line width = 0.6, font=\large, mydot/.style={circle, fill, inner sep=.7pt}]

\node[draw=Blue,circle,minimum size=1cm,fill=Blue!50] (h) at (6,8){};
\node[dashed, draw=Rhodamine,circle,minimum size=1.8cm,fill=Rhodamine!50] (s) at (3,2){};
\node[draw=Green,ellipse,minimum height=3cm, minimum width=1.1cm,fill=Green!50,rotate=-52] (j1) at (2,5){};
\node[draw=LimeGreen,ellipse,minimum height=3cm, minimum width=1.1cm,fill=LimeGreen!50] (j2) at (6,3.5){};
\node[draw=olive,ellipse,minimum height=3cm, minimum width=1.1cm,fill=olive!50,rotate=52] (jn) at (10,5){};

\node at (h) {$H$};
\node at (s) {$S$};
\node at (j1) {$J_1$};
\node at (j2) {$J_2$};
\node at (jn) {$J_K$};

\path (h) edge [double,double distance=2pt,color=Green] (j1) {};
\path (h) edge [double,double distance=2pt,color=LimeGreen] (j2) {};
\path (h) edge [double,double distance=2pt,color=olive] (jn) {};

\draw (s) edge [dashed,double,color=Rhodamine] (h) node [right] {};
\draw (s) edge [dashed,double,color=Rhodamine,bend left = 15] (j1) {};
\draw (s) edge [dashed,double,color=Rhodamine,bend right = 15] (j2) {};
\draw (s) edge [dashed,double,color=Rhodamine,bend right = 40] (jn) {};

\node (q1) at (4.5,9.5) {};
\node (q1p) at (4.7,9.7) {};
\node (qn) at (7.5,9.5) {};
\node (qnp) at (7.3,9.7) {};
\draw (q1) edge [color=Blue] (h) node [] {};
\draw (q1p) edge [color=Blue] (h) node [left] {$q_1$};
\draw (qn) edge [color=Blue] (h) node [] {};
\draw (qnp) edge [color=Blue] (h) node [right] {$q_L$};

\node (p1) at (0,3.5) {};
\node (p2) at (6,1) {};
\node (pn) at (12,3.5) {};
\draw (p1) edge [color=Green] (j1) node [below] {$p_1$};
\draw (p2) edge [color=LimeGreen] (j2) node [below] {$p_2$};
\draw (pn) edge [color=olive] (jn) node [below] {$p_K$};

\path (j2)-- node[mydot, pos=.333] {} node[mydot] {} node[mydot, pos=.666] {}(jn);
\path (q1)-- node[mydot, pos=.333] {} node[mydot] {} node[mydot, pos=.666] {}(qn);

\end{tikzpicture}
}
\vspace{-2em}
\caption{The partitioning of a generic wide-angle scattering graph into infrared subgraphs corresponding to a particular pinch surface appearing in the on-shell expansion, eq.~(\ref{wideangle_scattering_kinematics}). The doubled lines connecting different blobs represent any number of propagators.}
\label{classical_picture}
\end{figure}

Having interpreted the individual regions in the MoR in terms of the infrared subgraphs, we thus generalise the notion of asymptotically irreducible graphs introduced in Euclidean space by Smirnov, to Minkowski space.

This paper is organised as follows. In section~\ref{section-onshell_expansion_parameter_representation} we introduce the parameter representation approach to the method of regions, and relate the solutions of the Landau equations to the region vectors. Section~\ref{section-regions_onshell_momentum_expansion} then identifies necessary and sufficient conditions for a solution of the Landau equations to be a region of the on-shell expansion. We then turn to discuss two applications of these results. In section~\ref{section-graph_finding_algorithm_regions} we derive a graph-theoretic algorithm to construct generic regions. In section~\ref{section-analytic_structure_commutativity_multiple_expansions} we investigate the analytic structure of certain on-shell expansions, and derive conditions for the cases where multiple expansions commute. Finally, section~\ref{section-conclusions_outlook} summarises the results and discusses potential future research.

Some detailed analyses are presented in the appendices. Explicitly, in appendix~\ref{appendix-possible_regions_nonplanar_doublebox_graph} we study a nonplanar double-box graph, and show that although there are positive and negative kinematic invariants in the $\mathcal{F}$ polynomial for this specific example, there are no regions due to cancellation between terms. 
In appendix~\ref{appendix-schwinger-lp} we review the relation between the Schwinger and Lee-Pomeransky representations and demonstrate that the \emph{same} scaling law for the respective 
parameters corresponding to hard, jet and soft propagators applies for the integrand in these two representations in any infrared region.
In appendix~\ref{appendix-necessity_sufficiency_requirements} we present the detailed proof of the main statement of section~\ref{infrared_regions_in_wideangle_scattering}, regarding the requirements of the subgraphs $H$, $J$ and~$S$.

\section{The on-shell expansion in the parameter representation}
\label{section-onshell_expansion_parameter_representation}

While initially formulated in momentum space, the MoR has been further developed in parameter space, where it reached rather general algorithmic formulations~\cite{Plp08,PakSmn11,JtzSmnSmn12,SmnvSmnSmv19,AnthnrySkrRmn19,HrchJnsSlk22}. Specifically, this paper will be primarily based on the formulation in ref.~\cite{HrchJnsSlk22} using the Lee-Pomeransky representation, where regions can be identified geometrically as specific facets of the corresponding Newton polytope.

In section~\ref{general_setup_MoR} we present the general setup for the MoR in parameter representation, where we define the expansion operator and introduce notations we shall use in this paper. Then, in section~\ref{region_vectors_from_Newton_polytope} we determine the regions geometrically by identifying them with certain facets of the Newton polytope, following the method of ref.~\cite{HrchJnsSlk22}. Finally, we relate the regions of the on-shell expansion to the particular solutions of the Landau equations in section~\ref{region_vectors_from_LE}, and propose that they stand in one-to-one correspondence.

\subsection{General setup}
\label{general_setup_MoR}

In this section we explain how the MoR is viewed in the Lee-Pomeransky representation. We first introduce our notation and define expansion operators associated to a given kinematic limit. 
We then demonstrate this general notation using the one-loop Sudakov form factor as an example.

Throughout this paper, we will use $G$ to denote a wide-angle scattering graph, and $N$, $L$ and $V$ to denote respectively the numbers of propagators, loops and vertices of $G$. Similarly, for any subgraph $\gamma \subseteq G$, the numbers of propagators, loops and vertices are separately $N(\gamma)$, $L(\gamma)$ and $V(\gamma)$.

We further denote the dimensionally-regularised Feynman integral in $D= 4-2\epsilon$ spacetime dimensions corresponding to a graph $G$ as $\mathcal{I}(G)$. In the Lee-Pomeransky representation,\footnote{In our paper we only discuss scalar integrals, for example, in eqs.~(\ref{lee_pomeransky_definition}) and (\ref{feynman_representation}). This can be extended to more generic Feynman integrals with nontrivial numerators, because any Feynman integral can be reduced to a sum over scalar integrals \cite{Trsv96}, each of which has the same Lee-Pomeransky polynomial $\mathcal{P}(\x;\s)$.}
\begin{equation}
\I(G)=\int [d\x]\,I(\x;\s)\equiv \frac{\Gamma(D/2)}{\Gamma((L+1)D/2-\nu)\prod_{e\in G}\Gamma(\nu_e)} \int_0^\infty \left( \prod_{e\in G} \frac{dx_e}{x_e} \right) \,I(\x;\s)\,,
\label{lee_pomeransky_definition}
\end{equation}
where $\nu_e$ is the exponent of the denominator associated to the propagator $e$, and $\nu\equiv \sum_{e\in G}\nu_e$. The integration measure $[d\x]$ is defined through the equation above, where the product goes over all the edges $e\in G$ with $x_e$ being the Lee-Pomeransky parameters. We denote the set (or vector) of $N$ Lee-Pomeransky parameters by $\x$. We also denote the set (or vector) of Lorentz invariants (Mandelstam variables) by $\s$, which consists of scalar products formed amongst the momenta $p_1^\mu,\dots,p_K^\mu,q_1^\mu,\dots,q_L^\mu$.
Since the integral $\I(G)$ is a function of $\s$, we will occasionally use the notation $\I(\s)$. The corresponding integrand $I(\x;\s)$ reads
\begin{eqnarray}
I(\x;\s)= \left(\prod_{e\in G} x_e^{\nu_e} \right) \cdot \Big( \mathcal{P}(\x;\s) \Big)^{-D/2},\qquad \mathcal{P}(\x;\s)\equiv \mathcal{U}(\x)+\mathcal{F}(\x;\s),
\label{lee_pomeransky_integrand_definition}
\end{eqnarray}
where $\mathcal{P}(\x;\s)$ is the Lee-Pomeransky polynomial, and $\mathcal{U}(\x)$ and $\mathcal{F}(\x;\s)$ are the first and second Symanzik polynomials, given by:
\begin{eqnarray}
\mathcal{U}(\x)=\sum_{T^1}^{}\prod_{e\notin T^1}^{}x_e,\qquad \mathcal{F}(\x;\s)=-\sum_{T^2}^{} s_{T^2}^{} \prod_{e\notin T^2}^{}x_e +\mathcal{U}(\x)\sum_{e}^{}m_e^2 x_e\ . \label{UFterm_general_expression}
\end{eqnarray}
The notations $T^1$ and $T^2$ respectively denote a \emph{spanning tree} and a \emph{spanning 2-tree} of the graph $G$. The symbol $s_{T^2}^{}$ is the square of the momentum flowing into each of the components of the spanning 2-tree $T^2$. In this paper we consider only massless propagators, and hence we set all the internal masses to zero, so eq.~(\ref{UFterm_general_expression}) can be simplified to
\begin{eqnarray}
\mathcal{U}(\x)=\sum_{T^1}^{}\prod_{e\notin T^1}^{}x_e,\qquad \mathcal{F}(\x;\s)=-\sum_{T^2}^{}s_{T^2}^{} \prod_{e\notin T^2}^{}x_e.
\label{UFterm_onshell_expression}
\end{eqnarray}

Later we will also use the Feynman parameterised integral in momentum space, which may be written as 
\begin{eqnarray}
\I(\s)= \frac{\Gamma(\nu)}{\prod_{e\in G}\Gamma(\nu_e)}\int_0^\infty \bigg(\prod_{e\in G} d\alpha_e \alpha_e^{\nu_e-1}\bigg) \delta \bigg(\sum_{e\in G}\alpha_e -1\bigg) \int [d\boldsymbol{k}] \frac{1}{[\mathcal{D}(k,p,q;\alpha)]^\nu}.\nonumber\\
\label{feynman_parameterisation_form}
\end{eqnarray}
where we introduced $k=\{k_a\}$, $p=\{p_i\}$, $q=\{q_j\}$ and $\alpha=\{\alpha_e\}$. Here the integration measure $[d\boldsymbol{k}]\equiv \prod_{a=1}^L d^Dk_a/(2\pi)^D$, and the function in the denominator reads
\begin{eqnarray}
\mathcal{D}( k,p,q;\alpha)= \sum_{e\in G} \alpha_e \left(-l_e^2(k,p,q)+m_e^2 -i\varepsilon\right),
\label{feynman_parameterisation_denominator}
\end{eqnarray}
where $\alpha_e$ is the Feynman parameter associated to propagator $e$ and the corresponding momentum $l_e^\mu$ is a linear combination of the internal loop momenta $k_a^\mu$ and the external momenta $p_i^\mu$ and $q_j^\mu$. After integrating over the loop momenta, one obtains the Feynman parameter representation of $\I(\s)$:
\begin{eqnarray}
\I(\s)=\frac{\Gamma(\nu-LD/2)}{\prod_{e\in G}\Gamma(\nu_e)} \int_0^\infty \left(\prod_{e\in G} d\alpha_e \alpha_e^{\nu_e-1}\right) \delta \left(\sum_{e\in G}\alpha_e -1\right) \frac{\left[\mathcal{U}(\boldsymbol{\alpha})\right]^{\nu-(L+1)D/2}}{\left[\mathcal{F}(\boldsymbol{\alpha};\s)\right]^{\nu-LD/2}}.
\label{feynman_representation}
\end{eqnarray}
Note that the Feynman and Lee-Pomeransky representations can be obtained from each other via the following change of variables
\begin{eqnarray}
\alpha_e=\frac{x_e}{x_1+x_2+\dots+x_N}.
\label{Feynman_LP_relation}
\end{eqnarray}
Starting from eq.~(\ref{lee_pomeransky_definition}), one can first insert $1=\int_{-\infty}^\infty \delta (\sum_{e\in G} x_e -X) dX$, and then change the variables $x_e= X\alpha_e$ according to (\ref{Feynman_LP_relation}). At this point we use the homogeneity properties of the Symanzik polynomials in (\ref{UFterm_general_expression}), noting that the $\mathcal{U}$ polynomial scales as $X^L$ while the $\mathcal{F}$ polynomial scales as $X^{L+1}$, allowing us to integrate over $X$ and obtain a ratio of gamma functions. The result is exactly (\ref{feynman_representation}).

We now consider the expansion of eq.~(\ref{lee_pomeransky_definition}) around the kinematic limit where the Mandelstam variables $\t \subset \s$ become small. To this end it is convenient to introduce a scaling vector $\w_{\t}$ in the space of Mandelstam variables $\s$:
\begin{equation}
    \w_{\t}=(w_1,w_2,w_3,\dots).
    \label{kinematic_variables_scaling}
\end{equation}
For each $i$, $w_i=1$ if $s_i\in \t$ and $w_i=0$ otherwise. In the on-shell limit shown in eq.~(\ref{wideangle_scattering_kinematics}), we have $\boldsymbol{s} \to  \lambda^{\w_{\t}} \boldsymbol{s}$, or equivalently, $s_i\to  \lambda^{w_i}s_i$ for every $i$.
Here, we have introduced notation for the raising of a scalar to the power of a vector and for the (Hadamard/component-wise) product of two vectors. Given a scalar $\lambda$ and two vectors $\boldsymbol{a}=(a_1,a_2,a_3,\dots)$ and $\boldsymbol{b}=(b_1,b_2,b_3,\dots)$ we define
\begin{eqnarray}
\lambda^{\boldsymbol{b}} \boldsymbol{a} \equiv (\lambda^{b_1} a_1,\lambda^{b_2} a_2,\lambda^{b_3} a_3,\dots).
\label{lambda_power_b_times_a}
\end{eqnarray}
We will also use the following definition for the raising of a vector to the power of a vector
\begin{equation}
\boldsymbol{a}^{\boldsymbol{b}} \equiv (a_1^{b_1}, a_2^{b_2}, a_3^{b_3} \ldots).
\end{equation}

The MoR states that in order to obtain the correct asymptotic expansion, one needs to sum over a finite set of regions, which we denote by $\mathcal{R}(\I(\s),\t)$. In each region $R\in \mathcal{R}(\I(\s),\t)$, the relevant contribution arises from the scaling of the Lee-Pomeransky parameters $x_i\to \lambda^{u_i} x_i$ together with the external kinematic variables $s_i\to \lambda^{w_i} s_i$. The Lee-Pomeransky integrand in eq.~(\ref{lee_pomeransky_integrand_definition}) is then expanded in powers of $\lambda$, and the integral in eq.~(\ref{lee_pomeransky_definition}) is performed term by term. Note that $\lambda$ serves as a bookkeeping parameter and is finally set to $1$.

We can express the Feynman integral $\I(\s)$ as a sum over contributions from each region:
\begin{eqnarray}
\I(\s)\simeq  \T_{\t}\I(\s)&& \equiv \!\!\!\sum_{R \in \mathcal{R}(\I(\s),\t)} \!\!\!\I^{(R)},\qquad\quad \I^{(R)}= \T_{\t}^{(R)} \I(\s),
\label{expansion_by_regions_definition}
\end{eqnarray}
where the operator $\T_{\t}^{(R)}$ produces a Taylor expansion in $\lambda$ by acting on the integrand:
\begin{eqnarray}
\T_{\t}^{(R)} \I(\s)\equiv \left. \int[d\x]\ T_\lambda \left( \lambda^{-p_R^{}(\epsilon)} I(\lambda^{\boldsymbol{u}_R}\x;\lambda^{\w}\s) \right) \right|_{\lambda=1}.
\label{integral_expansion_operator_definition}
\end{eqnarray}
The symbol $\simeq$ in 
eq.~(\ref{expansion_by_regions_definition}) is used to indicate that $\T_{\t}\I(\s)$ is an approximation to $\I(\s)$ when the Taylor expansion is truncated, while it becomes an equality when the expansion is summed up to all orders.
In eq.~(\ref{integral_expansion_operator_definition})
we have rescaled the kinematic variables as well as the Feynman parameters according to 
\begin{equation}
\s\to \lambda^{\w}\s \qquad \text{and} \qquad \x\to \lambda^{\boldsymbol{u}_R}\x,
\end{equation}
where $\w$ is defined in (\ref{kinematic_variables_scaling}) and $\boldsymbol{u}_R$ is called the \emph{region vector} (we will see one method for determining region vectors in section~\ref{region_vectors_from_Newton_polytope}). We have also assumed in eq.~(\ref{integral_expansion_operator_definition}) that the asymptotic behaviour of the integrand in the region $R$ is
\begin{eqnarray}
I(\lambda^{\boldsymbol{u}_R}\x;\lambda^{\w}\s) \to \lambda^{p_R(\epsilon)} I_0^{(R)}(\x;\s),
\label{integrand_asymptotic_behaviour}
\end{eqnarray}
as $\lambda\to 0$, where the exponent $p_R^{}(\epsilon)$ is a linear function of the dimensional regularisation parameter~$\epsilon$, and $I_0$ represents the \emph{leading-order} (in $\lambda$) approximation of the integrand in the region $R$. Then, by multiplying the rescaled integrand $I(\lambda^{\boldsymbol{u}_R}\x;\lambda^{\w}\s)$ by $\lambda^{-p_R^{}(\epsilon)}$, we can take in eq.~(\ref{integral_expansion_operator_definition}) a regular Taylor expansion:
\begin{eqnarray}
T_\lambda\equiv \sum_{n=0}^\infty \lambda^n T_{\lambda,n},\qquad T_{\lambda,n} (\cdots)= \left.\frac{1}{n!} \frac{d^n}{d\lambda^n} (\cdots)\right|_{\lambda=0}.
\label{integrand_expansion_operator_definition}
\end{eqnarray}
It is convenient to define the Lee-Pomeransky polynomial corresponding to the leading behaviour of eq.~(\ref{integrand_asymptotic_behaviour}) in region $R$ as $\mathcal{P}_0^{(R)}(\x;\s)$. In each region, this polynomial contains only a subset of the monomials of the original Lee-Pomeransky polynomial $\mathcal{P}(\x;\s)$.

We stress the following key aspects~\cite{Smn02book,Jtz11,SmnvSmnSmv19} concerning eqs.~(\ref{expansion_by_regions_definition}) and (\ref{integral_expansion_operator_definition}).
First, the regions in the set $\mathcal{R}(\I(\s),\t)$ should be defined such that each point in the space of integration belongs to exactly one of them. Second, despite the fact that the expansion of the integrand in eq.~(\ref{integral_expansion_operator_definition}) is a valid approximation only within the region $R$, the integral $\I^{(R)}$ in eq.~(\ref{expansion_by_regions_definition}) is taken over the entire space. This relies on using dimensional regularisation, where any scaleless integral is set to zero, namely,
\begin{eqnarray}
\int [d\x] I_0^{(R)}(\x;\s)=0\qquad\text{ if }\ \ I_0^{(R)}(c^{\boldsymbol{u}} \x;\s) = c^q I_0^{(R)}(\x;\s),
\label{eq:uf_scaleless}
\end{eqnarray}
for some $c,q\in\mathbb{R}$ ($c\neq 1$) and $\boldsymbol{u} \in \mathbb{R}^N$.

To summarise, the MoR in eqs.~(\ref{expansion_by_regions_definition}) and (\ref{integral_expansion_operator_definition}) states that
\begin{equation}
\I(\s)\simeq \T_{\t} \I(\s)= \sum_{R \in \R(I(G),\t)}\int[d\x] \left. T_\lambda \left( \lambda^{-p_R^{}(\epsilon)} I(\lambda^{\boldsymbol{u}_R}\x;\lambda^{\w}\s) \right) \right|_{\lambda=1}.
\end{equation}
At a given order $n$ in the expansion we have
\begin{align}
\T_{\t,n} \equiv \sum_{R\in \R(\I,\t)} \T^{(R)}_{\t,n},\qquad \T_{\t,n}^{(R)} \I(\s)= \int[d\x] \left. T_{\lambda,n} \left( \lambda^{-p_R^{}(\epsilon)} I(\lambda^{\boldsymbol{u}_R}\x;\lambda^{\w}\s) \right) \right|_{\lambda=1}.
\label{integral_expansion_operator_definition_fixed_order}
\end{align}

Let us demonstrate the above definitions using a simple example  relevant to the one-loop Sudakov form factor. Consider the triangle graph, 
\begin{equation}
G_{3}\equiv
\begin{tikzpicture}[baseline=13ex,scale=1.]
\coordinate (x1) at (1.1340, 1.4999) ;
\coordinate (x2) at (2.8660, 1.5000) ;
\coordinate (x3) at (2,3) ;
\node (p1) at (0.2681, 0.9998) {$p_1$};
\node (p2) at (3.7320, 1.0000) {$p_2$};
\node (p3) at (2,4) {$q_1$};
\draw[color=green] (x1) -- (p1);
\draw[ultra thick,color=Blue] (x3) -- (p3);
\draw[color=ForestGreen] (x2) -- (p2);
\draw[ultra thick,color=Black] (x1) -- (x2);
\draw[ultra thick,color=Black] (x2) -- (x3);
\draw[ultra thick,color=Black] (x3) -- (x1);
\draw[fill,thick,color=Blue] (x1) circle (1pt);
\draw[fill,thick,color=Blue] (x2) circle (1pt);
\draw[fill,thick,color=Blue] (x3) circle (1pt);
\end{tikzpicture},
\end{equation}
whose external momenta $p_1^\mu$, $p_2^\mu$ and $q_1^\mu=-p_1^\mu-p_2^\mu$ satisfy $|p_1^2|\sim |p_2^2|\ll |q_1^2|$. We consider the expansion where the external momenta $p_1^\mu$ and $p_2^\mu$ become on-shell, so the relevant sets of small and large Mandelstam invariants are 
\begin{equation}
\begin{aligned}
&\t= ( p_1^2,p_2^2 )\subset \s=( p_1^2,p_2^2,q_1^2 )\,,
\end{aligned}
\end{equation}
the scaling vector of eq.~(\ref{kinematic_variables_scaling}) is $\w=(1,1,0)$, so $\lambda^{\w}\s=( \lambda^1 p_1^2,\lambda^1 p_2^2,\lambda^0 q_1^2)$,
and the Symanzik polynomials of eq.~(\ref{UFterm_onshell_expression}) are
\begin{equation}
\begin{aligned}
&\mathcal{U}= x_1 + x_2 + x_3,\qquad \mathcal{F}= (-p_1^2) x_1x_3+ (-p_2^2) x_2x_3 +(-q_1^2) x_1x_2.
\label{UFterms_oneloop_Sudakov_example}
\end{aligned}
\end{equation}
Consider the specific case where $\nu_1=\nu_2=\nu_3 =1$, and then the Lee-Pomeransky integrand, according to eq.~(\ref{lee_pomeransky_integrand_definition}), is 
\begin{eqnarray}
I(\x;\s)= x_1x_2x_3 (x_1 + x_2 + x_3 -p_1^2 x_1x_3- p_2^2 x_2x_3 - q_1^2 x_1x_2)^{-D/2}.
\label{integrand_oneloop_Sudakov_example}
\end{eqnarray}
As we will see in section~\ref{region_vectors_from_Newton_polytope}, there are four associated regions, which are defined by the following $N$-dimensional scaling vectors, $\boldsymbol{u}_R$,
\begin{equation}
\begin{aligned}
&\mathrm{Hard}\ (H): & &\boldsymbol{u}_H= (0,0,0),&\\
&\mathrm{Collinear\,1}\ (C_1):& &\boldsymbol{u}_{C_1}= (-1,0,-1),&\\
&\mathrm{Collinear\,2}\ (C_2):& &\boldsymbol{u}_{C_2}= (0,-1,-1),&\\
&\mathrm{Soft}\ (S): & &\boldsymbol{u}_S= (-1,-1,-2).&
\end{aligned}
\label{eq:sudakov_region_vectors}
\end{equation}

The MoR then claims that, for this example,
\begin{eqnarray}
\I(\s)= \T_{\t}^{(H)}\I(\s)+ \T_{\t}^{(C_1)}\I(\s)+ \T_{\t}^{(C_2)}\I(\s)+ \T_{\t}^{(S)}\I(\s).
\label{Sudakov_example_expansion}
\end{eqnarray}
To determine the leading contribution 
from each of the four regions in eq.~(\ref{eq:sudakov_region_vectors}), consider rescaling the integration parameters in eq.~(\ref{lee_pomeransky_definition})
according to~$\s\to \lambda^{\w}\s$ and $\x\to \lambda^{\boldsymbol{u}_R}\x$. Since the integration measure in eq.~(\ref{lee_pomeransky_definition}) is rescaling invariant, only the integrand of eq.~(\ref{integrand_oneloop_Sudakov_example}) changes under the rescaling. Consider for instance the case $R=S$, where $\boldsymbol{u}_S= (-1,-1,-2)$. The rescaled integrand for $D=4-2\epsilon$ reads
\begin{align}
I(\lambda^{\boldsymbol{u}_S}\x;\lambda^{\w}\s)&=I(\lambda^{-1}x_1, \lambda^{-1}x_2, \lambda^{-2}x_3;\lambda^{1}p_1^2, \lambda^{1}p_2^2, q_1^2)\nonumber\\
&=\lambda^{-4} x_1x_2x_3 \Big(\lambda^{-1}(x_1 + x_2) + \lambda^{-2}(x_3 -p_1^2 x_1x_3- p_2^2 x_2x_3 - q_1^2 x_1x_2) \Big)^{\epsilon-2} \nonumber \\
&= \lambda^{-2\epsilon} x_1x_2x_3 (\mathcal{P}_0^{(S)}(\x;\s))^{\epsilon-2} +\cdots,
\label{integrand_oneloop_Sudakov_example_rescaled}
\end{align}
where in the last expression we neglected terms that are suppressed by powers of $\lambda$, and identified the Lee-Pomeransky polynomial in the soft region, $R=S$, as
\begin{eqnarray}
\mathcal{P}_0^{(S)}(\x;\s)= x_3 +(-p_1^2)x_1x_3 +(-p_2^2)x_2x_3 +(-q_1^2)x_1x_2.
\end{eqnarray}
It follows from eq.~(\ref{integrand_oneloop_Sudakov_example_rescaled}) that $p_{S}^{}(\epsilon)=-2\epsilon$. Hence, according to eqs.~(\ref{integral_expansion_operator_definition}) and (\ref{integral_expansion_operator_definition_fixed_order}), the soft expansion operator acts on the integral $\I$ as follows:
\begin{eqnarray}
&&\T_{\t}^{(S)}\I(\s)= \int[d\x] \left. T_\lambda \left( \lambda^{2\epsilon} I(\lambda^{\boldsymbol{u}_S}\x;\lambda^{\w}\s) \right) \right|_{\lambda=1} \nonumber\\
&&\hspace{0.8cm} =\sum_{n=0}^\infty \frac{(\epsilon-2)_n}{n!}\int[d\x] x_1 x_2 x_3 \left(x_1+x_2\right)^n \left( x_3 -p_1^2x_1x_3 -p_2^2x_2x_3 -q_1^2x_1x_2 \right)^{\epsilon-2-n},\nonumber\\
&&
\label{1loop_Sudakov_soft_expansion}
\end{eqnarray}
where we have inserted the rescaled integrand of eq.~(\ref{integrand_oneloop_Sudakov_example_rescaled}) and expanded in $\lambda$ to all orders. Here $()_n$ denotes the falling factorial, i.e. $(a)_n\equiv a(a-1)\cdots(a-n+1)$. The right-hand side of the second equality of eq.~(\ref{1loop_Sudakov_soft_expansion}) represents a sum over the terms $\T_{\t,n}^{(S)}\I(\s)$. The same procedure can be carried out analogously for the other three regions $H$, $C_1$ and $C_2$. In particular we obtain
\begin{equation}
p_H^{}(\epsilon)=0,\quad 
p_{C_1}^{}(\epsilon)=-\epsilon,\quad 
p_{C_2}^{}(\epsilon)=-\epsilon,\quad 
p_{S}^{}(\epsilon)=-2\epsilon,
\label{leading_powers_oneloop_triangle_regions}
\end{equation}
reflecting a different analytic behaviour characteristic to each of the regions for small $p_i^2$. We shall return to analyse this in section~\ref{section-analytic_structure_commutativity_multiple_expansions}.

\subsection{Region vectors from the Newton polytope}
\label{region_vectors_from_Newton_polytope}

In this section, we briefly summarise the geometric formulation of the MoR in parameter space, following ref.~\cite{HrchJnsSlk22}. We review how the region vectors, $\boldsymbol{v}_R^{}$, which are used to rescale the Lee-Pomeransky parameters~$\x$, in each region~$R$, can be obtained by considering certain \emph{facets} of a \emph{Newton polytope} associated to the integral. We then revisit our previous example, the one-loop Sudakov form factor, to demonstrate how the regions are obtained directly from a geometric point of view.

In order to define the Newton polytope of the integral $\I(G)$ given by eqs.~(\ref{lee_pomeransky_definition}) and~(\ref{lee_pomeransky_integrand_definition}), we first consider the polynomial $\mathcal{P}(\x,\lambda^{\w} \s)$ obtained from $\mathcal{P}(\boldsymbol{x},\boldsymbol{s})$ by rescaling all of the Mandelstam variables by $\s\to \lambda^{\w} \s$.
In general, the polynomials we will consider are of the form:
\begin{align}
&\mathcal{P}(\boldsymbol{x};\lambda^{\w} \s) = \sum_{i=1}^m c_i(\s)\, x_1^{r_{i,1}} \dots x_N^{r_{i,N}}\, \lambda^{r_{i,N+1}} = \sum_{i=1}^m c_i(\s)\, \boldsymbol{x}^{\hat{\boldsymbol{r}}_i}\, \lambda^{r_{i,N+1}},
\label{LP_polynomial_rescaling_momenta}
\end{align}
where $r_{i,j} \in \{0,1\}$ for any $i=1,\dots,m$ with $m$ being the number of terms (monomials) in the Lee-Pomeransky polynomial, and $j=1,\dots,N+1$. (In the case of massive propagators, $r_{i,j} \in \{0,1,2\}$.)
We define the $(N+1)$-dimensional exponent vectors as $\boldsymbol{r}_i\equiv (\hat{\boldsymbol{r}}_i;r_{i,N+1}) \equiv (r_{i,1}, \ldots, r_{i,N}; r_{i,N+1})$ and for the work presented here, we assume $r_{i,N+1}=0$ or $1$. 
For the description of the geometric method below, we also demand ${c_i > 0\ \forall\ i}$, which forbids sets of monomials from cancelling each other at any point in the integration domain ${x_i \geqslant 0\ \forall\ i}$.
The method can be applied with some $c_i <0$, provided the cancellation of monomials does not lead to new regions. Generally speaking, the geometric method will identify only the regions that are present when all $c_i>0$. These regions correspond to endpoint singularities in parameter space.

The Newton polytope can now be defined as the \emph{convex hull} of the \emph{vertices} (dimension-0 \emph{faces}) given by the polynomial exponent vectors
\begin{eqnarray}
\Delta^{(N+1)}[\mathcal{P}] = \mathrm{convHull}(\boldsymbol{r}_1,\ldots,\boldsymbol{r}_m) = \left\{ \sum_i^m \alpha_i \boldsymbol{r}_i\  |\  \alpha_i \geqslant 0 \land \sum_i^m \alpha_i = 1 \right\}, \label{eq:convex_hull}
\end{eqnarray}
or, alternatively, as an intersection of half-spaces,
\begin{align}
&\Delta^{(N+1)}[\mathcal{P}] = \bigcap_{f\in F} \left\{ \boldsymbol{\rho}\in\mathbb{R}^{N+1} \mid  \boldsymbol{\rho}\cdot\boldsymbol{v}_f + a_f \geqslant 0 \right\},&
&a_f \in \mathbb{Z}\quad \forall\ f,&
\label{eq:NPhalfspaces}
\end{align}
where $F$ is the set of \emph{facets} (codimension 1 faces) of the polytope and each $\boldsymbol{v}_f$ is an inward-pointing vector normal to facet $f$. 
The vectors $\boldsymbol{v}_f$ take the form of eq.~(\ref{v_R_definition}). Let facets with an inward-pointing normal vector with a positive component in the $\lambda$ direction (i.e. with component $v_{f,N+1} > 0$) be called \emph{lower facets}, and let us denote the set of all such facets $F_+$. 
The region vectors of the contributing regions are given by these lower facets, i.e. $\{ \boldsymbol{v}_R \} = \{ \boldsymbol{v}_f, \ \forall\ f \in F_+ \}$.
Several computer packages exist for computing Newton polytopes (or convex hulls) and their representation in terms of facets, see for example refs.~\cite{Normaliz,Qhull}.

According to eq.~\eqref{eq:convex_hull}, the Newton polytope $\Delta^{(N+1)}[\mathcal{P}]$ is thus an $(N+1)$-dimensional polytope enclosing all points defined by the exponent vectors $\boldsymbol{r}_i$.
The first $N$ dimensions correspond to the Lee-Pomeransky parameters $x_1, \ldots, x_N$, which are integrated over, while the $(N+1)$-th dimension corresponds to the exponent of the expansion parameter, $\lambda$, which emerges from rescaling the external invariants $\s\to \lambda^{\w}\s$. 
Starting from eq.~(\ref{LP_polynomial_rescaling_momenta}), in each region, we will consider a polynomial of the form
\begin{align}
\mathcal{P}(\lambda^{\boldsymbol{u}_R} \boldsymbol{x}; \lambda^{\w} \s) &= \sum_{i=1}^m c_i(\s)\, \lambda^{\boldsymbol{u}_R \cdot \hat{\boldsymbol{r}}}\, \boldsymbol{x}^{\hat{\boldsymbol{r}}}\, \lambda^{r_{i,N+1}} = \sum_{i=1}^m c_i(\s)\, \lambda^{\boldsymbol{v}_R \cdot \boldsymbol{r}}\, x^{\hat{\boldsymbol{r}}} \nonumber\\
&= \lambda^{-a_R}\sum_{i=1}^m c_i(\s)\, \lambda^{\boldsymbol{v}_R \cdot \boldsymbol{r}+a_R}\, x^{\hat{\boldsymbol{r}}}\,,
\label{eq:LPscalingsubstitution}
\end{align}
where
\begin{equation}
\boldsymbol{v}_R \equiv (\boldsymbol{u}_R; v_{R,N+1})
=(\boldsymbol{u}_R; 1)
\,.
\label{v_R_definition}
\end{equation}
Here have made identification of the first $N$ components of $\boldsymbol{v}_R$ as $\boldsymbol{u}_R$ of eq.~(\ref{integral_expansion_operator_definition}), while the last component, is set to one: $v_{R,N+1}= 1$, representing the fact that $R$ corresponds to a lower facet.

The leading monomials in $\mathcal{P}(\lambda^{\boldsymbol{u}_R} \boldsymbol{x}, \lambda^{\w} \s)$, i.e. the ones having the smallest power of $\lambda$, are those that minimise $\boldsymbol{v}_R \cdot \boldsymbol{r}$.
Since each region vector is normal to a particular facet~$f$, of the Newton polytope, by comparing eqs.~(\ref{eq:NPhalfspaces}) and (\ref{eq:LPscalingsubstitution}) we see that the leading monomials, i.e.~those satisfying $\boldsymbol{v}_R \cdot \boldsymbol{r} + a_R =0$,  lie on and span the facet $f$, while all other points have $\boldsymbol{v}_R \cdot \boldsymbol{r} + a_R > 0$.
Crucially, after expanding in $\lambda$, we will obtain exactly the monomials which lie on the facet $f$.
By definition, facets have codimension 1, and so polynomials obtained from each region vector will thus live on an $N$-dimensional subspace of the Newton polytope.

We emphasise that in general, not all regions of the MoR are associated with facets~\cite{PakSmn11}.  It is generally accepted that the following holds. \emph{The regions contributing to the MoR fall into two categories, ones that are due to cancellations between terms of the Lee-Pomeransky polynomial $\mathcal{P}$, and ones that can be associated with facets of the corresponding Newton polytope $\Delta^{(N+1)}[\mathcal{P}]$.} The latter are identified with scaling vectors $\v_R$, the normal vectors of the lower facets, and in present paper, focusing on wide-angle scattering, we will assume that these are the only regions that are present. We briefly comment on the former at the end of this subsection.

In principle, as well as the facets of the Newton polytope, we could also consider lower dimensional faces (i.e. faces of codimension $> 1$), such faces correspond to intersections of the higher dimensional faces.
The normal vectors corresponding to faces of codimension $> 1$ select leading monomials in $\mathcal{P} (\x;\s)$, which span a face of dimension less than $N$. From eq.~(\ref{eq:uf_scaleless}) it follows that if the dimension of a lower face is less than $N$, then the corresponding dimensionally-regularised Feynman integral is scaleless, and therefore vanishes. Thus only facets contribute to the MoR.

Let us note that the lower facets of the Newton polytope obtained from the Lee-Pomeransky polynomial $\mathcal{P}(\x;\s) = \mathcal{U}(\x)+\mathcal{F}(\x;\s)$ are in one-to-one correspondence with the lower facets of the Newton polytope obtained using the product $\mathcal{U}(\x)\mathcal{F}(\x;\s)$. This implies that the region vectors obtained in the Feynman and Lee-Pomeransky representations are equivalent~\cite{SmnvSmnSmv19}, differing only by constant shifts of the form $(n\boldsymbol{1},0)$ with $n \in \mathbb{R}$. The Symanzik polynomials $\mathcal{U}(\x)$ and $\mathcal{F}(\x;\s)$ are homogeneous functions of $\x$ with degree $L$ and $L+1$, respectively. This means that in the $\x$ coordinates their Newton polytopes lie on separated parallel hyperplanes.
In this case, it follows from a geometric theorem known as the \emph{Cayley trick}, that the Newton polytope defined by the intersection of a hyperplane parallel to the $\mathcal{U}(\x)$ and $\mathcal{F}(\x;\s)$ hyperplanes with the Newton polytope of the sum $\mathcal{U}(\x) + \mathcal{F}(\x;\s)$ has the same convex hull, up to a rescaling, as the Newton polytope of the product $\mathcal{U}(\x) \mathcal{F}(\x;\s)$. We demonstrate the use of the Cayley trick in the example below, see, in particular, figure~\ref{fig:tri1_2}.

As an example of the geometric procedure for determining regions\footnote{Note that the monomials in ${\cal P}$ are all of the same sign provided we choose all kinematic invariants to be spacelike. In such Euclidean kinematics we can clearly exclude cancellations between terms in ${\cal P}$ and focus exclusively on the facets of the Newton polytope.  }, let us apply it to the one-loop triangle integral in the limit $p_1^2\sim p_2^2\ll q_1^2$ defined in eq.~(\ref{UFterms_oneloop_Sudakov_example}).
The Newton polytope is 4-dimensional and depends on $(r_{i,1}, r_{i,2}, r_{i,3}; r_{i,4})$, which are the exponents of the variables $( x_1,x_2,x_3; \lambda)$, it is given by
\begin{align}
\Delta^{(N+1)} \left[\mathcal{U}(G)+\mathcal{F}(G) \right] = \mathrm{convHull}(U_1,U_2,U_3,F_1,F_2,F_3),
\end{align}
where the vertices $U_1 = (1,0,0,0),\, U_2 = (0,1,0,0),\, U_3 = (0,0,1,0)$ correspond to monomials in $\mathcal{U}(G)$ and the vertices $F_1=(1,0,1,1),\, F_2=(0,1,1,1),\, F_3=(1,1,0,0)$ correspond to monomials in $\mathcal{F}(G)$.
Calculating the normal vector of each facet of the 4-dimensional Newton polytope and selecting those with a positive $\lambda$ component (i.e. the lower facets) immediately yields the region vectors,
\begin{align}
\begin{tabular}{lll}
&$\boldsymbol{v}_H = (0,0,0;1)$,\qquad
&$\boldsymbol{v}_S = (-1,-1,-2;1)$, \\
&$\boldsymbol{v}_{C_1} = (-1,0,-1;1)$,\qquad
&$\boldsymbol{v}_{C_2} = (0,-1,-1;1)$,
\end{tabular}
\end{align}
which match those stated in eq.~(\ref{eq:sudakov_region_vectors}) with the additional component $v_{R,N+1} = 1$, according to eq.~(\ref{v_R_definition}).

In order to visualise these regions, in figure~\ref{fig:tri1_1} we display a 3-dimensional projection of the Newton polytope, neglecting the $\lambda$ direction.
The colour of the vertices indicates whether they have a zero $\lambda$ exponent (black) or a nonzero $\lambda$ exponent (red).
The blue face corresponds to the $\mathcal{U}(G)$ polynomial and the green face corresponds to the $\mathcal{F}(G)$ polynomial.
The vertices of the Newton polytope obtained from the product of the Symanzik polynomials, $\Delta^{(N+1)}[\mathcal{U}(G)\mathcal{F}(G)]$, lie, up to a rescaling, on the edges connecting the vertices of the $\mathcal{U}(G)$ and $\mathcal{F}(G)$ faces, as depicted by the grey hexagon in figure~\ref{fig:tri1_2}. We use the notation $U_i F_j$ to denote the vertex obtained by taking the product of the $i$th term in $\mathcal{U}(G)$ with the $j$th term in $\mathcal{F}(G)$.

In figure~\ref{fig:tri1_3} we define a new $x_1^\prime x_2^\prime$ coordinate system on the $\mathcal{U}(G)\mathcal{F}(G)$ hyper-surface.
Note that because the $\mathcal{U}(G)\mathcal{F}(G)$ polynomial is homogeneous in the variables $\{x_1, x_2, x_3\}$ its corresponding hyper-surface is 3-dimensional (rather than 4-dimensional).
In figure~\ref{fig:tri1_4} we display the $\mathcal{U}(G)\mathcal{F}(G)$ polytope in the new coordinates and also show the $\lambda$ direction. Here the lower facets of the Newton polytope, which correspond to the different regions, are displayed in colour.
As we will demonstrate in the following sections, the upward-pointing light blue facet corresponds to the hard region $(H)$, the red facet corresponds to the soft region $(S)$ and the two green and yellow facets correspond to the collinear (jet) regions $(C_1)$ and $(C_2)$, respectively.
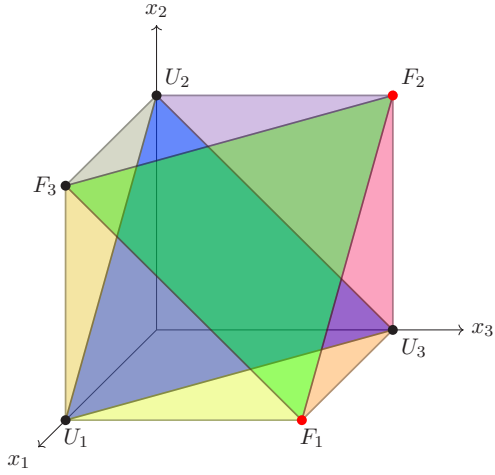
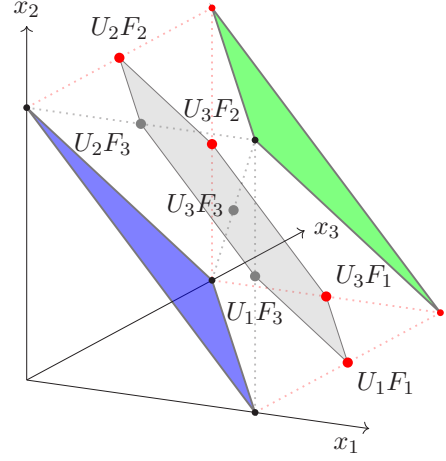
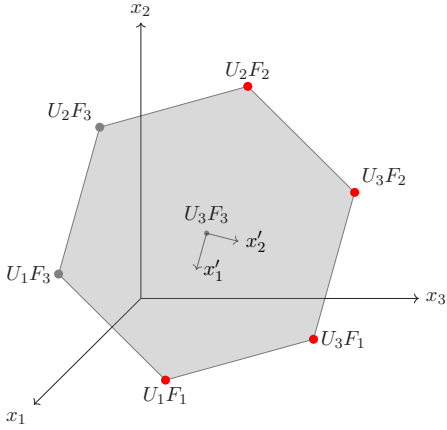
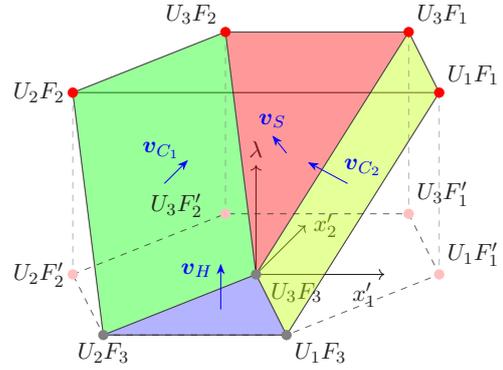
\begin{figure}[h]
  \centering
  \begin{subfigure}[b]{0.45\textwidth}
    \centering
      \resizebox{\textwidth}{!}{
    \begin{tikzpicture}[rotate around x=-90, rotate around z=-90, scale=4]
      \begin{scope}
        \draw[->] (0,0,0) -- (1.3,0,0) node[below left] {$x_1$};
        \draw[->] (0,0,0) -- (0,1.3,0) node[right] {$x_3$};
        \draw[->] (0,0,0) -- (0,0,1.3) node[above] {$x_2$};
        \begin{scope}[fill opacity=.5, draw opacity=.6, text opacity=1]
            %
            \draw [thick,fill=blue,fill] 
            (1,0,0) coordinate [label={[anchor=north west,xshift=-5pt]below:$U_1$}] (U0) 
            -- (0,0,1) coordinate [label=above right:$U_2$] (U1) 
            -- (0,1,0) coordinate [label={below right:$U_3$}] (U2)
            -- cycle;
            %
            \draw [thick,fill=green,fill] 
            (1,1,0) coordinate [label={[anchor=north west,xshift=-5pt]below:$F_1$}] (F0) 
            -- (0,1,1) coordinate [label=above right:$F_2$] (F1) 
            -- (1,0,1) coordinate [label={left:$F_3$}] (F2) 
            -- cycle;
            \draw [thick,fill,color=red,opacity=.2] (F0) -- (U2) -- (F1) -- cycle;
            \draw [thick,fill,color=cyan,opacity=.2] (F2) -- (U1) -- (F1) -- cycle;
            \draw [thick,fill,color=magenta,opacity=.2] (F1) -- (U2) -- (U1) -- cycle;
            \draw [thick,fill,color=yellow,opacity=.2] (U2) -- (U0) -- (F0) -- cycle;
            \draw [thick,fill,color=lime,opacity=.2] (F0) -- (U0) -- (F2) -- cycle;
            \draw [thick,fill,color=orange,opacity=.2] (U0) -- (F2) -- (U1) -- cycle;
            \draw [thick,opacity=.2] (F0) -- (U2) -- (F1) -- cycle;
            \draw [thick,opacity=.2] (F2) -- (U1) -- (F1) -- cycle;
            \draw [thick,opacity=.2] (F1) -- (U2) -- (U1) -- cycle;
            \draw [thick,opacity=.2] (U2) -- (U0) -- (F0) -- cycle;
            \draw [thick,opacity=.2] (F0) -- (U0) -- (F2) -- cycle;
            \draw [thick,opacity=.2] (U0) -- (F2) -- (U1) -- cycle;
            \node at (U0)[circle,fill,inner sep=1.7pt,opacity=1]{};
            \node at (U1)[circle,fill,inner sep=1.7pt,opacity=1]{};
            \node at (U2)[circle,fill,inner sep=1.7pt,opacity=1]{};
            \node at (F0)[circle,fill,inner sep=1.7pt,opacity=1,color=red]{};
            \node at (F1)[circle,fill,inner sep=1.7pt,opacity=1,color=red]{};
            \node at (F2)[circle,fill,inner sep=1.7pt,opacity=1]{};
        \end{scope}
      \end{scope}
    \end{tikzpicture}
    }
\vspace{-3em}
    \vspace{10pt}
    \caption{$\mathcal{U}+\mathcal{F}$ Newton polytope}
    \label{fig:tri1_1}
  \end{subfigure}
  \hfill
  \begin{subfigure}[b]{0.4\textwidth}
  \centering
          \resizebox{\textwidth}{!}{
    \begin{tikzpicture}[rotate around x=-90, rotate around z=-18, scale=4]
    \begin{scope}
    \draw[->] (0,0,0) -- (1.5,0,0) node[below left] {$x_1$};
    \draw[->] (0,0,0) -- (0,1.5,0) node[right] {$x_3$};
    \draw[->] (0,0,0) -- (0,0,1.3) node[above] {$x_2$};
    \begin{scope}[fill opacity=.5, draw opacity=.6, text opacity=1]
        %
        \draw [thick,fill=blue,fill] 
        (1,0,0) coordinate [label={[anchor=north west,xshift=-5pt]below:}] (U0) 
        -- (0,0,1) coordinate [label=above right:] (U1) 
        -- (0,1,0) coordinate [label={right:}] (U2) 
        -- cycle;
        %
        \draw [thick,fill=green,fill] 
        (1,1,0) coordinate [label={[anchor=north west,xshift=-5pt]below:}] (F0) 
        -- (0,1,1) coordinate [label=above right:] (F1) 
        -- (1,0,1) coordinate [label={right:}] (F2) 
        -- cycle;
        \draw [thick, dotted, color=red, opacity=0.3] (U1) -- (F1);
        \draw [thick, dotted, , opacity=0.3] (F2) -- (U1);
        \draw [thick, dotted, color=red, opacity=0.3] (U0) -- (F0);
        \draw [thick, dotted, opacity=0.3] (U0) -- (F2);
        \draw [thick, dotted, color=red, opacity=0.3] (U2) -- (F0);
        \draw [thick, dotted, color=red, opacity=0.3] (F1) -- (U2);
        \draw [thick, dotted, opacity=0.3] (U2) -- (F2);
        \node at (U0)[circle,fill,inner sep=1pt,opacity=1]{};
        \node at (U1)[circle,fill,inner sep=1pt,opacity=1]{};
        \node at (U2)[circle,fill,inner sep=1pt,opacity=1]{};
        \node at (F0)[circle,fill,inner sep=1pt,opacity=1,color=red]{};
        \node at (F1)[circle,fill,inner sep=1pt,opacity=1,color=red]{};
        \node at (F2)[circle,fill,inner sep=1pt,opacity=1]{};
        %
        \draw [fill=gray,fill opacity=0.2] 
        (0,1/2,1) coordinate [label={[anchor=south,yshift=4pt]$U_2 F_2$}] (U1F1) 
        -- (1/2,0,1) coordinate [label={below left:$U_2 F_3$}] (U1F2)
        -- (1,0,1/2) coordinate [label={[anchor=north,yshift=-7pt]:$U_1 F_3$}] (U0F2) 
        -- (1,1/2,0) coordinate [label={below right:$U_1 F_1$}] (U0F0) 
        -- (1/2,1,0) coordinate [label={above right:$U_3 F_1$}] (U2F0) 
        -- (0,1,1/2) coordinate [label={[anchor=south,yshift=6pt]:$U_3 F_2$}] (U2F1) 
        -- cycle;
        \node at (U1F1)[circle,fill,inner sep=1.5pt,opacity=1,color=red]{};
        \node at (U1F2)[circle,fill,inner sep=1.5pt,opacity=1,color=gray]{};
        \node at (U0F2)[circle,fill,inner sep=1.5pt,opacity=1,color=gray]{};
        \node at (U0F0)[circle,fill,inner sep=1.5pt,opacity=1,color=red]{};
        \node at (U2F0)[circle,fill,inner sep=1.5pt,opacity=1,color=red]{};
        \node at (U2F1)[circle,fill,inner sep=1.5pt,opacity=1,color=red]{};
        \node at (1/2,1/2,1/2)[circle,fill,inner sep=1.5pt,opacity=1,color=gray,label={[anchor=east]:$U_3 F_3$}]{};
    \end{scope}
    \end{scope}
    \end{tikzpicture}
    }
\vspace{-3em}
    \vspace{10pt}
    \caption{$\mathcal{U}$, $\mathcal{F}$ and $\mathcal{UF}$ faces}
    \label{fig:tri1_2}
  \end{subfigure}
  \begin{subfigure}[b]{0.38\textwidth}
      \vspace{1em}
\resizebox{!}{\textwidth}{
\begin{tikzpicture}[rotate around x=-90, rotate around z=-90, scale=4]
  \begin{scope}
    \draw[->] (0,0,0) -- (1.3,0,0) node[below left] {$x_1$};
    \draw[->] (0,0,0) -- (0,1.3,0) node[right] {$x_3$};
    \draw[->] (0,0,0) -- (0,0,1.3) node[above] {$x_2$};
    \begin{scope}[fill opacity=.5, draw opacity=.6, text opacity=1]
        %
        \draw [fill=gray,fill opacity=0.3] 
        (0,1/2,1) coordinate [label={above:$U_2 F_2$}] (U1F1) 
        -- (1/2,0,1) coordinate [label={above left:$U_2 F_3$}] (U1F2)
        -- (1,0,1/2) coordinate [label={left:$U_1 F_3$}] (U0F2) 
        -- (1,1/2,0) coordinate [label={below:$U_1 F_1$}] (U0F0) 
        -- (1/2,1,0) coordinate [label={right:$U_3 F_1$}] (U2F0) 
        -- (0,1,1/2) coordinate [label=above right:$U_3 F_2$] (U2F1) 
        -- cycle;

        \node at (U1F1)[circle,fill,inner sep=1.7pt,opacity=1,color=red]{};
        \node at (U0F2)[circle,fill,inner sep=1.7pt,opacity=1,color=gray]{};
        \node at (U1F2)[circle,fill,inner sep=1.7pt,opacity=1,color=gray]{};
        \node at (U2F1)[circle,fill,inner sep=1.7pt,opacity=1,color=red]{};
        \node at (U2F0)[circle,fill,inner sep=1.7pt,opacity=1,color=red]{};
        \node at (U0F0)[circle,fill,inner sep=1.7pt,opacity=1,color=red]{};
        \node at (1/2,1/2,1/2)[circle,fill,inner sep=1pt,opacity=1,color=gray,label=above:$U_3 F_3$]{};

        \draw[->] (1/2,1/2,1/2) -- (7/16,5/8,7/16)  node[right, text width=5em, color=black] {$x_2^\prime$};
        \draw[->] (1/2,1/2,1/2) -- (5/8,1/2,3/8)  node[right, text width=5em, color=black] {$x_1^\prime$};

    \end{scope}
  \end{scope}
\end{tikzpicture} 
}
\vspace{-3em}
\vspace{10pt}
\caption{$\mathcal{UF}$ face}
\label{fig:tri1_3}
\end{subfigure}
  \hfill
  \begin{subfigure}[b]{0.45\textwidth}
      \vspace{1em}
\resizebox{\textwidth}{!}{
  \begin{tikzpicture}[rotate around x=-90, rotate around z=-0, scale=3]
  \begin{scope}
    \draw[->] (0,0,0) -- (0.7,0,0) node[below left] {$x_1^\prime$};
    \draw[->] (0,0,0) -- (0,0.7,0) node[right] {$x_2^\prime$};
    \draw[->] (0,0,0) -- (0,0,0.6) node[above] {$\lambda$};
    \begin{scope}[fill opacity=.8, draw opacity=.6, text opacity=1]
        %
        \draw [fill opacity=0] 
        (-1,0,1) coordinate [label=left:$U_2 F_2$] (U2F2) 
        -- (-1/2,{sqrt(3)/2},1) coordinate [label=above left:$U_3 F_2$] (U3F2)
        -- (1/2,{sqrt(3)/2},1) coordinate [label=above right:$U_3 F_1$] (U3F1) 
        -- (1,0,1) coordinate [label=above right:$U_1 F_1$] (U1F1)
        -- cycle;
        \draw [fill opacity=0] 
        (U2F2)
        -- (U1F1)
        -- (1/2,{-sqrt(3)/2},0) coordinate [label=below right:$U_1 F_3$] (U1F3) 
        -- (-1/2,{-sqrt(3)/2},0) coordinate [label={below:$U_2 F_3$}] (U2F3) 
        -- cycle;
        \draw [fill=red,fill opacity=0.4] 
        (U3F2) coordinate
        -- (U3F1) coordinate
        -- (0,0,0) coordinate (U3F3)
        -- cycle;
        \draw [fill=green,fill opacity=0.4] 
        (U2F2) coordinate
        -- (U3F2) coordinate
        -- (U3F3) coordinate [label={[anchor=north west,xshift=3pt]:$U_3 F_3$}] 
        -- (U2F3) coordinate
        -- cycle;
        \draw [fill=lime,fill opacity=0.4]
        (U3F1) coordinate
        -- (U1F1) coordinate 
        -- (U1F3) coordinate
        -- (U3F3) coordinate
        -- cycle;
        \draw [fill=blue,fill opacity=0.3] 
        (U3F3)
        -- (U1F3) 
        -- (U2F3) 
        -- cycle;
        %
        \draw [opacity=0.7,dashed] 
        (-1,0,0) coordinate [label=left:$U_2 F_2^\prime$] (U2F2p) 
        -- ({-1/2},{sqrt(3)/2},0) coordinate [label={[anchor= east,xshift=-7pt,yshift=4pt]below:$U_3 F_2^\prime$}] (U3F2p)
        -- (1/2,{sqrt(3)/2},0) coordinate [label=above right:$U_3 F_1^\prime$] (U3F1p) 
        -- (1,0,0) coordinate [label=above right:$U_1 F_1^\prime$] (U1F1p) 
        -- (1/2,{-sqrt(3)/2},0) coordinate (UF3p) 
        -- (-1/2,{-sqrt(3)/2},0) coordinate (UF2p) 
        -- cycle;
        \draw [color=gray, dashed] (U2F2) -- (U2F2p);
        \draw [color=gray, dashed] (U3F2) -- (U3F2p);
        \draw [color=gray, dashed] (U3F1) -- (U3F1p);
        \draw [color=gray, dashed] (U1F1) -- (U1F1p);
        \node at (U2F2p)[circle,fill,inner sep=1.7pt,opacity=1,color=pink]{};
        \node at (U3F2p)[circle,fill,inner sep=1.7pt,opacity=1,color=pink]{};
        \node at (U3F1p)[circle,fill,inner sep=1.7pt,opacity=1,color=pink]{};
        \node at (U1F1p)[circle,fill,inner sep=1.7pt,opacity=1,color=pink]{};
        \node at (U2F2)[circle,fill,inner sep=1.7pt,opacity=1,color=red]{};
        \node at (U3F2)[circle,fill,inner sep=1.7pt,opacity=1,color=red]{};
        \node at (U3F1)[circle,fill,inner sep=1.7pt,opacity=1,color=red]{};
        \node at (U1F1)[circle,fill,inner sep=1.7pt,opacity=1,color=red]{};
        \node at (U1F3)[circle,fill,inner sep=1.7pt,opacity=1,color=gray]{};
        \node at (U2F3)[circle,fill,inner sep=1.7pt,opacity=1,color=gray]{};
        \node at (U3F3)[circle,fill,inner sep=1.7pt,opacity=1,color=gray]{};
    \end{scope}
    \draw[-{Stealth},color=blue] (0,-1/2,0) -- (0,-1/2,1/4) node[above left,yshift=-10pt,color=blue] {$\boldsymbol{v}_H$};
    \draw[-{Stealth},color=blue] (-1/2,0,1/2) -- ({-1/2+2/7*1/sqrt(3)},{-2/21},{1/2+2/7*1/sqrt(3)}) node[above left,yshift=-3pt,color=blue] {$\boldsymbol{v}_{C_1}$};
    \draw[-{Stealth},color=blue] (1/2,0,1/2) -- ({1/2-32/271*sqrt(3)},{-8/271},{1/2+20/271*sqrt(3)}) node[above right,xshift=13pt,yshift=-10pt,color=blue] {$\boldsymbol{v}_{C_2}$};
    \draw[-{Stealth},color=blue] (0,{sqrt(3)/4},1/2) -- ({0},{-1/5+sqrt(3)/4},{1/2+sqrt(3)/10}) node[above,color=blue] {$\boldsymbol{v}_{S}$};
  \end{scope}
\end{tikzpicture}
}
\vspace{-3em}
  \vspace{10pt}
  \caption{Facets of $\mathcal{UF}$ Newton polytope}
  \label{fig:tri1_4}
  \end{subfigure}
\caption{The geometry of the one-loop triangle graph $G_3$. The black (red) dots indicate that the vertices correspond to a term with an overall factor of $\lambda^0$ ($\lambda^1$). In figures~\ref{fig:tri1_1} and \ref{fig:tri1_2}, the blue face is defined by the exponents of the $\mathcal{U}$ polynomial and the green face is defined by the exponents of the $\mathcal{F}$ polynomial. Figures~\ref{fig:tri1_2} and \ref{fig:tri1_3} show how the \emph{Cayley trick} can be used to obtain the $\mathcal{U}\mathcal{F}$ polytope by slicing the $\mathcal{U} + \mathcal{F}$ polytope. In figure~\ref{fig:tri1_4} we display the facets of the $\mathcal{U}\mathcal{F}$ polytope corresponding to the different regions, as we will argue later, the blue facet corresponds to a hard region $H$, the red facet corresponds to a soft region $S$, and the green and yellow facets correspond to collinear regions $C_1$ and $C_2$, respectively.
}
\label{fig:tri1}
\end{figure}

Before concluding, let us emphasise once more that the geometric approach we just described is not guaranteed in general to identify all regions. The regions it identifies are only those corresponding to facets, i.e. those associated with a scaling vector,  $\x\to \lambda^{\boldsymbol{u}_R} \x$,  such that for $\lambda\to 0$ the singularity occur at the boundary of the domain of integration. These singularities are therefore all of the \emph{endpoint} type.
Other potential singularities, associated with cancellation between terms in ${\cal P}$, are non-endpoint, and would not be identified as potential regions in this approach. 

We stress that there is a wide class of region expansions for which scaling vectors are sufficient to classify all the required regions. First of all, this class encompasses all the kinematic limits that can be approached in  Euclidean kinematics. In such cases all the coefficients in the Lee-Pomeransky polynomial are positive definite, so the only way for an overall scaling with $\lambda$ to appear is by having homogeneous scaling of every leading monomial in the sum. When an expansion cannot be associated with Euclidean kinematics, the coefficients in the Lee-Pomeransky polynomial may have different signs. In this case a region could be associated to cancellations between different terms, e.g. the Glauber regions in the Regge limit \cite{JtzSmnSmn12,SmnvSmnSmv19,AnthnrySkrRmn19}. However, there are also expansions, for which all the regions are given by facets even though the coefficients cannot always be chosen to have equal signs. It is generally believed that the on-shell expansion for wide-angle scattering, on which we focus here, belongs to this category.
As an example, in appendix~\ref{appendix-possible_regions_nonplanar_doublebox_graph} we examine the nonplanar double-box graph, where the coefficients are necessarily of different signs in the on-shell limit. We show that regions due to cancellations do not emerge in this  case.

\subsection{Region vectors from the Landau equations}
\label{region_vectors_from_LE}

The Newton polytope approach provides us with a systematic way to determine the regions for a given Feynman graph $G$. In this section, focusing on wide-angle scattering, we endow the regions that appear in the expansion to any order in~$\lambda$ with a physical interpretation and relate them to the solutions of the Landau equations. The proposition is that for the on-shell expansion in wide-angle scattering, the general form for the region vectors $\v_R^{}$ is
\begin{align}
\label{region_vector_wideangle_scattering}
\begin{split}
    &\v_R= (u_{R,1},u_{R,2},\dots,u_{R,N};1),\qquad u_{R,e}\in \{0,-1,-2\},\\
    &\qquad u_{R,e}=0 \quad\Leftrightarrow\quad e\in H \\
    &\qquad u_{R,e}=-1 \quad\Leftrightarrow\quad e\in J\equiv \cup_{i=1}^K J_i\\
    &\qquad u_{R,e}=-2 \quad\Leftrightarrow\quad e\in S 
\end{split}
\end{align}
meaning that each edge $e$ falls into one of three possible categories, hard, jet and soft, and moreover, the region vectors must conform with the separation of the graph to subgraphs as in figure~\ref{classical_picture}, which manifests particular partitions of the set of propagators into $H$ (hard), $J_1,\dots,J_K$ (jets), and $S$ (soft). We will refer to all the Lee-Pomeransky parameters $x_e$ where $e \in H$ as $x^{[H]}$. Similarly, we use $x^{[J]}$ and $x^{[S]}$ to denote, respectively, parameters with $e \in J$ and $e \in S$.
Thus for $\lambda\to 0$, eq.~(\ref{region_vector_wideangle_scattering}) implies the following scaling of the Lee-Pomeransky parameters:
\begin{eqnarray}
x^{[H]}\sim \lambda^0,\ \quad x^{[J]}\sim \lambda^{-1},\quad \ x^{[S]}\sim \lambda^{-2}.
\label{LP_parameter_scales}
\end{eqnarray}
We note that the number of propagators of each type are denoted as $N(H)$, $N(J)$ and~$N(S)$, with $N(H)+N(J)+N(S) =N$.

Let us now recall the formulation of the Landau equations~\cite{Lnd59}. Considering the Feynman parameterised integral  (\ref{feynman_parameterisation_form}) with all the internal propagators massless, the combined denominator function $\mathcal{D}$ reads
\begin{eqnarray}
\mathcal{D}( k,p,q;\alpha)= \sum_{e\in G} \alpha_e \left(-l_e^2(k,p,q) -i\varepsilon\right).
\label{feynman_parameterisation_denominator_massless}
\end{eqnarray}
The Landau equations are then:
\begin{subequations}
\label{landau_equations}
\begin{align}
\label{landau_equations_I}
&\alpha_el_e^2(k,p,q) =0\qquad \forall e\in G \\
&\displaystyle{ \frac{\partial}{\partial k_a}} \mathcal{D} \left( k,p,q;\alpha \right) =0 \qquad \forall a\in\{1,\dots,L\}\,.
\label{landau_equations_II}
\end{align}
\end{subequations}
Here the first Landau equation, (\ref{landau_equations_I}), states that each propagator of $G$ is either on-shell, or associated with a vanishing Feynman parameter. This implies in particular that $\mathcal{D}=0$ and the integrand diverges. The second Landau condition, eq.~(\ref{landau_equations_II}), requires in addition that the integration contour would be pinched, that is, it is trapped between singularities. Each solution of this set of equations identifies a specific \emph{pinch surface} $\sigma$ of the Feynman integral $\I$ in eq.~(\ref{feynman_parameterisation_form}), where infrared singularities may arise. Below we explain the grounds for the proposition in eq.~(\ref{region_vector_wideangle_scattering}) (or eq.~(\ref{LP_parameter_scales})).

The hard region, where all line momenta are off-shell, is defined by accounting for the scaling of the external Mandelstam invariants with respect to the expansion parameter $\lambda$ in the integrand, without modifying the Feynman parameters, i.e. $\alpha^{[H]}\sim \lambda^0$. The region vector~in eq.~(\ref{region_vector_wideangle_scattering}) is then:
\begin{eqnarray}
\boldsymbol{v}_H^{}= (\underset{N}{\underbrace{0,0, \dots, 0}};1).
\label{hard_vector}
\end{eqnarray}
This region does not correspond to a solution of the Landau equations. Interpreting this region according to figure~\ref{classical_picture} amounts to $N(J)=N(S)=0$.

Let us now discuss the region vectors of the form (\ref{region_vector_wideangle_scattering}) for which $N(J)+N(S)\geqslant 1$. We will refer to them as \emph{infrared region vectors}, and show below that the corresponding regions of $\I$ in eq.~(\ref{feynman_parameterisation_form}) can be viewed as describing the specific scaling of the Feynman parameters $\boldsymbol{\alpha}$ in the vicinity of a solution of the Landau equations.

In detail, we take the momenta $p_i^\mu$ slightly off-shell, i.e. $p_i^2\sim \lambda Q^2$ as in eq.~(\ref{wideangle_scattering_kinematics}), and consider the neighbourhood of the pinch surface $\sigma$. The three relevant scalings of loop momenta in this neighbourhood are known to be~\cite{Stm78I,LbyStm78,Cls20,Ma20, EdgStm15}
\begin{subequations}
\label{infrared_region_momentum_scaling}
\begin{align}
\label{infrared_region_momentum_scaling_H} &\text{Hard: }k_H^\mu= (k_H^t,k_H^x,k_H^y,k_H^z) \sim Q (1,1,1,1),\\
\label{infrared_region_momentum_scaling_J} &\text{Jet: }k_{J_i}^\mu= \left(k_{J_i}\cdot \overline{\beta}_i,\ k_{J_i}\cdot \beta_i,\ k_{J_i}\cdot \beta_{i\perp}\right) \sim Q\left(1,\lambda,\lambda^{1/2}\right),\\
\label{infrared_region_momentum_scaling_S}        &\text{Soft: }k_S^\mu =(k_S^t,k_S^x,k_S^y,k_S^z) \sim Q (\lambda,\lambda,\lambda,\lambda).
\end{align}
\end{subequations}
In the scaling of the jet loop momenta, we have used light-cone coordinates, where $\beta_i^\mu$ is a null vector in the direction of the jet $p_i^\mu$, defined as $\beta_i^\mu=\frac{1}{\sqrt{2}}(1,\widehat{v}_i)$ where $\widehat{v}_i$ is the unit three-velocity of the jet. For each $\beta_i^\mu$, we also define $\overline{\beta}_i^\mu \equiv \frac{1}{\sqrt{2}} (1, -\widehat{v}_i)$, so that $\beta_i \cdot \overline{\beta}_i =1$.\footnote{For example, suppose the $i$-th jet is in the direction of the $z$-axis, then $\beta_i^\mu= (\beta_{i}^t,\beta_{i}^x,\beta_{i}^y,\beta_{i}^z)= (1,0,0,1)/\sqrt{2}$ and $\overline{\beta}_i^\mu=(1,0,0,-1)/\sqrt{2}.$ For any vector $k^\mu$, $k\cdot \overline{\beta}_i= (k^t+k^z)/\sqrt{2}$ and $k\cdot \beta_i= (k^t-k^z)/\sqrt{2}$.} For a given graph~$G$ and a region~$R$, the scaling of the loop momenta according to eq.~(\ref{infrared_region_momentum_scaling}) translates uniquely into the scaling of each of the line momenta of $G$, that is, each edge $l_e^\mu$ is either soft, jet-like or hard, as in eq.~(\ref{infrared_region_momentum_scaling}). The precise relation has been illustrated in section 3 of ref.~\cite{Ma20}.

From eq.~(\ref{infrared_region_momentum_scaling}) we can therefore find the scaling of the virtuality of each propagator:
\begin{equation}
\label{virtuality_scaling}
(l_e^{[H]})^2 \sim\lambda^0;\qquad 
(l_e^{[J]})^2 \sim\lambda^1;\qquad
(l_e^{[S]})^2 \sim\lambda^2\,.
\end{equation}
From the Schwinger representation of the propagator
\begin{equation}
\label{Schwinger_param}
\frac{1}{\left(l_e^2(k,p,q)\right)^{\nu_e}} = \int_0^\infty \frac{d \tilde{x}_e}{\tilde{x}_e}\, \frac{\tilde{x}_e^{\nu_e}}{\Gamma(\nu_e)} e^{
-\tilde{x}_e l_e^2(k,p,q)},
\end{equation}
one sees that the scaling limit of the Schwinger parameter~$\tilde{x}_e$ is related to the scaling of the corresponding propagator virtuality as (see ref.~\cite{Engel:2022kde} for example):
\begin{equation}
\label{x_scaling_rule}
x_e \sim \tilde{x}_e \sim \frac{1}{l_e^2(k,p,q)} \sim \lambda^{u_{R,e}}.
\end{equation}
Here we have further stated that the very same scaling applies also to the Lee-Pomeransky parameter~$x_e$. This is shown in appendix~\ref{appendix-schwinger-lp}, where we relate the Schwinger and Lee-Pomeransky representations and demonstrate that  
the region vectors $\boldsymbol{u}_R$ in the two are the same. 
From eqs.~(\ref{virtuality_scaling}) and (\ref{x_scaling_rule}) we can immediately read off the scaling of the Lee-Pomeransky parameters for hard ($H$), jet ($J$) and soft ($S$) propagators as summarised by eq.~(\ref{LP_parameter_scales}).
With this we have
established the scaling rule of eq.~(\ref{LP_parameter_scales})  based on the scaling of the momentum modes~(\ref{infrared_region_momentum_scaling}).

Using the relation in eq.~(\ref{Feynman_LP_relation}) we can also deduce the scaling of the Feynman parameter~$\alpha_e$.
To this end we must distinguish between two cases: infrared regions which feature soft propagators ($N(S)\geqslant 1$) versus those that have only jet and hard ones ($N(S)=0$). Using eqs.~(\ref{LP_parameter_scales}) and (\ref{Feynman_LP_relation}) we find
\begin{subequations}
\label{FP_scaling_any}
\begin{eqnarray}
N(J)\geqslant 1,\ N(S)=0:\quad&& \alpha_e^{[H]}\sim \lambda^{1},\quad \alpha_e^{[J]}\sim \lambda^{0};
\label{FP_scaling_hard_jet}\\
N(J),\ N(S)\geqslant 1:\quad&& \alpha_e^{[H]}\sim \lambda^{2},\quad  \alpha_e^{[J]}\sim \lambda^{1},\quad \alpha_e^{[S]}\sim \lambda^0.
\label{FP_scaling_hard_jet_soft}
\end{eqnarray}
\end{subequations}
Note that in either of these cases the propagators with the smallest virtuality (\emph{cf.}~eq.~(\ref{virtuality_scaling})) are of order ${\cal O}(\lambda^0)$, which is necessary for the condition $\sum_{e\in G} \alpha_e=1$ to be satisfied.  

In the remainder of this subsection we investigate the relation between pinch surfaces and infrared regions of the MoR. We begin in section~\ref{sec:neighbourhood} by introducing the concept of a neighbourhood of the pinch surface associated with the on-shell expansion. 
This allows us to study what constraints are imposed on the structure of infrared regions, making use of the Coleman-Norton analysis. Then in section~\ref{sec:proposition_1} we formulate our general proposition for the structure of the region vectors and discuss its applicability.

\subsubsection{The neighbourhood of pinch surfaces and infrared regions
\label{sec:neighbourhood}
}

The scaling of the momenta and the Feynman parameters discussed above naturally leads us to a definition of a \emph{neighbourhood of a pinch surface} associated with the expansion in~$\lambda$ by the following modified Landau equations:
\begin{subequations}
\label{landau_equations_modification}
\begin{align}
\label{landau_equations_modification_I}
&\alpha_e l_e^2(k,p,q) \sim\lambda^p \qquad \forall e\in G \\
&\displaystyle{ \frac{\partial}{\partial k_a}} \mathcal{D} \left( k,p,q;\alpha \right)\lesssim \lambda^{1/2} \qquad \forall a\in\{1,\dots,L\}\,,
\label{landau_equations_modification_II}
\end{align}
\end{subequations}
where $p$ in eq.~(\ref{landau_equations_modification_I}) is fixed according to the scaling of the line momentum with the smallest virtuality $l_{e'}^2\sim \lambda^p$: according to eq.~(\ref{virtuality_scaling}), for infrared regions $R$ which feature soft propagators,  we have $p=2$,  while for regions that feature only jet and hard propagators, $p=1$. 
In eq.~(\ref{landau_equations_modification_II}) the notation $\lesssim \lambda^{1/2}$ should be interpreted as allowing for scaling $\sim \lambda^q$ with $q \geqslant 1/2$. This is the weakest condition to be placed on the derivative, so as to reproduce eq.~(\ref{landau_equations_II}) in the limit $\lambda\to 0$.

While the Landau equations (\ref{landau_equations}) describe the pinch surfaces themselves, i.e. the manifolds where the integral is singular, the modified ones in eq.~(\ref{landau_equations_modification}) aim to capture the vicinity of pinch surfaces, describing the scaling of momenta and Feynman parameters when approaching the singularity as $\lambda$ tends to zero. 
The precise modification of the first Landau equation~(\ref{landau_equations_modification_I}) immediately follows from the discussion above (eqs.~(\ref{virtuality_scaling}) and (\ref{FP_scaling_any})). 

Let us now discuss the interpretation of eq.~(\ref{landau_equations_modification}). A first, key observation is that for a given pinch surface that corresponds to an infrared region of $\I(G)$, \emph{all} the terms $\alpha_e l_e^2$ in the combined denominator function ${\cal D}$ of eq.~(\ref{feynman_parameterisation_denominator_massless}) are characterised by a \emph{uniform} scaling.
From the geometric formulation of the MoR, excluding the possibility that some~$\alpha_{e'}$ is strictly zero\footnote{We note that $\alpha_{e'}=0$ may correspond to a region of a simpler graph where the propagator $e'$ has been contracted.}, amounts to excluding lower-dimensional faces of the Newton polytope of~$G$, which lead to vanishing scaleless integrals.

Eq.~(\ref{landau_equations_modification_II}) requires that the derivatives vanish as $\lambda^{1/2}$ (or faster) as $\lambda \to 0$, and places constraints on what configuration of hard, jet and soft momenta may constitute a  region. As we shall see below, these constraints relate to the Coleman-Norton interpretation of pinch surfaces~\cite{ClmNtn65}.

Consider first the case $N(J)\geqslant 1,\ N(S)=0$, where the denominator function of eq.~(\ref{feynman_parameterisation_denominator}) can be rewritten as: 
\begin{align}
\label{feynman_parameterisation_denominator_hard_jet}
\mathcal{D}= -\sum_{e\in H} \alpha_e^{[H]} \big( l_e^{[H]} \big)^2 - \sum_{e\in J} \alpha_e^{[J]} \big( l_e^{[J]} \big)^2 -i\epsilon.
\end{align}
Given a jet $J_i$ with the external momentum $p_i^\mu\sim Q\beta_i^\mu$, where the null vector $\beta_i^\mu$ is defined below eq.~(\ref{infrared_region_momentum_scaling}), then for any jet loop momentum $k_{a\mu}$ admitting (\ref{infrared_region_momentum_scaling_J}), the second Landau equation (\ref{landau_equations_modification_II}), $\partial D/\partial k_{a\mu}\lesssim \lambda^{1/2}$, implies
\begin{eqnarray}
\sum_{e\in H} 2\eta_{ae}^{} \alpha_e^{[H]} l_e^{[H]\mu} + \sum_{e\in J_i} 2\eta_{ae}^{} \alpha_e^{[J_i]} l_e^{[J_i]\mu} \lesssim \lambda^{1/2} \quad \Rightarrow \quad \sum_{e\in J_i} 2\eta_{ae}^{} \alpha_e^{[J_i]} l_e^{[J_i]\mu}\lesssim \lambda^{1/2}.
\label{second_LE_hard_jet}
\end{eqnarray}
The coefficients $\eta_{ae}^{}$ are either $0$ or $\pm 1$, depending on how the loop momentum $k_{a}$ enters the line momenta $l_e^{[H]}$ or $l_e^{[J_i]}$.
In deriving the second inequality, we have used the fact that $\alpha^{[H]}\sim \lambda$. We stress that the condition in eq.~(\ref{second_LE_hard_jet}) must be satisfied for any component~$\mu$, and below we analyze each component separately.

For the small lightcone component of any jet line momentum, $l_e^{[J_i]}\cdot \beta_i\ (\sim \lambda)$, and the transverse one \hbox{$l_e^{[J_i]}\cdot \beta_{i\perp}\ (\sim \lambda^{1/2})$}, eq.~(\ref{second_LE_hard_jet}) is automatically satisfied due to the scaling of the momentum components in eq.~(\ref{infrared_region_momentum_scaling_J}).
For the large lightcone component, \hbox{$l_e^{[J_i]}\cdot \overline{\beta}_i\ (\sim \lambda^0)$}, eq.~(\ref{second_LE_hard_jet}) provides a genuine constraint, as it requires cancellations between different terms in the sum, up to terms of $\mathcal{O}(\lambda^{1/2})$.
As we have pointed out above, this constraint corresponds to the Coleman-Norton analysis: each jet vertex $a$ can be associated to a ``position vector'' $x_a^\mu$, where any jet line $e$ in eq.~(\ref{second_LE_hard_jet}) can be seen to represent a massless particle propagating freely between the vertices $a$ and $b$, such that $x_b^\mu-x_a^\mu = \alpha_e l_e^\mu +\mathcal{O}(\lambda^{1/2})$. This sets a definite order of the jet vertices along the jet momentum flow in $J_i$, from the hard vertices to the final on-shell parton $i$. This allows us to view the propagation of jets as a classical process, constraining the configurations infrared regions may have.

One consequence of the Coleman-Norton analysis is that, after two jets depart from a hard subgraph, at vertices $x_a^\mu$ and $x_b^\mu$ respectively, $x_b^\mu-x_a^\mu \sim \lambda^0$, cannot be connected by propagators $e$ such that $\alpha_e l_e^{\mu} \ll \lambda^0$. Thus, a hard propagator connecting them is excluded. For example, in the $1\to 3$ decay process with external momenta $q_1^\mu$, $q_2^\mu$, $p_1^\mu$ and $p_2^\mu$, both the configurations in figures~\ref{Coleman_Norton_restriction_example}(i) and \ref{Coleman_Norton_restriction_example}(ii) are allowed by eq.~(\ref{FP_scaling_hard_jet}). However, figure~\ref{Coleman_Norton_restriction_example}(ii) is forbidden by the Coleman-Norton analysis, which implies that after the two jets $J_1$ and~$J_2$ depart from the hard subgraph $H_1$, they should propagate in different directions and should not interact at any other hard vertices again. Disconnected hard subgraphs are thus excluded.
\begin{figure}[t]
\centering
\begin{subfigure}[b]{0.32\textwidth}
\centering
\resizebox{\textwidth}{!}{
\begin{tikzpicture}[line width = 0.6, font=\large, mydot/.style={circle, fill, inner sep=.7pt},rotate=90,transform shape]

\node[draw=Black,circle,minimum size=4cm] (h) at (6,6){};

\node at (h) [rotate=-90] {$G$};


\node (q1) at (5.8,9.5) {};
\node (q1p) at (6.2,9.5) {};
\draw (q1) edge [color=Blue] (h) node [] {};
\draw (q1p) edge [color=Blue] (h) node [below,rotate=-90,xshift=-5pt] {$q_1$};

\node (qn) at (5.8,2.5) {};
\node (qnp) at (6.2,2.5) {};
\draw (qn) edge [color=Blue] (h) node [] {};
\draw (qnp) edge [color=Blue] (h) node [below,rotate=-90,xshift=5pt] {$q_2$};

\node (p1) at (2.6,3) {};
\node (pn) at (9.4,3) {};
\draw (p1) edge [color=LimeGreen] (h) node [right,rotate=-90] {$p_2$};
\draw (pn) edge [color=Green] (h) node [right,rotate=-90] {$p_1$};

\end{tikzpicture}
}
\vspace{-2em}
\caption{}
\end{subfigure}
\hfill
\setcounter{subfigure}{0}
\renewcommand\thesubfigure{\roman{subfigure}}
\begin{subfigure}[b]{0.32\textwidth}
\centering
\resizebox{\textwidth}{!}{
\begin{tikzpicture}[line width = 0.6, font=\large, mydot/.style={circle, fill, inner sep=.7pt},rotate=90,transform shape]

\node[draw=Blue,circle,minimum size=1cm,fill=Blue!50] (h) at (6,8){};
\node[draw=LimeGreen,ellipse,minimum height=3cm, minimum width=1.1cm,fill=LimeGreen!50,rotate=-35] (j1) at (4,5){};
\node[draw=Green,ellipse,minimum height=3cm, minimum width=1.1cm,fill=Green!50,rotate=35] (jn) at (8,5){};

\node at (h) [rotate=-90] {$H_1$};
\node at (j1) [rotate=-90] {$J_2$};
\node at (jn) [rotate=-90] {$J_1$};

\path (h) edge [double,double distance=2pt,color=LimeGreen] (j1) {};
\path (h) edge [double,double distance=2pt,color=Green] (jn) {};

\node (q1) at (5.8,9.5) {};
\node (q1p) at (6.2,9.5) {};
\draw (q1) edge [color=Blue] (h) node [] {};
\draw (q1p) edge [color=Blue] (h) node [below,rotate=-90,xshift=-5pt] {$q_1$};

\node (qn) at (5.8,6.5) {};
\node (qnp) at (6.2,6.5) {};
\draw (qn) edge [color=Blue] (h) node [] {};
\draw (qnp) edge [color=Blue] (h) node [below,rotate=-90,xshift=5pt] {$q_2$};

\node (p1) at (2.6,3) {};
\node (pn) at (9.4,3) {};
\draw (p1) edge [color=LimeGreen] (j1) node [right,rotate=-90] {$p_2$};
\draw (pn) edge [color=Green] (jn) node [right,rotate=-90] {$p_1$};

\end{tikzpicture}
}
\vspace{-2em}
\caption{}
\label{Coleman_Norton_restriction_example_i}
\end{subfigure}
\hfill
\begin{subfigure}[b]{0.32\textwidth}
\centering
\resizebox{\textwidth}{!}{
\begin{tikzpicture}[line width = 0.6, font=\large, mydot/.style={circle, fill, inner sep=.7pt},rotate=90,transform shape]

\node[draw=Blue,circle,minimum size=1cm,fill=Blue!50] (h) at (6,8){};
\node[draw=LimeGreen,ellipse,minimum height=3cm, minimum width=1.1cm,fill=LimeGreen!50,rotate=-35] (j1) at (4,5){};
\node[draw=Blue,circle,minimum size=1cm,fill=Blue!50] (h2) at (6,4){};
\node[draw=Green,ellipse,minimum height=3cm, minimum width=1.1cm,fill=Green!50,rotate=35] (jn) at (8,5){};

\node at (h) [rotate=-90] {$H_1$};
\node at (j1) [rotate=-90] {$J_2$};
\node at (h2) [rotate=-90] {$H_2$};
\node at (jn) [rotate=-90] {$J_1$};

\path (h) edge [double,double distance=2pt,color=LimeGreen] (j1) {};
\path (h) edge [double,double distance=2pt,color=Green] (jn) {};
\path (h2) edge [double,double distance=2pt,color=LimeGreen,bend left=20] (j1) {};
\path (h2) edge [double,double distance=2pt,color=Green,bend right=20] (jn) {};

\node (q1) at (5.8,9.5) {};
\node (q1p) at (6.2,9.5) {};
\draw (q1) edge [color=Blue] (h) node [] {};
\draw (q1p) edge [color=Blue] (h) node [below,rotate=-90,xshift=-5pt] {$q_1$};

\node (qn) at (5.8,2.5) {};
\node (qnp) at (6.2,2.5) {};
\draw (qn) edge [color=Blue] (h2) node [] {};
\draw (qnp) edge [color=Blue] (h2) node [below,rotate=-90,xshift=5pt] {$q_2$};

\node (p1) at (2.6,3) {};
\node (pn) at (9.4,3) {};
\draw (p1) edge [color=LimeGreen] (j1) node [right,rotate=-90] {$p_2$};
\draw (pn) edge [color=Green] (jn) node [right,rotate=-90] {$p_1$};

\end{tikzpicture}
}
\vspace{-2em}
\caption{}
\label{Coleman_Norton_restriction_example_ii}
\end{subfigure}
\caption{The restriction of the regions from the Coleman-Norton interpretation. For a wide-angle scattering process $q_1\to p_1,p_2,q_2$ on the left-hand side, where $p_1^\mu$ and $p_2^\mu$ are on-shell, and $q_1^\mu$ and $q_2^\mu$ are off-shell, both the configurations in (i) and (ii) on the right-hand side satisfy eqs.~(\ref{infrared_region_momentum_scaling}) and (\ref{landau_equations_modification_I}), but (ii) is ruled out because it does not satisfy~(\ref{landau_equations_modification_II}).}
\label{Coleman_Norton_restriction_example}
\end{figure}
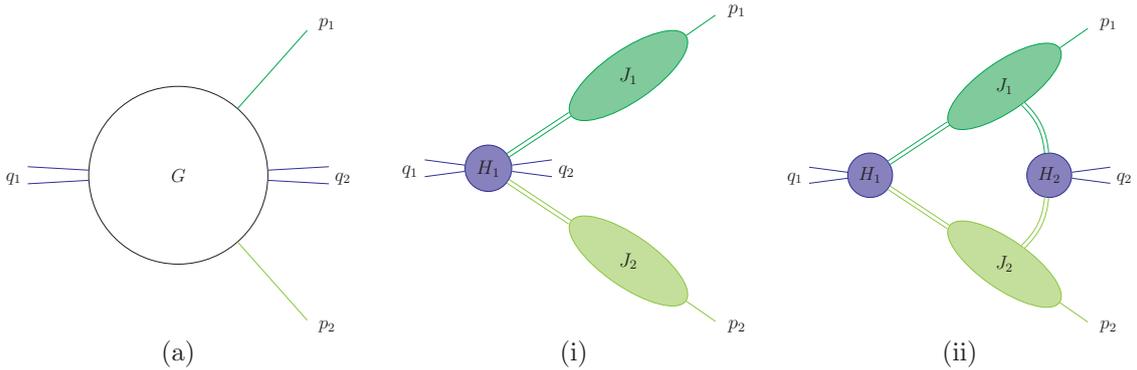

Another consequence of the Coleman-Norton analysis concerns the internal structure of any jet. Examples are provided in figure~\ref{jet_configuration_forbidden} where $(a)$ is consistent with the Coleman-Norton analysis, while $(b)$ and $(c)$ are inconsistent, and therefore cannot appear as part of an infrared region. 
Specifically, suppose two jet momenta $l_e^\mu$ ($e=1,2$) both connect two of the jet vertices $a$ and $b$, as is shown in figure~\ref{jet_configuration_forbidden}$(b)$. Then upon considering the edge $e=1$, according to the Coleman-Norton analysis, one can associate $x_a^\mu$ to $a$ and $x_b^\mu$ to $b$ such that eq.~(\ref{second_LE_hard_jet}) becomes $x_b^\mu-x_a^\mu = \alpha_1 l_1^\mu +\mathcal{O}(\lambda^{1/2})$. For the large lightcone component (suppose it is the $+$ direction), one then obtains $x_b^+>x_a^+$. However, the same reasoning works for the edge $e=2$, which yields $x_a^+>x_b^+$. This contradiction forbids the configuration of figure~\ref{jet_configuration_forbidden}$(b)$. A similar analysis can be carried out for figure~\ref{jet_configuration_forbidden}$(c)$, where the soft lines do not change the argument, since the soft momenta are negligible compared to the large lightcone component of the jet.
\begin{figure}[t]
\centering
\begin{subfigure}[b]{0.32\textwidth}
\centering
\resizebox{\textwidth}{!}{
\begin{tikzpicture}[line width = 0.6, font=\large, mydot/.style={circle, fill, inner sep=.7pt},transform shape]

\node (pi) at (0.9,8.2) {\Huge $p_i$};
\node (Ji) at (2.6,5) {\huge $J_i$};
\node (H) at (6,1) {\huge $H$};
\node (S) at (6,7) {\huge $S$};

\draw (1,9) edge [ultra thick, color=LimeGreen] (2,7.5) node [right] {};
\draw (3,6) edge [ultra thick, color=LimeGreen] (4,4.5) node [right] {};
\node[ultra thick, draw=LimeGreen,ellipse,minimum height=3cm, minimum width=1.1cm,fill=LimeGreen!66,rotate=33.7] (j1) at (2.5,6.83){};
\node[ultra thick, draw=LimeGreen,ellipse,minimum height=3cm, minimum width=1.1cm,fill=LimeGreen!66,rotate=33.7] (j1) at (4.33,4){};
\draw (4.5,3) edge [ultra thick, color=LimeGreen, bend right =10] (5.33,1.67) node [right] {};
\draw (5.16,3.5) edge [ultra thick, color=LimeGreen, bend left =10] (6,2.1) node [right] {};
\draw (4,0) edge [ultra thick, color=Blue, bend left =30] (8,2.66) node [right] {};
\draw (3,7) edge [dashed, ultra thick, color=Rhodamine, bend left =20] (5,9) node [right] {};
\draw (5,4) edge [dashed, ultra thick, color=Rhodamine, bend right =20] (7.5,5.2) node [right] {};

\path (4.9,1.9)-- node[mydot, pos=.333] {} node[mydot] {} node[mydot, pos=.666] {}(6,2.6);
\path (4,7)-- node[mydot, pos=.333] {} node[mydot] {} node[mydot, pos=.666] {}(5.5,5.2);

\end{tikzpicture}
}
\caption{}
\label{jet_configuration_forbidden_good}
\end{subfigure}
\hfill
\begin{subfigure}[b]{0.32\textwidth}
\centering
\resizebox{\textwidth}{!}{
\begin{tikzpicture}[line width = 0.6, font=\large, mydot/.style={circle, fill, inner sep=.7pt},transform shape]

\node (H) at (6,1) {\huge $H$};
\node (Ji) at (2.5,6.5) {\huge $J_i$};
\node (S) at (7.5,7) {\huge $S$};
\node (a) at (3.7,3) {\huge $a$};
\node (b) at (5.3,3.95) {\huge $b$};
\node (l1) at (3.7,4.4) {\huge $l_1$};
\node (l2) at (5.3,2.5) {\huge $l_2$};

\draw (1.5,5.5) edge [ultra thick, color=LimeGreen, bend right =10] (4.6,1) node [right] {};
\draw (3,8) edge [ultra thick, color=LimeGreen, bend left =10] (7,2.5) node [right] {};
\draw (4.05,3) edge [ultra thick, color=LimeGreen, bend right =10] (5.33,1.67) node [right] {};
\draw (5,3.95) edge [ultra thick, color=LimeGreen, bend left =10] (6,2.1) node [right] {};
\node[ultra thick, draw=LimeGreen,ellipse,minimum height=1.5cm, minimum width=1.3cm,rotate=33.7] (j1) at (4.5,3.5){};
\draw[fill,thick,color=LimeGreen] (4.05,3) circle (3pt);
\draw[fill,thick,color=LimeGreen] (5,3.95) circle (3pt);
\draw (4,0) edge [ultra thick, color=Blue, bend left =30] (8,2.66) node [right] {};
\draw (4,7) edge [dashed, ultra thick, color=Rhodamine, bend left =20] (8,9) node [right] {};
\draw (6.2,4) edge [dashed, ultra thick, color=Rhodamine, bend right =10] (9,5) node [right] {};

\draw (3.75,3.9) edge [] (4.1,4.15) node [right] {};
\draw (4,3.7) edge [] (4.1,4.15) node [right] {};
\draw (5,3.26) edge [] (4.8,2.85) node [right] {};
\draw (5.26,3) edge [] (4.8,2.85) node [right] {};

\path (1.9,4.5)-- node[mydot, pos=.333] {} node[mydot] {} node[mydot, pos=.666] {}(5.1,6.2);
\path (6,8)-- node[mydot, pos=.333] {} node[mydot] {} node[mydot, pos=.666] {}(8,5);

\end{tikzpicture}
}
\caption{}
\label{jet_configuration_forbidden_bad1}
\end{subfigure}
\hfill
\begin{subfigure}[b]{0.32\textwidth}
\centering
\resizebox{\textwidth}{!}{
\begin{tikzpicture}[line width = 0.6, font=\large, mydot/.style={circle, fill, inner sep=.7pt},transform shape]

\node (pi) at (0.9,7.2) {\Huge $p_i$};
\node (Ji) at (4.5,5) {\huge $J_i$};
\node (H) at (6,1) {\huge $H$};
\node (a) at (3.6,7.3) {\huge $a$};
\node (b) at (5.2,5.5) {\huge $b$};
\node (l1) at (5.4,6.9) {\huge $l_1$};
\node (l2) at (4.5,6.3) {\huge $l_2$};
\node (S) at (7.5,6.5) {\huge $S$};

\draw (1,8) edge [ultra thick, color=LimeGreen] (2,6.5) node [right] {};
\node[ultra thick, draw=LimeGreen,ellipse,minimum height=4cm, minimum width=1.1cm,fill=LimeGreen!66,rotate=33.7] (j1) at (3,5.1){};
\draw (3.2,4) edge [ultra thick, color=LimeGreen, bend right =10] (4.9,1.36) node [right] {};
\draw (4,4.5) edge [ultra thick, color=LimeGreen, bend left =10] (5.7,2) node [right] {};
\draw[fill,thick,color=LimeGreen] (3.4,5.5) circle (3pt);
\node[ultra thick, draw=LimeGreen,ellipse,minimum height=2cm, minimum width=1.1cm,rotate=33.7] (j1) at (4.5,6.5){};
\draw (3.4,5.5) edge [ultra thick, color=LimeGreen] (3.78,7) node [right] {};
\draw (3.4,5.5) edge [ultra thick, color=LimeGreen] (4.8,5.6) node [right] {};
\draw[fill,thick,color=LimeGreen] (3.78,7) circle (3pt);
\draw[fill,thick,color=LimeGreen] (4.8,5.6) circle (3pt);
\draw (4,0) edge [ultra thick, color=Blue, bend left =30] (8,2.66) node [right] {};
\draw (4.5,7.3) edge [dashed, ultra thick, color=Rhodamine, bend left =20] (8,9) node [right] {};
\draw (5,3.33) edge [dashed, ultra thick, color=Rhodamine, bend right =10] (9,5) node [right] {};

\draw (4.2,6.2) edge [] (3.9,6.5) node [right] {};
\draw (3.9,6.05) edge [] (3.9,6.5) node [right] {};
\draw (5,7) edge [] (5.1,6.6) node [right] {};
\draw (4.74,6.8) edge [] (5.1,6.6) node [right] {};

\path (3.4,6.3)-- node[mydot, pos=.333] {} node[mydot] {} node[mydot, pos=.666] {}(4.1,5.3);
\path (6,8)-- node[mydot, pos=.333] {} node[mydot] {} node[mydot, pos=.666] {}(8,4.5);

\end{tikzpicture}
}
\caption{}
\label{jet_configuration_forbidden_bad2}
\end{subfigure}
\caption{Some configurations of the jet $J_i$ that are allowed or forbidden by eq.~(\ref{landau_equations_modification_II}), according to the Coleman-Norton analysis (the conclusion is not affected by the presence of the soft lines). $(a)$ an allowed jet configuration. $(b)$ a configuration of $J_i$ that is forbidden, because a connected component of $J_i$ is attached to the hard process only. $(c)$ another configuration of $J_i$ that is forbidden, because the jet contains a ``tadpole'' component upon removing the soft lines.}
\label{jet_configuration_forbidden}
\end{figure}
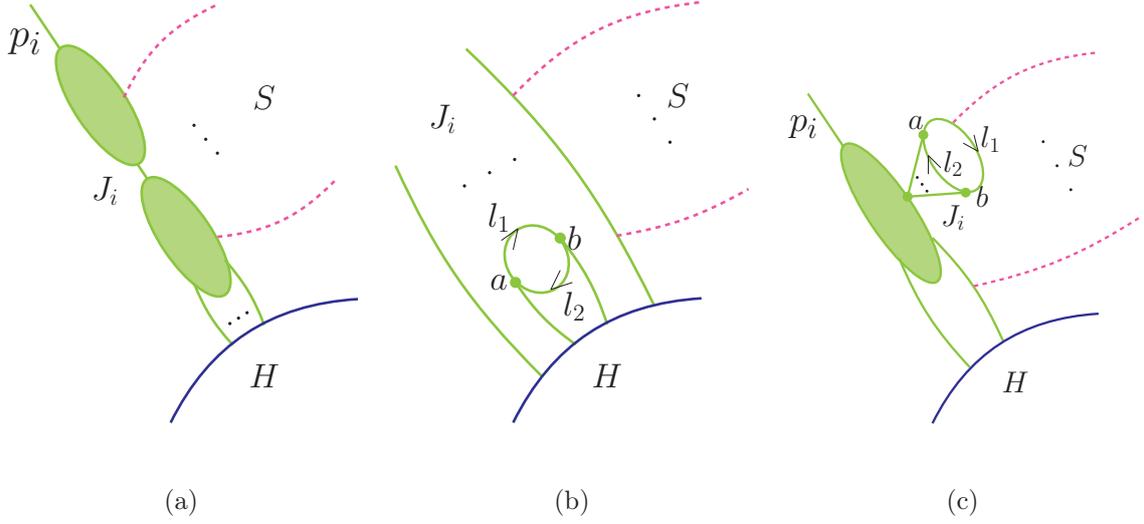

Consider now the case $N(J),\ N(S)\geqslant 1$ where the denominator function of eq.~(\ref{feynman_parameterisation_denominator}) can be rewritten as: 
\begin{eqnarray}
\mathcal{D}= -\sum_{e\in H} \alpha_e^{[H]} \big( l_e^{[H]} \big)^2 -\sum_{e\in J} \alpha_e^{[J]} \big( l_e^{[J]} \big)^2 -\sum_{e\in S} \alpha_e^{[S]} \big( l_e^{[S]} \big)^2 -i\epsilon.
\label{feynman_parameterisation_denominator_hard_jet_soft}
\end{eqnarray}
In this case, the second Landau equation (\ref{landau_equations_modification_II}) automatically holds, because it is of the following possible forms:
\begin{subequations}
\label{second_LE_hard_jet_soft}
\begin{align}
\label{second_LE_hard_jet_soft_H} &\frac{\partial D}{\partial k_{a\mu}^{[H]}}\lesssim \lambda^{1/2} \Leftrightarrow \sum_{e\in H}2\eta_{ae}^{} \alpha_e^{[H]} l_e^{[H]\mu} \lesssim \lambda^{1/2}; \\
\label{second_LE_hard_jet_soft_J} &\frac{\partial D}{\partial k_{a\mu}^{[J]}}\lesssim \lambda^{1/2} \Leftrightarrow \sum_{e\in H}2\eta_{ae}^{} \alpha_e^{[H]} l_e^{[H]\mu} +\sum_{i=1}^n \sum_{e\in J_i} 2\eta_{ae}^{} \alpha_e^{[J_i]} l_e^{[J_i]\mu} \lesssim \lambda^{1/2}; \\
\label{second_LE_hard_jet_soft_S} &\frac{\partial D}{\partial k_{a\mu}^{[S]}}\lesssim \lambda^{1/2} \Leftrightarrow \sum_{e\in H}2\eta_{ae}^{} \alpha_e^{[H]} l_e^{[H]\mu} +\sum_{i=1}^n \sum_{e\in J_i} 2\eta_{ae}^{} \alpha_e^{[J_i]} l_e^{[J_i]\mu} +\sum_{e\in S} 2\eta_{ae}^{} \alpha_e^{[S]} l_e^{[S]\mu} \lesssim \lambda^{1/2}.
\end{align}
\end{subequations}
In each of the relations above, the coefficients $\eta_{ae}^{}$ are either $0$ or $\pm 1$, depending on how the loop momentum $k_{a}$ enters the line momenta $l_e^{[H]}$, $l_e^{[J_i]}$ or $l_e^{[S]}$. The three relations in eq.~(\ref{second_LE_hard_jet_soft}) can be justified immediately based on the momentum scaling of eq.~(\ref{FP_scaling_hard_jet_soft}). Nevertheless, we can always restrict the Landau equation analysis to the subset of propagators containing the hard and the jet momenta only. (This can be demonstrated by starting with a Feynman parameterisation of these propagators, and leaving out all the soft ones.) The outcome of this analysis would lead to the same constraints we have seen in eq.~(\ref{second_LE_hard_jet}) in the case of $N(J)\geqslant0, N(S)=0$. This is equivalent to the argument used in the discussion of figure~\ref{jet_configuration_forbidden}, where soft lines can be removed when obtaining the constraints on the configuration of an infrared region from the second Landau condition.

\subsubsection{The region vectors in the on-shell expansion\label{sec:proposition_1}}

Summarising the analysis above, we propose that the solutions of the Landau equations for massless wide-angle scattering are all endpoint singularities which stand in one-to-one correspondence with the regions in the on-shell expansion associated with the facets of the Newton polytope~$\Delta^{(N+1)}[\mathcal{P}]$. This leads to the following proposition, which we present without a formal proof.

\begin{proposition}
In the on-shell expansion for any wide-angle scattering graph, all regions are given by region vectors which take the form of eq.~(\ref{region_vector_wideangle_scattering}), where the entries $0$, $-1$ and~$-2$ correspond to the propagators in the hard ($H$), jet ($J_1,\dots,J_K$) and soft ($S$) subgraphs separately. Moreover, these subgraphs meet the following basic requirements.
\begin{enumerate}
    \item $H$, to which all the off-shell external momenta $q_j^\mu$ are attached, is a connected subgraph.
    \item For each $i=1,\dots,K$, $J_i$ is a connected subgraph attached to $H$, and its only external momentum is the $p_i^\mu$ line.
    \item $S$ is a subgraph without external momenta, and each connected component of $S$ is attached to $H\cup J$, where $J\equiv \cup_{i=1}^K J_i$.
\end{enumerate}
\label{proposition-region_vectors_are_hard_and_infrared}
\end{proposition}

While the basic requirements set out above are necessary for any infrared region, it is clear at the outset that additional requirements on the configurations of jets arise from the Coleman-Norton analysis, as shown in figure~\ref{jet_configuration_forbidden}. This amounts to excluding certain candidates for regions, which yield scaleless integrals. We will return to study this aspect in section~\ref{infrared_regions_in_wideangle_scattering}, where we formulate the precise requirements for the momentum configuration in any infrared regions.

It is natural to ask what the precise conditions are for the applicability of this proposition, as the momentum scalings of eq.~(\ref{infrared_region_momentum_scaling}) may be relevant beyond the scope of the on-shell expansion in wide-angle scattering, for which it has been formulated. The reality is, that depending on the kinematic setup and on the expansion considered, different regions may arise, going far beyond this simple picture. We will give below three examples to illustrate this.

Consider first the mass expansion described in figure~\ref{mass_expansion_counterexample} where we assume the hierarchy \hbox{$p^2=m^2 \ll M^2=Q_1^2\sim Q_2^2$}. Here the following two region vectors (amongst others) can be found:
\begin{samepage}
\begin{itemize}
    \item [$R_1$:] $\v_{R_1}= (0,-1,-1,-1,-1,-2,-1,-2,-3,1)$;
    \item [$R_2$:] $\v_{R_2}= (0,-1,-1,-1,-1,-1,-1,-2,-2,1)$.
\end{itemize}
\end{samepage}
In momentum space, the loop momenta in the regions $R_1$ and $R_2$ are respectively:
\begin{samepage}
\begin{itemize}
    \item [$R_1$:] $k_1^\mu=(k_1^+,k_1^-,k_{1\perp})\sim (1,\lambda,\lambda^{1/2})$, $k_2^\mu\sim \lambda$, $k_3^\mu= (k_1^+,k_1^-,k_{1\perp})\sim \lambda(1,\lambda,\lambda^{1/2})$;
    \item [$R_2$:] $k_1^\mu=(k_1^+,k_1^-,k_{1\perp})\sim (1,\lambda,\lambda^{1/2})$, $k_2^\mu\sim \lambda$, $k_3^\mu\sim \lambda$.
\end{itemize}
\end{samepage}
\begin{figure}[t]
\centering
\resizebox{0.3\textwidth}{!}{
\begin{tikzpicture}[decoration={markings,mark=at position 0.55 with {\arrow{latex}}}] 
\node (q1m) at (3,0) {$Q_2$};
\node [draw,circle,minimum size=4pt,fill=Blue,color=Blue,inner sep=0pt,outer sep=0pt] (q1) at (3, 1) {};
\node (q2) at (1, 5) {$Q_1$};
\node (p1) at (5,5) {$p$};

\path (q1)-- node[pos=.3,minimum size=0pt,inner sep=0pt,outer sep=0pt,color=Red,label={[color=red,anchor=north west,xshift=-5pt,yshift=-12pt]below:$1$}] (q1a) {} node[pos=.6,minimum size=0pt,inner sep=0pt,outer sep=0pt,color=Red,label={[color=red,anchor=north west,xshift=-5pt,yshift=-12pt]below:$2$}] (q1b) {} node[pos=.9,minimum size=0pt,inner sep=0pt,outer sep=0pt,color=Red,label={[color=red,anchor=north west,xshift=-5pt,yshift=-12pt]below:$3$}] (q1c) {}(q2);
\path (q1)-- node[pos=.3,minimum size=0pt,inner sep=0pt,outer sep=0pt,color=Red,label={[color=red,anchor=north west,xshift=-8pt,yshift=-12pt]below:$4$}] (p1a) {} node[pos=.6,minimum size=0pt,inner sep=0pt,outer sep=0pt,color=Red,label={[color=red,anchor=north west,xshift=-8pt,yshift=-13pt]below:$5$}] (p1b) {} node[pos=.9,minimum size=0pt,inner sep=0pt,outer sep=0pt,color=Red,label={[color=red,anchor=north west,xshift=-8pt,yshift=-13pt]below:$6$}] (p1c) {}(p1);

\node [draw,circle,minimum size=4pt,fill=Black,inner sep=0pt,outer sep=0pt] () at (q1a) {};
\node [draw,circle,minimum size=4pt,fill=Black,inner sep=0pt,outer sep=0pt] () at (q1b) {};
\node [draw,circle,minimum size=4pt,fill=Black,inner sep=0pt,outer sep=0pt] () at (q1c) {};
\node [draw,circle,minimum size=4pt,fill=Black,inner sep=0pt,outer sep=0pt] () at (p1a) {};
\node [draw,circle,minimum size=4pt,fill=Black,inner sep=0pt,outer sep=0pt] () at (p1b) {};
\node [draw,circle,minimum size=4pt,fill=Black,inner sep=0pt,outer sep=0pt] () at (p1c) {};

\draw [ultra thick,color=Blue] (q1m) -- (q1);
\draw [ultra thick,color=Black] (q1) -- (q2); 
\draw [color=Black] (q1) -- (p1); 

\draw (q1a) edge [dashed,bend left = 20,postaction={decorate}] (p1a) {};
\path (q1a) -- (p1a) node [midway,above,yshift=4pt] {$k_1$};
\path (q1a) -- (p1a) node [midway,below,yshift=1pt,color=Red] {$7$};

\draw (q1b) edge [dashed,bend left = 20,postaction={decorate}] (p1b) {};
\path (q1b) -- (p1b) node [midway,above,yshift=7pt] {$k_2$};
\path (q1b) -- (p1b) node [midway,below,yshift=4pt,color=Red] {$8$};

\draw (q1c) edge [dashed,bend left = 20,postaction={decorate}] (p1c) {};
\path (q1c) -- (p1c) node [midway,above,yshift=10pt] {$k_3$};
\path (q1c) -- (p1c) node [midway,below,yshift=7pt,color=Red] {$9$};
\end{tikzpicture}
}
\vspace{-3em}
\vspace{10pt}
\caption{The mass expansion example.}
\label{mass_expansion_counterexample}
\end{figure}
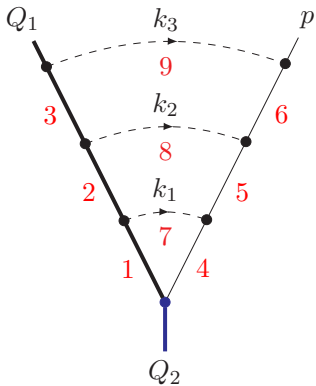
There is an essential difference between the scaling of the two regions $R_1$ and $R_2$. While the latter is fully consistent with eq.~(\ref{infrared_region_momentum_scaling}), the former departs from it in that $k_3^\mu$ does not correspond to any of the expected scalings we have discussed. This illustrates the fact that the scaling vectors of eq.~(\ref{region_vector_wideangle_scattering}) may not suffice to describe all the regions in a generic expansion.

A second example illustrates the fact that some solutions of the Landau equations may not correspond to facet of the Newton polytope, but to lower dimensional faces, in which case they would not appear as regions. As a consequence, only a subset of the pinch surfaces would be represented as regions. 
This occurs for example in the threshold expansion for Drell-Yan or Higgs production. Specifically, refs.~\cite{AnstsDhrDltHzgmMstbg13,AnstsDhrDltFrlHzgmMstbg15} computed the mixed real and virtual corrections at N$^3$LO  using the threshold expansion where the emitted partons are soft. For certain integrals the regions were found (using a somewhat heuristic method) to correspond to only a subset of the soft and collinear momentum scalings of the form of eq.~(\ref{region_vector_wideangle_scattering}).

Another scenario where more complex regions emerge is upon departing from wide-angle scattering kinematics. As explained in section~\ref{region_vectors_from_Newton_polytope}, we expect that all the regions appearing in the wide-angle scattering are given by the facets of $\Delta^{(N+1)}[\mathcal{P}]$. This does not hold, however, for other processes such as the forward scattering, where Glauber modes exist. 
Also in this case, solutions of the Landau equations are expected to be useful in identifying regions, but some of these may arise from cancellation between terms in ${\cal P}$ rather then from scaling vectors associated with the facets of $\Delta^{(N+1)}[\mathcal{P}]$.

To summarise, we determined the general form of region vectors that arise in the on-shell expansion in wide-angle scattering. In the following sections we will show that vectors of the form of eq.~(\ref{region_vector_wideangle_scattering}) correspond to regions, if and only if certain extra requirements on $H$, $J$ and~$S$ are satisfied. These requirements will be spelt out and proven in section~\ref{infrared_regions_in_wideangle_scattering}.
Based on these, we formulate in section~\ref{section-graph_finding_algorithm_regions} a graph-finding algorithm. Furthermore, using this algorithm we verify proposition~\ref{proposition-region_vectors_are_hard_and_infrared} up to five loops, showing that it produces precisely the regions found via the geometric approach of section~\ref{region_vectors_from_Newton_polytope}.

\section{Regions in the on-shell expansion}
\label{section-regions_onshell_momentum_expansion}

In this section we explore the properties of the region vectors in the on-shell expansion of wide-angle scattering, which take the form of eq.~(\ref{region_vector_wideangle_scattering}) as proposed above. Our motivation is to formulate and prove the precise conditions under which such a vector corresponds to a region of a given Feynman integral, namely, a lower facet of the associated Newton polytope $\Delta^{(N+1)}[\mathcal{P}]$.

This section is organised as follows. We begin in section~\ref{space_spanning_trees} by deriving the conditions for the existence of the hard region of the integral $\I(\s)$ corresponding to a graph $G$. To this end we introduce a vector space that is spanned by $\boldsymbol{r}_i- \boldsymbol{r}_j$, the difference between the exponents of any two different monomials in the $\mathcal{U}$ polynomial, each of which is associated with a specific spanning tree of the graph. If this space is $N$-dimensional, then there is a hard region. In section~\ref{leading_terms_in_infrared_regions} we investigate the generic form of monomials that may appear in the leading Lee-Pomeransky polynomial~$\mathcal{P}_0^{(R)}(\x;\s)$, in a generic infrared region $R$.
Based on these results, we prove in section~\ref{infrared_regions_in_wideangle_scattering} that any vector of the form of eq.~(\ref{region_vector_wideangle_scattering}) defines a region vector of the Newton polytope $\Delta^{(N+1)}[\mathcal{P}(G)]$, if and only if some extra requirements on the subgraphs of $G$ are satisfied, eliminating any potential scaleless integrals.

A few notations are introduced here for convenience.
Recall that we use $T^1$ and $T^2$ to denote, respectively, a spanning tree and a spanning 2-tree of $G$. Similarly, for any subgraph $\gamma\subseteq G$, $T^1(\gamma)$ and $T^2(\gamma)$ are used to denote, respectively, a spanning tree and a spanning 2-tree of $\gamma$. We further denote the monomials in $\mathcal{F}$ carrying a kinematic factor $p_i^2\sim \lambda$ as $\mathcal{F}^{(p_i^2)}$ terms, and those carrying a kinematic factor $q_j^2$ or $p_k\cdot p_l$ ($k\neq l$), which scale as $\lambda^0$, as $\mathcal{F}^{(q^2)}$ terms. In a region $R$, the $\mathcal{U}$, $\mathcal{F}^{(p_i^2)}$ and $\mathcal{F}^{(q^2)}$ terms that appear in the leading Lee-Pomeransky polynomial $\mathcal{P}_0^{(R)}(\x;\s)$ (defined below eq.~(\ref{integrand_expansion_operator_definition})) are denoted respectively as $\mathcal{U}^{(R)}$, $\mathcal{F}^{(p_i^2,R)}$ and $\mathcal{F}^{(q^2,R)}$ terms.

\subsection{Spanning trees}
\label{space_spanning_trees}

For any connected graph $\gamma$, we denote the parameter subspace that is spanned by the points corresponding to $T^1(\gamma)$ as $W_\gamma$. By definition, $W_\gamma$ is generated by vectors of the form $\Delta \boldsymbol{r}_{ij}\equiv \boldsymbol{r}_i - \boldsymbol{r}_j$, where $\boldsymbol{r}_i$ and $\boldsymbol{r}_j$ are any two points corresponding to two spanning trees $T_i^1(\gamma)$ and $T_j^1(\gamma)$. Here we aim to derive the dimension, $\text{dim}(W_\gamma)$, by obtaining a set of basis vectors of $W_\gamma$.

We recall that a graph is \emph{one-vertex irreducible} (1VI) if it remains connected after any one of its vertices is removed. (In graph theory, a 1VI graph is also called a biconnected graph.)
Let us consider the case where $\gamma$ is 1VI, and derive the following theorem. 
\begin{theorem}
For any 1VI graph $\gamma$, the parameter subspace $W_\gamma$, which contains all the points corresponding to the spanning trees of $\gamma$, satisfies
\begin{eqnarray}
\textup{dim}(W_\gamma)= N(\gamma)-1,
\label{Uterms_space_dimensionality}
\end{eqnarray}
where $N(\gamma)$ is the number of propagators of $\gamma$. Furthermore, the $(N(\gamma)-1)$ independent basis vectors of $W_\gamma$ can be expressed as:
\begin{eqnarray}
\Big ( \underset{N(\gamma)}{\underbrace{1,-1,0,\dots,0}};0 \Big ),\ \Big ( \underset{N(\gamma)}{\underbrace{1,0,-1,0,\dots,0}};0 \Big ),\ \dots,\ \Big ( \underset{N(\gamma)}{\underbrace{1,0,\dots,0,-1}};0 \Big ).
\label{Nminus1_basis_vector_U}
\end{eqnarray}
\label{theorem-Uterms_space_dimensionality}
\end{theorem}

\begin{proof}
First recall that vectors defining the monomials in $\mathcal{U}(\x)$ take the form of $\boldsymbol{r}_i\equiv (r_{i,1}, \ldots, r_{i,N}; 0)$, where $r_{i,e}=1$ for each edge $e$ that has been removed in forming the corresponding spanning tree, while the other entries are $0$.

It immediately follows from the definition above that $\textup{dim}(W_\gamma) \leqslant N(\gamma)$. Assume now that $\gamma$ contains $L(\gamma)$ loops. In order to form a spanning tree of $\gamma$, we must remove exactly one propagator per loop from $\gamma$, hence every point in $W_\gamma$ satisfies the constraint \hbox{$\sum_{e\in \gamma} r_{i,e} =L(\gamma)$}. Thus, we can tighten the upper bound on the dimension:
\begin{eqnarray}
\textup{dim}(W_\gamma) \leqslant N(\gamma)-1.
\label{proof_theorem1_step1}
\end{eqnarray}
Below we aim to prove $\textup{dim}(W_\gamma) \geqslant N(\gamma)-1$, by finding the basis vectors shown in eq.~(\ref{Nminus1_basis_vector_U}).

Since $\gamma$ is 1VI, a result in graph theory is that for any two propagators of $\gamma$ there is a loop containing them both \cite{West01book}. So for any $i,j\in \{1,\dots,N(\gamma)\}$, we consider a loop $L_{ij}$ that simultaneously contains $e_i$ and $e_j$. We now aim to obtain a subgraph $\gamma' \subseteq \gamma$, which includes all the vertices of $\gamma$ and contains exactly one loop, $L_{ij}$.

The subgraph $\gamma'$ can be obtained through the following operations. Define $\mathcal{V}_0$ as the set of vertices of $L_{ij}$ and $\mathcal{V}'_0$ as the set of vertices of $\gamma \setminus L_{ij}$. Since $\gamma$ is connected, there must be an edge $e'_1$ in $\gamma$ that connects a vertex $v_0\in \mathcal{V}_0$ and $v'_0\in \mathcal{V}'_0$. The graph $\gamma_1\equiv L_{ij}\cup e'_1$ is then also a connected subgraph of $\gamma$, which still contains exactly one loop $L_{ij}$.
We then define $\mathcal{V}_1$ to be the set of vertices of $\gamma_1$, and $\mathcal{V}'_1$ as the set of vertices of $\gamma \setminus \gamma_1$. Now there must be another edge $e'_2$ joining a vertex $v_1\in \mathcal{V}_1$ and $v'_1\in \mathcal{V}'_1$. The graph $\gamma_2\equiv L_{ij}\cup e'_1\cup e'_2$, is again a connected subgraph of $\gamma$ that contains exactly one loop $L_{ij}$. This procedure can be carried out recursively as above, terminating with the final graph $\gamma'\equiv L_{ij}\cup e'_1 \cup\dots\cup e'_n$, with $n=V(\gamma)-V(L_{ij})$. An example illustrating this procedure is shown in figure~\ref{theorem1_procedure}.

\begin{figure}[t]
\centering
\begin{subfigure}[b]{0.13\textwidth}
\centering
    \resizebox{\textwidth}{!}{
\begin{tikzpicture}[scale=0.3]
\node [draw,circle,minimum size=3pt,fill=Black,inner sep=0pt,outer sep=0pt] (v1) at (0,0) {};
\node [draw,circle,minimum size=3pt,fill=Black,inner sep=0pt,outer sep=0pt] (v2) at (2,0) {};
\node [draw,circle,minimum size=3pt,fill=Black,inner sep=0pt,outer sep=0pt] (v3) at (2,2) {};
\node [draw,circle,minimum size=3pt,fill=Black,inner sep=0pt,outer sep=0pt]  (v4) at (0,2) {};
\node [draw,circle,minimum size=3pt,fill=Black,inner sep=0pt,outer sep=0pt]  (v5) at (0,5) {};
\node [draw,circle,minimum size=3pt,fill=Black,inner sep=0pt,outer sep=0pt]  (v6) at (5,0) {};
\node [draw,circle,minimum size=3pt,fill=Black,inner sep=0pt,outer sep=0pt]  (v7) at (5,5) {};
\node [draw,circle,minimum size=3pt,fill=Black,inner sep=0pt,outer sep=0pt]  (v8) at (3.5,3.5) {};

\draw [thick] (v1) -- (v2) node [midway,below] {\textcolor{red}{$e_j$}};
\draw [thick] (v2) -- (v3);
\draw [thick] (v3) -- (v4) node [midway,above] {\textcolor{red}{$e_i$}};
\draw [thick] (v4) -- (v1);

\draw [thick] (v2) -- (v6);
\draw [thick] (v6) -- (v7);
\draw [thick] (v7) -- (v5);
\draw [thick] (v5) -- (v4);

\draw [thick] (v3) -- (v8);
\draw [thick] (v8) -- (v7);
\draw [thick] (v8) -- (v5);
\end{tikzpicture}
}
\vspace{-2pt}
\caption{$\gamma$}
\end{subfigure}

\setcounter{subfigure}{0}
\renewcommand\thesubfigure{\roman{subfigure}}
\begin{subfigure}[b]{0.13\textwidth}
\centering
    \resizebox{\textwidth}{!}{
\begin{tikzpicture}[scale=0.3]
\node [draw,circle,minimum size=3pt,fill=Black,inner sep=0pt,outer sep=0pt] (v1) at (0,0) {};
\node [draw,circle,minimum size=3pt,fill=Black,inner sep=0pt,outer sep=0pt] (v2) at (2,0) {};
\node [draw,circle,minimum size=3pt,fill=Black,inner sep=0pt,outer sep=0pt] (v3) at (2,2) {};
\node [draw,circle,minimum size=3pt,fill=Black,inner sep=0pt,outer sep=0pt]  (v4) at (0,2) {};
\node [draw,circle,minimum size=3pt,fill=Black,inner sep=0pt,outer sep=0pt]  (v5) at (0,5) {};
\node [draw,circle,minimum size=3pt,fill=Black,inner sep=0pt,outer sep=0pt]  (v6) at (5,0) {};
\node [draw,circle,minimum size=3pt,fill=Black,inner sep=0pt,outer sep=0pt]  (v7) at (5,5) {};
\node [draw,circle,minimum size=3pt,fill=Black,inner sep=0pt,outer sep=0pt]  (v8) at (3.5,3.5) {};

\draw [thick] (v1) -- (v2);
\draw [thick] (v2) -- (v3);
\draw [thick] (v3) -- (v4);
\draw [thick] (v4) -- (v1);



\end{tikzpicture}
}
\caption{$L_{ij}$}
\end{subfigure}
\hfill
\begin{subfigure}[b]{0.13\textwidth}
\centering
    \resizebox{\textwidth}{!}{
\begin{tikzpicture}[scale=0.3]
\node [draw,circle,minimum size=3pt,fill=Black,inner sep=0pt,outer sep=0pt] (v1) at (0,0) {};
\node [draw,circle,minimum size=3pt,fill=Black,inner sep=0pt,outer sep=0pt] (v2) at (2,0) {};
\node [draw,circle,minimum size=3pt,fill=Black,inner sep=0pt,outer sep=0pt] (v3) at (2,2) {};
\node [draw,circle,minimum size=3pt,fill=Black,inner sep=0pt,outer sep=0pt]  (v4) at (0,2) {};
\node [draw,circle,minimum size=3pt,fill=Black,inner sep=0pt,outer sep=0pt]  (v5) at (0,5) {};
\node [draw,circle,minimum size=3pt,fill=Black,inner sep=0pt,outer sep=0pt]  (v6) at (5,0) {};
\node [draw,circle,minimum size=3pt,fill=Black,inner sep=0pt,outer sep=0pt]  (v7) at (5,5) {};
\node [draw,circle,minimum size=3pt,fill=Black,inner sep=0pt,outer sep=0pt]  (v8) at (3.5,3.5) {};

\draw [thick] (v1) -- (v2);
\draw [thick] (v2) -- (v3);
\draw [thick] (v3) -- (v4);
\draw [thick] (v4) -- (v1);

\draw [thick] (v5) -- (v4) node [midway,right] {\textcolor{red}{$e_1^\prime$}};


\end{tikzpicture}
}
\caption{$\gamma_1$}
\end{subfigure}
\hfill
\begin{subfigure}[b]{0.13\textwidth}
\centering
    \resizebox{\textwidth}{!}{
\begin{tikzpicture}[scale=0.3]
\node [draw,circle,minimum size=3pt,fill=Black,inner sep=0pt,outer sep=0pt] (v1) at (0,0) {};
\node [draw,circle,minimum size=3pt,fill=Black,inner sep=0pt,outer sep=0pt] (v2) at (2,0) {};
\node [draw,circle,minimum size=3pt,fill=Black,inner sep=0pt,outer sep=0pt] (v3) at (2,2) {};
\node [draw,circle,minimum size=3pt,fill=Black,inner sep=0pt,outer sep=0pt]  (v4) at (0,2) {};
\node [draw,circle,minimum size=3pt,fill=Black,inner sep=0pt,outer sep=0pt]  (v5) at (0,5) {};
\node [draw,circle,minimum size=3pt,fill=Black,inner sep=0pt,outer sep=0pt]  (v6) at (5,0) {};
\node [draw,circle,minimum size=3pt,fill=Black,inner sep=0pt,outer sep=0pt]  (v7) at (5,5) {};
\node [draw,circle,minimum size=3pt,fill=Black,inner sep=0pt,outer sep=0pt]  (v8) at (3.5,3.5) {};

\draw [thick] (v1) -- (v2);
\draw [thick] (v2) -- (v3);
\draw [thick] (v3) -- (v4);
\draw [thick] (v4) -- (v1);

\draw [thick] (v5) -- (v4);

\draw [thick] (v3) -- (v8) node [midway,above left] {\textcolor{red}{$e_2^\prime$}};

\end{tikzpicture}
}
\caption{$\gamma_2$}
\end{subfigure}
\hfill
\begin{subfigure}[b]{0.13\textwidth}
\centering
    \resizebox{\textwidth}{!}{
\begin{tikzpicture}[scale=0.3]
\node [draw,circle,minimum size=3pt,fill=Black,inner sep=0pt,outer sep=0pt] (v1) at (0,0) {};
\node [draw,circle,minimum size=3pt,fill=Black,inner sep=0pt,outer sep=0pt] (v2) at (2,0) {};
\node [draw,circle,minimum size=3pt,fill=Black,inner sep=0pt,outer sep=0pt] (v3) at (2,2) {};
\node [draw,circle,minimum size=3pt,fill=Black,inner sep=0pt,outer sep=0pt]  (v4) at (0,2) {};
\node [draw,circle,minimum size=3pt,fill=Black,inner sep=0pt,outer sep=0pt]  (v5) at (0,5) {};
\node [draw,circle,minimum size=3pt,fill=Black,inner sep=0pt,outer sep=0pt]  (v6) at (5,0) {};
\node [draw,circle,minimum size=3pt,fill=Black,inner sep=0pt,outer sep=0pt]  (v7) at (5,5) {};
\node [draw,circle,minimum size=3pt,fill=Black,inner sep=0pt,outer sep=0pt]  (v8) at (3.5,3.5) {};

\draw [thick] (v1) -- (v2);
\draw [thick] (v2) -- (v3);
\draw [thick] (v3) -- (v4);
\draw [thick] (v4) -- (v1);

\draw [thick] (v2) -- (v6) node [midway,above] {\textcolor{red}{$e_3^\prime$}};
\draw [thick] (v5) -- (v4);

\draw [thick] (v3) -- (v8);

\end{tikzpicture}
}
\caption{$\gamma_3$}
\end{subfigure}
\hfill
\begin{subfigure}[b]{0.13\textwidth}
\centering
    \resizebox{\textwidth}{!}{
\begin{tikzpicture}[scale=0.3]
\node [draw,circle,minimum size=3pt,fill=Black,inner sep=0pt,outer sep=0pt] (v1) at (0,0) {};
\node [draw,circle,minimum size=3pt,fill=Black,inner sep=0pt,outer sep=0pt] (v2) at (2,0) {};
\node [draw,circle,minimum size=3pt,fill=Black,inner sep=0pt,outer sep=0pt] (v3) at (2,2) {};
\node [draw,circle,minimum size=3pt,fill=Black,inner sep=0pt,outer sep=0pt]  (v4) at (0,2) {};
\node [draw,circle,minimum size=3pt,fill=Black,inner sep=0pt,outer sep=0pt]  (v5) at (0,5) {};
\node [draw,circle,minimum size=3pt,fill=Black,inner sep=0pt,outer sep=0pt]  (v6) at (5,0) {};
\node [draw,circle,minimum size=3pt,fill=Black,inner sep=0pt,outer sep=0pt]  (v7) at (5,5) {};
\node [draw,circle,minimum size=3pt,fill=Black,inner sep=0pt,outer sep=0pt]  (v8) at (3.5,3.5) {};

\draw [thick] (v1) -- (v2);
\draw [thick] (v2) -- (v3);
\draw [thick] (v3) -- (v4);
\draw [thick] (v4) -- (v1);

\draw [thick] (v2) -- (v6);
\draw [thick] (v5) -- (v4);

\draw [thick] (v3) -- (v8);
\draw [thick] (v8) -- (v7) node [midway,above left] {\textcolor{red}{$e_4^\prime$}};

\end{tikzpicture}
}
\caption{$\gamma^\prime=\gamma_4$}
\end{subfigure}
\caption{The procedure in the proof of theorem~\ref{theorem-Uterms_space_dimensionality}.}
\label{theorem1_procedure}
\end{figure}
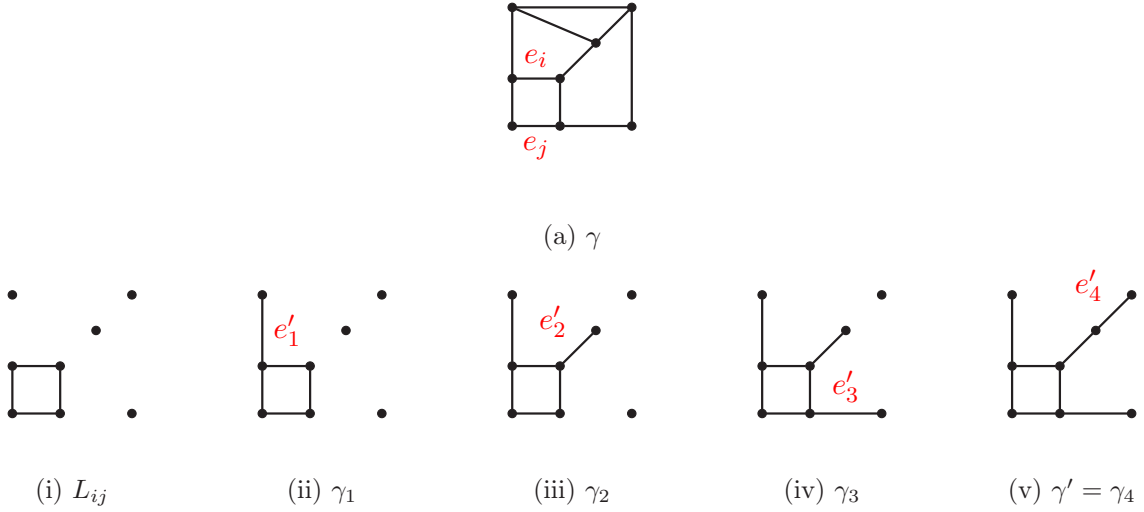

From the construction above, we can obtain a spanning tree of $\gamma$ by removing any edge of $L_{ij}$ from $\gamma'$. We consider the spanning trees $T_i^1\equiv \gamma' \setminus e_i$ and $T_j^1\equiv \gamma' \setminus e_j$. These two spanning trees correspond respectively to two points $\boldsymbol{r}_i$ and $\boldsymbol{r}_j$ in the parameter space, and $\Delta \boldsymbol{r}_{ij}\equiv \boldsymbol{r}_i -\boldsymbol{r}_j$ is
\begin{eqnarray}
\Delta \boldsymbol{r}_{ij}= \Big( \underset{N(\gamma)}{\underbrace{0,\dots,0,1,0,\dots,0,-1,0,\dots,0}};0 \Big).
\label{delta_rij_form_1VI_spanning_tree}
\end{eqnarray}
We can parameterise the graph $\gamma$ in specific ways such that $i=1$ and $j=2,\dots,N(\gamma)$. Then all the vectors in eq.~(\ref{Nminus1_basis_vector_U}) are obtained, and there are $N(\gamma)-1$ of them. Thus $\textup{dim}(W_\gamma) \geqslant N(\gamma)-1$. Combining this with eq.~(\ref{proof_theorem1_step1}) above, we have determined the dimension of $W_\gamma$, which is (\ref{Uterms_space_dimensionality}).
\end{proof}

Note that every connected graph can be seen as the union of several 1VI graphs. Then theorem~\ref{theorem-Uterms_space_dimensionality} could be directly generalised to any connected graph $\gamma$, i.e.
\begin{corollary}
Suppose a graph $\gamma$ has $n$ 1VI components, $\gamma_1, \gamma_2, \dots, \gamma_n$, then we have
\begin{eqnarray}
\label{Uterms_space_direct_sum_1VI_components}
W_\gamma &&= \bigoplus_{i=1}^n W_{\gamma_n},\\
\textup{dim}(W_\gamma) &&= \sum_{i=1}^n (N(\gamma_i)-1).
\label{Uterms_spacedim_sum_1VI_dimensions}
\end{eqnarray}
\label{theorem-Uterms_space_dimensionality_corollary1}
\end{corollary}
\begin{proof}
First, it is a result of graph theory that each spanning tree of $\gamma$ is the union of the spanning trees of the $\gamma_i$ ($i=1,\dots,n$), and vice versa. Then since $W_\gamma$ is the space that contains all the points corresponding to the spanning trees of $\gamma$, we immediately have eq.~(\ref{Uterms_space_direct_sum_1VI_components}). Eq.~(\ref{Uterms_spacedim_sum_1VI_dimensions}) then follows directly from eq.~(\ref{Uterms_space_dimensionality}).
\end{proof}

Note that for any 1VI component $\gamma_i$ with a single propagator, then $N(\gamma_i)-1=0$, and $\gamma_i$ does not contribute to $\text{dim} (W_\gamma)$. We can define the \emph{nontrivial 1VI components} as those containing at least one loop, so that the word ``1VI components'' could get replaced by ``nontrivial 1VI components'' in corollary~\ref{theorem-Uterms_space_dimensionality_corollary1}.

From this observation, we construct a simpler form of every graph $G$ and call it the \emph{reduced form of $G$}. It is obtained from $G$ by contracting each of its nontrivial 1VI components into a vertex. If two nontrivial 1VI components $\gamma_1$ and $\gamma_2$ share a vertex, we add an auxiliary propagator connecting $\gamma_1$ and $\gamma_2$ before contracting them. The reduced form of $G$ is denoted by $G_\text{red}$, which is always a tree graph.

Some examples of the reduced form of $G$ are given in figure~\ref{reduced_form_examples}. Here we note a relation between the edges of $G$ and $G_\text{red}$:
\begin{eqnarray}
N(G) = \widehat{N}(G_\text{red})+ \sum_{i=1}^n N(\gamma_i),
\label{relation_propagator_number_reduced}
\end{eqnarray}
where $n$ is the number of nontrivial 1VI components of $G$, and $\widehat{N}(G_\text{red})$ is the number of edges that are in both $G_\text{red}$ and $G$ (i.e. the non-auxiliary edges of $G_\text{red}$).
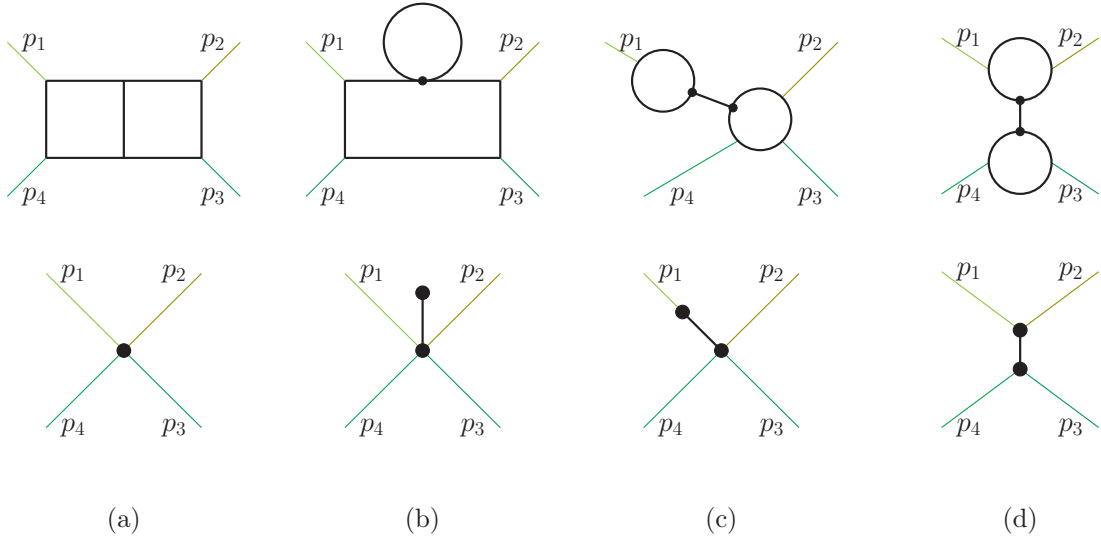
\begin{figure}[t]
\centering
\begin{subfigure}[b]{0.22\textwidth}
\centering
\resizebox{\textwidth}{!}{
\begin{tikzpicture}[line width = 0.6, font=\large, mydot/.style={circle, fill, inner sep=.7pt}]

\node (p1) at (0.7,11) {\huge $p_1$};
\node (p2) at (5.3,11) {\huge $p_2$};
\node (p3) at (5.3,7) {\huge $p_3$};
\node (p4) at (0.7,7) {\huge $p_4$};
\node (p1r) at (1.7,5) {\huge $p_1$};
\node (p2r) at (4.3,5) {\huge $p_2$};
\node (p3r) at (4.3,1) {\huge $p_3$};
\node (p4r) at (1.7,1) {\huge $p_4$};

\draw (0,11) edge [thick, color=LimeGreen] (1,10) node [] {};
\draw (0,7) edge [thick, color=Green] (1,8) node [] {};
\draw (6,11) edge [thick, color=olive] (5,10) node [] {};
\draw (6,7) edge [thick, color=ForestGreen] (5,8) node [] {};
\draw (1,10) edge [ultra thick] (5,10) node [] {};
\draw (1,8) edge [ultra thick] (5,8) node [] {};
\draw (1,10) edge [ultra thick] (1,8) node [] {};
\draw (3,10) edge [ultra thick] (3,8) node [] {};
\draw (5,10) edge [ultra thick] (5,8) node [] {};

\draw (1,5) edge [thick, color=LimeGreen] (3,3) node [] {};
\draw (1,1) edge [thick, color=Green] (3,3) node [] {};
\draw (5,5) edge [thick, color=olive] (3,3) node [] {};
\draw (5,1) edge [thick, color=ForestGreen] (3,3) node [] {};

\draw[fill,thick] (3,3) circle (5pt);
\end{tikzpicture}
}
\vspace{2pt}
\caption{}
\label{reduced_form_examples_a}
\end{subfigure}
\hfill
\begin{subfigure}[b]{0.22\textwidth}
\centering
\resizebox{\textwidth}{!}{
\begin{tikzpicture}[line width = 0.6, font=\large, mydot/.style={circle, fill, inner sep=.7pt}]

\node (p1) at (0.7,11) {\huge $p_1$};
\node (p2) at (5.3,11) {\huge $p_2$};
\node (p3) at (5.3,7) {\huge $p_3$};
\node (p4) at (0.7,7) {\huge $p_4$};
\node (p1r) at (1.7,5) {\huge $p_1$};
\node (p2r) at (4.3,5) {\huge $p_2$};
\node (p3r) at (4.3,1) {\huge $p_3$};
\node (p4r) at (1.7,1) {\huge $p_4$};

\draw (0,11) edge [thick, color=LimeGreen] (1,10) node [] {};
\draw (0,7) edge [thick, color=Green] (1,8) node [] {};
\draw (6,11) edge [thick, color=olive] (5,10) node [] {};
\draw (6,7) edge [thick, color=ForestGreen] (5,8) node [] {};
\draw (1,10) edge [ultra thick] (5,10) node [] {};
\draw (1,8) edge [ultra thick] (5,8) node [] {};
\draw (1,10) edge [ultra thick] (1,8) node [] {};
\draw (5,10) edge [ultra thick] (5,8) node [] {};
\draw[fill,thick] (3,10) circle (3pt);
\node[ultra thick, draw=Black, circle, minimum size=2cm] () at (3,11){};

\draw (1,5) edge [thick, color=LimeGreen] (3,3) node [] {};
\draw (1,1) edge [thick, color=Green] (3,3) node [] {};
\draw (5,5) edge [thick, color=olive] (3,3) node [] {};
\draw (5,1) edge [thick, color=ForestGreen] (3,3) node [] {};
\draw (3,3) edge [ultra thick] (3,4.5) node [] {};

\draw[fill,thick] (3,3) circle (5pt);
\draw[fill,thick] (3,4.5) circle (5pt);
\end{tikzpicture}
}
\vspace{2pt}
\caption{}
\label{reduced_form_examples_b}
\end{subfigure}
\hfill
\begin{subfigure}[b]{0.22\textwidth}
\centering
\resizebox{\textwidth}{!}{
\begin{tikzpicture}[line width = 0.6, font=\large, mydot/.style={circle, fill, inner sep=.7pt}]

\node (p1) at (0.7,11) {\huge $p_1$};
\node (p2) at (5.3,11) {\huge $p_2$};
\node (p3) at (5.3,7) {\huge $p_3$};
\node (p4) at (2,7) {\huge $p_4$};
\node (p1r) at (1.7,5) {\huge $p_1$};
\node (p2r) at (4.3,5) {\huge $p_2$};
\node (p3r) at (4.3,1) {\huge $p_3$};
\node (p4r) at (1.7,1) {\huge $p_4$};

\draw (0,11) edge [thick, color=LimeGreen] (0.83,10.5) node [] {};
\draw (1,7) edge [thick, color=Green] (3.46,8.46) node [] {};
\draw (6,11) edge [thick, color=olive] (4.54,9.54) node [] {};
\draw (6,7) edge [thick, color=ForestGreen] (4.54,8.46) node [] {};
\node[ultra thick, draw=Black, circle, minimum size=1.6cm] () at (4,9){};
\node[ultra thick, draw=Black, circle, minimum size=1.6cm] () at (1.5,10){};
\draw (2.25,9.7) edge [ultra thick] (3.3,9.3) node [] {};
\draw[fill,thick] (2.25,9.7) circle (3pt);
\draw[fill,thick] (3.3,9.3) circle (3pt);

\draw (1,5) edge [thick, color=LimeGreen] (2,4) node [] {};
\draw (1,1) edge [thick, color=Green] (3,3) node [] {};
\draw (5,5) edge [thick, color=olive] (3,3) node [] {};
\draw (5,1) edge [thick, color=ForestGreen] (3,3) node [] {};
\draw[fill,thick] (2,4) circle (5pt);
\draw (3,3) edge [ultra thick] (2,4) node [] {};

\draw[fill,thick] (3,3) circle (5pt);
\end{tikzpicture}
}
\vspace{2pt}
\caption{}
\label{reduced_form_examples_c}
\end{subfigure}
\hfill
\begin{subfigure}[b]{0.22\textwidth}
\centering
\resizebox{0.7\textwidth}{!}{
\begin{tikzpicture}[line width = 0.6, font=\large, mydot/.style={circle, fill, inner sep=.7pt}]


\node (p1) at (1.7,11) {\huge $p_1$};
\node (p2) at (4.3,11) {\huge $p_2$};
\node (p3) at (4.3,7) {\huge $p_3$};
\node (p4) at (1.7,7) {\huge $p_4$};
\node (p1r) at (1.7,5) {\huge $p_1$};
\node (p2r) at (4.3,5) {\huge $p_2$};
\node (p3r) at (4.3,1) {\huge $p_3$};
\node (p4r) at (1.7,1) {\huge $p_4$};

\draw (1,11) edge [thick, color=LimeGreen] (2.2,10.2) node [] {};
\draw (1,7) edge [thick, color=Green] (2.2,7.8) node [] {};
\draw (5,11) edge [thick, color=olive] (3.8,10.2) node [] {};
\draw (5,7) edge [thick, color=ForestGreen] (3.8,7.8) node [] {};
\node[ultra thick, draw=Black, circle, minimum size=1.6cm] () at (3,10.2){};
\node[ultra thick, draw=Black, circle, minimum size=1.6cm] () at (3,7.8){};
\draw[fill,thick] (3,9.4) circle (3pt);
\draw[fill,thick] (3,8.6) circle (3pt);
\draw (3,9.4) edge [ultra thick] (3,8.6) node [] {};

\draw (1,5) edge [thick, color=LimeGreen] (3,3.5) node [] {};
\draw (1,1) edge [thick, color=Green] (3,2.5) node [] {};
\draw (5,5) edge [thick, color=olive] (3,3.5) node [] {};
\draw (5,1) edge [thick, color=ForestGreen] (3,2.5) node [] {};
\draw (3,3.5) edge [ultra thick] (3,2.5) node [] {};

\draw[fill,thick] (3,3.5) circle (5pt);
\draw[fill,thick] (3,2.5) circle (5pt);
\end{tikzpicture}
}
\vspace{2pt}
\caption{}
\label{reduced_form_examples_d}
\end{subfigure}
\caption{Some examples of the reduced forms of certain $2\to 2$ scattering graphs. Top: the original Feynman graphs. Bottom: the corresponding reduced forms.}
\label{reduced_form_examples}
\end{figure}

The concept of the reduced form enables us to obtain a necessary and sufficient condition that $G$ has a hard region, which can be seen as another corollary of theorem~\ref{theorem-Uterms_space_dimensionality}.

\begin{corollary}
The hard vector $\boldsymbol{v}_H$ is normal to a lower facet of $\Delta^{(N+1)}[\mathcal{P}(G)]$, if and only if all the internal propagators of $G_\textup{red}$ are off-shell.
\label{theorem-Uterms_space_dimensionality_corollary2}
\end{corollary}

\begin{proof}
First we show that if all the propagators of $G_\text{red}$ are off-shell, then $\boldsymbol{v}_H$ is normal to a lower facet.
Suppose there are $n$ nontrivial 1VI components of $G$. We can parameterise the propagators of $G$ such that the first $\sum_{i=1}^n N(\gamma_i)$ parameters correspond to the propagators of the $n$ nontrivial 1VI components, and the next (and the last) $\widehat{N}(G_\text{red})$ parameters correspond to the propagators that are in both $G_\text{red}$ and $G$. According to corollary~\ref{theorem-Uterms_space_dimensionality_corollary1}, we can find the following $\sum_{i=1}^n (N(\gamma_i)-1)$ linearly independent vectors of the form:
\begin{eqnarray}
\Big ( \underset{N(\gamma_1)+\dots+N(\gamma_{i-1})}{\underbrace{0,\dots,0}}, \underset{N(\gamma_i)}{\underbrace{1,0,\dots,0,-1,0,\dots,0}}, \underset{N(\gamma_{i+1})+\dots+N(\gamma_{n})}{\underbrace{0,\dots,0}}, \underset{\widehat{N}(G_\text{red})}{\underbrace{0,\dots,0}};0 \Big ),
\label{hard_basis_vectors_U_polynomial}
\end{eqnarray}
and the dimension of the space spanned by the vectors above is $\sum_{i=1}^n N(\gamma_i)-n$. According to eq.~(\ref{relation_propagator_number_reduced}), we still need to find another $n+ \widehat{N}(G_\text{red})$ vectors that are linearly independent of them, in order to show that the lower face normal to $\boldsymbol{v}_H$ is $N$-dimensional, hence a lower facet.

For the spanning tree $T^1$ corresponding to  a given term in $\mathcal{U}(\x)$, the associated point in the parameter space is of the following form:
\begin{eqnarray}
\boldsymbol{r}_1= \Big ( \underset{N(\gamma_1)}{\underbrace{a_1^{(1)},\dots,a_{N(\gamma_1)}^{(1)}}},\dots,\underset{N(\gamma_n)}{\underbrace{a_1^{(n)},\dots,a_{N(\gamma_n)}^{(n)}}},\underset{\widehat{N}(G_\text{red})}{\underbrace{0,\dots,0}};0 \Big ),
\end{eqnarray}
where for each $i=1,\dots,n$, there are $L(\gamma_i)$ entries with value $1$ and $N(\gamma_i)-L(\gamma_i)$ entries with value $0$ in $\{a_1^{(i)},\dots,a_{N(\gamma_i)}^{(i)}\}$. We now consider the spanning 2-tree corresponding to an $\mathcal{F}^{(q^2)}$ term by removing one propagator from $T^1$. Since every line momentum of $G_\text{red}$ is off-shell, this propagator can either be in $G_\text{red}$, or in the $n$ nontrivial 1VI components. The point $\boldsymbol{r}_2$ that is associated to this spanning 2-tree, is identical to $\boldsymbol{r}_1$ except for one entry, which is $0$ in $\boldsymbol{r}_1$ and $1$ in $\boldsymbol{r}_2$. As a result, the vector $\boldsymbol{r}_2 -\boldsymbol{r}_1$ may take one of the two following forms:
\begin{subequations}
\label{hard_UandFs_difference}
\begin{align}
        \Big (\underset{N(\gamma_1)+\dots+N(\gamma_{i-1})}{\underbrace{0,\dots,0}},\ \underset{N(\gamma_i)}{\underbrace{0,\dots,0,1,0,\dots,0}}\ , \underset{N(\gamma_{i+1})+\dots+N(\gamma_n)}{\underbrace{0,\dots,0}}, \underset{\widehat{N}(G_\text{red})}{\underbrace{0,\dots,0}},;0 \Big ),
        \label{hard_UandFs_differenceI}
        \\
        \Big (\underset{N(\gamma_1)+\dots+N(\gamma_n)}{\underbrace{0,\dots,0}},\underset{\widehat{N}(G_\text{red})}{\underbrace{0,\dots,0,1,0,\dots,0}};0 \Big ).
        \label{hard_UandFs_differenceII}
\end{align}
\end{subequations}
In each of the vectors above, there is exactly one entry with value $1$, while all the others are~$0$. Note that in (\ref{hard_UandFs_differenceI}), we may choose the entry $1$ at $n$ distinct positions, each corresponding to removing some off-shell propagator in a different nontrivial 1VI component $\gamma_i$ with $1 \leqslant i\leqslant n$; these are bound to be independent. 
In (\ref{hard_UandFs_differenceII}), we may choose the entry~$1$ at any of the last $\widehat{N}(G_\text{red})$ entries, each corresponding to removing a different internal propagator in~$G_\text{red}$. 
It is straightforward to check that the vectors in eqs.~(\ref{hard_basis_vectors_U_polynomial}) and (\ref{hard_UandFs_difference}) are linearly independent, and their total number is $N(G)$, hence they constitute a set of basis vectors of the hard facet.

Let us now consider the case where there is a propagator $e'\in G_\text{red}$ that carries an on-shell momentum (either lightlike or vanishing). In this case, such a basis cannot be constructed.
For example, if one propagator of $G_\text{red}$ carries a momentum $p_1^\mu$, then removing it from a spanning tree, yields a $\mathcal{F}^{(p_1^2)}$ term, rather than a $\mathcal{F}^{(q^2)}$ term, which is not in the hard facet.
\end{proof}

As an example, all the propagators in the reduced form of figures~\ref{reduced_form_examples_a} and \ref{reduced_form_examples_d} carry off-shell momenta, so the hard region is present in those graphs; meanwhile the reduced form in figures~\ref{reduced_form_examples_b} and \ref{reduced_form_examples_c} contains propagators that carry vanishing or lightlike momenta, respectively, so  there are no hard regions for these graphs.

With theorem~\ref{theorem-Uterms_space_dimensionality} and its corollaries, we can proceed with the study of the space generated by the polynomial $\mathcal{P}_0^{(R)}(\x;\s)$ for a given infrared region $R$. This will be our objective in the rest of this section.

\subsection{Leading terms in the infrared regions}
\label{leading_terms_in_infrared_regions}

We now focus on the infrared regions by investigating the generic properties of the terms of $\mathcal{P}_0^{(R)} (\x;\s)$, where $R$ is a given infrared region characterised by the region vector $\boldsymbol{v}_R$. By definition, these terms correspond to the points $\boldsymbol{r}$ of the Newton polytope with the minimum of $\boldsymbol{r} \cdot \boldsymbol{v}_R$. Below we will see that these terms can be classified into four types having distinct features.

We start with two examples of infrared regions from figure~\ref{figure-UF_polynomial_onshell_expansion_example}, which describes a hard process with one off-shell external momentum $Q^\mu$ and three nearly on-shell external momenta $p_i^\mu$ for $i=1,2,3$. In the on-shell expansion, we assume that $p_i^2 \sim \lambda Q^2$ with $\lambda\ll 1$, while the hard scales $s_{kl}\equiv (p_k+p_l)^2$ and $Q^2$ are all of the same order of magnitude. By definition, the Lee-Pomeransky polynomial $\mathcal{P}(\x;\s)$ is
\begin{eqnarray}
\mathcal{P}(\x;\s)=&& x_1x_2 +x_1x_3 +x_1x_5 +x_1x_6 +x_2x_3 +x_2x_4\nonumber\\
&& +x_2x_6 +x_3x_4 +x_3x_5 +x_4x_5 +x_4x_6 +x_5x_6 +(-Q^2)x_1x_2x_3\nonumber\\
&& +(-p_1^2)\big( x_1x_2x_4 +x_1x_3x_4 +x_1x_4x_5 +x_1x_4x_6 +x_1x_5x_6 \big)\nonumber\\
&& +(-p_2^2)\big( x_1x_2x_5 +x_2x_3x_5 +x_2x_4x_5 +x_2x_4x_6 +x_2x_5x_6 \big)\nonumber\\
&& +(-p_3^2)\big( x_1x_3x_6 +x_2x_3x_6 +x_3x_4x_6 +x_3x_5x_6 +x_3x_4x_5 \big)\nonumber\\
&& +(-s_{12})x_1x_2x_6 +(-s_{23})x_2x_3x_4 +(-s_{13})x_1x_3x_5.
\label{UF_polynomial_onshell_expansion_example}
\end{eqnarray}
The two regions we focus on are denoted by $\text{SS}$ and $\text{C}_1\text{S}$. In the $\text{SS}$-region, the two loop momenta are soft; in the $\text{C}_1\text{S}$-region, the loop momentum running through the propagators $4,5,2$ and $1$ is collinear to $p_1^\mu$, while the other loop momentum is soft. Using eq.~(\ref{LP_parameter_scales}), the Lee-Pomeransky parameters in these regions scale as follows:
\begin{subequations}
\label{SS_and_C1S_scaling}
\begin{align}
    \text{SS: }&\quad x_1,x_2,x_3\sim \lambda^{-1},\quad x_4,x_5,x_6\sim\lambda^{-2};\\
    \text{C}_1\text{S: }&\quad x_2\sim\lambda^0,\quad x_1,x_3,x_4,x_5\sim\lambda^{-1},\quad x_6\sim\lambda^{-2}.
\end{align}
\end{subequations}
\begin{figure}[t]
\centering
\resizebox{0.45\textwidth}{!}{
\begin{tikzpicture}[decoration={markings,mark=at position 0.55 with {\arrow{latex}}}] 
\node (q1m) at (6,6) {$Q$};
\node [draw,circle,minimum size=4pt,fill=Black,color=Black,inner sep=0pt,outer sep=0pt] (q1) at (6, 5) {};
\node (p1) at (3, 2) {$p_1$};
\node (p2) at (6,1.5) {$p_2$};
\node (p3) at (9,2) {$p_3$};
\node [minimum size=0pt,inner sep=0pt,outer sep=0pt] (s) at (5.1,2.8) {};

\path (q1)-- node[pos=.5,minimum size=0pt,inner sep=0pt,outer sep=0pt,color=Red] (q1b) {}(p1);
\path (q1)-- node[pos=.51,circle,fill=white,opacity=1,minimum size=10pt,inner sep=0pt,outer sep=0 pt] (q2c) {} node[pos=.8,minimum size=0pt,inner sep=0pt,outer sep=0pt,color=Red] (q2b) {}(p2);
\path (q1)-- node[pos=.5,minimum size=0pt,inner sep=0pt,outer sep=0pt,color=Red] (q3b) {}(p3);

\node [draw,circle,minimum size=4pt,fill=Black,inner sep=0pt,outer sep=0pt] () at (q1b) {};
\node [draw,circle,minimum size=4pt,fill=Black,inner sep=0pt,outer sep=0pt] () at (q2b) {};
\node [draw,circle,minimum size=4pt,fill=Black,inner sep=0pt,outer sep=0pt] () at (q3b) {};
\node [draw,circle,minimum size=4pt,fill=Black,inner sep=0pt,outer sep=0pt] () at (s) {};

\path (q1) -- (q1b) node [midway,above,xshift=-3pt,color=Red] {$1$};
\path (q1) -- (q3b) node [midway,above,xshift=3pt,color=Red] {$3$};
\path (q1) -- (q2c) node [midway,xshift=8pt,yshift=-5pt,color=Red] {$2$};
\path (s) -- (q1b) node [midway,below,xshift=-8pt,yshift=3pt,color=Red] {$4$};
\path (s) -- (q2b) node [midway,below,xshift=-5pt,yshift=-2pt,color=Red] {$5$};
\path (s) -- (q3b) node [midway,below,xshift=1pt,yshift=2pt,color=Red] {$6$};

\draw [ultra thick,color=Blue] (q1m) -- (q1);
\draw [thick,color=Black] (q1) -- (p1); 
\draw [thick,color=Black] (q1) -- (q2c);
\draw [thick,color=Black] (q2c) -- (p2);
\draw [thick,color=Black] (q1) -- (p3);
\draw (s) edge [thick,color=Black,bend left=15] (q1b) {};
\draw (s) edge [thick,color=Black,bend right=15] (q2b) {};
\draw (s) edge [thick,color=Black,bend left=15] (q3b) {};
\end{tikzpicture}
}
\vspace{-3em}
\vspace{12pt}
\caption{A two-loop four-point wide-angle scattering example with three nearly on-shell external momenta $p_i^\mu$ ($i=1,2,3$) and one off-shell external momentum $Q^\mu$.}
\label{figure-UF_polynomial_onshell_expansion_example}
\end{figure}
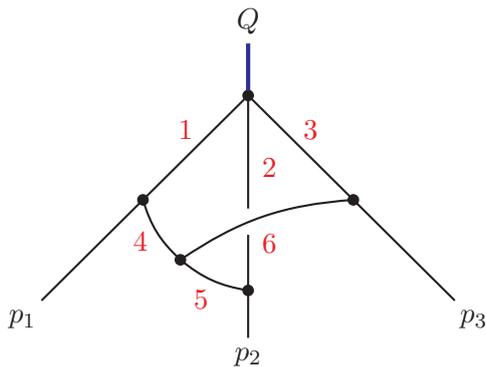
The leading Lee-Pomeransky polynomial $\mathcal{P}_0^{(R)}(\x;\s)$, with $R=\text{SS}$ reads:
\begin{eqnarray}
\label{UF_polynomial_SS_leading_terms}
\mathcal{P}_0^{(\text{SS})}(\x;\s)=&& {\color{Red} x_4x_5} + {\color{Red} x_4x_6} + {\color{Red} x_5x_6} + (-p_1^2) ({\color{Green} x_1}{\color{Red} x_4x_5} + {\color{Green} x_1}{\color{Red} x_4x_6} + {\color{Green} x_1}{\color{Red} x_5x_6}) \\ \nonumber
&&+(-p_2^2)({\color{Green} x_2}{\color{Red} x_4x_5} + {\color{Green} x_2}{\color{Red} x_4x_6} + {\color{Green} x_2}{\color{Red} x_5x_6}) + (-p_3^2)({\color{Green} x_3}{\color{Red} x_4x_5} + {\color{Green} x_3}{\color{Red} x_4x_6} + {\color{Green} x_3}{\color{Red} x_5x_6}) \\ \nonumber
&&+(-s_{12}) {\color{Green} x_1x_2}{\color{Red} x_6} +(-s_{23}){\color{Green} x_2x_3}{\color{Red} x_4} +(-s_{13}){\color{Green} x_1x_3}{\color{Red} x_5},
\end{eqnarray}
while the leading polynomial $\mathcal{P}_0^{(R)}(\x;\s)$ with $R=\text{C}_1\text{S}$ is:
\begin{eqnarray}
\label{UF_polynomial_C1S_leading_terms}
\mathcal{P}_0^{(\text{C}_1\text{S})}(\x;\s)=&& {\color{Green} x_1}{\color{Red} x_6}+ {\color{Green} x_4}{\color{Red} x_6}+ {\color{Green} x_5}{\color{Red} x_6} + (-p_1^2) ({\color{Green} x_1x_4}{\color{Red} x_6} + {\color{Green} x_1x_5}{\color{Red} x_6})\\
&&+ (-p_3^2)({\color{Green} x_1x_3}{\color{Red} x_6} + {\color{Green} x_3x_4}{\color{Red} x_6} + {\color{Green} x_3x_5}{\color{Red} x_6}) +(-s_{12}){\color{Green} x_1}{\color{Blue} x_2}{\color{Red} x_6} +(-s_{13}){\color{Green} x_1x_3x_5}.\qquad\nonumber
\end{eqnarray}
Note that in the polynomials~(\ref{UF_polynomial_SS_leading_terms}) and~(\ref{UF_polynomial_C1S_leading_terms})
we used colour coding to identify the parameters associated with 
{\color{Blue} hard} (blue), {\color{Green} jet} (green) and {\color{Red} soft} (red), respectively.
In each of these polynomials, the first three terms are the $\mathcal{U}^{(R)}$ terms, the terms carrying $(-p_i^2)$ dependence are the $\mathcal{F}^{(p_i^2,R)}$ terms, e.g. $\mathcal{F}^{(p_i^2,\text{C}_1\text{S})}= (-p_1^2) ({\color{Green} x_1x_4}{\color{Red} x_6} + {\color{Green} x_1x_5}{\color{Red} x_6})$, and those carrying $s_{kl}$ dependence are $\mathcal{F}^{(q^2,R)}$ terms, e.g. $\mathcal{F}^{(q^2,\text{C}_1\text{S})}= (-s_{12}){\color{Green} x_1}{\color{Blue} x_2}{\color{Red} x_6} +(-s_{13}){\color{Green} x_1x_3x_5}$. Using eq.~(\ref{SS_and_C1S_scaling}), one can verify that all the terms in (\ref{UF_polynomial_SS_leading_terms}) scale as $\lambda^{-4}$ while all the terms in (\ref{UF_polynomial_C1S_leading_terms}) scale as $\lambda^{-3}$. Note that the remaining terms in (\ref{UF_polynomial_onshell_expansion_example}), which have been dropped in $\mathcal{P}_0^{(\text{SS})}$ and $\mathcal{P}_0^{(\text{C}_1\text{S})}$, would be sub-leading in $\lambda$.

In the two polynomials above, one observes that only certain combinations of {\color{Blue} hard}, {\color{Green} jet} and {\color{Red} soft} parameters can be consistent with aforementioned $\lambda$ scaling. Specifically, for $R=\text{SS}$ in (\ref{UF_polynomial_SS_leading_terms}), every term in $\mathcal{U}^{(R)}$ has two soft parameters, every term in $\mathcal{F}^{(p^2,R)}$ has two soft and one jet parameters, and every term in $\mathcal{F}^{(q^2,R)}$ has one soft and two jet parameters. For $R=\text{C}_1\text{S}$ in (\ref{UF_polynomial_C1S_leading_terms}), every term in $\mathcal{U}^{(R)}$ has one soft and one jet parameters, every term in $\mathcal{F}^{(p^2,R)}$ has one soft and two jet parameters. In this case we note that there are two types of terms in $\mathcal{F}^{(q^2,R)}$: one of them ($s_{12}{\color{Green} x_1}{\color{Blue} x_2}{\color{Red} x_6}$) has one soft, one jet and one hard parameters, while the other ($s_{13}{\color{Green} x_1x_3x_5}$) has three jet parameters.

Motivated by these observations, we consider a generic region with hard, jet and soft propagators. In order to study the terms of $\mathcal{U}^{(R)}$ and $\mathcal{F}^{(R)}$ and classify them into different types, we define the subgraphs $H$, $J_i$ ($i=1,\dots,K$) and $S$, as a partition of $G$, such that every edge and vertex in $G$ can be uniquely associated to exactly one of them. Furthermore, we associate edges and vertices to these subgraphs such that:
\begin{itemize}
    \item An edge is contained in $H$, $J_i$ ($i=1,\dots,K$) and $S$ respectively if and only if its momentum is hard, collinear to $p_i$, or soft.  
    \item A vertex is in $S$ if and only if it connects to soft edges only. A vertex is in $J_i$ if and only if it connects to at least one edge in $J_i$ and possibly edges in $S$ but no edges from $H$ nor $J_j$ with $j\neq i$. All the other vertices are in $H$.
\end{itemize}
Note that from the definitions above, the subgraph containing the endpoint of an edge may be different to the subgraph containing the edge itself. For example, the endpoints of a soft propagator can be jet vertices, and one endpoint of a jet propagator can be a hard vertex. Thus, what we are calling a subgraph in this section is subtly different from the strict notion of a subgraph in graph theory.

\begin{figure}[t]
\centering
\begin{subfigure}[b]{0.32\textwidth}
\centering
    \resizebox{0.6\textwidth}{!}{
\begin{tikzpicture}[line width = 0.6, font=\large, mydot/.style={circle, fill, inner sep=.7pt}]
\node[draw=Blue,circle,minimum size=1cm,fill=Blue!50] (h) at (6,8){};
\node at (h) {$H$};
\node (q1) at (4.5,9.5) {};
\node (q1p) at (4.7,9.7) {};
\node (qn) at (7.5,9.5) {};
\node (qnp) at (7.3,9.7) {};
\draw (q1) edge [color=Blue] (h) node [] {};
\draw (q1p) edge [color=Blue] (h) node [left] {$q_1$};
\draw (qn) edge [color=Blue] (h) node [] {};
\draw (qnp) edge [color=Blue] (h) node [right] {$q_L$};
\path (q1)-- node[mydot, pos=.333] {} node[mydot] {} node[mydot, pos=.666] {}(qn);
\end{tikzpicture}
}
\vspace{3em}
\vspace{3pt}
\caption{A hard subgraph $H$}
\label{H}
\end{subfigure}
\begin{subfigure}[b]{0.32\textwidth}
\centering
    \resizebox{0.2\textwidth}{!}{
\begin{tikzpicture}[line width = 0.6, font=\large, mydot/.style={circle, fill, inner sep=.7pt}]
\node(h) at (2,7){};
\node[draw=LimeGreen,ellipse,minimum height=3cm, minimum width=1.1cm,fill=LimeGreen!50] (j2) at (2,4.5){};
\node at (j2) {$J_i$};
\node (p2) at (2,2) {};
\draw[fill,thick,color=magenta] (h) circle (4pt);
\path (h) edge [color=LimeGreen] [double,double distance=3pt] (j2) {};
\draw (p2) edge [color=LimeGreen] (j2) node [below] {$p_i$};
\end{tikzpicture}
}
\vspace{3pt}
\caption{A contracted jet graph $\widetilde{J}_i$}
\label{Jtilde}
\end{subfigure}
\hfill
\begin{subfigure}[b]{0.32\textwidth}
\centering
    \resizebox{0.65\textwidth}{!}{
\begin{tikzpicture}[line width = 0.6, font=\large, mydot/.style={circle, fill, inner sep=.7pt}]
\node (h) at (3,7){};
\node[dashed,draw=Rhodamine,circle,minimum size=1.8cm,fill=Rhodamine!50] (s) at (3,2){};
\node (ji) at (3.5,4.5){};
\node (jn) at (4.5,4.5){};
\node at (s) {$S$};
\draw[fill,thick,color=magenta] (h) circle (3pt);
\draw (s) edge [dashed,double,bend left = 60,color=Rhodamine] (h) node [right] {};
\draw (s) edge [dashed,double,bend left = 15,color=Rhodamine] (h) {};
\draw (s) edge [dashed,double,bend right = 15,color=Rhodamine] (h) {};
\draw (s) edge [dashed,double,bend right = 60,color=Rhodamine] (h) {};
\path (ji)-- node[mydot, pos=.15] {} node[mydot] {} node[mydot, pos=.85] {}(jn);
\end{tikzpicture}
}
\vspace{3pt}
\caption{A contracted soft graph $\widetilde{S}$}
\label{Stilde}
\end{subfigure}
\caption{The $H$, $\widetilde{J}_i$ and $\widetilde{S}$ graphs. The pink vertex in the $\widetilde{J}_i$ (or $\widetilde{S}$) graph is a newly introduced vertex, which identifies all the vertices of $H$ (or $H\cup J$) that $J_i$ (or $S$) is attached to. Note that what is contracted is the rest of the graph, not the $J$ or $S$ themselves.}
\label{HJStilde}
\end{figure}
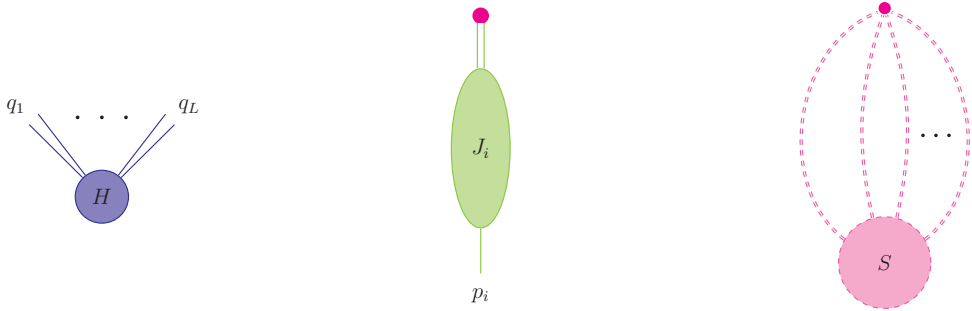

Next, let us define the \emph{contracted graphs} $\widetilde{J}_i$ and $\widetilde{S}$ that are constructed from the subgraphs $J_i$ and $S$. The contracted soft subgraph $\widetilde{S}$ is obtained from $S$ by contracting to a point all the vertices of $H$ and $J$ into which soft momenta flow, and replacing them by a new soft vertex, which we will refer to as the \emph{auxiliary soft vertex}. Similarly, $\widetilde{J}_i$ is obtained from $J_i$ by contracting to a point all the vertices of $H$ into which the $J_i$ momenta flow, and replacing them by a new vertex of $J_i$ (the \emph{auxiliary $J_i$ vertex}).
In figure~\ref{HJStilde} we show an illustration of the hard subgraphs $H$ and the contracted subgraphs $\widetilde{J}_i$ and $\widetilde{S}$ that would be obtained from the generic scattering graph shown in figure~\ref{classical_picture}.
By definition, for each of the contracted graphs $\widetilde{S}$ and $\widetilde{J}_i$ ($i=1,\dots,K$), there is one more vertex than the corresponding subgraph $S$ and $J_i$. The contracted graphs can also be viewed as follows: $\widetilde{S}=G/(H\cup J)$ and $\widetilde{J}_i = (G\setminus S)/(H \bigcup_{k\neq i} J_k)$.

The introduction of the contracted graphs is convenient for the following reason. We observe that the total\footnote{Note that the number of loops in each of the graphs separately is not the same between the original and the contracted graphs. For example, $\widetilde{S}$ has more loops than $S$.} number of the independent hard, jet and soft loop momenta in $\I(G)$ is equal to the number of loops in $H$, $\widetilde{J}$ and $\widetilde{S}$, respectively.
It then follows that
\begin{eqnarray}
L(G) =&& L(G\setminus S) + L(\widetilde{S})\nonumber\\
=&& L(H) +\sum_{i=1}^K L(\widetilde{J}_i) + L(\widetilde{S}).
\label{loop_relation}
\end{eqnarray}

Another benefit of defining the contracted subgraphs is that we can use this concept to describe the leading polynomials. As we will see below, the polynomials $\mathcal{U}^{(R)}$ and $\mathcal{F}^{(R)}$ can be classified into certain types, according to the structures of $H$, $\widetilde{J}$ and $\widetilde{S}$.

Any term of $\mathcal{P}_0^{(R)}$ is obtained either from a spanning tree or a spanning 2-tree of $G$, which we denote as $T$. Further, let there be $n_H^{}$ ($n_J^{}$, $n_S^{}$) hard (jet, soft) propagators in the complement of $T$. The following fundamental constraints on the values of $n_H^{}$ and $n_S^{}$ can be observed:
\begin{eqnarray}
n_H^{} \geqslant L(H),\qquad\quad n_S^{} \leqslant L(\widetilde{S}).
\label{leading_terms_fundamental_constraints}
\end{eqnarray}
The first inequality $n_H^{} \geqslant L(H)$ must hold, because we need to remove at least $L(H)$ propagators in the hard subgraph to turn $G$ into a spanning tree. To arrive at \hbox{$n_S^{}\leqslant L(\widetilde{S})$} let us construct a proof by contradiction. To this end we assume that $n_S>L(\widetilde{S})$ and show that it is inconsistent. First of all, $n_S>L(\widetilde{S})$ implies that the graph $\gamma_{\widetilde S}\equiv \widetilde{S} \cap T$ would be disconnected.\footnote{This can be understood from Euler's theorem. For $\widetilde S$ we have $V(\widetilde S)-N(\widetilde S)+L(\widetilde S)=1$. After removing $L(\widetilde S)+x$ edges from $\widetilde S$ we obtain $\gamma_{\widetilde S}$ for which Euler's theorem yields $V(\widetilde S)-(N(\widetilde S)-L(\widetilde S)-x)+L(\gamma_{\widetilde S})=k$ and $k$ denotes the number of disjoint connected components of $\gamma_{\widetilde S}$ . Combining these equations one obtains $k=L(\gamma_{\widetilde S})+x+1>1$  as long as $x>0$ proving that $\gamma_{\widetilde S}$ is disconnected.} One of its components consists exclusively of soft edges which do not attach to $H$ or $J$ in $G$.
Since by definition (see point~\textit{3} of proposition~\ref{proposition-region_vectors_are_hard_and_infrared}), the soft subgraph $S$ does not attach to any external momenta, momentum conservation of this disconnected component implies that the total momentum flowing into it vanish, and hence it would not contribute to the Symanzik polynomials. We have thus established eq.~(\ref{leading_terms_fundamental_constraints}).

\begin{figure}[t]
  \centering
  \begin{subfigure}[b]{0.49\textwidth}
    \centering
        \resizebox{\textwidth}{!}{
\begin{tikzpicture}[line width = 0.6, font=\large, mydot/.style={circle, fill, inner sep=.7pt}]

\node[draw=Blue,circle,minimum size=1cm,fill=Blue!50] (h) at (6,8){};
\node[dashed, draw=Rhodamine,circle,minimum size=1.8cm,fill=Rhodamine!50] (s) at (4,4.1){};
\node[draw=Green,ellipse,minimum height=3cm, minimum width=1.1cm,fill=Green!50,rotate=-52] (j1) at (2,5){};
\node[draw=LimeGreen,ellipse,minimum height=3cm, minimum width=1.1cm,fill=LimeGreen!50] (j2) at (6,3.5){};
\node[draw=olive,ellipse,minimum height=3cm, minimum width=1.1cm,fill=olive!50,rotate=52] (jn) at (10,5){};

\node at (h) {ST};
\node at (s) {ST};
\node at (j1) {ST};
\node at (j2) {ST};
\node at (jn) {ST};

\path (h) edge [double,double distance=2pt,color=Green] (j1) {};
\path (h) edge [double,double distance=2pt,color=LimeGreen] (j2) {};
\path (h) edge [double,double distance=2pt,color=olive] (jn) {};

\draw (s) edge [dashed,double,color=Rhodamine] (h) node [right] {};
\draw (s) edge [dashed,double,color=Rhodamine] (j1) {};
\draw (s) edge [dashed,double,color=Rhodamine] (j2) {};
\draw (s) edge [dashed,double,color=Rhodamine,bend left = 35] (jn) {};

\node (q1) at (4.5,9.5) {};
\node (q1p) at (4.7,9.7) {};
\node (qn) at (7.5,9.5) {};
\node (qnp) at (7.3,9.7) {};
\draw (q1) edge [color=Blue] (h) node [] {};
\draw (q1p) edge [color=Blue] (h) node [left] {$q_1$};
\draw (qn) edge [color=Blue] (h) node [] {};
\draw (qnp) edge [color=Blue] (h) node [right] {$q_L$};

\node (p1) at (0,3.5) {};
\node (p2) at (6,1) {};
\node (pn) at (12,3.5) {};
\draw (p1) edge [color=Green] (j1) node [below] {$p_1$};
\draw (p2) edge [color=LimeGreen] (j2) node [below] {$p_i$};
\draw (pn) edge [color=olive] (jn) node [below] {$p_K$};

\path (p2)-- node[mydot, pos=.333] {} node[mydot] {} node[mydot, pos=.666] {}(pn);
\path (p1)-- node[mydot, pos=.333] {} node[mydot] {} node[mydot, pos=.666] {}(p2);
\path (q1)-- node[mydot, pos=.333] {} node[mydot] {} node[mydot, pos=.666] {}(qn);

\end{tikzpicture}
}
        \vspace{-2em}
        \vspace{5pt}
    \caption{}
    \label{fig:UF_terms_graphic_description_1}
    \vspace{3pt}
  \end{subfigure}
  \hfill
  \begin{subfigure}[b]{0.49\textwidth}
  \centering
        \resizebox{\textwidth}{!}{
\begin{tikzpicture}[line width = 0.6, font=\large, mydot/.style={circle, fill, inner sep=.7pt}]

\node[draw=Blue,circle,minimum size=1cm,fill=Blue!50] (h) at (6,8){};
\node[dashed, draw=Rhodamine,circle,minimum size=1.8cm,fill=Rhodamine!50] (s) at (4,4.1){};
\node[draw=Green,ellipse,minimum height=3cm, minimum width=1.1cm,fill=Green!50,rotate=-52] (j1) at (2,5){};
\node[draw=LimeGreen,ellipse,minimum height=3cm, minimum width=1.1cm,fill=LimeGreen!50] (j2) at (6,3.5){};
\node[draw=olive,ellipse,minimum height=3cm, minimum width=1.1cm,fill=olive!50,rotate=52] (jn) at (10,5){};

\node at (h) {ST};
\node at (s) {ST};
\node at (j1) {ST};
\node at (j2) {S2T};
\node at (jn) {ST};

\path (h) edge [double,double distance=2pt,color=Green] (j1) {};
\path (h) edge [double,double distance=2pt,color=LimeGreen] (j2) {};
\path (h) edge [double,double distance=2pt,color=olive] (jn) {};

\draw (s) edge [dashed,double,color=Rhodamine] (h) node [right] {};
\draw (s) edge [dashed,double,color=Rhodamine] (j1) {};
\draw (s) edge [dashed,double,color=Rhodamine] (j2) {};
\draw (s) edge [dashed,double,color=Rhodamine,bend left = 35] (jn) {};

\node (q1) at (4.5,9.5) {};
\node (q1p) at (4.7,9.7) {};
\node (qn) at (7.5,9.5) {};
\node (qnp) at (7.3,9.7) {};
\draw (q1) edge [color=Blue] (h) node [] {};
\draw (q1p) edge [color=Blue] (h) node [left] {$q_1$};
\draw (qn) edge [color=Blue] (h) node [] {};
\draw (qnp) edge [color=Blue] (h) node [right] {$q_L$};

\node (p1) at (0,3.5) {};
\node (p2) at (6,1) {};
\node (pn) at (12,3.5) {};
\draw (p1) edge [color=Green] (j1) node [below] {$p_1$};
\draw (p2) edge [color=LimeGreen] (j2) node [below] {$p_i$};
\draw (pn) edge [color=olive] (jn) node [below] {$p_K$};

\path (p2)-- node[mydot, pos=.333] {} node[mydot] {} node[mydot, pos=.666] {}(pn);
\path (p1)-- node[mydot, pos=.333] {} node[mydot] {} node[mydot, pos=.666] {}(p2);
\path (q1)-- node[mydot, pos=.333] {} node[mydot] {} node[mydot, pos=.666] {}(qn);

\end{tikzpicture}
}
        \vspace{-2em}
        \vspace{5pt}
    \caption{}
    \label{fig:UF_terms_graphic_description_2}
    \vspace{3pt}
  \end{subfigure}
  \begin{subfigure}[b]{0.49\textwidth}
  \centering
        \resizebox{\textwidth}{!}{
\begin{tikzpicture}[line width = 0.6, font=\large, mydot/.style={circle, fill, inner sep=.7pt}]

\node[draw=Blue,circle,minimum size=1cm,fill=Blue!50] (h) at (6,8){};
\node[dashed, draw=Rhodamine,circle,minimum size=1.8cm,fill=Rhodamine!50] (s) at (4,4.1){};
\node[draw=Green,ellipse,minimum height=3cm, minimum width=1.1cm,fill=Green!50,rotate=-52] (j1) at (2,5){};
\node[draw=LimeGreen,ellipse,minimum height=3cm, minimum width=1.1cm,fill=LimeGreen!50] (j2) at (6,3.5){};
\node[draw=olive,ellipse,minimum height=3cm, minimum width=1.1cm,fill=olive!50,rotate=52] (jn) at (10,5){};

\node at (h) {S2T};
\node at (s) {ST};
\node at (j1) {ST};
\node at (j2) {ST};
\node at (jn) {ST};

\path (h) edge [double,double distance=2pt,color=Green] (j1) {};
\path (h) edge [double,double distance=2pt,color=LimeGreen] (j2) {};
\path (h) edge [double,double distance=2pt,color=olive] (jn) {};

\draw (s) edge [dashed,double,color=Rhodamine] (h) node [right] {};
\draw (s) edge [dashed,double,color=Rhodamine] (j1) {};
\draw (s) edge [dashed,double,color=Rhodamine] (j2) {};
\draw (s) edge [dashed,double,color=Rhodamine,bend left = 35] (jn) {};

\node (q1) at (4.5,9.5) {};
\node (q1p) at (4.7,9.7) {};
\node (qn) at (7.5,9.5) {};
\node (qnp) at (7.3,9.7) {};
\draw (q1) edge [color=Blue] (h) node [] {};
\draw (q1p) edge [color=Blue] (h) node [left] {$q_1$};
\draw (qn) edge [color=Blue] (h) node [] {};
\draw (qnp) edge [color=Blue] (h) node [right] {$q_L$};

\node (p1) at (0,3.5) {};
\node (p2) at (6,1) {};
\node (pn) at (12,3.5) {};
\draw (p1) edge [color=Green] (j1) node [below] {$p_1$};
\draw (p2) edge [color=LimeGreen] (j2) node [below] {$p_i$};
\draw (pn) edge [color=olive] (jn) node [below] {$p_K$};

\path (p2)-- node[mydot, pos=.333] {} node[mydot] {} node[mydot, pos=.666] {}(pn);
\path (p1)-- node[mydot, pos=.333] {} node[mydot] {} node[mydot, pos=.666] {}(p2);
\path (q1)-- node[mydot, pos=.333] {} node[mydot] {} node[mydot, pos=.666] {}(qn);

\end{tikzpicture}
}
        \vspace{-2em}
        \vspace{5pt}
    \caption{}
    \label{fig:UF_terms_graphic_description_3}
  \end{subfigure}
  \hfill
    \begin{subfigure}[b]{0.49\textwidth}
      \vspace{0.5em}
  \centering
        \resizebox{\textwidth}{!}{
\begin{tikzpicture}[line width = 0.6, font=\large, mydot/.style={circle, fill, inner sep=.7pt}]

\node[draw=Blue,circle,minimum size=1cm,fill=Blue!50] (h) at (6,8){};
\node[dashed, draw=Rhodamine,circle,minimum size=1.8cm,fill=Rhodamine!50] (s) at (4,4.1){};
\node[draw=Green,ellipse,minimum height=3cm, minimum width=1.1cm,fill=Green!50,rotate=-52] (j1) at (2,5){};
\node[draw=LimeGreen,ellipse,minimum height=3cm, minimum width=1.1cm,fill=LimeGreen!50] (j2) at (6,3.5){};
\node[draw=olive,ellipse,minimum height=3cm, minimum width=1.1cm,fill=olive!50,rotate=52] (jn) at (10,5){};

\node at (h) {ST};
\node at (s) {1-loop};
\node at (j1) {S2T};
\node at (j2) {ST};
\node at (jn) {S2T};

\path (h) edge [double,double distance=2pt,color=Green] (j1) {};
\path (h) edge [double,double distance=2pt,color=LimeGreen] (j2) {};
\path (h) edge [double,double distance=2pt,color=olive] (jn) {};

\draw (s) edge [dashed,double,color=Rhodamine] (h) node [right] {};
\draw (s) edge [dashed,double,color=Rhodamine] (j1) {};
\draw (s) edge [dashed,double,color=Rhodamine] (j2) {};
\draw (s) edge [dashed,double,color=Rhodamine,bend left = 35] (jn) {};

\node (q1) at (4.5,9.5) {};
\node (q1p) at (4.7,9.7) {};
\node (qn) at (7.5,9.5) {};
\node (qnp) at (7.3,9.7) {};
\draw (q1) edge [color=Blue] (h) node [] {};
\draw (q1p) edge [color=Blue] (h) node [left] {$q_1$};
\draw (qn) edge [color=Blue] (h) node [] {};
\draw (qnp) edge [color=Blue] (h) node [right] {$q_L$};

\node (p1) at (0,3.5) {};
\node (p2) at (6,1) {};
\node (pn) at (12,3.5) {};
\draw (p1) edge [color=Green] (j1) node [below] {$p_1$};
\draw (p2) edge [color=LimeGreen] (j2) node [below] {$p_i$};
\draw (pn) edge [color=olive] (jn) node [below] {$p_K$};

\path (p2)-- node[mydot, pos=.333] {} node[mydot] {} node[mydot, pos=.666] {}(pn);
\path (p1)-- node[mydot, pos=.333] {} node[mydot] {} node[mydot, pos=.666] {}(p2);
\path (q1)-- node[mydot, pos=.333] {} node[mydot] {} node[mydot, pos=.666] {}(qn);

\end{tikzpicture}
}
        \vspace{-2em}
        \vspace{5pt}
    \caption{}
    \label{fig:UF_terms_graphic_description_4}
  \end{subfigure}
  \caption{The graphic descriptions of the terms of $\mathcal{U}^{(R)}(\x)$ and $\mathcal{F}^{(R)}(\x;\s)$ with $R$ being an infrared region. The lines connecting different blobs represent any number of propagators, and the notation ST (S2T) stands for spanning trees (spanning 2-trees). The labels (ST, S2T, one-loop) imply that the contracted graphs $H$, $\widetilde{J}$ and $\widetilde{S}$ is respectively a (spanning tree, spanning 2-tree, one-loop graph) in the corresponding terms. Any of the terms in $\mathcal{P}_0^{(R)}$ is of one of the following four types: $(a)$ the terms of $\mathcal{U}^{(R)}(\x)$; $(b)$ the terms of $\mathcal{F}^{(p_i^2,R)}(\x;\s)$, note that the spanning 2-tree structure can be within any one of the jet subgraphs; $(c)$ one possibility for the terms of $\mathcal{F}^{(q^2,R)}(\x;\s)$, denoted as the $\mathcal{F}_\text{I}^{(q^2,R)}$ terms, where $H\cap T$ is a spanning 2-tree of $H$, identifying an off-shell momentum $q^\mu$ flowing between its components;
  $(d)$ another possibility for the terms of $\mathcal{F}^{(q^2,R)}(\x;\s)$, denoted as the $\mathcal{F}_\text{II}^{(q^2,R)}$ terms, where for two of the jet graphs $\widetilde{J}_i$ and $\widetilde{J}_j$ ($i=1$ and $j=K$ in the figure), $\widetilde{J}_i\cap T$ and $\widetilde{J}_j\cap T$ are spanning 2-trees of $\widetilde{J}_i$ and $\widetilde{J}_j$ respectively. While in $(a)$, $(b)$ and $(c)$ the graph $\widetilde{S}\cap T$ is a spanning tree of $\widetilde{S}$, in $(d)$ this graph has (exactly) one loop.}
  \label{UF_terms_graphic_description}
\end{figure}
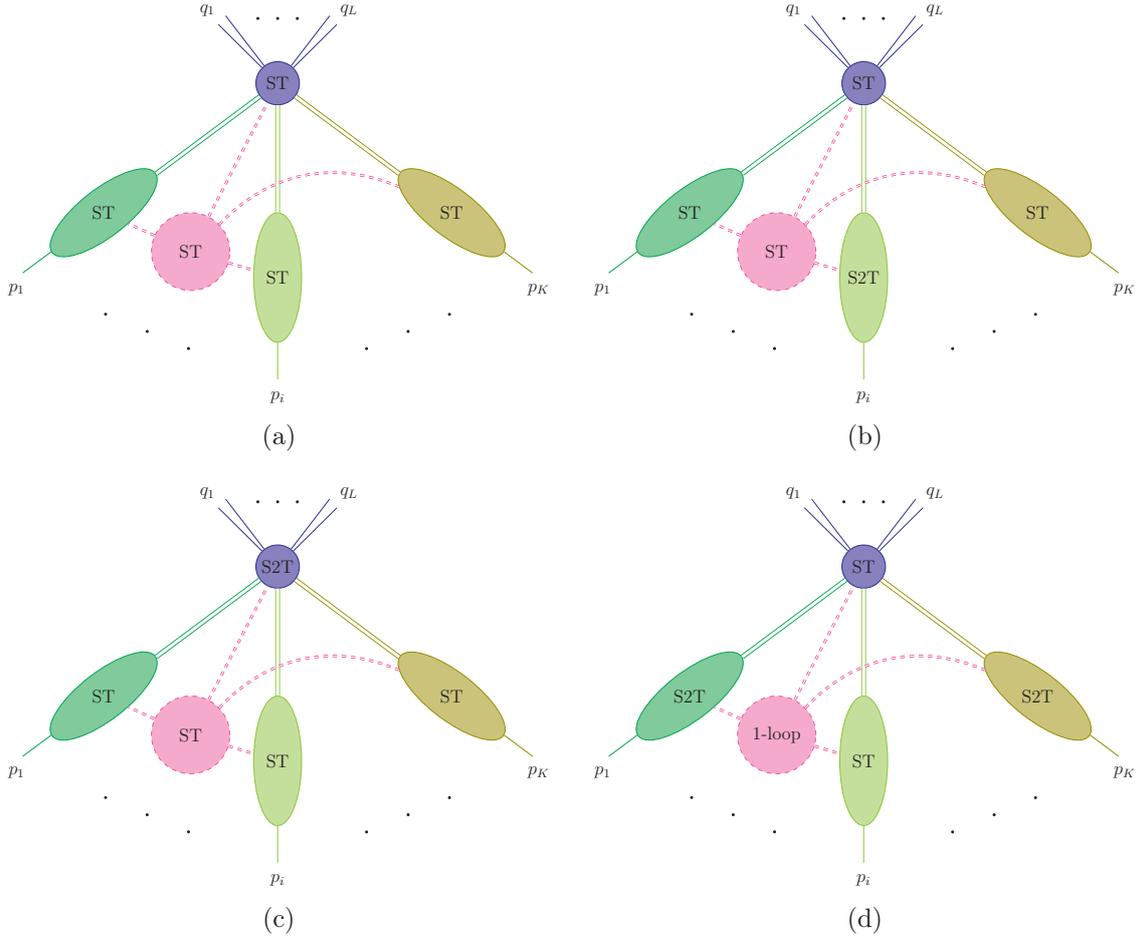

With the observations above, one can derive the generic form of the terms in $\mathcal{P^{(R)}}$ by considering all the possible cases.
We first focus on any term of $\mathcal{U}^{(R)}$, denoting the corresponding spanning tree $T^1$. The $\mathcal{U}^{(R)}$ polynomial is a homogeneous, degree $L(G)$ polynomial, which implies that $n_H + n_J + n_S = L(H) + \sum_{i=1}^K L(\widetilde{J}_i) + L(\widetilde{S})$. Recall that a facet contains the points that minimise $\boldsymbol{r}\cdot \boldsymbol{v}_R$, as follows from eq.~(\ref{eq:LPscalingsubstitution}).
Because for the $\mathcal{U}$ terms $\boldsymbol{r}\cdot \boldsymbol{v}_R = \hat{\boldsymbol{r}}\cdot \boldsymbol{u}_R$ (i.e. $r_{i,N+1}=0$), $\boldsymbol{r}\cdot \boldsymbol{v}_R$ corresponds to summing the subset of entries of $\boldsymbol{v}_R$ associated with edges $e$, for which $r_e=1$ (see eq.~(\ref{LP_polynomial_rescaling_momenta})).
Since the hard, jet and soft propagators have a corresponding $0$, $-1$, or $-2$ entry in $\boldsymbol{v}_R$ respectively, in order to achieve the minimum of $\boldsymbol{r}\cdot \boldsymbol{v}_R$, the value of $n_H^{}$ must be minimised while the value of $n_S^{}$ must be maximised. With this in mind, eq.~(\ref{leading_terms_fundamental_constraints}) implies:
\begin{eqnarray}
n_H^{} = L(H),\quad n_J^{} = L(\widetilde{J}),\quad n_S^{}= L(\widetilde{S}).
\label{minimal_u_term_condition}
\end{eqnarray}
This implies that each of the subgraphs, $H\cap T^1$, $\widetilde{J}_i\cap T^1$ for each $i\in \{1,\dots,K\}$, and $\widetilde{S}\cap T^1$ are spanning trees of $H$, $\widetilde{J}_i$ and $\widetilde{S}$, respectively. This configuration is depicted in figure~\ref{UF_terms_graphic_description}$(a)$. Here the notation ST stands for a spanning tree. Note that it is the \emph{contracted} graphs that should be considered, namely, $\gamma\cap T^1$ is a spanning tree of $\gamma$, for $\gamma= H, \widetilde{J}_1,\dots,\widetilde{J}_K, \widetilde{S}$.

For the analysis that follows, it is useful to examine the structure of the soft subgraph in more detail. Suppose there are $n$ connected components of $S$, denoted by $S_1,\dots,S_n$, corresponding to the $n$ 1VI components of $\widetilde{S}$, i.e. $\widetilde{S}_1,\cdots,\widetilde{S}_n$, which all share the auxiliary soft vertex. As illustrated above, $\widetilde{S} \cap T^1$ is a spanning tree of $\widetilde{S}$. Then recall corollary~\ref{theorem-Uterms_space_dimensionality_corollary1}, $\widetilde{S}_i \cap T^1$ is a spanning tree of $\widetilde{S}_i$ for $i=1,\dots,n$. Eq.~(\ref{minimal_u_term_condition}) can hence be rewritten in a more precise way as
\begin{eqnarray}
\mathcal{U}^{(R)}(\x):&&\ n_H^{} = L(H),\nonumber \\
&&\ n_{J_i}^{} = L(\widetilde{J}_i) \quad (i=1,\dots,K), \nonumber\\
&&\ n_{S_k}^{}= L(\widetilde{S}_k) \quad (k=1,\dots,n).
\label{minimal_u_term_condition_precise}
\end{eqnarray}

Using the graph-theoretical knowledge that the combination of any chosen spanning tree of $H$, $\widetilde{J}_i$ and $\widetilde{S}$ is also a spanning tree of $G$, along with eq.~(\ref{minimal_u_term_condition_precise}), we can deduce a useful property: upon denoting the parameters $\x$ that are associated with edges in $H$, $J_i$, and $S = \cup_{k=1}^n S_k$ by $\x^{[H]}$, $\x^{[J_i]}$ and $\x^{[S]}$ respectively, the polynomial $\mathcal{U}^{(R)}(\x)$ can be factorised as follows
\begin{eqnarray}
\mathcal{U}^{(R)}(\x) = \mathcal{U}_H(\x^{[H]}) \cdot \Big( \prod_{i=1}^K \mathcal{U}_{J_i}(\x^{[J_i]}) \Big) \cdot \mathcal{U}_S(\x^{[S]}),
\label{leading_Uterms_factorise}
\end{eqnarray}
where $\mathcal{U}_H(\x^{[H]})$, $\mathcal{U}_{J_i}(\x^{[J_i]})$ and $\mathcal{U}_S(\x^{[S]})= \prod_{k=1}^n \mathcal{U}_{S_k}(\x^{[S_k]})$ are the polynomials that contain only the parameters $\x^{[H]}$, $\x^{[J_i]}$ and $\x^{[S]}$, and correspond to the spanning trees of the graphs $H$, $\widetilde{J}_i$ and $\widetilde{S}$ respectively.
For example, in eq.~(\ref{UF_polynomial_SS_leading_terms}) we have $\mathcal{U}_H(\x^{[H]}) = \mathcal{U}_{J_i}(\x^{[J_i]}) =1$, and $\mathcal{U}_S(\x^{[S]})= x_4x_5 + x_4x_6 + x_5x_6$. Meanwhile in eq.~(\ref{UF_polynomial_C1S_leading_terms}) we can see that $\mathcal{U}_H(\x^{[H]}) =1$, $\mathcal{U}_{J_1}(\x^{[J_1]})= x_1+x_4+x_5$, $\mathcal{U}_{J_i}(\x^{[J_i]})=1$ ($i\neq 1$), and $\mathcal{U}_S(\x^{[S]})= x_6$.

Note that eq.~(\ref{minimal_u_term_condition}) can also be viewed as a result of the uniqueness of minimum spanning trees in graph theory, i.e. by associating a weight to each edge of a given graph, the set of weights in a minimum spanning tree of this graph is uniquely determined. In the case under consideration, these weights are $0$, $-1$ and $-2$ for hard, jet and soft edges, respectively. The minimum of $\boldsymbol{r}\cdot \boldsymbol{v}_R$, where $\boldsymbol{r}$ corresponds to any term in $\mathcal{P}_0^{(R)}(\x;\s)$, can then be directly read from~(\ref{minimal_u_term_condition}):
\begin{eqnarray}
\text{min}\left( \boldsymbol{r}\cdot \boldsymbol{v}_R \right)= -2L(\widetilde{S}) -L(\widetilde{J}).
\label{min_value_rdotv}
\end{eqnarray}

\bigbreak
Now we study the terms in $\mathcal{F}^{(R)}(\x;\s)$. We denote their corresponding spanning 2-trees as $T^2$. Here we have $n_H^{}+n_J^{}+n_S^{}= L(H)+L(\widetilde{J})+L(\widetilde{S})+1$. Using the fundamental constraint, eq.~(\ref{leading_terms_fundamental_constraints}), all the possible values of $n_H^{}$, $n_J^{}$ and $n_S^{}$ can be summarised as follows:
\begin{eqnarray}
n_H^{}= L(H)+k_1,\ \ n_J^{}= L(\widetilde{J})+k_2+1-k_1,\ \ n_S^{}=L(\widetilde{S})-k_2,
\label{condition_generic_F_terms}
\end{eqnarray}
where $k_1$ and $k_2$ are non-negative integers ($k_1,k_2\in \mathbb{N}_0$). Below we discuss three cases of $k_1$ and $k_2$, and show that they enumerate all possible $\mathcal{F}^{(R)}$~terms, which correspond to, respectively, figures~\ref{fig:UF_terms_graphic_description_2}, \ref{fig:UF_terms_graphic_description_3} and \ref{fig:UF_terms_graphic_description_4}.

\begin{enumerate}
    \item [\textbf{I. }]$\boldsymbol{k_1=k_2=0}$ \textbf{(figure~\ref{fig:UF_terms_graphic_description_2}).} In this case, we have
    \begin{eqnarray}
    n_H^{} = L(H),\ n_J^{} = L(\widetilde{J})+1,\ n_S^{}= L(\widetilde{S}).
    \label{minimal_fp2_term_condition}
    \end{eqnarray}
    This relation implies that $H\cap T^2$ and $\widetilde{S}\cap T^2$ are spanning trees of $H$ and $\widetilde{S}$ respectively. The subgraph $\widetilde{J}\cap T^2$, on the other hand, is a spanning 2-tree of $\widetilde{J}$. More precisely, for a specific jet $i\in\{1,\dots,K\}$, $\widetilde{J}_i \cap T^2$ is a spanning 2-tree of $\widetilde{J}_i$, while all the subgraphs $\widetilde{J}_j \cap T^2$ with $j\neq i$ are spanning trees of $\widetilde{J}_j$. This configuration is then depicted in figure~\ref{fig:UF_terms_graphic_description_2}, where the notation S2T is the abbreviation of spanning 2-trees. In this configuration, one component of $T^2$ only has a single nearly on-shell external momentum $p_i^\mu$, hence it corresponds to $\mathcal{F}^{(p_i^2)}(\x;\s)$. Similar to the analysis above (\ref{minimal_u_term_condition_precise}), $\widetilde{S}_k \cap T^1$ is a spanning tree of $\widetilde{S}_k$ for $k=1,\dots,n$. Therefore, eq.~(\ref{minimal_fp2_term_condition}) can be written as:
    \begin{eqnarray}
    \mathcal{F}^{(p_i^2,R)}(\x;\s):
    &&\ n_H^{} = L(H),\nonumber \\
    &&\ n_{J_i}^{} = L(\widetilde{J}_i)+1,\quad n_{J_j}^{} = L(\widetilde{J}_j)\quad (j\neq i,\ j=1,\dots,K), \nonumber\\
    &&\ n_{S_k}^{}= L(\widetilde{S}_k)\quad (k=1,\dots,n).
    \label{minimal_fp2_term_condition_precise}
    \end{eqnarray}
    
    It remains to be shown that these terms are associated with the same minimum of $\boldsymbol{r}\cdot \boldsymbol{v}_R$ as in eq.~(\ref{min_value_rdotv}). In comparison with the terms of $\mathcal{U}^{(R)}(\x)$, which are characterised by (\ref{minimal_u_term_condition_precise}), there is an extra $-1$ from the jet parameter contribution to $\boldsymbol{r}\cdot \boldsymbol{v}_R$, and an extra $+1$ from the kinematic contribution, which is the $(N+1)$-th entry of $\boldsymbol{r}$ (see eq.~(\ref{eq:LPscalingsubstitution})). The terms of $\mathcal{F}^{(p_i^2,R)}(\x;\s)$ are thus described by figure~\ref{fig:UF_terms_graphic_description_2}.
    
    We note that a factorisation formula, similar to eq.~(\ref{leading_Uterms_factorise}), also applies to the  $\mathcal{F}^{(p_i^2,R)}$ terms, namely,
    \begin{eqnarray}
    \mathcal{F}^{(p_i^2,R)}(\x;\s) = \mathcal{U}_H(\x^{[H]}) \cdot \mathcal{F}_{J_i}^{(p_i^2)}(\x^{[J_i]}) \cdot \Big( \prod_{j\neq i}^K\mathcal{U}_{J_j}(\x^{[J_j]}) \Big) \cdot\mathcal{U}_S(\x^{[S]}).
    \label{leading_Fp2terms_factorise}
    \end{eqnarray}
    The polynomials $\mathcal{U}_H(\x^{[H]})$, $\mathcal{U}_{J_j}(\x^{[J_j]})$ and $\mathcal{U}_S(\x^{[S]})$ correspond to the spanning trees of $H$, $\widetilde{J}_i$ and $\widetilde{S}$, respectively. The polynomial $\mathcal{F}^{(p_i^2)}(\x^{[J_i]})$ depends on the parameters $x^{[J_i]}$ only, and corresponds to the spanning 2-trees of $\widetilde{J}_i$. Moreover, the momentum flowing between the two connected components is exactly $p_i^\mu$. As in the case of eq.~(\ref{leading_Uterms_factorise}), this factorisation reflects a graph-theoretical property: by combining any set of spanning trees of $H$, $\widetilde{J}_j$ ($j\neq i$) and $\widetilde{S}$, together with a $p_i$ spanning 2-tree of $\widetilde{J}_i$, one obtains a $p_i$ spanning 2-tree of $G$.

    \item [\textbf{II. }]$\boldsymbol{k_1=1,\ k_2=0}$ \textbf{(figure~\ref{fig:UF_terms_graphic_description_3}).} In this case, we have
    \begin{eqnarray}
    n_H^{} = L(H)+1,\ n_J^{} = L(\widetilde{J}),\ n_S^{}= L(\widetilde{S}).
    \label{minimal_fs_term_condition1}
    \end{eqnarray}
    In the example of eq.~(\ref{UF_polynomial_C1S_leading_terms}) as we showed above, the term $s_{12}{\color{Green} x_1}{\color{Blue} x_2}{\color{Red} x_6}$ is of this type. More generically, for a given $\boldsymbol{r}$ corresponding to a term of this form, we first examine the value of $\boldsymbol{r}\cdot \boldsymbol{v}_R$, using eq.~(\ref{eq:LPscalingsubstitution}), obtaining
    \begin{eqnarray}
    \boldsymbol{r}\cdot \boldsymbol{v}_R = -2L(\widetilde{S}) -L(\widetilde{J}) +r_{N+1},
    \end{eqnarray}
    where $r_{N+1}=0$ or $1$ is the contribution from the kinematic prefactor of the term (i.e. a factor of $\lambda^0$ or $\lambda^1$ generated by the rescaling of the kinematic invariants $\s \rightarrow \lambda^{\w}\s$). In order that a term satisfying (\ref{minimal_fs_term_condition1}) is in $\mathcal{F}^{(R)}(\x;\s)$, we must have $r_{N+1}=0$, which means that $T^2$ corresponds to a term of $\mathcal{F}^{(q^2,R)}(\x;\s)$ in this case. We denote these terms as $\mathcal{F}_\text{I}^{(q^2,R)}(\x;\s)$. In each $T^2$ that corresponds to a term in $\mathcal{F}_\text{I}^{(q^2,R)}(\x;\s)$, the graph $H\cap T^2$ is a spanning 2-tree of $H$, such that the momentum flowing between its components of $T^2$ is $q^\mu$, which is off shell. In contrast, the remaining subgraphs in $T^2$ are spanning trees. In particular, $\widetilde{J}_i\cap T^2$ is a spanning tree of $\widetilde{J}_i$ for any $i=1,\dots,K$, and $\widetilde{S}_k \cap T^2$ is a spanning tree of $\widetilde{S}_k$ for any $k=1,\dots,n$. With this analysis, we may rewrite (\ref{minimal_fs_term_condition1}) as
    \begin{eqnarray}
    \mathcal{F}_\text{I}^{(q^2,R)}(\x;\s):
    &&\ n_H^{} = L(H)+1,\nonumber \\
    &&\ n_{J_i}^{} = L(\widetilde{J}_i)\quad (i=1,\dots,K), \nonumber\\
    &&\ n_{S_k}^{}= L(\widetilde{S}_k)\quad (k=1,\dots,n).
    \label{minimal_fs_term_condition1_precise}
    \end{eqnarray}
    The configuration of $T^2$ in this case is depicted in figure~\ref{fig:UF_terms_graphic_description_3}. The polynomial of $\mathcal{F}_\text{I}^{(q^2)}$ can be factorised as follows
    \begin{eqnarray}
    \mathcal{F}_\textup{I}^{(q^2)}(\x;\s) = \mathcal{F}_{H\cup J}^{(q^2)}(\x^{[H]},\x^{[J]}) \cdot \mathcal{U}_S(\x^{[S]}),
    \label{leading_first_Fq2terms_factorise}
    \end{eqnarray}
    where the polynomial $\mathcal{F}_{H\cup J}^{(q^2)}(\x^{[H]},\x^{[J]})$ consists of all the spanning 2-trees of $H\cup J$ such that the momentum flowing between its components is $q^\mu$. In comparison with the previous factorisation formulas, eqs.~(\ref{leading_Uterms_factorise}) and (\ref{leading_Fp2terms_factorise}), we note that $\mathcal{F}_{H\cup J}^{(q^2)}$ cannot be factorised any further.

    \item [\textbf{III. }]$\boldsymbol{k_1=0,\ k_2=1}$ \textbf{(figure~\ref{fig:UF_terms_graphic_description_4}).} In this case, we have
    \begin{eqnarray}
    n_H^{} = L(H),\ n_J^{} = L(\widetilde{J})+2,\ n_S^{}= L(\widetilde{S})-1.
    \label{minimal_fs_term_condition2}
    \end{eqnarray}
    In the example of eq.~(\ref{UF_polynomial_C1S_leading_terms}) as we showed above, the term $s_{13}{\color{Green} x_1x_3x_5}$ is of this type. In generic cases, we first examine the value of $\boldsymbol{r}\cdot \boldsymbol{v}_R$ as above, which is
    \begin{eqnarray}
    \boldsymbol{r}\cdot \boldsymbol{v}_R = -2\big( L(\widetilde{S})-1 \big) -\big( L(\widetilde{J})+2 \big) +r_{N+1} =-2L(\widetilde{S}) -L(\widetilde{J}) +r_{N+1},
    \end{eqnarray}
    where $r_{N+1}=0$ or $1$ is the kinematic contribution. In order for a term satisfying eq.~(\ref{minimal_fs_term_condition2}) to appear in $\mathcal{F}^{(R)}(\x;\s)$, we must have $r_{N+1}=0$, which means that this case corresponds to another contribution to $\mathcal{F}^{(q^2,R)}(\x;\s)$, distinct from case \textbf{II} above. We denote the polynomial consisting of these terms as $\mathcal{F}_\text{II}^{(R)}(\x;\s)$.
    
    To understand the configuration of $T^2$ in this case, we first notice that eq.~(\ref{minimal_fs_term_condition2}) implies that $H\cap T^2$ is a spanning tree of $H$, and $\widetilde{S}\cap T^2$ contains one loop. If the auxiliary soft vertex were not part of this loop, then $T^2$ would have contained this loop as well, which contradicts $T^2$ being a spanning 2-tree. Therefore, the loop of $\widetilde{S}\cap T^2$ must correspond to a path $P$ that consists of soft propagators, joining two distinct vertices $A$ and $B$ in $H\cup J$. From the definition of $\widetilde{S}$, this path becomes a loop of $\widetilde{S}$ upon contracting $H\cup J$ and introducing the auxiliary soft vertex.
    
    Now we argue that $A$ and $B$ must be in two different jets $J_i$ and $J_j$. Notice that $n_J^{}= L(\widetilde{J})+2$ implies that $\widetilde{J}\cap T^2$ has three connected components, two of which have nearly on-shell external momenta $p_i$ and $p_j$, respectively. If the path $P$ were not connecting them, then at least one would have been a component of $T^2$, thus an $\mathcal{F}^{(p_i^2,R)}$ term, in conflict with the conclusion above.

    Therefore, in the configuration of $T^2$ there is a path of soft propagators that joins two of the nearly on-shell external momenta, say $p_i$ and $p_j$. We denote the connected component of $S$ that contains this path as $S_{k[i,j]}$, and $q_{ij}\equiv p_i+p_j$. The graph $\widetilde{S}_{k[i,j]} \cap T^2$ then contains exactly one loop, and the graphs $\widetilde{J}_i\cap T^2$ and $\widetilde{J}_j\cap T^2$ are spanning 2-trees of $\widetilde{J}_i$ and $\widetilde{J}_j$. That is,
    \begin{eqnarray}
    \mathcal{F}_\text{II}^{(q^2,R)}(\x;\s):
    &&\ n_H^{} = L(H),\nonumber \\
    &&\ n_{J_{i}}^{} = L(\widetilde{J}_i)+1,\quad n_{J_j}^{} = L(\widetilde{J}_j)+1,\quad n_{J_l}=L(\widetilde{J}_l)\ \ (\forall l\neq i,j),\nonumber\\
    &&\ n_{S_{k[i,j]}}^{}= L(\widetilde{S}_{k[i,j]})-1,\quad n_{S_k}^{}= L(\widetilde{S}_k)\ \ (\forall k\neq k[i,j]).
    \label{minimal_fs_term_condition2_precise}
    \end{eqnarray}
    The configuration of $T^2$ in this case is depicted in figure~\ref{UF_terms_graphic_description}$(d)$. Similar to the previous three cases, the $\mathcal{F}_\textup{II}^{(q^2,R)}$ polynomial can be factorised as follows
    \begin{eqnarray}
    \label{leading_second_Fq2terms_factorise1}
    \mathcal{F}_\textup{II}^{(q_{ij}^2,R)}(\x;\s)= \mathcal{U}_H(\x^{[H]}) \cdot \sum_{i,j} \Big( \mathcal{F}_{J_i\cup J_j\cup S}^{q_{ij}^2}(\x^{[J_i]},\x^{[J_j]},\x^{[S]}) \cdot \prod_{k\neq i,j}\mathcal{U}_{J_k}(\x^{[J_k]}) \Big).
    \end{eqnarray}
    Here the polynomial
    $\mathcal{F}_{J_i\cup J_j\cup S}^{(q_{ij}^2)}$ consists of the spanning 2-trees of the graph $J_i\cup J_j\cup S$, such that the momentum flowing between the components is exactly $q_{ij}\equiv p_i+p_j$.
\end{enumerate}

\bigbreak
We note that the three cases above have covered all the possibilities of the terms of $\mathcal{F}^{(R)}(\x;\s)$, because for any other values of $k_1$ and $k_2$ in (\ref{condition_generic_F_terms}), we have $k_1+k_2 \geqslant 2$, and the corresponding value of $\boldsymbol{r}\cdot \boldsymbol{v}_R$ is
\begin{eqnarray}
\boldsymbol{r}\cdot \boldsymbol{v}_R &&= -2 (L(\widetilde{S})-k_2) -(L(\widetilde{J}) +k_2+1-k_1) +r_{N+1} \nonumber\\
&&=-2L(\widetilde{S}) -L(\widetilde{J}) +k_2+k_1-1 +r_{N+1} \nonumber\\
&&> -2L(\widetilde{S}) -L(\widetilde{J}),
\label{rdotv_other_potential_possibilities}
\end{eqnarray}
where $r_{N+1}$ is the kinematic contribution of the point $\boldsymbol{r}$, satisfying $r_{N+1}=0$ or $1$. We emphasise that (\ref{rdotv_other_potential_possibilities}) is a strict inequality. The right-hand side is $\text{min}(\boldsymbol{r}\cdot \boldsymbol{v}_R)$, which is attained for all of the four cases above. It follows that cases with $k_1+k_2 \geqslant 2$ do not correspond to any potential terms of $\mathcal{P}_0^{(R)}(\x;\s)$.

We conclude that the following theorem must hold.
\begin{theorem}
For any region $R$ in the on-shell expansion of a wide-angle scattering graph $G$, the leading Lee-Pomeransky polynomial takes the form
\begin{eqnarray}
\mathcal{P}_0^{(R)}(\x;\s)&=& \mathcal{U}^{(R)}(\x) + 
\mathcal{F}^{(R)}(\x;\s)\\
\mathcal{F}^{(R)}(\x;\s)&=&\sum_{i=1}^K \mathcal{F}^{(p_i^2,R)}(\x;\s)+ \mathcal{F}_\textup{I}^{(q^2,R)}(\x;\s)+ 
\sum_{i>j=1}^K \mathcal{F}_\textup{II}^{(q_{ij}^2,R)}(\x;\s)
\end{eqnarray}
\noindent These polynomials factorise as follows
\begin{align}
    \begin{split}
        &\mathcal{U}^{(R)}(\x) = \mathcal{U}_H(\x^{[H]}) \cdot \Big( \prod_{i=1}^K \mathcal{U}_{J_i}(\x^{[J_i]}) \Big) \cdot \mathcal{U}_S(\x^{[S]})\,,\\
        &\mathcal{F}^{(p_i^2,R)}(\x;\s) = \mathcal{U}_H(x^{[H]}) \cdot \mathcal{F}^{(p_i^2)}_{J_i}(\x^{[J_i]};\s) \cdot \Big( \prod_{j\neq i}^K\mathcal{U}_{J_j}(\x^{[J_j]}) \Big) \cdot\mathcal{U}_S(\x^{[S]})\,,\\
        &\mathcal{F}_\textup{I}^{(q^2,R)}(\x;\s) = \mathcal{F}_{H\cup J}^{(q^2)}(\x^{[H]},\x^{[J]}) \cdot \mathcal{U}_S(\x^{[S]})\,,\phantom{\prod_{k\neq i,j}}\\
        &\mathcal{F}_\textup{II}^{(q_{ij}^2,R)}(\x;\s)= \mathcal{U}_H(\x^{[H]}) \cdot 
        \mathcal{F}_{J_i\cup J_j\cup S}^{(q_{ij}^2)}(\x^{[J_i]},\x^{[J_j]},\x^{[S]}) \cdot \prod_{k\neq i,j}\mathcal{U}_{J_k}(\x^{[J_k]})\,.
    \end{split}
    \label{summary_leading_terms_factorise}
\end{align}
\label{theorem-leading_terms_form}
\end{theorem}

\begin{corollary}
Let us denote the number of parameters which each monomial contains in $H$, $J_i$ and $S_k$ as $n_H$, $n_{J_i}$ and $n_{S_k}$, respectively. These numbers are summarised in the following table:
\begin{table}[htbp]
    \centering
    \begin{tblr}{
      colspec={
      |>{\centering \arraybackslash}X[1]
      |>{\centering \arraybackslash}X[1.2]
      |>{\centering \arraybackslash}X[4.5]
      |>{\centering \arraybackslash}X[3.5]|
      }, row{1} = {c}, hlines,
    }
  & $n_H$ & $n_{J_1},\dots,n_{J_K}$ & $n_{S_1},\dots,n_{S_n}$ \\
 $\mathcal{U}^{(R)}$ &  $L(H)$  & $L(\widetilde{J}_1),\dots,L(\widetilde{J}_K)$  & $L(\widetilde{S}_1),\dots,L(\widetilde{S}_n)$ \\
 $\mathcal{F}^{(p_i^2,R)}$ &  $L(H)$  & $L(\widetilde{J}_1),\dots,L(\widetilde{J}_i)+1,\dots,L(\widetilde{J}_K)$  & $L(\widetilde{S}_1),\dots,L(\widetilde{S}_n)$\\
 $\mathcal{F}_\textup{I}^{(q^2,R)}$ &  $L(H)+1$    & $L(\widetilde{J}_1),\dots,L(\widetilde{J}_K)$  & $L(\widetilde{S}_1),\dots,L(\widetilde{S}_n)$\\
 $\mathcal{F}_\textup{II}^{(q_{ij}^2,R)}$ &  $L(H)$    & $\scriptstyle{L(H),L(\widetilde{J}_1),\dots,L(\widetilde{J}_i)+1,\dots,L(\widetilde{J}_j)+1,\dots,L(\widetilde{J}_K)}$  & $\scriptstyle{L(\widetilde{S}_1),\ \dots,\ L(\widetilde{S}_{k[i,j]})-1,\ \dots,\ L(\widetilde{S}_n)}$ \\
    \end{tblr}
\end{table}
\label{theorem-leading_terms_form-corollary}
\end{corollary}

\subsection{The infrared regions in wide-angle scattering}
\label{infrared_regions_in_wideangle_scattering}

As is discussed above, any infrared region $R$ contributing to the MoR in the on-shell expansion of a wide-angle graph $G$, may be described by figure~\ref{classical_picture} with the subgraphs $H,J_1,\dots,J_K,S$ satisfying the set of basic requirements stated in proposition~\ref{proposition-region_vectors_are_hard_and_infrared}. Our goal is to fully characterise the possible structure of these subgraphs. To this end, we must exclude any configuration that passes these basic requirements and yet produces a scaleless integral, which therefore does not contribute to the MoR. In order to avoid these pathological configurations, some extra requirements on $H$, $J$ and $S$ would be needed. We will show that the following three requirements provide a necessary and sufficient condition for any configuration satisfying proposition~\ref{proposition-region_vectors_are_hard_and_infrared} to be an infrared region.
\begin{itemize}
    \item \emph{Requirement of $H$: all the internal propagators of $H_\textup{red}$, which is the reduced form of~$H$, are off-shell.}
    \item \emph{Requirement of $J$: all the internal propagators of $\widetilde{J}_{i,\textup{red}}$, which is the reduced form of the contracted graph $\widetilde{J}_i$, carry exactly the momentum $p_i^\mu$.}
    \item \emph{Requirement of $S$: every connected component of $S$ must connect at least two different jet subgraphs $J_i$ and $J_j$.}
\end{itemize}
To understand these requirements, we recall the concept of the reduced form of a graph, which is introduced in section~\ref{space_spanning_trees}. The reduced form of $H$ is obtained by contracting each of the nontrivial 1VI components of $H$ to a vertex, and if two nontrivial 1VI components $\gamma_1$ and $\gamma_2$ share a vertex, we add an auxiliary propagator connecting $\gamma_1$ and $\gamma_2$ before contracting them. The same construction applies to any of the contracted jet graphs $\widetilde{J}_i$. Similar to eq.~(\ref{relation_propagator_number_reduced}), we can further derive the following relations:
\begin{eqnarray}
N(H) = \widehat{N}(H_\text{red})+ \sum_{j=1}^{n} N(\gamma_j^H),\qquad N(\widetilde{J}_i) = \widehat{N}(\widetilde{J}_{i,\text{red}})+ \sum_{j=1}^{n_i} N(\gamma_j^{\widetilde{J}_i}),
\label{relation_propagator_number_reduced_HandJ}
\end{eqnarray}
where $\gamma_j^H$ denotes the 1VI components of $H$ and $\widehat{N}(H_\text{red})$ the number of edges that are in both $H_\text{red}$ and $H$ (i.e. $\widehat{N}(H_\text{red})$ excludes the auxiliary edges that have been introduced upon separating nontrivial 1VI components when constructing $H_\text{red}$). Similarly,  $\gamma_j^{\widetilde{J}_i}$ represents the nontrivial 1VI components of $\widetilde{J}_i$, where $\widehat{N}(\widetilde{J}_{i,\text{red}})$ is the number of edges that are in both $\widetilde{J}_{i,\text{red}}$ and $\widetilde{J}_i$.

The requirement of $H$ can be seen as a generalisation of corollary~\ref{theorem-Uterms_space_dimensionality_corollary2}. The requirement of $J$ follows directly from the second Landau condition, eq.~(\ref{landau_equations_modification_II}). Namely, for example, the jet configurations in figure~\ref{jet_configuration_forbidden}$(a)$ are allowed, while those in figure~\ref{jet_configuration_forbidden}$(b)$ and $(c)$ instead lead to scaleless integrals over some jet loop momenta, and are hence forbidden in any infrared region. Finally, the requirement of $S$ rules out the configurations where a connected component of $S$ is only attached to the hard subgraph and/or one of the jets. In other words, allowed soft subgraphs must connect at least two jets. As we show in appendix~\ref{appendix-necessity_sufficiency_requirements}, for any configuration satisfying 
proposition~\ref{proposition-region_vectors_are_hard_and_infrared}, the requirements of $H,J$ and $S$ above provide a necessary and sufficient condition for the existence of a corresponding region $R$, such that $\boldsymbol{v}_R$ is normal to a lower facet of the Newton polytope $\Delta^{(N+1)}[{\cal P}(G)]$. Should any of the requirements be violated, such a configuration may correspond to a face with dimension less than $N$, rather than a facet.

\section{A graph-finding algorithm for regions}
\label{section-graph_finding_algorithm_regions}

In the last section we classified the configurations of the hard and infrared regions, where the latter emerge as specific solutions of the Landau equations. The subgraphs $H$, $J$ and $S$ introduced in proposition~\ref{proposition-region_vectors_are_hard_and_infrared}, supplemented with the extra requirements in section~\ref{infrared_regions_in_wideangle_scattering}, allow us to identify precisely the infrared regions of the Feynman integral $\I(G)$. The purpose of this section is to construct an algorithm which would identify all the infrared regions of a graph $G$ directly, without referring to the geometric MoR polytope construction. To this end, we will translate the properties of $H$, $J$ and $S$ into a strict graph-theoretical description of the infrared regions.

In section~\ref{graph_theoretical_condition} we propose a set of graph-theoretical conditions for the subgraphs of $G$ and show that these conditions are equivalent to the extra requirements in section~\ref{infrared_regions_in_wideangle_scattering}. In section~\ref{algorithm_for_regions} we then construct a general algorithm for the regions, which ensures that the properties in proposition~\ref{proposition-region_vectors_are_hard_and_infrared} as well as the requirements in section~\ref{infrared_regions_in_wideangle_scattering} are satisfied. Finally in section~\ref{more_examples_implementation}, we apply this algorithm to a variety of wide-angle massless scattering graph examples including all graphs having up to 5 loops and 3 external legs, and verify that it agrees with the geometric algorithm of ref.~\cite{HrchJnsSlk22}.

\subsection{Graph-theoretical conditions of infrared subgraphs}
\label{graph_theoretical_condition}

A useful concept in the classification of Euclidean infrared subgraphs, due to Brown~\cite{Brn15}, is that of a \emph{motic subgraph} (or asymptotically irreducible graph by Smirnov~\cite{Smn02book}) which can be used to define the infrared-cograph, or hard subgraph. A motic subgraph may be defined as \emph{a subgraph whose connected components each become 1PI after connecting all the external lines to an auxiliary vertex.} In the original work by Brown \cite{Brn15} subgraphs only inherit external momenta if all external momenta of the parent graph are incident on the subgraph. Here we allow a subgraph to inherit any of the external momenta which enter any of its vertices.

The concept of motic subgraphs is useful in at least two ways. The first is that it facilitates the construction of a simple algorithm for finding infrared subgraphs~\cite{HzgRjl17}. The second is that motic subgraphs are associated with a Hopf algebra structure~\cite{Krm97}, which has been important in placing the $R^*$-operation~\cite{BkvBrskHzg20}, an operation which subtracts both infrared and ultraviolet divergences from any Euclidean Feynman integral, onto a firmer mathematical ground.

We find that a small but essential modification of the notion of motic graphs is useful in the classification of Minkowskian infrared subgraphs, and specifically, those appearing in the on-shell expansion. In particular, we define a graph to be \emph{mojetic} if \emph{it becomes 1VI after connecting all of its external edges to an auxiliary vertex} (a mojetic graph is necessarily motic). It turns out that mojetic is equivalent to connected, motic and scaleful (i.e. free of scaleless subgraphs) for the case under consideration. As we will see below, the notion of mojetic graphs is useful for the construction of a Minkowskian infrared-subgraph finder. In figure~\ref{motic_mojetic} we show a few examples to illustrate the concepts of motic and mojetic graphs.
\begin{figure}[t]
\centering
\centering
\begin{subfigure}[b]{0.45\textwidth}
\centering
    \resizebox{\textwidth}{!}{
\begin{tikzpicture}[scale=0.5]

\node (v1) at (0,5) {};
\node [draw,circle,minimum size=4pt,fill=Black,inner sep=0pt,outer sep=0pt] (v2) at (2,5) {};
\node [draw,circle,minimum size=4pt,fill=Black,inner sep=0pt,outer sep=0pt] (v3) at (4,7) {};
\node [draw,circle,minimum size=4pt,fill=Black,inner sep=0pt,outer sep=0pt]  (v4) at (4,3) {};
\node [draw,circle,minimum size=4pt,fill=Black,inner sep=0pt,outer sep=0pt]  (v5) at (6,5) {};
\node (v6) at (8,5) {};

\draw [thick] (v1) -- (v2);
\path (v2) edge [thick,bend left = 45] (v3) {};
\path (v2) edge [thick,bend right = 45] (v4) {};
\draw [thick] (v3) -- (v4);
\path (v4) edge [thick,bend right = 45] (v5) {};
\path (v3) edge [thick,bend left = 45] (v5) {};
\draw [thick] (v5) -- (v6);

\node (x1) at (10,5) {};
\node (x2) at (12,5) {};
\draw [->,thick] (x1) -- (x2);

\node [draw,circle,minimum size=4pt,fill=Black,inner sep=0pt,outer sep=0pt] (vv2) at (15,5) {};
\node [draw,circle,minimum size=4pt,fill=Black,inner sep=0pt,outer sep=0pt] (vv3) at (17,7) {};
\node [draw,circle,minimum size=4pt,fill=Black,inner sep=0pt,outer sep=0pt]  (vv4) at (17,3) {};
\node [draw,circle,minimum size=4pt,fill=Black,inner sep=0pt,outer sep=0pt]  (vv5) at (19,5) {};
\node [draw=magenta,circle,minimum size=4pt,fill=magenta,inner sep=0pt,outer sep=0pt] (vv7) at (17,1.5) {};

\node (y8) at (20,1) {};

\path (vv7) edge [thick,bend left = 60] (vv2) {};
\path (vv2) edge [thick,bend left = 45] (vv3) {};
\path (vv2) edge [thick,bend right = 45] (vv4) {};
\draw [thick] (vv3) -- (vv4);
\path (vv4) edge [thick,bend right = 45] (vv5) {};
\path (vv3) edge [thick,bend left = 45] (vv5) {};
\path (vv5) edge [thick,bend left = 60] (vv7) {};

\end{tikzpicture}
}
\vspace{-5em}
\vspace{3em}
\vspace{-2em}
\vspace{9pt}
\caption{Both motic and mojetic.}
\label{fig:mojetic1}
\end{subfigure}
\hfill
\begin{subfigure}[b]{0.45\textwidth}
\centering
    \resizebox{\textwidth}{!}{
\begin{tikzpicture}[scale=0.5]

\node (v1) at (0,5) {};
\node [draw,circle,minimum size=4pt,fill=Black,inner sep=0pt,outer sep=0pt] (v2) at (1.5,5) {};
\node [minimum size=0pt,inner sep=0pt,outer sep=0pt] (v3) at (2.5,6) {};
\node [minimum size=0pt,inner sep=0pt,outer sep=0pt] (v4) at (2.5,4) {};
\node [draw,circle,minimum size=4pt,fill=Black,inner sep=0pt,outer sep=0pt]  (v5) at (3.5,5) {};
\node [draw,circle,minimum size=4pt,fill=Black,inner sep=0pt,outer sep=0pt]  (v6) at (4.5,5) {};
\node [minimum size=0pt,inner sep=0pt,outer sep=0pt] (v7) at (5.5,6) {};
\node [minimum size=0pt,inner sep=0pt,outer sep=0pt] (v8) at (5.5,4) {};
\node [draw,circle,minimum size=4pt,fill=Black,inner sep=0pt,outer sep=0pt]  (v9) at (6.5,5) {};
\node (v10) at (8,5) {};

\draw [thick] (v1) -- (v2);
\path (v2) edge [thick,bend left = 45] (v3) {};
\path (v2) edge [thick,bend right = 45] (v4) {};
\path (v4) edge [thick,bend right = 45] (v5) {};
\path (v3) edge [thick,bend left = 45] (v5) {};
\draw [thick] (v5) -- (v6);
\path (v6) edge [thick,bend left = 45] (v7) {};
\path (v6) edge [thick,bend right = 45] (v8) {};
\path (v8) edge [thick,bend right = 45] (v9) {};
\path (v7) edge [thick,bend left = 45] (v9) {};
\draw [thick] (v9) -- (v10);

\node (x1) at (10,5) {};
\node (x2) at (12,5) {};
\draw [->,thick] (x1) -- (x2);

\node [draw,circle,minimum size=4pt,fill=Black,inner sep=0pt,outer sep=0pt] (vv2) at (14.5,5) {};
\node [minimum size=0pt,inner sep=0pt,outer sep=0pt] (vv3) at (15.5,6) {};
\node [minimum size=0pt,inner sep=0pt,outer sep=0pt] (vv4) at (15.5,4) {};
\node [draw,circle,minimum size=4pt,fill=Black,inner sep=0pt,outer sep=0pt]  (vv5) at (16.5,5) {};
\node [draw,circle,minimum size=4pt,fill=Black,inner sep=0pt,outer sep=0pt]  (vv6) at (17.5,5) {};
\node [minimum size=0pt,inner sep=0pt,outer sep=0pt] (vv7) at (18.5,6) {};
\node [minimum size=0pt,inner sep=0pt,outer sep=0pt] (vv8) at (18.5,4) {};
\node [draw,circle,minimum size=4pt,fill=Black,inner sep=0pt,outer sep=0pt]  (vv9) at (19.5,5) {};

\node [draw=magenta,circle,minimum size=4pt,fill=magenta,inner sep=0pt,outer sep=0pt]  (vv10) at (17,1.5) {};

\node (y8) at (20,1) {};

\path (vv10) edge [thick,bend left = 50] (vv2) {};
\path (vv2) edge [thick,bend left = 45] (vv3) {};
\path (vv2) edge [thick,bend right = 45] (vv4) {};
\path (vv4) edge [thick,bend right = 45] (vv5) {};
\path (vv3) edge [thick,bend left = 45] (vv5) {};
\draw [thick] (vv5) -- (vv6);
\path (vv6) edge [thick,bend left = 45] (vv7) {};
\path (vv6) edge [thick,bend right = 45] (vv8) {};
\path (vv8) edge [thick,bend right = 45] (vv9) {};
\path (vv7) edge [thick,bend left = 45] (vv9) {};
\path (vv9) edge [thick,bend left = 50] (vv10) {};

\end{tikzpicture}
}
\vspace{-2em}
\vspace{-2em}
\vspace{9pt}
\caption{Both motic and mojetic.}
\label{fig:mojetic2}
\end{subfigure}

\begin{subfigure}[b]{0.45\textwidth}
\centering
    \vspace{1em}
\resizebox{\textwidth}{!}{
\begin{tikzpicture}[scale=0.5]

\node (v1) at (0,5) {};
\node [draw,circle,minimum size=4pt,fill=Black,inner sep=0pt,outer sep=0pt] (v2) at (3,5) {};
\node [draw,circle,minimum size=4pt,fill=Black,inner sep=0pt,outer sep=0pt] (v3) at (4,6) {};
\node [minimum size=0pt,inner sep=0pt,outer sep=0pt] (v4) at (4,4) {};
\node [draw,circle,minimum size=4pt,fill=Black,inner sep=0pt,outer sep=0pt]  (v5) at (5,5) {};
\node (v6) at (8,5) {};
\node [minimum size=0pt,inner sep=0pt,outer sep=0pt] (v7) at (4,7.5) {};

\draw [thick] (v1) -- (v2);
\path (v2) edge [thick,bend left = 45] (v3) {};
\path (v2) edge [thick,bend right = 45] (v4) {};
\path (v4) edge [thick,bend right = 45] (v5) {};
\path (v3) edge [thick,bend left = 45] (v5) {};
\draw [thick] (v5) -- (v6);
\path (v3) edge [thick,bend left = 80] (v7) {};
\path (v3) edge [thick,bend right = 80] (v7) {};

\node (x1) at (10,5) {};
\node (x2) at (12,5) {};
\draw [->,thick] (x1) -- (x2);

\node [draw,circle,minimum size=4pt,fill=Black,inner sep=0pt,outer sep=0pt] (vv2) at (16,5) {};
\node [draw,circle,minimum size=4pt,fill=Black,inner sep=0pt,outer sep=0pt] (vv3) at (17,6) {};
\node [minimum size=0pt,inner sep=0pt,outer sep=0pt] (vv4) at (17,4) {};
\node [draw,circle,minimum size=4pt,fill=Black,inner sep=0pt,outer sep=0pt]  (vv5) at (18,5) {};
\node [minimum size=0pt,inner sep=0pt,outer sep=0pt] (vv7) at (17,7.5) {};

\node [draw=magenta,circle,minimum size=4pt,fill=magenta,inner sep=0pt,outer sep=0pt]  (vv10) at (17,2.5) {};

\path (vv10) edge [thick,bend left = 50] (vv2) {};
\path (vv2) edge [thick,bend left = 45] (vv3) {};
\path (vv2) edge [thick,bend right = 45] (vv4) {};
\path (vv4) edge [thick,bend right = 45] (vv5) {};
\path (vv3) edge [thick,bend left = 45] (vv5) {};
\path (vv5) edge [thick,bend left = 50] (vv10) {};
\path (vv3) edge [thick,bend left = 80] (vv7) {};
\path (vv3) edge [thick,bend right = 80] (vv7) {};

\node (y8) at (20,1) {};

\end{tikzpicture}
}
\vspace{-3em}
\vspace{-2em}
\vspace{9pt}
\caption{Motic but not mojetic.}
\label{fig:mojetic3}
\end{subfigure}
\hfill
\begin{subfigure}[b]{0.45\textwidth}
\centering
    \vspace{1em}
\resizebox{\textwidth}{!}{
\begin{tikzpicture}[scale=0.5]

\node (v1) at (0,5) {};
\node [draw,circle,minimum size=4pt,fill=Black,inner sep=0pt,outer sep=0pt] (v2) at (3,5) {};
\node [draw,circle,minimum size=4pt,fill=Black,inner sep=0pt,outer sep=0pt] (v3) at (4,6) {};
\node [minimum size=0pt,inner sep=0pt,outer sep=0pt] (v4) at (4,4) {};
\node [draw,circle,minimum size=4pt,fill=Black,inner sep=0pt,outer sep=0pt]  (v5) at (5,5) {};
\node (v6) at (8,5) {};
\node [draw,circle,minimum size=4pt,fill=Black,inner sep=0pt,outer sep=0pt] (v7) at (4,7) {};
\node [minimum size=0pt,inner sep=0pt,outer sep=0pt] (v8) at (4,8.5) {};

\draw [thick] (v1) -- (v2);
\path (v2) edge [thick,bend left = 45] (v3) {};
\path (v2) edge [thick,bend right = 45] (v4) {};
\path (v4) edge [thick,bend right = 45] (v5) {};
\path (v3) edge [thick,bend left = 45] (v5) {};
\draw [thick] (v5) -- (v6);
\path (v3) edge [thick] (v7) {};
\path (v7) edge [thick,bend right = 90] (v8) {};
\path (v7) edge [thick,bend left = 90] (v8) {};

\node (x1) at (10,5) {};
\node (x2) at (12,5) {};
\draw [->,thick] (x1) -- (x2);

\node [draw,circle,minimum size=4pt,fill=Black,inner sep=0pt,outer sep=0pt] (vv2) at (16,5) {};
\node [draw,circle,minimum size=4pt,fill=Black,inner sep=0pt,outer sep=0pt] (vv3) at (17,6) {};
\node [minimum size=0pt,inner sep=0pt,outer sep=0pt] (vv4) at (17,4) {};
\node [draw,circle,minimum size=4pt,fill=Black,inner sep=0pt,outer sep=0pt]  (vv5) at (18,5) {};
\node [draw,circle,minimum size=4pt,fill=Black,inner sep=0pt,outer sep=0pt] (vv7) at (17,7) {};
\node [minimum size=0pt,inner sep=0pt,outer sep=0pt] (vv8) at (17,8.5) {};

\node [draw=magenta,circle,minimum size=4pt,fill=magenta,inner sep=0pt,outer sep=0pt]  (vv10) at (17,2.5) {};

\path (vv10) edge [thick,bend left = 50] (vv2) {};
\path (vv2) edge [thick,bend left = 45] (vv3) {};
\path (vv2) edge [thick,bend right = 45] (vv4) {};
\path (vv4) edge [thick,bend right = 45] (vv5) {};
\path (vv3) edge [thick,bend left = 45] (vv5) {};
\path (vv5) edge [thick,bend left = 50] (vv10) {};
\path (vv3) edge [thick] (vv7) {};
\path (vv7) edge [thick,bend right = 80] (vv8) {};
\path (vv7) edge [thick,bend left = 80] (vv8) {};

\node (y8) at (20,1) {};

\end{tikzpicture}
}
\vspace{-3em}
\vspace{-2em}
\vspace{9pt}
\caption{Neither motic nor mojetic.}
\label{fig:mojetic4}
\end{subfigure}
\caption{Examples of motic and mojetic graphs. The pink vertex is introduced to connect all external lines to a single point.}
\label{motic_mojetic}
\end{figure}
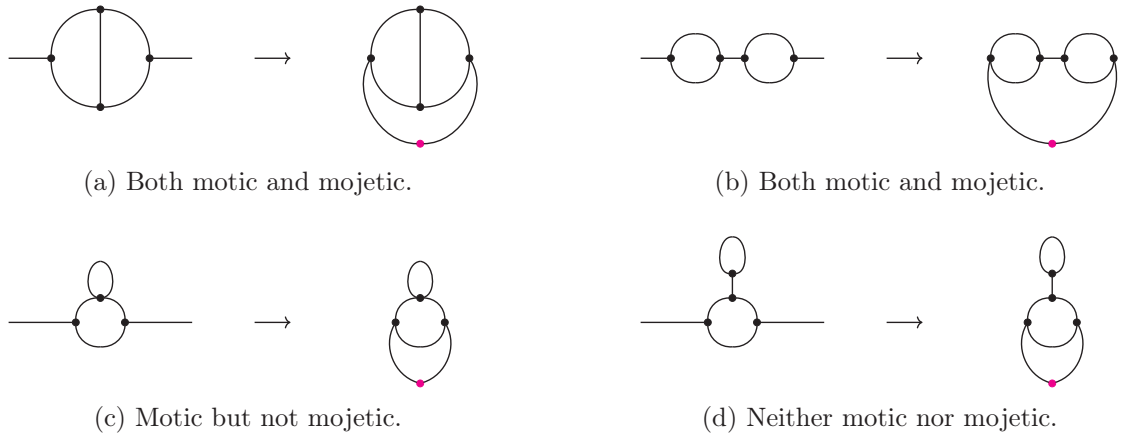

With the help of these concepts, we claim that the requirements regarding~$H$ and $J$ in section~\ref{infrared_regions_in_wideangle_scattering}, are equivalent to the requirement that all the $K$ subgraphs $H\cup J\setminus J_i$ ($i=1,\dots,K$) are mojetic. Therefore, the following theorem characterising the infrared regions of $G$ may be formulated.

\begin{theorem}
Any configuration of subgraphs in $G$ compatible with proposition~\ref{proposition-region_vectors_are_hard_and_infrared} is an infrared region of $G$, if and only if the following graph-theoretical conditions for $H$, $J$ and $S$ are satisfied:
\begin{enumerate}
    \item[$\textit{1.}$] {For any $i=1,\dots,K$, the subgraph $H\cup J\setminus J_i$ is mojetic.}
    \item[$\textit{2.}$] {Every connected component of $S$ must connect at least two jet subgraphs $J_i$ and $J_j$ for some $i \neq j$.}
\end{enumerate}
\label{theorem-algorithm_necessary_sufficient}
\end{theorem}
Note that the second condition in the theorem above is identical to the requirement of $S$ stated in section~\ref{infrared_regions_in_wideangle_scattering}. Below we will prove that the first condition is \emph{equivalent} to the requirements regarding $H$ and $J$ in section~\ref{infrared_regions_in_wideangle_scattering}.

\begin{proof}
First, we show that as long as the requirements of $H$ and $J$ of section~\ref{infrared_regions_in_wideangle_scattering} are satisfied, each of the subgraphs $H\cup J\setminus J_i$ ($i=1,\dots,K$) is mojetic. Equivalently, we show that for any vertex $v\in H\cup J\setminus J_i$, the graph $(H\cup J\setminus J_i) \setminus v$ is connected after we connect all its external momenta to an auxiliary vertex.

In order to show this, we first define a certain type of vertex of the hard subgraph $H$: a \emph{hard-hard vertex} is a vertex that is shared by two 1VI components of $H$, or attached to an external off-shell momentum $q_j$.

From the requirement of $H$ in section~\ref{infrared_regions_in_wideangle_scattering}, i.e.~the internal propagators of $H_\text{red}$ are off-shell, it follows that each 1VI component in $H$ is one of the following types (see figure~\ref{hard_subgraph_structure_original}):
\begin{enumerate}
    \item[(1)] it has two or more hard-hard vertices (and any number of jet edges attached);
    \item[(2)] it has one hard-hard vertex, and it attaches to jet edges of at least \emph{two} different jets;
    \item[(3)] it has no hard-hard vertices, and it attaches to jet edges of at least \emph{four} different jets.
\end{enumerate}
This is because $H_\text{red}$ is a tree graph, and every 1VI component of $H$ is either an edge or a vertex of $H_\text{red}$, hence for the propagators of $H_\text{red}$ to be off-shell, momentum conservation requires one of the configurations listed above. For example, in figure~\ref{hard_subgraph_structure_original} $\gamma_2$ and $\gamma_3$ have at least two hard-hard vertices (marked black), and $\gamma_1$ and $\gamma_4$ each have one hard-hard vertex, and are attached to jet edges of two different jets (through the yellow vertices).
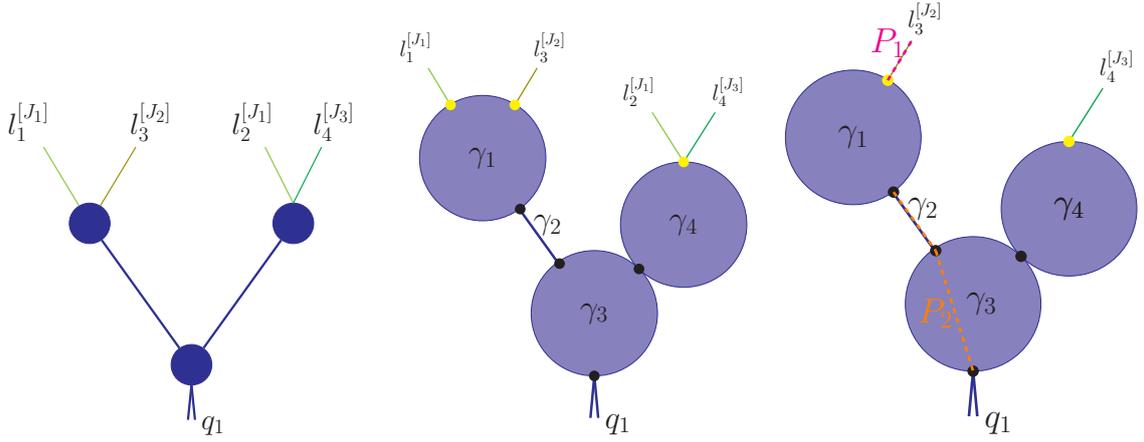
\begin{figure}[t]
\centering
\begin{subfigure}[b]{0.32\textwidth}
\centering
\resizebox{\textwidth}{!}{
\begin{tikzpicture}[line width = 0.6, font=\large, mydot/.style={circle, fill, inner sep=.7pt},transform shape]

\node[draw=Blue,circle,minimum size=1cm,fill=Blue] (h) at (5,3){};
\node[draw=Blue,circle,minimum size=1cm,fill=Blue] (g1) at (2.5,6.5){};
\node[draw=Blue,circle,minimum size=1cm,fill=Blue] (g4) at (7.5,6.5){};

\path (h) edge [ultra thick, color=Blue] (g1) {};
\path (h) edge [ultra thick, color=Blue] (g4) {};

\node (q1) at (4.9,1.5) {};
\node (q1p) at (5.1,1.5) {};
\draw (q1) edge [ultra thick, color=Blue] (5,2.5) node [] {};
\draw (q1p) edge [ultra thick, color=Blue] (5,2.5) node [right] {\huge $q_1$};

\node (p1) at (1,9) {\huge $l_1^{[J_1]}$};
\node (p2) at (4,9) {\huge $l_3^{[J_2]}$};
\node (p3) at (6.5,9) {\huge $l_2^{[J_1]}$};
\node (p4) at (8.5,9) {\huge $l_4^{[J_3]}$};
\draw (p1) edge [thick, color=LimeGreen] (g1) node [right] {};
\draw (p2) edge [thick, color=olive] (g1) node [right] {};
\draw (p3) edge [thick, color=LimeGreen] (7.5,7) node [right] {};
\draw (p4) edge [thick, color=Green] (7.5,7) node [right] {};

\draw (3.5,5) node [right] {};

\end{tikzpicture}
}
\vspace{5pt}
\vspace{-2em}
\caption{An example of $H_\text{red}$, and the jet and external momenta attached to it.}
\label{hard_subgraph_structure_reduced}
\end{subfigure}
\hfill
\begin{subfigure}[b]{0.32\textwidth}
\centering
\resizebox{\textwidth}{!}{
\begin{tikzpicture}[line width = 0.6, font=\large, mydot/.style={circle, fill, inner sep=.7pt},transform shape]

\node[draw=Blue,circle,minimum size=2.828cm,fill=Blue!50] (h) at (5,3){};
\node[draw=Blue,circle,minimum size=2.828cm,fill=Blue!50] (g1) at (2.5,6.5){};
\node[draw=Blue,circle,minimum size=2.828cm,fill=Blue!50] (g4) at (7,5){};

\node at (h) {\huge $\gamma_3$};
\node at (g1) {\huge $\gamma_1$};
\node at (g4) {\huge $\gamma_4$};

\path (h) edge [ultra thick, color=Blue] (g1) {};

\node (q1) at (4.9,0.5) {};
\node (q1p) at (5.1,0.5) {};
\draw (q1) edge [ultra thick, color=Blue] (5,1.5858) node [] {};
\draw (q1p) edge [ultra thick, color=Blue] (5,1.5858) node [right] {\huge $q_1$};

\node (p1) at (1,9) {\Large $l_1^{[J_1]}$};
\node (p2) at (4,9) {\Large $l_3^{[J_2]}$};
\node (p3) at (6,8) {\Large $l_2^{[J_1]}$};
\node (p4) at (8,8) {\Large $l_4^{[J_3]}$};
\draw (p1) edge [thick, color=LimeGreen] (g1) node [right] {};
\draw (p2) edge [thick, color=olive] (g1) node [right] {};
\draw (p3) edge [thick, color=LimeGreen] (7,6.414) node [right] {};
\draw (p4) edge [thick, color=Green] (7,6.414) node [right] {};

\draw[fill,thick,color=Black] (6,4) circle (3pt);
\draw[fill,thick,color=Black] (4.22,4.12) circle (3pt);
\draw[fill,thick,color=Black] (3.34,5.35) circle (3pt);
\draw[fill,thick,color=Black] (5,1.5858) circle (3pt);
\draw[fill,thick,color=yellow] (3.22,7.7) circle (3pt);
\draw[fill,thick,color=yellow] (1.78,7.7) circle (3pt);
\draw[fill,thick,color=yellow] (7,6.414) circle (3pt);

\draw (3.5,5) node [right] {\huge $\gamma_2$};

\end{tikzpicture}
}
\vspace{5pt}
\vspace{-2em}
\caption{A possible $H$ corresponding to $H_\text{red}$ and its attached jet and external momenta.}
\label{hard_subgraph_structure_original}
\end{subfigure}
\hfill
\begin{subfigure}[b]{0.32\textwidth}
\centering
\resizebox{\textwidth}{!}{
\begin{tikzpicture}[line width = 0.6, font=\large, mydot/.style={circle, fill, inner sep=.7pt},transform shape]

\node[draw=Blue,circle,minimum size=2.828cm,fill=Blue!50] (h) at (5,3){};
\node[draw=Blue,circle,minimum size=2.828cm,fill=Blue!50] (g1) at (2.5,6.5){};
\node[draw=Blue,circle,minimum size=2.828cm,fill=Blue!50] (g4) at (7,5){};

\node at (h) {};
\node at (g1) {};
\node at (g4) {\huge $\gamma_4$};

\path (h) edge [ultra thick, color=Blue] (g1) {};

\node (q1) at (4.9,0.5) {};
\node (q1p) at (5.1,0.5) {};
\draw (q1) edge [ultra thick, color=Blue] (5,1.5858) node [] {};
\draw (q1p) edge [ultra thick, color=Blue] (5,1.5858) node [right] {\huge $q_1$};

\node (p2) at (4,9) {\Large $l_3^{[J_2]}$};
\node (p4) at (8,8) {\Large $l_4^{[J_3]}$};
\draw (p2) edge [thick, color=olive] (g1) node [right] {};
\draw (p4) edge [thick, color=Green] (7,6.414) node [right] {};

\draw[fill,thick,color=Black] (6,4) circle (3pt);
\draw[fill,thick,color=Black] (4.22,4.12) circle (3pt);
\draw[fill,thick,color=Black] (3.34,5.35) circle (3pt);
\draw[fill,thick,color=Black] (5,1.5858) circle (3pt);
\draw[fill,thick,color=yellow] (3.22,7.7) circle (3pt);
\draw[fill,thick,color=yellow] (7,6.414) circle (3pt);
\draw (3.22,7.7) edge [ultra thick, dashed, color=magenta] (p2) node [above, xshift=0pt, yshift=10pt] {\huge {\color{magenta}$P_1$}};
\draw (3.34,5.35) edge [ultra thick, dashed, color=orange] (4.22,4.12) node [] {};
\draw (4.22,4.12) edge [ultra thick, dashed, color=orange] (5,1.5858) node [below, yshift=-25pt] {\huge {\color{orange}$P_2$}};

\draw (3.5,5) node [right] {\huge $\gamma_2$};

\node at (h) [xshift=5pt] {\huge $\gamma_3$};
\node at (g1) [] {\huge $\gamma_1$};
\node at (g4) {\huge $\gamma_4$};

\end{tikzpicture}
}
\vspace{-2em}
\vspace{5pt}
\caption{$H$ and its attached jet and external momenta, where $J_1$ is removed from $G$.}
\label{hard_subgraph_structure_removal}
\end{subfigure}
\caption{An example illustrating the structure of the hard subgraph and the momenta attached to it (discarding soft momenta), where $q_1$ is the off-shell external momentum of $G$, while $l_1^{[J_1]}$ and~$l_2^{[J_1]}$ are momenta of propagators belonging to the jet $J_1$,  $l_3^{[J_2]}$ is the momentum of a propagator of $J_2$, and $l_4^{[J_3]}$ of $J_3$. The subgraphs $\gamma_i$ (i=1,\dots,4) are the 1VI components of $H$.}
\label{hard_subgraph_structure}
\end{figure}

It then follows that upon removing all the edges from a specific jet $J_i$, every 1VI component of $H$ is one of the following types (see figure~\ref{hard_subgraph_structure_removal}):
\begin{enumerate}
    \item[(1)] it has two or more hard-hard vertices (and any number of jet edges attached);
    \item[(2)] it has one hard-hard vertex, and it attaches to at least \emph{one} jet;
    \item[(3)] it has no hard-hard vertices, and it attaches to jet edges of at least \emph{three} different jets.
\end{enumerate}
This observation implies that each 1VI component of $H$ can be connected to two external momenta of $H\cup J\setminus J_i$, through some distinct paths~$P_1$ and~$P_2$, respectively, where $P_1$ and $P_2$ do not share any vertices. For example, in figure~\ref{hard_subgraph_structure_removal} which depicts $H\cup J\setminus J_1$, the 1VI component~$\gamma_1$ can be connected to the external momentum $p_2$ (the external momentum of jet $J_2$) via $P_1$ (the magenta path), and to $q_1$ via $P_2$ (the orange path), where $P_1$ and $P_2$ do not share any vertices. Similar conclusions hold for other 1VI components $\gamma_2$, $\gamma_3$ and $\gamma_4$ for some external momenta in each case.

Moreover, from the requirement of $J$ in section~\ref{infrared_regions_in_wideangle_scattering}, i.e. that every internal propagator of~$\widetilde{J}_{j,\text{red}}$ carries exactly the momentum $p_j$, we deduce that the configuration of $\widetilde{J}_{j,\text{red}}$ must be described by figure~\ref{reduced_form_jetsubgraph}. Since every jet attaches to $H$, we know that each 1VI component of~$J_j$ can also be connected to two external momenta in $H\cup J\setminus J_i$ (with $i\neq j$), one of which is $p_j$ while the other connects via a path going through the hard subgraph $H$. These two paths do not share any vertex.

As a result, by removing any vertex $v\in H\cup J\setminus J_i$, one may break the path $P_1$ or the path $P_2$, but never both. Thus, there exists no $v\in G$, such that $(H\cup J\setminus J_i) \setminus v$ is disconnected after we connect all the external momenta to an auxiliary vertex. Thus, by definition, $H\cup J\setminus J_i$ is mojetic for any $i$.
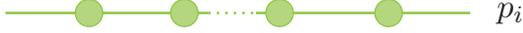
\begin{figure}[t]
\centering
\resizebox{0.5\textwidth}{!}{
\begin{tikzpicture}[scale=1.1]
\node (v1) at (0,0) {};
\node [draw=LimeGreen,circle,minimum size=9pt,fill=LimeGreen!66,inner sep=0pt,outer sep=0pt] (v2) at (1,0) {};
\node [draw=LimeGreen,circle,minimum size=9pt,fill=LimeGreen!66,inner sep=0pt,outer sep=0pt] (v3) at (2,0) {};
\node [minimum size=0pt,inner sep=0pt,outer sep=0pt] (v3p) at (2.25,0) {};
\node [minimum size=0pt,inner sep=0pt,outer sep=0pt] (v4p) at (2.75,0) {};
\node [draw=LimeGreen,circle,minimum size=9pt,fill=LimeGreen!66,inner sep=0pt,outer sep=0pt]  (v4) at (3,0) {};
\node [draw=LimeGreen,circle,minimum size=9pt,fill=LimeGreen!66,inner sep=0pt,outer sep=0pt,right]  (v5) at (4,0) {};
\node [minimum size=0pt,inner sep=0pt,outer sep=0pt]  (v6) at (5,0) {};

\node [right,xshift=5pt] at (v6) {$p_i$};

\draw [thick,color=LimeGreen] (v1) -- (v2);
\draw [thick,color=LimeGreen] (v2) -- (v3);
\draw [thick,color=LimeGreen] (v3) -- (v3p);
\draw [thick,dotted,color=LimeGreen] (v3p) -- (v4p);
\draw [thick,color=LimeGreen] (v4p) -- (v4);
\draw [thick,color=LimeGreen] (v4) -- (v5);
\draw [thick,color=LimeGreen] (v5) -- (v6);
\end{tikzpicture}
}
\vspace{-2em}
\vspace{8pt}
\caption{The reduced form of any contracted jet subgraphs that satisfy the requirement on $J$.}
\label{reduced_form_jetsubgraph}
\end{figure}

Finally, we show that $H\cup J\setminus J_i$ being mojetic for any $i=1,\dots,K$ implies the requirements of $H$ and $J$ in section~\ref{infrared_regions_in_wideangle_scattering}. Let us demonstrate this by contradiction, i.e. assume that there is an edge $e\in H_\text{red}$ that carries an on-shell momentum $p_i$, and then show that the graph $H\cup J\setminus J_i$ is not mojetic.
In the case that $e\in H$, i.e. $e$ is not an auxiliary propagator of $H_\text{red}$, then $e$ itself is a (trivial) 1VI component of $H\cup J\setminus J_i$. After we connect all the external momenta of $H\cup J\setminus J_i$ to an auxiliary vertex, $e$ remains a 1VI component: since the momentum it carries is $p_i$, and $p_i$ is not an external momentum of $H\cup J\setminus J_i$, the branch of $H_\text{red}$ to which $e$ belongs does not get connected to the auxiliary vertex. It then follows that $H\cup J\setminus J_i$ is not mojetic.
In the case that $e\notin H$, i.e. $e$ is an auxiliary propagator, then there must be a nontrivial 1VI component $\gamma^H\subset H$ that is connected to~$e$, and the total momentum flowing into (and out of) $\gamma^H$ is exactly $p_i$ (by the contradiction assumption). 
Using the same reasoning as above, one can claim that $\gamma^H$ remains a 1VI component of $H\cup J\setminus J_i$, even after all the external momenta are connected to an auxiliary vertex. Hence $H\cup J\setminus J_i$ is not mojetic.
We conclude that all the edges of $H_\text{red}$ must carry off-shell momenta. The same reasoning allows one to prove that $\widetilde{J}_{i,\text{red}}$ must carry exactly the momentum~$p_i$. (The only other possibility for edges in $\widetilde{J}_{i,\text{red}}$ would be that they  carry zero momentum, as in a dead-end structure, which is incompatible with $H\cup J\setminus J_i$ being mojetic.)
\end{proof}

\subsection{An algorithm for regions}
\label{algorithm_for_regions}

In this section we develop an algorithmic way to construct all the infrared regions for the on-shell expansion of a generic wide-angle scattering graph $G$. Throughout the construction and with the help of theorem~\ref{theorem-algorithm_necessary_sufficient} above, we will see that the Landau equation as well as the three requirements in section~\ref{infrared_regions_in_wideangle_scattering} are all satisfied automatically.

Before doing so, we introduce the following concept that is helpful in the construction of the jet subgraphs. Consider drawing a cut through a set of edges $\{e_c\}$ which disconnects a graph $G$ into two connected subgraphs, one of which, denoted by $\widehat{\gamma}_i$, contains only a single external momentum $p_i$, compatible with a unitarity cut on the $p_i$ channel (see for example refs.~\cite{Vltm63,Vltm94book,Abreu:2014cla}). Furthermore, $\widehat{\gamma}_i$ must include any vertex whose removal disconnects $p_i^\mu$. We define the \emph{one-external subgraph in the $p_i$ channel} $\gamma_i\equiv \widehat{\gamma}_i \cup \{e_c\}$. From this construction, it follows that a one-external subgraph in the $p_i$ channel, $\gamma_i$, satisfies
\begin{itemize}
    \item[(1)] $\gamma_i$ is connected;
    \item[(2)] every edge of $\gamma_i$ satisfies either
    \begin{enumerate}
        \item[(i)] both endpoints are in $\gamma_i$;
        \item[(ii)] one endpoint is in $\gamma_i$ while the other is in $G\setminus \gamma_i$.
    \end{enumerate}
\end{itemize}

Our algorithm for constructing the regions can then be described by the following steps.
\begin{itemize}
    \item \emph{Step 1}: For each nearly on-shell external momentum $p_i$ ($i=1,\dots,K$), construct the one-external subgraph $\gamma_{i}$ in the $p_i$ channel, such that the subgraph $H_i\equiv G \setminus \gamma_{i}$ is mojetic.
    \item \emph{Step 2}: 
    Consider all possible sets $\{\gamma_1,\dots,\gamma_K\}$. For each such set focus on each edge $e\in G$. If~$e$ has been assigned to two or more $\gamma_{i}$, it belongs to the soft subgraph $S$; if~$e$ has been assigned to exactly one $\gamma_{i}$, it belongs to the jet subgraph $J_i$; if~$e$ has not been assigned to any $\gamma_{i}$, it belongs to $H$. We also denote $J\equiv \cup_{i=1}^K J_i$.
    \item \emph{Step 3}: We now require that the result obtained in Step 2 satisfies the following three further constraints: (i) each jet subgraph $J_i$ is connected; (ii) each hard subgraph $H$ is connected; (iii) each of the $K$ subgraphs $H \cup J \setminus J_i$ ($i=1,\dots,K$) is mojetic. The region would be ruled out if any of these conditions are not satisfied.
\end{itemize}

In the steps above, we have associated every edge $e\in G$ to a specific subgraph. Let us recall the association of vertices, $v\in G$, to these subgraphs. According to the definitions in section~\ref{leading_terms_in_infrared_regions}, any vertex $v\in S$ if and only if it connects to soft edges only; $v\in J_i$ if and only if it connects to at least one edge in $J_i$ but no edges from $H$ nor $J_j$ with $j\neq i$; otherwise $v\in H$. A special case of $v\in H$ occurs when $v$ is met by two different jets $J_i$ and~$J_j$ simultaneously.

In the following we will consider a few examples to demonstrate the workings of the algorithm. Let us start with the one-loop triangle example $G_3$ with one off-shell external momentum $q_1$ and two nearly on-shell external momenta $p_1$ and $p_2$:

\begin{equation}
G_{3}=
\begin{tikzpicture}[baseline=13ex,scale=1.]
\coordinate (x1) at (1.1340, 1.4999) ;
\coordinate (x2) at (2.8660, 1.5000) ;
\coordinate (x3) at (2,3) ;
\node (p1) at (0.2681, 0.9998) {$p_1$};
\node (p2) at (3.7320, 1.0000) {$p_2$};
\node (p3) at (2,4) {$q_1$};
\draw[color=green] (x1) -- (p1);
\draw[ultra thick,color=Blue] (x3) -- (p3);
\draw[color=ForestGreen] (x2) -- (p2);
\draw[ultra thick,color=Black] (x1) -- (x2);
\draw[ultra thick,color=Black] (x2) -- (x3);
\draw[ultra thick,color=Black] (x3) -- (x1);
\draw[fill,thick,color=Blue] (x1) circle (1pt);
\draw[fill,thick,color=Blue] (x2) circle (1pt);
\draw[fill,thick,color=Blue] (x3) circle (1pt);
\end{tikzpicture}
\end{equation}
The algorithm above can generate a single one-external subgraph in the $p_1$ channel, which we denote by $\gamma_1$, and another one in the $p_2$ channel denoted by $\gamma_2$:
\begin{equation}
\gamma_1:
\begin{tikzpicture}[baseline=13ex,scale=1.]
\coordinate (x1) at (1.1340, 1.4999) ;
\coordinate (x2) at (2.8660, 1.5000) ;
\coordinate (x3) at (2,3) ;
\draw[ultra thick, dotted] (1.1340, 2.4999) arc (70:0:1.5);
\node (p1) at (0.2681, 0.9998) {$p_1$};
\node (p2) at (3.7320, 1.0000) {$p_2$};
\node (p3) at (2,4) {$q_1$};
\draw[color=green] (x1) -- (p1);
\draw[color=green] (x1) -- (x2);
\draw[fill,thick,color=green] (x1) circle (1pt);
\draw[color=green] (x3) -- (x1);
\draw[ultra thick,color=Blue] (x3) -- (p3);
\draw[color=ForestGreen] (x2) -- (p2);
\draw[ultra thick,color=Blue] (x2) -- (x3);
\draw[fill,thick,color=Blue] (x2) circle (1pt);
\draw[fill,thick,color=Blue] (x3) circle (1pt);
\end{tikzpicture}
\gamma_2:
\begin{tikzpicture}[baseline=13ex,scale=1.]
\draw[ultra thick, dotted] (2.866, 2.4999) arc (110:180:1.5);
\coordinate (x1) at (1.1340, 1.4999) ;
\coordinate (x2) at (2.8660, 1.5000) ;
\coordinate (x3) at (2,3) ;
\node (p1) at (0.2681, 0.9998) {$p_1$};
\node (p2) at (3.7320, 1.0000) {$p_2$};
\node (p3) at (2,4) {$q_1$};
\draw[color=green] (x1) -- (p1);
\draw[color=ForestGreen] (x2) -- (p2);
\draw[fill,thick,color=ForestGreen] (x2) circle (1pt);
\draw[color=ForestGreen] (x1) -- (x2);
\draw[color=ForestGreen] (x2) -- (x3);
\draw[ultra thick,color=Blue] (x3) -- (p3);
\draw[ultra thick,color=Blue] (x3) -- (x1);
\draw[fill,thick,color=Blue] (x1) circle (1pt);
\draw[fill,thick,color=Blue] (x3) circle (1pt);
\end{tikzpicture}
\end{equation}
Note that $\gamma_i$ ($i=1,2$) denotes the green subgraph that is enclosed by the cut (the dotted curve) that contains $p_i$, including the cut propagators.
For each of these one-external subgraphs we now check whether or not $H_i\equiv G_3 \setminus \gamma_{i}$ is mojetic. Let us explicitly check this for $G_3 \setminus \gamma_1$,
\begin{equation}
G_3\setminus \gamma_1=
\begin{tikzpicture}[baseline=13ex,scale=1.]
\coordinate (x1) at (1.1340, 1.4999) ;
\coordinate (x2) at (2.8660, 1.5000) ;
\coordinate (x3) at (2,3) ;
\node (p2) at (3.7320, 1.0000) {$p_2$};
\node (p3) at (2,4) {$q_1$};
\draw[ultra thick,color=Blue] (x3) -- (p3);
\draw[color=ForestGreen] (x2) -- (p2);
\draw[ultra thick,color=Blue] (x2) -- (x3);
\draw[fill,thick,color=Blue] (x2) circle (1pt);
\draw[fill,thick,color=Blue] (x3) circle (1pt);
\end{tikzpicture},
\qquad
\left[G_3\setminus \gamma_1\right]_c=
\begin{tikzpicture}[baseline=13ex,scale=1.]
\coordinate (x1) at (1.1340, 1.4999) ;
\coordinate (x2) at (2.8660, 1.5000) ;
\coordinate (x3) at (2,3) ;
\draw[color=ForestGreen] (x1) -- (x2);
\draw[fill,ultra thick,color=magenta] (x1) circle (1pt);
\draw[ultra thick,color=Blue] (x3) -- (x1);
\draw[ultra thick,color=Blue] (x2) -- (x3);
\draw[fill,thick,color=Blue] (x2) circle (1pt);
\draw[fill,thick,color=Blue] (x3) circle (1pt);
\end{tikzpicture} .
\label{algorithm_example_step1_check}
\end{equation}
Here, in the second equation, we have added an auxiliary (pink) vertex connecting the edges carrying the external momenta $p_1$ and $q_1$. This graph, which we denote by $\left[G_3\setminus \gamma_1\right]_c$, is evidently 1VI, so $G_3 \setminus \gamma_1$ is mojetic. By symmetry, $G_3 \setminus \gamma_2$ is also mojetic. This concludes Step 1 of the algorithm.

In Step 2 we combine the subgraphs $\gamma_1$ and $\gamma_2$, by an \emph{overlay operation} which we denote by the symbol~$\sqcup$. The resulting graph $\gamma_1 \sqcup \gamma_2$ is
\begin{equation}
\gamma_1 \sqcup \gamma_2: 
\begin{tikzpicture}[baseline=13ex,scale=1.]
\coordinate (x1) at (1.1340, 1.4999) ;
\coordinate (x2) at (2.8660, 1.5000) ;
\coordinate (x3) at (2,3) ;
\node (p1) at (0.2681, 0.9998) {$p_1$};
\node (p2) at (3.7320, 1.0000) {$p_2$};
\node (p3) at (2,4) {$q_1$};

\draw[color=red,dashed] (x1) -- (x2);

\draw[fill,thick,color=green] (x1) circle (1pt);
\draw[color=green] (x1) -- (p1);
\draw[color=green] (x1) -- (x3);
\draw[color=ForestGreen] (x2) -- (p2);
\draw[fill,thick,color=ForestGreen] (x2) circle (1pt);
\draw[color=ForestGreen] (x3) -- (x2);
\draw[ultra thick,color=Blue] (x3) -- (p3);
\draw[fill,thick,color=Blue] (x3) circle (1pt);
\end{tikzpicture}
\end{equation}
The edge connecting $p_1$ and $q_1$ belongs only to $\gamma_1$, and hence is part of $J_1$. Similarly, the edge connecting $p_2$ and $q_1$ is part of $J_2$. The edge connecting $p_1$ and $p_2$ belongs to both $\gamma_1$ and $\gamma_2$, and hence is part of $S$. Note that we have used different shades of green to denote different jets, and a red dashed line to denote the soft edge, while the hard subgraph is drawn in blue. We will keep using the same colour scheme for all examples.

Step 3 then demands that we check that the resulting graph $\gamma_1 \sqcup \gamma_2$ satisfies three properties. The first two properties about the connectivity of each $J_i$ and $H$ are clearly satisfied. The third property requires for all $H\cup J\setminus J_i$ to be mojetic, which in this particular example coincides with the check in eq.~(\ref{algorithm_example_step1_check}).

We note that the unitarity cuts that construct the one-external subgraphs can be ``crossed''\footnote{It is interesting to note that in the context of computing iterated discontinuities, crossed cuts of this type are excluded, see~ref.~\cite{Abreu:2014cla}.}, i.e. there can be vertices $v\in \gamma_i\cap \gamma_j$ for $i\neq j$. Indeed, such crossed untarity cuts are necessary to generate soft vertices. A three-loop example is given by:\\
\begin{equation}
\gamma_1:
\begin{tikzpicture}[baseline=13ex,scale=1.]
\coordinate (x1) at (1, 1.5) ;
\coordinate (x2) at (3, 1.5) ;
\coordinate (x3) at (2,3) ;
\draw (1.5,2.75) edge [dotted, ultra thick, color=Black, bend left =30] (2.5,1.25) node [right] {};
\node (p1) at (0.5, 1) {$p_1$};
\node (p2) at (3.5, 1) {$p_2$};
\node (p3) at (2,3.8) {$q_1$};
\draw[color=green] (x1) -- (p1);
\draw[color=green] (x1) -- (x2);
\draw[color=green] (1.4,2.1) -- (2,2.1);
\draw[color=green] (2,2.1) -- (2.6,2.1);
\draw[color=green] (2,2.1) -- (2,1.5);
\draw[fill,thick,color=green] (x1) circle (1pt);
\draw[fill,thick,color=green] (1.4,2.1) circle (1pt);
\draw[fill,thick,color=green] (2,2.1) circle (1pt);
\draw[fill,thick,color=green] (2,1.5) circle (1pt);
\draw[color=green] (x3) -- (x1);
\draw[ultra thick,color=Blue] (x3) -- (p3);
\draw[color=ForestGreen] (x2) -- (p2);
\draw[ultra thick,color=Blue] (x2) -- (x3);
\draw[fill,thick,color=Blue] (x2) circle (1pt);
\draw[fill,thick,color=Blue] (x3) circle (1pt);
\draw[fill,thick,color=Blue] (2.6,2.1) circle (1pt);
\end{tikzpicture}
\gamma_2:
\begin{tikzpicture}[baseline=13ex,scale=1.]
\draw (2.5,2.75) edge [dotted, ultra thick, color=Black, bend right =30] (1.5,1.25) node [right] {};
\coordinate (x1) at (1,1.5) ;
\coordinate (x2) at (3, 1.5) ;
\coordinate (x3) at (2,3) ;
\node (p1) at (0.5, 1) {$p_1$};
\node (p2) at (3.5, 1) {$p_2$};
\node (p3) at (2,3.8) {$q_1$};
\draw[color=green] (x1) -- (p1);
\draw[color=ForestGreen] (x2) -- (p2);
\draw[fill,thick,color=ForestGreen] (x2) circle (1pt);
\draw[color=ForestGreen] (x1) -- (x2);
\draw[color=ForestGreen] (x2) -- (x3);
\draw[color=ForestGreen] (1.4,2.1) -- (2,2.1);
\draw[color=ForestGreen] (2,2.1) -- (2.6,2.1);
\draw[color=ForestGreen] (2,2.1) -- (2,1.5);
\draw[fill,thick,color=ForestGreen] (2.6,2.1) circle (1pt);
\draw[fill,thick,color=ForestGreen] (2,2.1) circle (1pt);
\draw[fill,thick,color=ForestGreen] (2,1.5) circle (1pt);
\draw[ultra thick,color=Blue] (x3) -- (p3);
\draw[ultra thick,color=Blue] (x3) -- (x1);
\draw[fill,thick,color=Blue] (x1) circle (1pt);
\draw[fill,thick,color=Blue] (x3) circle (1pt);
\draw[fill,thick,color=Blue] (1.4,2.1) circle (1pt);
\end{tikzpicture}
\Rightarrow \gamma_1\sqcup\gamma_2:
\begin{tikzpicture}[baseline=13ex,scale=1.]
\coordinate (x1) at (1,1.5) ;
\coordinate (x2) at (3, 1.5) ;
\coordinate (x3) at (2,3) ;
\node (p1) at (0.5, 1) {$p_1$};
\node (p2) at (3.5, 1) {$p_2$};
\node (p3) at (2,3.8) {$q_1$};
\draw[color=green] (x1) -- (p1);
\draw[color=ForestGreen] (x2) -- (p2);
\draw[fill,thick,color=ForestGreen] (x2) circle (1pt);
\draw[dashed, color=Red] (x1) -- (x2);
\draw[color=ForestGreen] (x2) -- (x3);
\draw[dashed, color=Red] (1.4,2.1) -- (2,2.1);
\draw[dashed, color=Red] (2,2.1) -- (2.6,2.1);
\draw[dashed, color=Red] (2,2.1) -- (2,1.5);
\draw[fill,thick,color=ForestGreen] (2.6,2.1) circle (1pt);
\draw[fill,thick,color=Red] (2,2.1) circle (1pt);
\draw[fill,thick,color=Red] (2,1.5) circle (1pt);
\draw[ultra thick,color=Blue] (x3) -- (p3);
\draw[color=green] (x3) -- (x1);
\draw[fill,thick,color=green] (x1) circle (1pt);
\draw[fill,thick,color=Blue] (x3) circle (1pt);
\draw[fill,thick,color=green] (1.4,2.1) circle (1pt);
\end{tikzpicture}
\vspace*{10pt}
\end{equation}

We now show an example of a candidate region that is incompatible with the mojetic criterion in Step 1. Consider the following graph with external momenta $p_1$, $p_2$ and~$q_1$:
\begin{equation}
\gamma_2:
\begin{tikzpicture}[baseline=7ex,scale=1.]
\coordinate (x1) at (1.1340, 1.5) ;
\coordinate (x2) at (2.8660, 1.5) ;
\coordinate (x3) at (2,3) ;
\coordinate (x4) at (1.1340,0) ;
\coordinate (x5) at (2.8660,0) ;
\node (p1) at (0.2681, -0.9998) {$p_1$};
\node (q1) at (3.7320, -1.0000) {$q_1$};
\node (p2) at (2,4) {$p_2$};

\draw[color=green] (x4) -- (p1);
\draw[color=Green] (x3) -- (p2);
\draw[color=Green] (x2) -- (x3);
\draw[color=Green] (x1) -- (x2);
\draw[color=Green] (x3) -- (x1);
\draw[color=Green] (x2) -- (x5);
\draw[fill,thick,color=Green] (x2) circle (1pt);
\draw[fill,thick,color=Green] (x3) circle (1pt);

\draw[ultra thick,color=Blue] (x5) -- (q1);

\draw [ultra thick,color=Blue] (x4)  arc (230:130:0.97);
\draw [ultra thick,color=Blue] (x4)  arc (310:410:0.97);
\draw[ultra thick,color=Blue] (x4) -- (x5);

\draw[fill,thick,color=Blue] (x1) circle (1pt);
\draw[fill,thick,color=Blue] (x4) circle (1pt);
\draw[fill,thick,color=Blue] (x5) circle (1pt);
\draw[ultra thick, dotted] (1.3, 2.6) arc (190:265:2);

\end{tikzpicture}
\qquad \left[ G\setminus \gamma_2 \right]_c=
\quad 
\begin{tikzpicture}[baseline=0ex,scale=1.]
\coordinate (x1) at (1.1340, 1.5) ;
\coordinate (x2) at (2.8660, 1.5) ;
\coordinate (x3) at (2,3) ;
\coordinate (x4) at (1.1340,0) ;
\coordinate (x5) at (2.8660,0) ;
\coordinate (p1) at (2, -1) ;
\coordinate (q1) at (2, -1) ;

\draw[color=green] (x4) -- (p1);

\draw[ultra thick,color=Blue] (x5) -- (q1);

\draw [ultra thick,color=Blue] (x4)  arc (230:130:0.97);
\draw [ultra thick,color=Blue] (x4)  arc (310:410:0.97);
\draw[ultra thick,color=Blue] (x4) -- (x5);

\draw[fill,thick,color=Blue] (x1) circle (1pt);
\draw[fill,thick,color=Blue] (x4) circle (1pt);
\draw[fill,thick,color=Blue] (x5) circle (1pt);
\draw[fill,thick,color=magenta] (p1) circle (1pt);
\end{tikzpicture}
\end{equation}
The subgraph $G\setminus \gamma_2$, which is the complement of the jet subgraph $\gamma_2$ above is not mojetic, since $\left[ G\setminus \gamma_2 \right]_c$ is not 1VI. (Note that here $J_1$ corresponds to a trivial jet, containing only the external massless momentum $p_1$.) Therefore this choice of $\gamma_2$ is ruled out in Step 1.

\bigbreak
Above we have introduced an algorithm to construct the regions that appear in the on-shell expansion of a given wide-angle scattering graph $G$, and explained it using some low-order examples. We now show that this algorithm constructs exactly those regions satisfying the two conditions in theorem~\ref{theorem-algorithm_necessary_sufficient}, i.e. $H\cup J\setminus J_i$ is mojetic for every $i$, and every connected component of $S$ is attached to at least two different jets.

Let us begin by showing that any given region obtained from the algorithm above satisfies these two conditions. First, in Step 3 we have discarded all the unqualified regions that violate the property that $H\cup J\setminus J_i$ are mojetic. Second, since all the soft edges are obtained from the intersection of different jet edges, and each jet is internally connected (from Step 3), it follows that any connected component of $S$ necessarily connects at least two different jets. So both conditions in theorem~\ref{theorem-algorithm_necessary_sufficient} are automatically satisfied.

We now show that any region $R$ satisfying proposition~\ref{proposition-region_vectors_are_hard_and_infrared} and the conditions of theorem~\ref{theorem-algorithm_necessary_sufficient} can be obtained from a suitable choice of $\{\gamma_i^{(R)}\}$ for $i=1,\dots,K$, each of which is a one-external subgraph in the $p_i$ channel. To construct these $\gamma_i^{(R)}$, we first denote by $S_i$ the union of the connected components of the soft subgraph $S$ that are attached to $J_i$. The vertices of $G$ are then automatically partitioned into the following two types: those within $J_i\cup S_i$ and those within $G\setminus (J_i\cup S_i)$. It follows that $S_i\cup J_i$ has the following properties: (1) any (internal) edge $e$ in $S_i\cup J_i$ connects to at least one vertex within $J_i\cup S_i$; (2) it has a single external momentum $p_i$. These properties imply that $S_i\cup J_i$ is a one-external subgraph in the $p_i$ channel, and we define $\gamma_i^{(R)}\equiv S_i\cup J_i$. Since $S=\cup_i S_i$ and $S\cup J=\cup_i (J_i\cup S_i)$, it follows that any region $R$ can be obtained from a set of $\gamma_i$. This completes the proof that the algorithm reproduces all possible regions satisfying the conditions in theorem~\ref{theorem-algorithm_necessary_sufficient}.

\subsection{More examples and implementation}
\label{more_examples_implementation}

As another example, we show below how to find the regions for the following $3\times 2$ fishnet graph.
\begin{equation}
G_{3\times2}=
\begin{tikzpicture}[baseline=5ex,scale=1.]
\coordinate (x1) at (0,0) ;
\coordinate (x2) at (0,1) ;
\coordinate (x3) at (0,2) ;
\coordinate (x6) at (1,2) ;
\coordinate (x5) at (1,1) ;
\coordinate (x4) at (1,0) ;
\coordinate (x9) at (2,2) ;
\coordinate (x8) at (2,1) ;
\coordinate (x7) at (2,0) ;
\coordinate (x12) at (3,2) ;
\coordinate (x11) at (3,1) ;
\coordinate (x10) at (3,0) ;

\node (p4) at (-0.5,-0.5) {$4$};
\node (p1) at (-0.5,2.5) {$1$};
\node (p3) at (3.5,-0.5) {$3$};
\node (p2) at (3.5,2.5) {$2$};

\draw[ultra thick] (x1) -- (p4);
\draw[ultra thick] (x3) -- (p1);
\draw[ultra thick] (x12) -- (p2);
\draw[ultra thick] (x10) -- (p3);

\draw[ultra thick,color=Black] (x1) -- (x2);
\draw[ultra thick,color=Black] (x2) -- (x3);
\draw[ultra thick,color=Black] (x1) -- (x4);
\draw[ultra thick,color=Black] (x2) -- (x5);
\draw[ultra thick,color=Black] (x3) -- (x6);
\draw[ultra thick,color=Black] (x4) -- (x5);
\draw[ultra thick,color=Black] (x5) -- (x6);
\draw[ultra thick,color=Black] (x4) -- (x7);
\draw[ultra thick,color=Black] (x5) -- (x8);
\draw[ultra thick,color=Black] (x6) -- (x9);
\draw[ultra thick,color=Black] (x7) -- (x8);
\draw[ultra thick,color=Black] (x8) -- (x9);
\draw[ultra thick,color=Black] (x7) -- (x10);
\draw[ultra thick,color=Black] (x8) -- (x11);
\draw[ultra thick,color=Black] (x9) -- (x12);
\draw[ultra thick,color=Black] (x10) -- (x11);
\draw[ultra thick,color=Black] (x11) -- (x12);

\draw[fill,thick,color=Black] (x1) circle (1pt);
\draw[fill,thick,color=Black] (x2) circle (1pt);
\draw[fill,thick,color=Black] (x3) circle (1pt);
\draw[fill,thick,color=Black] (x4) circle (1pt);
\draw[fill,thick,color=Black] (x5) circle (1pt);
\draw[fill,thick,color=Black] (x6) circle (1pt);
\draw[fill,thick,color=Black] (x7) circle (1pt);
\draw[fill,thick,color=Black] (x8) circle (1pt);
\draw[fill,thick,color=Black] (x9) circle (1pt);
\draw[fill,thick,color=Black] (x10) circle (1pt);
\draw[fill,thick,color=Black] (x11) circle (1pt);
\draw[fill,thick,color=Black] (x12) circle (1pt);
\end{tikzpicture}
\end{equation}

We first consider the condition that the momenta at the diagonal positions $1$ and $3$ are nearly on-shell (in different directions), and those at $2$ and $4$ are off-shell, i.e. using the conventions above, the external momenta of $G_{3\times2}$ are $p_1$, $p_3$, $q_2$ and $q_4$. Using the algorithm of section~\ref{algorithm_for_regions}, one possible choice for the one-external subgraphs $\gamma_1$ and $\gamma_3$ is highlighted with thinner green lines in the following:
\begin{equation}
\gamma_1:
\begin{tikzpicture}[baseline=5ex,scale=1.]
\coordinate (x1) at (0,0) ;
\coordinate (x2) at (0,1) ;
\coordinate (x3) at (0,2) ;
\coordinate (x6) at (1,2) ;
\coordinate (x5) at (1,1) ;
\coordinate (x4) at (1,0) ;
\coordinate (x9) at (2,2) ;
\coordinate (x8) at (2,1) ;
\coordinate (x7) at (2,0) ;
\coordinate (x12) at (3,2) ;
\coordinate (x11) at (3,1) ;
\coordinate (x10) at (3,0) ;

\node (q4) at (-0.5,-0.5) {$q_4$};
\node (p1) at (-0.5,2.5) {$p_1$};
\node (p3) at (3.5,-0.5) {$p_3$};
\node (q2) at (3.5,2.5) {$q_2$};

\draw[ultra thick, color=Blue] (x1) -- (q4);
\draw[color=LimeGreen] (x3) -- (p1);
\draw[color=LimeGreen] (x3) -- (x6);
\draw[color=LimeGreen] (x2) -- (x3);
\draw[fill,color=LimeGreen] (x3) circle (1pt);

\draw[ultra thick, color=Blue] (x12) -- (q2);
\draw[color=Green] (x10) -- (p3);

\draw[ultra thick,color=Blue] (x1) -- (x2);
\draw[ultra thick,color=Blue] (x1) -- (x4);
\draw[ultra thick,color=Blue] (x2) -- (x5);

\draw[ultra thick,color=Blue] (x4) -- (x5);
\draw[ultra thick,color=Blue] (x5) -- (x6);
\draw[ultra thick,color=Blue] (x4) -- (x7);
\draw[ultra thick,color=Blue] (x5) -- (x8);
\draw[ultra thick,color=Blue] (x6) -- (x9);
\draw[ultra thick,color=Blue] (x7) -- (x8);
\draw[ultra thick,color=Blue] (x8) -- (x9);
\draw[ultra thick,color=Blue] (x7) -- (x10);
\draw[ultra thick,color=Blue] (x8) -- (x11);
\draw[ultra thick,color=Blue] (x9) -- (x12);
\draw[ultra thick,color=Blue] (x10) -- (x11);
\draw[ultra thick,color=Blue] (x11) -- (x12);

\draw[fill,thick,color=Blue] (x1) circle (1pt);
\draw[fill,thick,color=Blue] (x2) circle (1pt);

\draw[fill,thick,color=Blue] (x4) circle (1pt);
\draw[fill,thick,color=Blue] (x5) circle (1pt);
\draw[fill,thick,color=Blue] (x6) circle (1pt);
\draw[fill,thick,color=Blue] (x7) circle (1pt);
\draw[fill,thick,color=Blue] (x8) circle (1pt);
\draw[fill,thick,color=Blue] (x9) circle (1pt);
\draw[fill,thick,color=Blue] (x10) circle (1pt);
\draw[fill,thick,color=Blue] (x11) circle (1pt);
\draw[fill,thick,color=Blue] (x12) circle (1pt);
\end{tikzpicture}
,\qquad
\gamma_3:
\begin{tikzpicture}[baseline=5ex,scale=1.]
\coordinate (x1) at (0,0) ;
\coordinate (x2) at (0,1) ;
\coordinate (x3) at (0,2) ;
\coordinate (x6) at (1,2) ;
\coordinate (x5) at (1,1) ;
\coordinate (x4) at (1,0) ;
\coordinate (x9) at (2,2) ;
\coordinate (x8) at (2,1) ;
\coordinate (x7) at (2,0) ;
\coordinate (x12) at (3,2) ;
\coordinate (x11) at (3,1) ;
\coordinate (x10) at (3,0) ;

\node (q4) at (-0.5,-0.5) {$q_4$};
\node (p1) at (-0.5,2.5) {$p_1$};
\node (p3) at (3.5,-0.5) {$p_3$};
\node (q2) at (3.5,2.5) {$q_2$};

\draw[ultra thick, color=Blue] (x1) -- (q4);
\draw[color=LimeGreen] (x3) -- (p1);
\draw[ultra thick, color=Blue] (x12) -- (q2);

\draw[color=Green] (x10) -- (p3);
\draw[fill,thick,color=Green] (x7) circle (1pt);
\draw[fill,thick,color=Green] (x11) circle (1pt);
\draw[color=Green] (x4) -- (x7);
\draw[color=Green] (x7) -- (x8);
\draw[color=Green] (x7) -- (x10);
\draw[color=Green] (x8) -- (x11);
\draw[color=Green] (x10) -- (x11);
\draw[color=Green] (x11) -- (x12);

\draw[ultra thick,color=Blue] (x1) -- (x2);
\draw[ultra thick,color=Blue] (x2) -- (x3);
\draw[ultra thick,color=Blue] (x1) -- (x4);
\draw[ultra thick,color=Blue] (x2) -- (x5);
\draw[ultra thick,color=Blue] (x3) -- (x6);
\draw[ultra thick,color=Blue] (x4) -- (x5);
\draw[ultra thick,color=Blue] (x5) -- (x6);
\draw[ultra thick,color=Blue] (x5) -- (x8);
\draw[ultra thick,color=Blue] (x6) -- (x9);
\draw[ultra thick,color=Blue] (x8) -- (x9);
\draw[ultra thick,color=Blue] (x9) -- (x12);

\draw[fill,thick,color=Blue] (x1) circle (1pt);
\draw[fill,thick,color=Blue] (x2) circle (1pt);
\draw[fill,thick,color=Blue] (x3) circle (1pt);
\draw[fill,thick,color=Blue] (x4) circle (1pt);
\draw[fill,thick,color=Blue] (x5) circle (1pt);
\draw[fill,thick,color=Blue] (x6) circle (1pt);
\draw[fill,thick,color=Blue] (x8) circle (1pt);
\draw[fill,thick,color=Blue] (x9) circle (1pt);
\draw[fill,thick,color=Green] (x10) circle (1pt);
\draw[fill,thick,color=Blue] (x12) circle (1pt);
\end{tikzpicture}
\end{equation}

The graphs $\gamma_1$ and $\gamma_3$ do not overlap so there are no soft propagators when we combine (overlay) them:
\begin{equation}
\gamma_1 \sqcup \gamma_3 :
\begin{tikzpicture}[baseline=5ex,scale=1.]
\coordinate (x1) at (0,0) ;
\coordinate (x2) at (0,1) ;
\coordinate (x3) at (0,2) ;
\coordinate (x6) at (1,2) ;
\coordinate (x5) at (1,1) ;
\coordinate (x4) at (1,0) ;
\coordinate (x9) at (2,2) ;
\coordinate (x8) at (2,1) ;
\coordinate (x7) at (2,0) ;
\coordinate (x12) at (3,2) ;
\coordinate (x11) at (3,1) ;
\coordinate (x10) at (3,0) ;

\node (q4) at (-0.5,-0.5) {$q_4$};
\node (p1) at (-0.5,2.5) {$p_1$};
\node (p3) at (3.5,-0.5) {$p_3$};
\node (q2) at (3.5,2.5) {$q_2$};

\draw[ultra thick, color=Blue] (x1) -- (q4);

\draw[color=LimeGreen] (x3) -- (p1);
\draw[color=LimeGreen] (x3) -- (x6);
\draw[color=LimeGreen] (x2) -- (x3);
\draw[fill,color=LimeGreen] (x3) circle (1pt);

\draw[ultra thick, color=Blue] (x12) -- (q2);

\draw[color=Green] (x10) -- (p3);
\draw[fill,thick,color=Green] (x7) circle (1pt);
\draw[fill,thick,color=Green] (x11) circle (1pt);
\draw[fill,thick,color=Green] (x10) circle (1pt);
\draw[color=Green] (x4) -- (x7);
\draw[color=Green] (x7) -- (x8);
\draw[color=Green] (x7) -- (x10);
\draw[color=Green] (x8) -- (x11);
\draw[color=Green] (x10) -- (x11);
\draw[color=Green] (x11) -- (x12);

\draw[ultra thick,color=Blue] (x1) -- (x2);
\draw[ultra thick,color=Blue] (x1) -- (x4);
\draw[ultra thick,color=Blue] (x2) -- (x5);
\draw[ultra thick,color=Blue] (x4) -- (x5);
\draw[ultra thick,color=Blue] (x5) -- (x6);
\draw[ultra thick,color=Blue] (x5) -- (x8);
\draw[ultra thick,color=Blue] (x6) -- (x9);
\draw[ultra thick,color=Blue] (x8) -- (x9);
\draw[ultra thick,color=Blue] (x9) -- (x12);

\draw[fill,thick,color=Blue] (x1) circle (1pt);
\draw[fill,thick,color=Blue] (x2) circle (1pt);
\draw[fill,thick,color=Blue] (x4) circle (1pt);
\draw[fill,thick,color=Blue] (x5) circle (1pt);
\draw[fill,thick,color=Blue] (x6) circle (1pt);
\draw[fill,thick,color=Blue] (x8) circle (1pt);
\draw[fill,thick,color=Blue] (x9) circle (1pt);
\draw[fill,thick,color=Blue] (x12) circle (1pt);
\end{tikzpicture}
\end{equation}
One can readily check that the hard subgraph (thick-blue edges and vertices) fulfills all the required criteria in Step 3 of the algorithm and thus contributes a valid region. 

The following example illustrates an invalid choice of $\gamma_3$.
\begin{equation}
\gamma'_{3}:
\begin{tikzpicture}[baseline=5ex,scale=1.]
\coordinate (x1) at (0,0) ;
\coordinate (x2) at (0,1) ;
\coordinate (x3) at (0,2) ;
\coordinate (x6) at (1,2) ;
\coordinate (x5) at (1,1) ;
\coordinate (x4) at (1,0) ;
\coordinate (x9) at (2,2) ;
\coordinate (x8) at (2,1) ;
\coordinate (x7) at (2,0) ;
\coordinate (x12) at (3,2) ;
\coordinate (x11) at (3,1) ;
\coordinate (x10) at (3,0) ;

\node (q4) at (-0.5,-0.5) {$q_4$};
\node (p1) at (-0.5,2.5) {$p_1$};
\node (p3) at (3.5,-0.5) {$p_3$};
\node (q2) at (3.5,2.5) {$q_2$};

\draw[ultra thick, color=Blue] (x1) -- (q4);
\draw[color=LimeGreen] (x3) -- (p1);
\draw[ultra thick, color=Blue] (x12) -- (q2);

\draw[color=Green] (x10) -- (p3);
\draw[fill,thick,color=Green] (x11) circle (1pt);
\draw[fill,thick,color=Green] (x10) circle (1pt);
\draw[color=Green] (x7) -- (x10);
\draw[color=Green] (x8) -- (x11);
\draw[color=Green] (x10) -- (x11);
\draw[color=Green] (x11) -- (x12);
\draw[color=Green] (x7) -- (x8);

\draw[ultra thick,color=Blue] (x4) -- (x7);
\draw[ultra thick,color=Blue] (x1) -- (x2);
\draw[ultra thick,color=Blue] (x2) -- (x3);
\draw[ultra thick,color=Blue] (x1) -- (x4);
\draw[ultra thick,color=Blue] (x2) -- (x5);
\draw[ultra thick,color=Blue] (x3) -- (x6);
\draw[ultra thick,color=Blue] (x4) -- (x5);
\draw[ultra thick,color=Blue] (x5) -- (x6);
\draw[ultra thick,color=Blue] (x5) -- (x8);
\draw[ultra thick,color=Blue] (x6) -- (x9);
\draw[ultra thick,color=Blue] (x8) -- (x9);
\draw[ultra thick,color=Blue] (x9) -- (x12);

\draw[fill,thick,color=Blue] (x1) circle (1pt);
\draw[fill,thick,color=Blue] (x2) circle (1pt);
\draw[fill,thick,color=Blue] (x3) circle (1pt);
\draw[fill,thick,color=Blue] (x4) circle (1pt);
\draw[fill,thick,color=Blue] (x5) circle (1pt);
\draw[fill,thick,color=Blue] (x6) circle (1pt);
\draw[fill,thick,color=Blue] (x7) circle (1pt);
\draw[fill,thick,color=Blue] (x8) circle (1pt);
\draw[fill,thick,color=Blue] (x9) circle (1pt);
\draw[fill,thick,color=Blue] (x12) circle (1pt);
\end{tikzpicture}    
\end{equation}
The reason is that $\gamma'_3$ does not correspond to a one-external subgraph in the $p_3$ channel, since there is no unitarity cut that would be consistent with $\gamma'_3$.

Furthermore, any choices with overlapping edges between $\gamma_1$ and $\gamma_3$ ($\gamma_1\cap \gamma_3 \neq \varnothing$) cannot be candidates for regions in this case, because that would imply that the subgraph $J\cup S$ separates the two off-shell momenta $q_2$ and $q_4$, creating two disconnected hard components, which do not qualify as a hard subgraph $H$.

We then consider another kinematic configuration, where $p_1$ and $p_2$ are nearly on-shell while $q_3$ and $q_4$ are off-shell. One choice for $\gamma_1$ and $\gamma_2$ is:
\begin{equation}
\gamma_1:
\begin{tikzpicture}[baseline=5ex,scale=1.]
\coordinate (x1) at (0,0) ;
\coordinate (x2) at (0,1) ;
\coordinate (x3) at (0,2) ;
\coordinate (x6) at (1,2) ;
\coordinate (x5) at (1,1) ;
\coordinate (x4) at (1,0) ;
\coordinate (x9) at (2,2) ;
\coordinate (x8) at (2,1) ;
\coordinate (x7) at (2,0) ;
\coordinate (x12) at (3,2) ;
\coordinate (x11) at (3,1) ;
\coordinate (x10) at (3,0) ;

\node (q4) at (-0.5,-0.5) {$q_4$};
\node (p1) at (-0.5,2.5) {$p_1$};
\node (q3) at (3.5,-0.5) {$q_3$};
\node (p2) at (3.5,2.5) {$p_2$};


\draw[ultra thick, color=Blue] (x1) -- (q4);

\draw[color=LimeGreen] (x3) -- (p1);
\draw[color=LimeGreen] (x3) -- (x6);
\draw[color=LimeGreen] (x2) -- (x3);
\draw[color=LimeGreen] (x1) -- (x2);
\draw[color=LimeGreen] (x2) -- (x5);
\draw[color=LimeGreen] (x5) -- (x6);
\draw[color=LimeGreen] (x6) -- (x9);
\draw[color=LimeGreen] (x4) -- (x5);
\draw[color=LimeGreen] (x5) -- (x8);
\draw[fill,thick,color=LimeGreen] (x6) circle (1pt);
\draw[fill,thick,color=LimeGreen] (x2) circle (1pt);
\draw[fill,thick,color=LimeGreen] (x3) circle (1pt);
\draw[fill,thick,color=LimeGreen] (x5) circle (1pt);

\draw[color=ForestGreen] (x12) -- (p2);

\draw[ultra thick, color=Blue] (x10) -- (q3);


\draw[ultra thick,color=Blue] (x1) -- (x4);
\draw[ultra thick,color=Blue] (x4) -- (x7);
\draw[ultra thick,color=Blue] (x7) -- (x8);
\draw[ultra thick,color=Blue] (x8) -- (x9);
\draw[ultra thick,color=Blue] (x7) -- (x10);
\draw[ultra thick,color=Blue] (x8) -- (x11);
\draw[ultra thick,color=Blue] (x9) -- (x12);
\draw[ultra thick,color=Blue] (x10) -- (x11);
\draw[ultra thick,color=Blue] (x11) -- (x12);

\draw[fill,thick,color=Blue] (x1) circle (1pt);
\draw[fill,thick,color=Blue] (x4) circle (1pt);
\draw[fill,thick,color=Blue] (x7) circle (1pt);
\draw[fill,thick,color=Blue] (x8) circle (1pt);
\draw[fill,thick,color=Blue] (x9) circle (1pt);
\draw[fill,thick,color=Blue] (x10) circle (1pt);
\draw[fill,thick,color=Blue] (x11) circle (1pt);
\draw[fill,thick,color=Blue] (x12) circle (1pt);
\end{tikzpicture}
\quad
\gamma_2:
\begin{tikzpicture}[baseline=5ex,scale=1.]
\coordinate (x1) at (0,0) ;
\coordinate (x2) at (0,1) ;
\coordinate (x3) at (0,2) ;
\coordinate (x6) at (1,2) ;
\coordinate (x5) at (1,1) ;
\coordinate (x4) at (1,0) ;
\coordinate (x9) at (2,2) ;
\coordinate (x8) at (2,1) ;
\coordinate (x7) at (2,0) ;
\coordinate (x12) at (3,2) ;
\coordinate (x11) at (3,1) ;
\coordinate (x10) at (3,0) ;

\node (q4) at (-0.5,-0.5) {$q_4$};
\node (p1) at (-0.5,2.5) {$p_1$};
\node (q3) at (3.5,-0.5) {$q_3$};
\node (p2) at (3.5,2.5) {$p_2$};

\draw[ultra thick, color=Blue] (x1) -- (q4);
\draw[color=LimeGreen] (x3) -- (p1);
\draw[color=ForestGreen] (x12) -- (p2);
\draw[fill,thick,color=ForestGreen] (x8) circle (1pt);
\draw[fill,thick,color=ForestGreen] (x9) circle (1pt);
\draw[fill,thick,color=ForestGreen] (x11) circle (1pt);
\draw[fill,thick,color=ForestGreen] (x12) circle (1pt);
\draw[color=ForestGreen] (x5) -- (x8);
\draw[color=ForestGreen] (x6) -- (x9);
\draw[color=ForestGreen] (x7) -- (x8);
\draw[color=ForestGreen] (x8) -- (x9);
\draw[color=ForestGreen] (x8) -- (x11);
\draw[color=ForestGreen] (x9) -- (x12);
\draw[color=ForestGreen] (x10) -- (x11);
\draw[color=ForestGreen] (x11) -- (x12);

\draw[ultra thick, color=Blue] (x10) -- (q3);

\draw[ultra thick,color=Blue] (x1) -- (x2);
\draw[ultra thick,color=Blue] (x2) -- (x3);
\draw[ultra thick,color=Blue] (x1) -- (x4);
\draw[ultra thick,color=Blue] (x2) -- (x5);
\draw[ultra thick,color=Blue] (x3) -- (x6);
\draw[ultra thick,color=Blue] (x4) -- (x5);
\draw[ultra thick,color=Blue] (x5) -- (x6);
\draw[ultra thick,color=Blue] (x4) -- (x7);
\draw[ultra thick,color=Blue] (x7) -- (x10);

\draw[fill,thick,color=Blue] (x1) circle (1pt);
\draw[fill,thick,color=Blue] (x2) circle (1pt);
\draw[fill,thick,color=Blue] (x3) circle (1pt);
\draw[fill,thick,color=Blue] (x4) circle (1pt);
\draw[fill,thick,color=Blue] (x5) circle (1pt);
\draw[fill,thick,color=Blue] (x6) circle (1pt);
\draw[fill,thick,color=Blue] (x7) circle (1pt);
\draw[fill,thick,color=Blue] (x10) circle (1pt);
\end{tikzpicture}
\end{equation}
In this case we do have a nontrivial overlap between $\gamma_1$ and $\gamma_2$, namely, $\gamma_1\cap\gamma_2\ne\varnothing$. Hence overlaying the two we obtain the following graph, which also includes soft lines (the red dashed lines):
\begin{equation}
\gamma_1\sqcup\gamma_2:    
\begin{tikzpicture}[baseline=5ex,scale=1.]
\coordinate (x1) at (0,0) ;
\coordinate (x2) at (0,1) ;
\coordinate (x3) at (0,2) ;
\coordinate (x6) at (1,2) ;
\coordinate (x5) at (1,1) ;
\coordinate (x4) at (1,0) ;
\coordinate (x9) at (2,2) ;
\coordinate (x8) at (2,1) ;
\coordinate (x7) at (2,0) ;
\coordinate (x12) at (3,2) ;
\coordinate (x11) at (3,1) ;
\coordinate (x10) at (3,0) ;

\node (q4) at (-0.5,-0.5) {$q_4$};
\node (p1) at (-0.5,2.5) {$p_1$};
\node (q3) at (3.5,-0.5) {$q_3$};
\node (p2) at (3.5,2.5) {$p_2$};

\draw[color=Red,dashed] (x6) -- (x9);
\draw[color=Red,dashed] (x5) -- (x8);

\draw[ultra thick, color=Blue] (x1) -- (p4);
\draw[color=LimeGreen] (x3) -- (p1);
\draw[color=LimeGreen] (x3) -- (x6);
\draw[color=LimeGreen] (x2) -- (x3);
\draw[color=LimeGreen] (x1) -- (x2);
\draw[color=LimeGreen] (x2) -- (x5);
\draw[color=LimeGreen] (x5) -- (x6);
\draw[color=LimeGreen] (x4) -- (x5);
\draw[fill,thick,color=LimeGreen] (x6) circle (1pt);
\draw[fill,thick,color=LimeGreen] (x2) circle (1pt);
\draw[fill,thick,color=LimeGreen] (x3) circle (1pt);
\draw[fill,thick,color=LimeGreen] (x5) circle (1pt);

\draw[color=ForestGreen] (x12) -- (p2);
\draw[fill,thick,color=ForestGreen] (x8) circle (1pt);
\draw[fill,thick,color=ForestGreen] (x9) circle (1pt);
\draw[fill,thick,color=ForestGreen] (x11) circle (1pt);
\draw[fill,thick,color=ForestGreen] (x12) circle (1pt);
\draw[color=ForestGreen] (x7) -- (x8);
\draw[color=ForestGreen] (x8) -- (x9);
\draw[color=ForestGreen] (x8) -- (x11);
\draw[color=ForestGreen] (x9) -- (x12);
\draw[color=ForestGreen] (x10) -- (x11);
\draw[color=ForestGreen] (x11) -- (x12);

\draw[ultra thick, color=Blue] (x10) -- (p3);

\draw[ultra thick,color=Blue] (x1) -- (x4);
\draw[ultra thick,color=Blue] (x4) -- (x7);
\draw[ultra thick,color=Blue] (x7) -- (x10);

\draw[fill,thick,color=Blue] (x1) circle (1pt);
\draw[fill,thick,color=Blue] (x4) circle (1pt);
\draw[fill,thick,color=Blue] (x7) circle (1pt);
\draw[fill,thick,color=Blue] (x10) circle (1pt);
\end{tikzpicture}
\end{equation}
One may verify $H \cup J_1$ and $H\cup J_2$ are both mojetic.

Finally, we consider the kinematic configuration where all the external momenta are on-shell.
In particular let us choose of the following four one-external subgraphs $\gamma_1$, $\gamma_2$, $\gamma_3$ and $\gamma_4$.
\begin{equation}
\begin{aligned}
& \gamma_1:
\begin{tikzpicture}[baseline=5ex,scale=1.]
\coordinate (x1) at (0,0) ;
\coordinate (x2) at (0,1) ;
\coordinate (x3) at (0,2) ;
\coordinate (x6) at (1,2) ;
\coordinate (x5) at (1,1) ;
\coordinate (x4) at (1,0) ;
\coordinate (x9) at (2,2) ;
\coordinate (x8) at (2,1) ;
\coordinate (x7) at (2,0) ;
\coordinate (x12) at (3,2) ;
\coordinate (x11) at (3,1) ;
\coordinate (x10) at (3,0) ;

\node (q4) at (-0.5,-0.5) {$p_4$};
\node (p1) at (-0.5,2.5) {$p_1$};
\node (q3) at (3.5,-0.5) {$p_3$};
\node (p2) at (3.5,2.5) {$p_2$};


\draw[color=green] (x1) -- (q4);

\draw[color=LimeGreen] (x3) -- (p1);
\draw[color=LimeGreen] (x3) -- (x6);
\draw[color=LimeGreen] (x2) -- (x3);
\draw[color=LimeGreen] (x1) -- (x2);
\draw[color=LimeGreen] (x2) -- (x5);
\draw[color=LimeGreen] (x5) -- (x6);
\draw[color=LimeGreen] (x6) -- (x9);
\draw[color=LimeGreen] (x4) -- (x5);
\draw[color=LimeGreen] (x5) -- (x8);
\draw[fill,thick,color=LimeGreen] (x6) circle (1pt);
\draw[fill,thick,color=LimeGreen] (x2) circle (1pt);
\draw[fill,thick,color=LimeGreen] (x3) circle (1pt);
\draw[fill,thick,color=LimeGreen] (x5) circle (1pt);

\draw[color=ForestGreen] (x12) -- (p2);

\draw[color=Green] (x10) -- (q3);


\draw[ultra thick,color=Blue] (x1) -- (x4);
\draw[ultra thick,color=Blue] (x4) -- (x7);
\draw[ultra thick,color=Blue] (x7) -- (x8);
\draw[ultra thick,color=Blue] (x8) -- (x9);
\draw[ultra thick,color=Blue] (x7) -- (x10);
\draw[ultra thick,color=Blue] (x8) -- (x11);
\draw[ultra thick,color=Blue] (x9) -- (x12);
\draw[ultra thick,color=Blue] (x10) -- (x11);
\draw[ultra thick,color=Blue] (x11) -- (x12);

\draw[fill,thick,color=Blue] (x1) circle (1pt);
\draw[fill,thick,color=Blue] (x4) circle (1pt);
\draw[fill,thick,color=Blue] (x7) circle (1pt);
\draw[fill,thick,color=Blue] (x8) circle (1pt);
\draw[fill,thick,color=Blue] (x9) circle (1pt);
\draw[fill,thick,color=Blue] (x10) circle (1pt);
\draw[fill,thick,color=Blue] (x11) circle (1pt);
\draw[fill,thick,color=Blue] (x12) circle (1pt);
\end{tikzpicture}
\quad
\gamma_2:
\begin{tikzpicture}[baseline=5ex,scale=1.]
\coordinate (x1) at (0,0) ;
\coordinate (x2) at (0,1) ;
\coordinate (x3) at (0,2) ;
\coordinate (x6) at (1,2) ;
\coordinate (x5) at (1,1) ;
\coordinate (x4) at (1,0) ;
\coordinate (x9) at (2,2) ;
\coordinate (x8) at (2,1) ;
\coordinate (x7) at (2,0) ;
\coordinate (x12) at (3,2) ;
\coordinate (x11) at (3,1) ;
\coordinate (x10) at (3,0) ;

\node (q4) at (-0.5,-0.5) {$p_4$};
\node (p1) at (-0.5,2.5) {$p_1$};
\node (q3) at (3.5,-0.5) {$p_3$};
\node (p2) at (3.5,2.5) {$p_2$};

\draw[color=green] (x1) -- (q4);
\draw[color=LimeGreen] (x3) -- (p1);
\draw[color=ForestGreen] (x12) -- (p2);
\draw[fill,thick,color=ForestGreen] (x8) circle (1pt);
\draw[fill,thick,color=ForestGreen] (x9) circle (1pt);
\draw[fill,thick,color=ForestGreen] (x11) circle (1pt);
\draw[fill,thick,color=ForestGreen] (x12) circle (1pt);
\draw[color=ForestGreen] (x5) -- (x8);
\draw[color=ForestGreen] (x6) -- (x9);
\draw[color=ForestGreen] (x7) -- (x8);
\draw[color=ForestGreen] (x8) -- (x9);
\draw[color=ForestGreen] (x8) -- (x11);
\draw[color=ForestGreen] (x9) -- (x12);
\draw[color=ForestGreen] (x10) -- (x11);
\draw[color=ForestGreen] (x11) -- (x12);

\draw[color=Green] (x10) -- (q3);

\draw[ultra thick,color=Blue] (x1) -- (x2);
\draw[ultra thick,color=Blue] (x2) -- (x3);
\draw[ultra thick,color=Blue] (x1) -- (x4);
\draw[ultra thick,color=Blue] (x2) -- (x5);
\draw[ultra thick,color=Blue] (x3) -- (x6);
\draw[ultra thick,color=Blue] (x4) -- (x5);
\draw[ultra thick,color=Blue] (x5) -- (x6);
\draw[ultra thick,color=Blue] (x4) -- (x7);
\draw[ultra thick,color=Blue] (x7) -- (x10);

\draw[fill,thick,color=Blue] (x1) circle (1pt);
\draw[fill,thick,color=Blue] (x2) circle (1pt);
\draw[fill,thick,color=Blue] (x3) circle (1pt);
\draw[fill,thick,color=Blue] (x4) circle (1pt);
\draw[fill,thick,color=Blue] (x5) circle (1pt);
\draw[fill,thick,color=Blue] (x6) circle (1pt);
\draw[fill,thick,color=Blue] (x7) circle (1pt);
\draw[fill,thick,color=Blue] (x10) circle (1pt);
\end{tikzpicture}
\\
& \gamma_3:
\begin{tikzpicture}[baseline=5ex,scale=1.]
\coordinate (x1) at (0,0) ;
\coordinate (x2) at (0,1) ;
\coordinate (x3) at (0,2) ;
\coordinate (x6) at (1,2) ;
\coordinate (x5) at (1,1) ;
\coordinate (x4) at (1,0) ;
\coordinate (x9) at (2,2) ;
\coordinate (x8) at (2,1) ;
\coordinate (x7) at (2,0) ;
\coordinate (x12) at (3,2) ;
\coordinate (x11) at (3,1) ;
\coordinate (x10) at (3,0) ;

\node (p4) at (-0.5,-0.5) {$p_4$};
\node (p1) at (-0.5,2.5) {$p_1$};
\node (p3) at (3.5,-0.5) {$p_3$};
\node (p2) at (3.5,2.5) {$p_2$};

\draw[color=green] (x1) -- (p4);
\draw[color=LimeGreen] (x3) -- (p1);
\draw[color=ForestGreen] (x12) -- (p2);
\draw[color=Green] (x10) -- (p3);
\draw[fill,thick,color=Green] (x10) circle (1pt);
\draw[color=Green] (x7) -- (x10);
\draw[color=Green] (x10) -- (x11);

\draw[ultra thick,color=Blue] (x4) -- (x7);
\draw[ultra thick,color=Blue] (x1) -- (x2);
\draw[ultra thick,color=Blue] (x2) -- (x3);
\draw[ultra thick,color=Blue] (x1) -- (x4);
\draw[ultra thick,color=Blue] (x2) -- (x5);
\draw[ultra thick,color=Blue] (x3) -- (x6);
\draw[ultra thick,color=Blue] (x4) -- (x5);
\draw[ultra thick,color=Blue] (x5) -- (x6);
\draw[ultra thick,color=Blue] (x5) -- (x8);
\draw[ultra thick,color=Blue] (x6) -- (x9);
\draw[ultra thick,color=Blue] (x8) -- (x9);
\draw[ultra thick,color=Blue] (x9) -- (x12);
\draw[ultra thick,color=Blue] (x8) -- (x11);
\draw[ultra thick,color=Blue] (x11) -- (x12);
\draw[ultra thick,color=Blue] (x7) -- (x8);

\draw[fill,thick,color=Blue] (x11) circle (1pt);
\draw[fill,thick,color=Blue] (x1) circle (1pt);
\draw[fill,thick,color=Blue] (x2) circle (1pt);
\draw[fill,thick,color=Blue] (x3) circle (1pt);
\draw[fill,thick,color=Blue] (x4) circle (1pt);
\draw[fill,thick,color=Blue] (x5) circle (1pt);
\draw[fill,thick,color=Blue] (x6) circle (1pt);
\draw[fill,thick,color=Blue] (x7) circle (1pt);
\draw[fill,thick,color=Blue] (x8) circle (1pt);
\draw[fill,thick,color=Blue] (x9) circle (1pt);
\draw[fill,thick,color=Blue] (x12) circle (1pt);
\end{tikzpicture}    
\quad
\gamma_4:
\begin{tikzpicture}[baseline=5ex,scale=1.]
\coordinate (x1) at (0,0) ;
\coordinate (x2) at (0,1) ;
\coordinate (x3) at (0,2) ;
\coordinate (x6) at (1,2) ;
\coordinate (x5) at (1,1) ;
\coordinate (x4) at (1,0) ;
\coordinate (x9) at (2,2) ;
\coordinate (x8) at (2,1) ;
\coordinate (x7) at (2,0) ;
\coordinate (x12) at (3,2) ;
\coordinate (x11) at (3,1) ;
\coordinate (x10) at (3,0) ;

\node (p4) at (-0.5,-0.5) {$p_4$};
\node (p1) at (-0.5,2.5) {$p_1$};
\node (p3) at (3.5,-0.5) {$p_3$};
\node (p2) at (3.5,2.5) {$p_2$};

\draw[color=green] (x1) -- (p4);
\draw[fill,thick,color=green] (x1) circle (1pt);
\draw[color=green] (x1) -- (x2);
\draw[color=green] (x1) -- (x4);
\draw[color=LimeGreen] (x3) -- (p1);
\draw[color=ForestGreen] (x12) -- (p2);
\draw[color=Green] (x10) -- (p3);

\draw[ultra thick,color=Blue] (x2) -- (x3);
\draw[ultra thick,color=Blue] (x2) -- (x5);
\draw[ultra thick,color=Blue] (x3) -- (x6);
\draw[ultra thick,color=Blue] (x4) -- (x5);
\draw[ultra thick,color=Blue] (x5) -- (x6);
\draw[ultra thick,color=Blue] (x4) -- (x7);
\draw[ultra thick,color=Blue] (x5) -- (x8);
\draw[ultra thick,color=Blue] (x6) -- (x9);
\draw[ultra thick,color=Blue] (x7) -- (x8);
\draw[ultra thick,color=Blue] (x8) -- (x9);
\draw[ultra thick,color=Blue] (x7) -- (x10);
\draw[ultra thick,color=Blue] (x8) -- (x11);
\draw[ultra thick,color=Blue] (x9) -- (x12);
\draw[ultra thick,color=Blue] (x10) -- (x11);
\draw[ultra thick,color=Blue] (x11) -- (x12);

\draw[fill,thick,color=Blue] (x2) circle (1pt);
\draw[fill,thick,color=Blue] (x3) circle (1pt);
\draw[fill,thick,color=Blue] (x4) circle (1pt);
\draw[fill,thick,color=Blue] (x5) circle (1pt);
\draw[fill,thick,color=Blue] (x6) circle (1pt);
\draw[fill,thick,color=Blue] (x7) circle (1pt);
\draw[fill,thick,color=Blue] (x8) circle (1pt);
\draw[fill,thick,color=Blue] (x9) circle (1pt);
\draw[fill,thick,color=Blue] (x10) circle (1pt);
\draw[fill,thick,color=Blue] (x11) circle (1pt);
\draw[fill,thick,color=Blue] (x12) circle (1pt);
\end{tikzpicture}
\end{aligned}
\end{equation}
By overlaying all four we obtain
\begin{equation}
\gamma_1\sqcup \gamma_2\sqcup\gamma_3\sqcup\gamma_4:
\begin{tikzpicture}[baseline=5ex,scale=1.]
\coordinate (x1) at (0,0) ;
\coordinate (x2) at (0,1) ;
\coordinate (x3) at (0,2) ;
\coordinate (x6) at (1,2) ;
\coordinate (x5) at (1,1) ;
\coordinate (x4) at (1,0) ;
\coordinate (x9) at (2,2) ;
\coordinate (x8) at (2,1) ;
\coordinate (x7) at (2,0) ;
\coordinate (x12) at (3,2) ;
\coordinate (x11) at (3,1) ;
\coordinate (x10) at (3,0) ;

\node (p4) at (-0.5,-0.5) {$p_4$};
\node (p1) at (-0.5,2.5) {$p_1$};
\node (p3) at (3.5,-0.5) {$p_3$};
\node (p2) at (3.5,2.5) {$p_2$};

\draw[color=Red,dashed] (x6) -- (x9);
\draw[color=Red,dashed] (x5) -- (x8);
\draw[color=Red,dashed] (x10) -- (x11);
\draw[color=Red,dashed] (x1) -- (x2);

\draw[color=green] (x1) -- (p4);
\draw[fill,thick,color=green] (x1) circle (1pt);
\draw[green,green] (x1) -- (x4);
\draw[color=LimeGreen] (x3) -- (p1);
\draw[color=LimeGreen] (x3) -- (x6);
\draw[color=LimeGreen] (x2) -- (x3);

\draw[color=LimeGreen] (x2) -- (x5);
\draw[color=LimeGreen] (x5) -- (x6);
\draw[color=LimeGreen] (x4) -- (x5);
\draw[fill,thick,color=LimeGreen] (x6) circle (1pt);
\draw[fill,thick,color=LimeGreen] (x2) circle (1pt);
\draw[fill,thick,color=LimeGreen] (x3) circle (1pt);
\draw[fill,thick,color=LimeGreen] (x5) circle (1pt);

\draw[color=ForestGreen] (x12) -- (p2);
\draw[fill,thick,color=ForestGreen] (x8) circle (1pt);
\draw[fill,thick,color=ForestGreen] (x9) circle (1pt);
\draw[fill,thick,color=ForestGreen] (x11) circle (1pt);
\draw[fill,thick,color=ForestGreen] (x12) circle (1pt);
\draw[color=ForestGreen] (x7) -- (x8);
\draw[color=ForestGreen] (x8) -- (x9);
\draw[color=ForestGreen] (x8) -- (x11);
\draw[color=ForestGreen] (x9) -- (x12);

\draw[color=ForestGreen] (x11) -- (x12);

\draw[color=Green] (x10) -- (p3);
\draw[fill,thick,color=Green] (x10) circle (1pt);
\draw[color=Green] (x7) -- (x10);

\draw[ultra thick,color=Blue] (x4) -- (x7);

\draw[fill,thick,color=Blue] (x4) circle (1pt);
\draw[fill,thick,color=Blue] (x7) circle (1pt);
\end{tikzpicture}
\end{equation}
where the soft subgraph consists of four disconnected components. Note that each of the $H\cup J\setminus J_i$ is mojetic.

We have implemented the new graph-finding algorithm of section~\ref{algorithm_for_regions} in Maple \cite{Maple} and checked that the result of the program agrees with the output of pySecDec~\cite{HrchJnsSlk22} in a scalar theory containing both three- and four-point interactions. In particular, we verified that the two algorithms agree for all possible three-point graphs up to five loops and four-point graphs up to four loops, where, in both cases, we expanded around different on-shell limits, including the case of any $n-1$ - as well as all $n$ - external momenta becoming on-shell. For the case of $n-1$ external momenta we also checked all five-point graphs with up to three loops. For a particular three-loop graph with one off-shell momentum and three on-shell ones we display the complete list of regions in figure~\ref{fig:3Lexample}, as obtained both with the Maple implementation of our new algorithm and with pySecDec. 
\begin{figure}
\includegraphics[scale=0.85]{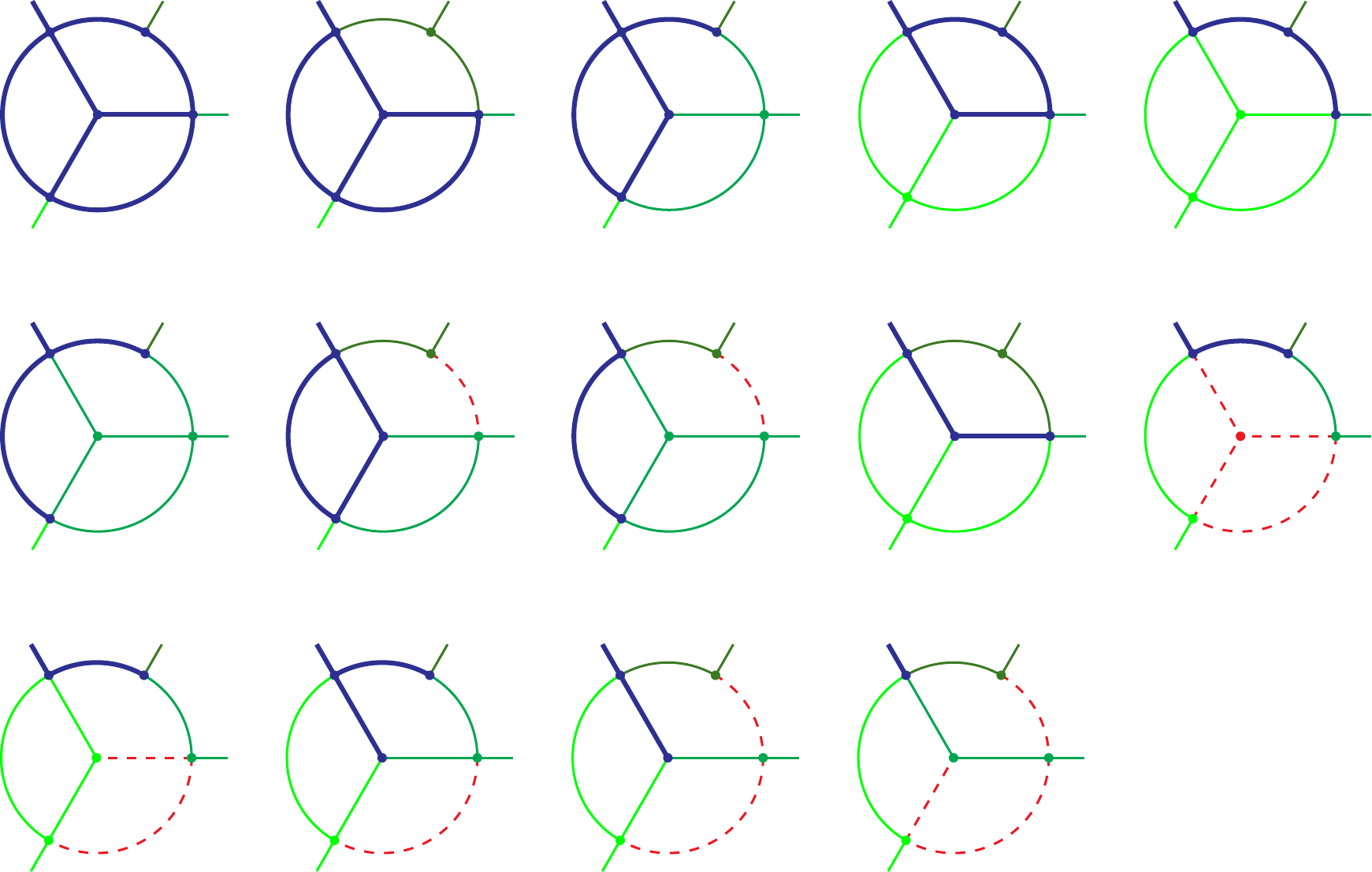}
\caption{All the regions required for the expansion of this four-point integral around the limit where the three green external legs become on-shell. As before, blue denotes the hard subgraph, three shades of green represent the three jets, while red represents the soft subgraph.}
\label{fig:3Lexample}
\end{figure}

In terms of performance we find that the graph-finding algorithm implemented in Maple is considerably faster than pySecDec, especially for more complicated problems at five loops and beyond. This is despite the Maple program not being highly optimised. Of course, the pySecDec program is far more general in the sense that it can deal with a large set of region expansions, while the Maple implementation has been developed specifically for the on-shell expansion in wide-angle scattering. 
The graph-finding algorithm does not refer to any particular representation of the integral and can thus be used in momentum and/or parameter space.  
Another advantage of having a graph-finding algorithm is that it can simplify the development of the region expansion for Feynman integrals in a general quantum field theory, without first reducing them to scalar integrals.

\section{Analytic structure and commutativity of multiple on-shell expansions}
\label{section-analytic_structure_commutativity_multiple_expansions}

In this section we discuss the commutativity of multiple on-shell expansions, and its connection with the analytic properties of wide-angle scattering graphs.
This can be viewed as another application of the results in section~\ref{section-regions_onshell_momentum_expansion}. To motivate this analysis we start with an analytic expression for the three-mass triangle,
\begin{equation}
G_{3}=
\begin{tikzpicture}[baseline=13ex,scale=1.]
\coordinate (x1) at (1.1340, 1.4999) ;
\coordinate (x2) at (2.8660, 1.5000) ;
\coordinate (x3) at (2,3) ;
\node (p1) at (0.2681, 0.9998) {$p_1$};
\node (p2) at (3.7320, 1.0000) {$p_2$};
\node (p3) at (2,4) {$q_1$};
\draw[color=green] (x1) -- (p1);
\draw[ultra thick,color=Blue] (x3) -- (p3);
\draw[color=ForestGreen] (x2) -- (p2);
\draw[ultra thick,color=Black] (x1) -- (x2);
\draw[ultra thick,color=Black] (x2) -- (x3);
\draw[ultra thick,color=Black] (x3) -- (x1);
\draw[fill,thick,color=Blue] (x1) circle (1pt);
\draw[fill,thick,color=Blue] (x2) circle (1pt);
\draw[fill,thick,color=Blue] (x3) circle (1pt);
\end{tikzpicture},
\end{equation}
which is originally derived in ref.~\cite{AnstsGlvOlr00, BoosDvdchv90}. The explicit expression for the corresponding integral, whose integrand we defined in eq.~(\ref{integrand_oneloop_Sudakov_example}), can be written as:
\begin{align}
 \I(G_3)&= \I^{(H)}(G_3)+  \I^{(C_1)}(G_3) +   \I^{(C_2)}(G_3)+ \I^{(S)}(G_3)
 \label{eq:tri3all}
\end{align}
with
\begin{subequations}
\begin{align}
\I^{(H)}(G_3)&= c_3~F_4 \left( 1,1+\eps,1+\eps,1+\eps;\frac{p_1^2}{q_1^2},\frac{p_2^2}{q_1^2} \right);
\label{triangle_one_loop_result_H}
\\ 
\I^{(C_1)}(G_3)&= -c_3~\left( \frac{p_1^2}{q_1^2} \right)^{-\eps}~F_4 \left( 1,1-\eps,1-\eps,1+\eps;\frac{p_1^2}{q_1^2},\frac{p_2^2}{q_1^2} \right);
\label{triangle_one_loop_result_C1}
\\
\I^{(C_2)}(G_3)&= -c_3~\left( \frac{p_2^2}{q_1^2}\right)^{-\eps}~F_4\left(1,1-\eps,1+\eps,1-\eps;\frac{p_1^2}{q_1^2},\frac{p_2^2}{q_1^2}\right);
\label{triangle_one_loop_result_C2}
\\
\I^{(S)}(G_3)&=c_3c_s~\left( \frac{p_1^2p_2^2}{(q_1^2)^2} \right)^{-\eps}~F_4 \left( 1-2\eps,1-\eps,1-\eps,1-\eps;\frac{p_1^2}{q_1^2},\frac{p_2^2}{q_1^2} \right) 
,\phantom{aaaa}
\label{triangle_one_loop_result_S}
\end{align}
\label{triangle_one_loop_result}
\end{subequations}
where
\begin{equation}
c_3=  -(-q_1^2)^{-1-\eps} \frac{\Gamma^2(1-\eps)\Gamma(1+\eps)}{\eps^2\Gamma(1-2\eps)} \,,
\quad c_s=\frac{\Gamma(1+\eps)\Gamma(1-2\eps)}{\Gamma(1-\eps)}\,,
\end{equation}
and $F_4$ is one of Appell's hypergeometric functions. Its series representation around the point $x=0=y$, which is valid as long as $\sqrt{x}+\sqrt{y}\leqslant 1$, is given by
\begin{equation}
F_4(\alpha,\beta,\gamma,\gamma';x,y) = \sum_{m,n=0}^{\infty} \frac{(\alpha)^{(m+n)}(\beta)^{(m+n)}}
{(\gamma)^{(m)}(\gamma')^{(n)}} \, \frac{x^m}{m!} \, \frac{y^n}{n!}\,,
\label{appell_definition}
\end{equation} 
with the rising factorial defined as $(z)^{(n)}=\Gamma(z+n)/\Gamma(z)$. For the on-shell expansion
\begin{equation}
\label{Triangle_on_shell}
|p_1^2/q_1^2|\sim |p_2^2/q_1^2|\sim \lambda \ll 1,
\end{equation}
eqs.~(\ref{triangle_one_loop_result_H})-(\ref{triangle_one_loop_result_S}), as suggested by their notations, can be identified with the hard, collinear-to-$p_1$, collinear-to-$p_2$ and soft regions, respectively.

Eqs.~(\ref{eq:tri3all})-(\ref{appell_definition}) make manifest the analytic behaviour associated with the contribution of each region around the limit $p_i^2/q_1^2\to 0$.
Specifically, in each region, the non-analytic behaviour is associated with a particular power of $(p_i^2/q_1^2)^\epsilon$, which multiplies a function with a \emph{regular} Taylor expansion in powers of $p_i^2/q_1^2$, having a \emph{finite} radius of convergence. The latter is therefore an analytic function of both~$p_1^2/q_1^2$ and $p_2^2/q_1^2$ in the vicinity of the origin.
Note that the characteristic power of the expansion parameter associated with each region can be seen directly from the scaling of the integrand, as summarised in eq.~(\ref{leading_powers_oneloop_triangle_regions}). The scaling with $\lambda$ is associated to scaling with $p_1^2/q_1^2$, or with $p_2^2/q_1^2$, or with both, in line with eq.~(\ref{Triangle_on_shell}).

We further observe that the above analytic structure implies that expansions of $\I(G_3)$ in $p_1^2/q_1^2$ and $p_2^2/q_1^2$ commute. In fact, we will show that the above properties --- that the non-analytic $p_i^2$-dependence can be factorised in each region, and the remaining function of $p_i^2$ has a regular multivariate Taylor expansion near the origin, leading to a commuting expansion in different $p_i^2$ --- hold for the on-shell expansion of a broad range of graphs, including \emph{any} Sudakov form factor and \emph{any} planar graph of wide-angle scattering. Meanwhile, commutativity may break down for nonplanar graphs having three or more external legs taken on shell, for regions which contain nontrivial soft subgraphs with at least one soft vertex. An example is provided by the graph in figure~\ref{figure-UF_polynomial_onshell_expansion_example} with all three $p_i^2$ simultaneously taken to the on-shell limit. Integrals arising in these soft regions may involve dependence on composite variables such as\footnote{In the graph of figure~\ref{figure-UF_polynomial_onshell_expansion_example} such dependence appears in the SS region of eq.~(\ref{SS_and_C1S_scaling}).}
\begin{equation}
\label{non-commutativity_X}
X(\{z_{ij}\}) = \frac{4 z_{12} z_{23} z_{13}}{(z_{12} + z_{23} + z_{13})^3}\qquad \text{with} \qquad z_{ij}=\frac{2p_i\cdot p_j}{p_i^2 p_j^2}\,.
\end{equation}
Given the homogeneity of this function,  a simultaneous expansion in all three $p_i^2$ preserves it, $X(\{z_{ij}\})\sim \lambda^0$, while a further expansion in, say $p_1^2\sim p_3^2$ (assuming they are both taken to be smaller than $p_2^2$) would necessarily introduce negative powers of $p_2^2$, and hence does not commute with the expansion where $p_2^2$ is taken to zero first.

In what follows we will study the analytic structure of a specific kind of region, which we call \emph{jet-pairing soft} regions, where every connected component of the soft subgraph only interacts with two jets. For such regions we can explicitly describe the factorisation of non-analytic dependence on $p_i^2$, and show that it leaves behind a function with a regular Taylor expansion. We formulate and prove this result in section~\ref{commutativity_general_prescription}. On this basis we prove the commutativity of multiple on-shell expansions in certain cases in section~\ref{commutativity_multiple_onshell_expansions}.

\subsection{The analytic structure of jet-pairing soft regions}
\label{commutativity_general_prescription}

In this section we focus on on-shell expansions in wide-angle scattering for a specific kind of region, one in which every connected component of the corresponding soft subgraph connects to exactly two jets. We call them the \emph{jet-pairing soft} regions. 
Note that regions with no soft propagators also belong to this class. The key result of this subsection is summarised in the following theorem.

\begin{theorem}
If $R$ is a jet-pairing soft region that appears in the on-shell expansion of a wide-angle scattering graph $G$, then the all-order expansion of $\I(G)$ in this region can be written as follows:
\begin{eqnarray}
\T_{\t}^{(R)}\I(\s)= \bigg( \prod_{p_i^2\in \t} (p_i^2)^{\rho_{R,i}(\epsilon)} \bigg)\cdot \sum_{k_1,\dots,k_{|\t|}\geqslant 0} \Big( \prod_{p_i^2\in \t} (-p_i^2)^{k_i} \Big) \cdot \overline{\I}_{\{k\}}^{(R)}\left( \s \setminus \t \right),
\label{invariant_mass_overall_factors}
\end{eqnarray}
where $\rho_{R,i}(\epsilon)$ is a linear function of~$\epsilon$, $k_i$ are non-negative integer powers and $\overline{\I}_{\{k\}}^{(R)}\left( \s \setminus \t \right)$ is a function of the off-shell kinematics, independent of any $p_i^2\in \t$.  
\label{theorem-analytic_structure_jet_pairing_soft}
\end{theorem}

The essence of this theorem is that the \emph{non-analytic dependence} on the virtualities $p_i^2$ involves only powers of~$(-p_i^2)^\epsilon$, and it factorises from the rest of the integral, which admits a regular multivariate Taylor expansion in $p_i^2$.

\begin{proof}[Proof of theorem~\ref{theorem-analytic_structure_jet_pairing_soft}]
We first focus on the leading contribution of $\T_{\t}^{(R)}\I(G)$, according to the definitions in section~\ref{general_setup_MoR}. Namely,
\begin{eqnarray}
\T_{\t,0}^{(R)}\I(G)= \int[d\x] \prod_{e\in G} x_e^{\nu_e} \cdot \Big( \mathcal{P}_0^{(R)}(\x;\s) \Big)^{-D/2},
\label{leading_contribution_approximation}
\end{eqnarray}
with
\begin{eqnarray}
\mathcal{P}_0^{(R)}(\x;\s)= \mathcal{U}^{(R)}(\x) +\sum_{i=1}^K \mathcal{F}^{(p_i^2,R)}(\x;\s)+ \mathcal{F}_\textup{I}^{(q^2,R)}(\x;\s)+\sum_{i>j=1}^K \mathcal{F}_\textup{II}^{(q_{ij}^2,R)}(\x;\s),
\label{leading_LP_polynomial_generic_form}
\end{eqnarray}
where the terms on the right-hand side factorise according to eq.~(\ref{summary_leading_terms_factorise}).

In the case that $R$ is a jet-pairing soft region, every connected component of the soft subgraph $S$ is attached to exactly two jets by definition. In other words, every parameter corresponding to a soft edge is associated to an $(i,j)$ pair, with $i>j =1,\dots,K$ labelling the jets. We denote these parameters by $x^{[S_{ij}]}$. We then rescale the integration variables as follows:
\begin{eqnarray}
\hspace{-10pt} x'^{[H]}\equiv x^{[H]},\qquad x'^{[J_i]}\equiv \frac{-p_i^2}{-Q^2}x^{[J_i]}\ (i=1,\dots,K),\qquad x'^{[S_{ij}]} \equiv \frac{(-p_i^2)(-p_j^2)} {(-Q^2)^2} x^{[S_{ij}]},
\label{wideangle_scattering_parameters_rescale}
\end{eqnarray}
where $|Q^2|\gg |p_i^2|$ can be \emph{any} large kinematic scale, $Q^2\in\s \setminus \t$. One can check that, according to corollary~\ref{theorem-leading_terms_form-corollary}, under this specific rescaling, each monomial of~$\mathcal{P}_0^{(R)}(\x;\s)$ has the same factor of $\prod_{i=1}^K (-p_i^2/-Q^2)^{-\kappa_{R,i}^{}}$, where $\kappa_{R,i}^{}$ are non-negative integers. Specifically,
\begin{eqnarray}
\mathcal{P}_0^{(R)}(\x;\s)= \overline{\mathcal{P}}_0^{(R)}(\x';\s\setminus \t)\cdot \prod_{i=1}^K {\Big( \frac{-p_i^2}{-Q^2} \Big)}^{-L(\widetilde{J}_i)-\sum_{j\neq i} L(\widetilde{S}_{ij})},
\end{eqnarray}
where the polynomial $\overline{\mathcal{P}}_0^{(R)} (\x'; \s\setminus \t)$ is defined by factoring out of $\mathcal{P}_0^{(R)}(\x;\s)$ the dependence on~$p_i^2$ after the rescaling, hence obtaining a polynomial that only depends on the kinematic variables in $\s\setminus \t$. Under the same rescaling of the integration variables in eq.~(\ref{leading_contribution_approximation}), $\T_{\t,0}^{(R)}\I(G)$ can be rewritten as
\begin{align}
    \begin{split}
        \T_{\t,0}^{(R)}\I(G)= (-Q^2)^{\rho_{R,0}} \prod_{i=1}^K (-p_i^2)^{\rho_{R,i}} \int[d\x'] \prod_{e\in G} (x'_e)^{\nu_e} \big( \overline{\mathcal{P}}_0^{(R)}(\x'; \s\setminus \t) \big)^{-D/2},
    \end{split}
    \label{wideangle_scattering_leading_contribution_pi2_dependence}
\end{align}
where
\begin{subequations}
\begin{align}
\rho_{R,i}\equiv& \bigg(L(\widetilde{J}_i)+\sum_{j\neq i} L(\widetilde{S}_{ij})\bigg)D/2 -\sum_{e\in J_i} \nu_{e} -\sum_{j\neq i} \sum_{e\in S_{ij}} \nu_{e}\ ,
\\
\rho_{R,0}\equiv& \sum_{e\in J}\nu_{e} + 2\sum_{e\in S}\nu_{e}\ .
\end{align}
\end{subequations}
Note that the integral of eq.~(\ref{wideangle_scattering_leading_contribution_pi2_dependence}) does not depend on the kinematic variables in $\t$, which indicates that we have verified theorem~\ref{theorem-analytic_structure_jet_pairing_soft} at the leading order. We now include all the other (higher-order) terms of $\T_{\t}^{(R)} \I$, i.e.
\begin{eqnarray}
\label{expansion_proof}
\T_{\t}^{(R)} \I(G) &&= \sum_{n=0}^\infty \left. \int [d\x] T_{\lambda,n} \Big(\lambda^{-p_R^{}(\epsilon)} I(\lambda^{\boldsymbol{u}_R}\x;\lambda^{\w}\s)\Big)\right|_{\lambda=1}\\ \nonumber
&&= \left. \sum_{n=0}^\infty \int [d\x] \Big( \prod_{e\in G} x_e^{\nu_e} \Big) T_{\lambda,n} \left[ \lambda^{-p_R^{}(\epsilon)} \big( \mathcal{P}(\lambda^{\boldsymbol{u}_R}\x;\lambda^{\w}\s)\big)^{-D/2} \right] \right|_{\lambda=1}
\end{eqnarray}
Defining 
\begin{eqnarray}
\Delta\mathcal{P}^{(R)}(\x;\s)\equiv \mathcal{P}(\x;\s) - \mathcal{P}_0^{(R)}(\x;\s)
\label{proof_commutativity_step1_notations}
\end{eqnarray}
to capture the non-leading terms of $\mathcal{P}(\x;\s)$ in the region $R$, we may write (\ref{expansion_proof}) as follows
\begin{eqnarray}
\T_{\t}^{(R)} \I(G) 
&&= \sum_{n=0}^\infty\int [d\x] 
\Big( \prod_{e\in G} x_e^{\nu_e} \Big)
T_{\lambda,n} \left[ \lambda^{-p_R^{}(\epsilon)}  \big( \mathcal{P}_0^{(R)}(\lambda^{\boldsymbol{u}_R}\x;\lambda^{\w}\s)\big)^{-D/2} \right. \nonumber\\
&&\hspace{4.5cm} \left. \left.\left( 1+\frac{\Delta\mathcal{P}^{(R)}(\lambda^{\v_R^{}}\x;\lambda^{\w}\s)}{\mathcal{P}_0^{(R)}(\lambda^{\v_R^{}}\x;\lambda^{\w}\s)} \right)^{-D/2} \right]\right|_{\lambda=1}.
\label{proof_commutativity_step1}
\end{eqnarray}
Eq.~(\ref{proof_commutativity_step1}) is equivalent to the Taylor expansion
\begin{eqnarray}
\T_{\t}^{(R)} \I(G) = \sum_{n=0}^\infty \frac{(-D/2)_n}{n!} \int [d\x] \Big( \prod_{e\in G} x_e^{\nu_e} \Big) \big( \mathcal{P}_0^{(R)}(\x;\s)\big)^{-D/2} \bigg( \frac{\Delta\mathcal{P}^{(R)}(\x;\s)}{\mathcal{P}_0^{(R)}(\x;\s)} \bigg)^{n},
\label{proof_commutativity_step2}
\end{eqnarray}
where $()_n$ denotes the falling factorial. For each fixed $n$, $(\Delta\mathcal{P}^{(R)}(\x;\s))^{n}$ is a polynomial of $\x$, so by expanding $(\Delta\mathcal{P}^{(R)}(\x;\s))^{n}$, one can view (\ref{proof_commutativity_step2}) as a sum over integrals of the following form:
\begin{eqnarray}
\T_{\t}^{(R)} \I(G) = \sum_{n=0}^\infty \sum_l \mathcal{C}_{n,l} \int [d\x] \Big( \prod_{e\in G} x_e^{\nu_e} \Big) \big( \mathcal{P}_0^{(R)}(\x;\s) \big)^{-D/2} \cdot \frac{\s^{\boldsymbol{p}_{n,l}} \x^{\boldsymbol{q}_{n,l}}}{\big( \mathcal{P}_0^{(R)}(\x;\s) \big)^{n}},
\label{proof_commutativity_step3}
\end{eqnarray}
where $l$ enumerates the different terms, while
$\mathcal{C}_{n,l}$ denote multinomial coefficients. $\s^{\boldsymbol{p}_{n,l}} \x^{\boldsymbol{q}_{n,l}}$ is the product of $n$ monomials from $\Delta\mathcal{P}^{(R)}(\x;\s)$ and the powers $\boldsymbol{p}_{n,l}$ and $\boldsymbol{q}_{n,l}$ are vectors of non-negative integers. From the definition of $\Delta\mathcal{P}^{(R)}$ it follows that the $n$-th order term in eq.~(\ref{proof_commutativity_step3}) is suppressed by $n$ powers of the expansion parameter. We note that eq.~(\ref{proof_commutativity_step3}) is completely general, and in particular it is valid for any region.

In the case of jet-pairing soft regions we can use the rescaling of $\x$ as in eq.~(\ref{wideangle_scattering_parameters_rescale}) to make this suppression manifest in terms of powers of $p_i^2/Q^2$, namely,
\begin{eqnarray}
\T_{\t}^{(R)} \I(G)=  (-Q^2)^{\rho_{R,0}} 
\prod_{i=1}^K (-p_i^2)^{\rho_{R,i}}&& 
\sum_{n=0}^\infty \sum_l \mathcal{C}_{n,l}
 \int [d\x'] \Big( \prod_{e\in G} {x_e'}^{\nu_e} \Big) \big( \overline{\mathcal{P}}_0^{(R)}(\x';\s\setminus \t) \big)^{-D/2}\nonumber\\
&&  \cdot \frac{(\s\setminus\t)^{\boldsymbol{p}'_{n,l}} {\x'}^{\boldsymbol{q}_{n,l}}}{\big( \overline{\mathcal{P}}_0^{(R)}(\x';\s\setminus \t) \big)^{n}}\cdot \prod_{i=1}^K \Big( \frac{-p_i^2}{-Q^2} \Big)^{k_i(n,l)}.
\label{proof_commutativity_step4}
\end{eqnarray}
Aside from the sum over $n$ and $l$, the first line of eq.~(\ref{proof_commutativity_step4}) is identical to eq.~(\ref{wideangle_scattering_leading_contribution_pi2_dependence}), while the second line contains subleading powers, which emerge from the last factor in eq.~(\ref{proof_commutativity_step3}).
Note that all the powers $k_i(n,l)$ are non-negative integers, such that $\sum_{i=1}^K k_i(n,l)=n$. The powers $\boldsymbol{p}'_{n,l}$ are vectors of non-negative integers, which characterise the dependence on the off-shell external kinematic variables $\s\setminus\t$. We note that the non-analytic dependence on~$p_i^2$ factors out of the sum over $n$ and $l$, and the remaining part is a regular Taylor expansion in all $p_i^2\in \t$. This is exactly the statement of eq.~(\ref{invariant_mass_overall_factors}). Hence we have proved theorem~\ref{theorem-analytic_structure_jet_pairing_soft}.
\end{proof}

Theorem~\ref{theorem-analytic_structure_jet_pairing_soft} can be used to investigate the commutativity of multiple on-shell expansions. We study this in the next subsection.

\subsection{Commutativity of multiple on-shell expansions}
\label{commutativity_multiple_onshell_expansions}

Given a graph $G$, we say that two region expansions, associated to sets of kinematic invariants $\t_1$ and $\t_2$ respectively, \emph{commute} at the $n$-th order if and only if
\begin{equation}
\T_{\t_1 \cup \t_2,n} \I(G)=
\sum_{n_1+n_2=n} \T_{\t_1,n_1} \T_{\t_2,n_2} \I(G) 
=\sum_{n_1+n_2=n} \T_{\t_2,n_1} \T_{\t_1,n_2} \I(G) \,,
\label{commutativity_definition}
\end{equation}
where the expansion operators $\T$ are defined as in section~\ref{general_setup_MoR}, where in each case $\T_{\t}$ represents an expansion in the limit in which all elements of $\t$ are simultaneously taken small, of order $\lambda$, while all other scales, $\s\setminus \t$, are kept fixed.
Eq.~(\ref{commutativity_definition}) means that we can either first expand in $\t_1$ and then in $\t_2$ or vice versa or just immediately do the expansion in set $\t_1\cup \t_2$; all these three expansions should agree for the $n$-th order term.

Let us consider the special case where $\t_1=\{p_1^2\}$ and $\t_2=\{p_2^2\}$. In this case we can identify their series expansions in $\lambda$ with expansions in $p_1^2$ and $p_2^2$, respectively:
\begin{align}
\T_{\t_i,n} \I    =\sum_{R\in \R(\I,\t_i)} \T^{(R)}_{\t_i,n}\I\,\qquad \text{for}\,\,i=1,2\,.
\end{align}
Considering now a general term in the double expansion in $p_1^2$ and $p_2^2$, we thus obtain
\begin{equation}
\T_{\t_1,n_1} \T_{\t_2,n_2} \I =
\sum_{R_2\in \R_2} \sum_{R_1\in \R_1(R_2)} \T^{(R_1)}_{\t_1,n_1} \T^{(R_2)}_{\t_2,n_2}\I \,,
\end{equation}
where $\R_2=\R(\I,\t_2)$ is the set of regions of $\I$ expanding in $\t_2$, and $\R_1(R_2)=\R(\T^{(R_2)}_{\t_2,n_2}\I,\t_1)$ is the set of regions of $\T^{(R_2)}_{\t_2,n_2}\I$ expanding in $\t_1$.

Naturally, we can associate this double sum in $p_1^2$ and $p_2^2$ to the limit $|p_1^2|\ll |p_2^2|\ll |Q^2|$, since we first expanded in $p_2^2$ assuming $|p_2^2|\ll |Q^2|$ (where $Q^2$ is some hard scale) and then assumed $|p_1^2|\ll |p_2^2|$ in the second expansion. In other words, this is the limit where the double expansion is expected to converge well. Were we to consider the other order, $\T_{\t_2,n_2} \T_{\t_1,n_1} \I$, we would get a similar double sum in $p_1^2$ and $p_2^2$ that would naively converge only for $|p_2^2|\ll |p_1^2|\ll |Q^2|$. Only if the two expansions commute would $\T_{\t_1,n_1} \T_{\t_2,n_2} \I$ and $\T_{\t_2,n_2} \T_{\t_1,n_1} \I$ be identical for any given $n_1$ and $n_2$. Therefore, eq.~(\ref{commutativity_definition}) is a nontrivial identity.

We point out that if all the regions of a graph $G$ are jet-pairing soft, then from theorem~\ref{theorem-analytic_structure_jet_pairing_soft}, the entire region expansion can be written as
\begin{eqnarray}
\T_{\t}\I(\s)= \sum_{R\in\R(I,\t)} \sum_{k_1,\dots,k_{|\t|}\geqslant 0} \prod_{p_i^2\in \t} (p_i^2)^{\rho_{R,i}(\epsilon)+k_i} \cdot \overline{\I}_{\{k\}}^{(R)}\left( \s \setminus \t \right).
\label{invariant_mass_overall_factors_allR}
\end{eqnarray}
Eq.~(\ref{invariant_mass_overall_factors_allR}) implies that $\T_{\t}\I(\s)$ can be written as a sum over terms, where in each of them, the non-analytic dependence on $p_i^2\to 0$, which is characteristic of the region, factors out of the expansion. Eq.~(\ref{triangle_one_loop_result}) provides a simple example of this, which is hence also an example of the commutativity of expansions (in $p_1^2$ and $p_2^2$). Namely, $\T_{\t_2,n_2} \T_{\t_1,n_1} \I= \T_{\t_1,n_1} \T_{\t_2,n_2} \I$ for any given $n_1$ and $n_2$, because the non-analytic behaviour near $p_i^2\to 0$ is associated exclusively with powers of $(p_i^2/q_1^2)^\epsilon$ in each term of eq.~(\ref{eq:tri3all}).

We conclude that in general, all the graphs for which all the regions are jet-pairing soft, satisfy eq.~(\ref{invariant_mass_overall_factors_allR}), and hence the commutativity relation (\ref{commutativity_definition}). These graphs include in particular any Sudakov form factor and any planar graph of wide-angle scattering.

Note that these commutativity relations also provide us with a simple means to check when eq.~(\ref{invariant_mass_overall_factors_allR}) does not hold. For instance, the nonplanar two-loop graph shown in figure~\ref{figure-UF_polynomial_onshell_expansion_example} contains a soft region whose connected component connects 3 jets. We will discuss this example in more details in the following subsection.

\subsection{Non-commuting example}

As discussed above, multiple on-shell expansions commute for large classes of graphs, hence non-commuting examples are by nature nontrivial. The first such example appears as a nonplanar two-loop graph with four external legs in an on-shell expansions of at least three external legs. We already met such a graph in  figure~\ref{figure-UF_polynomial_onshell_expansion_example} where the regions in the on-shell expansion in the set $\t_{123}=\{p_1^2,p_2^2,p_3^2\}$ were investigated. Our expectation is  that the expansion in $\t_{123}$ would not lead to a factorising dependence on the $p_i^2$ but that there would remain nontrivial dependence on scaleless ratios of the form $p_i^2/p_j^2$, such that subsequent expansions in subsets of $\t_{123}$ would lead to further simpler regions. A special feature of this particular example is that the leading term in the expansion in $\t_{123}$ is of order $\lambda^{-1}$, i.e. it is super-leading. The region expansion at this order is given by a single (double-soft) region, whose corresponding region vector is $\v_\text{SS}=(-1,-1,-1,-2,-2,-2;1)$ with the labeling corresponding to that in figure~\ref{figure-UF_polynomial_onshell_expansion_example}. The soft subgraph of this double-soft region connects to all three jets, and $\text{SS}$ is therefore not a jet-pairing soft region. Its properties are thus not dictated by theorem~\ref{theorem-analytic_structure_jet_pairing_soft}.
 
The leading Lee-Pomeransky polynomial in this region was already provided in eq.~(\ref{UF_polynomial_SS_leading_terms}). We rewrite it here in a slightly more compact form, after simplifying its kinematic dependence by rescaling $x_{i}\to x_i/(-p_i^2)$ for $i=1,2,3$, 
\begin{eqnarray}
\label{UF_polynomial_SS_leading_terms_simp}
\mathcal{P}_0^{(\text{SS})}(\x;\s)=&& ({\color{Red} x_4x_5} + {\color{Red} x_4x_6} + {\color{Red} x_5x_6}) (1 +  x_1+  x_2 + x_3) \\ \nonumber
&&+{\color{Red}z_{12}}  x_1x_2{\color{Red} x_6} +{\color{Red}z_{23}}x_2x_3{\color{Red} x_4} +{\color{Red}z_{13}} x_1x_3{\color{Red} x_5},
\end{eqnarray}
which depends only on the kinematic variables
\begin{equation}
z_{ij}=\frac{(-s_{ij})}{(-p_i^2)(-p_j^2)}    \,.
\end{equation}
The corresponding Feynman integral $\I_0^{(\text{SS})}(G_{4np})$ can be shown to depend, nontrivially, only on two ratios of the $z_{ij}$ variables as one of them can be scaled out (this can be seen by dimensional analysis). The integral can thus be expressed in the form
\begin{equation}
\label{eq:ISS}
\I_0^{(\text{SS})}(G_{4np})=\frac{z_{12}^{-1+2\eps}}{p_1^2p_2^2p_3^2}\; \bar\I_0^{(\text{SS})}(t_{13},t_{23})\,,
\end{equation}
where the dimensionless ratios $t_{13}$ and $t_{23}$ are given by
\begin{equation}
\label{t13t23_ratios}
t_{13}=\frac{z_{13}}{z_{12}}=\frac{(-s_{13})(-p_2^2)}{(-s_{12})(-p_3^2)},\qquad t_{23}=\frac{z_{23}}{z_{12}}=\frac{(-s_{23})(-p_1^2)}{(-s_{12})(-p_3^2)}\,,
\end{equation}
and $\bar\I_0^{(\text{SS})}$ admits the following integral representation:
\begin{equation}
\bar\I_0^{(\text{SS})}(t_{13},t_{23};\eps)=
\Gamma(1+2\eps) \int_0^\infty d^4 x\, 
\frac{(x_1+x_2+x_1x_2)^{-1-2\eps}(x_3+x_4+x_3x_4)^{-1-\eps}}{1+t_{13} x_1 x_3 + t_{23} x_2 x_4 }\,,
\end{equation}
with $d^4x\equiv dx_1\,dx_2\,dx_3\,dx_4$.

Note that the two ratios in eq.~(\ref{t13t23_ratios}) have no overall scaling in the corresponding region expansion parameter $\lambda\sim p_i^2$. Direct integration in terms of generalised polylogarithms (e.g. using HyperInt~\cite{Panzer:2014caa}) is obstructed by a square root of a cubic polynomial, which contains dependence on the variable of eq.~(\ref{non-commutativity_X}) above, namely,
\begin{equation}
X=\frac{4 t_{13}t_{23}}{(1+t_{13}+t_{23})^3}\,.    
\end{equation}
One can thus expect that this integral falls into the class of elliptic multiple polylogarithms, which has been studied extensively in recent years, see for instance~\cite{Bloch:2013tra,Adams:2013nia,Remiddi:2017har,Broedel:2017kkb,Broedel:2017siw,Broedel:2019hyg} and \cite{Bourjaily:2022bwx} for a recent review and references therein. Here we shall not pursue such analytic integration, but rather investigate the region expansion numerically.

To demonstrate that the dependence on $t_{ij}$ in $\bar\I_0^{(\text{SS})}(t_{13},t_{23};\eps)$ does indeed feature non-commuting expansions, we studied the subsequent on-shell expansion in $\boldsymbol{t}_1=\{p_1^2\}$. The program pySecDec reveals three further regions:
\begin{align}
\label{new_regionsNP}
\begin{split}
&\v_\text{H}=(0,0,0,0,0,0;1);\\
&\v_{\text{C}_1}=(-1,0,0,-1,-1,-1;1);\\
&\v_{\text{C}_1\text{C}_1}=(-1,0,0,-1,0,0;1),
\end{split}
\end{align}
which are all present within the double soft region of $\t_{123}$. Summing up these three regions we then demonstrate numerically that 
\begin{equation}
\T_{\t_1,0}\T_{\t_{123},0}\,\I(G_{4np})\ \neq\ \T_{\t_{123},0}\,\I(G_{4np}) \,.
\label{eq:demonstrateNC}
\end{equation}
This may be contrasted with the properties of multiple expansions for jet-pairing soft regions, taking the form of eq.~(\ref{proof_commutativity_step4}), where a simultaneous on-shell expansion in some $\{p_i^2\}$ is regular in each of the $p_i^2$, so a further expansion in any of them does not have any effect whatsoever. This is very different in the case considered here, where a subsequent expansion in $p_1^2$ yields new regions, eq.~(\ref{new_regionsNP}).

\begin{figure}[t]
\centering
\includegraphics[width=0.9\textwidth]{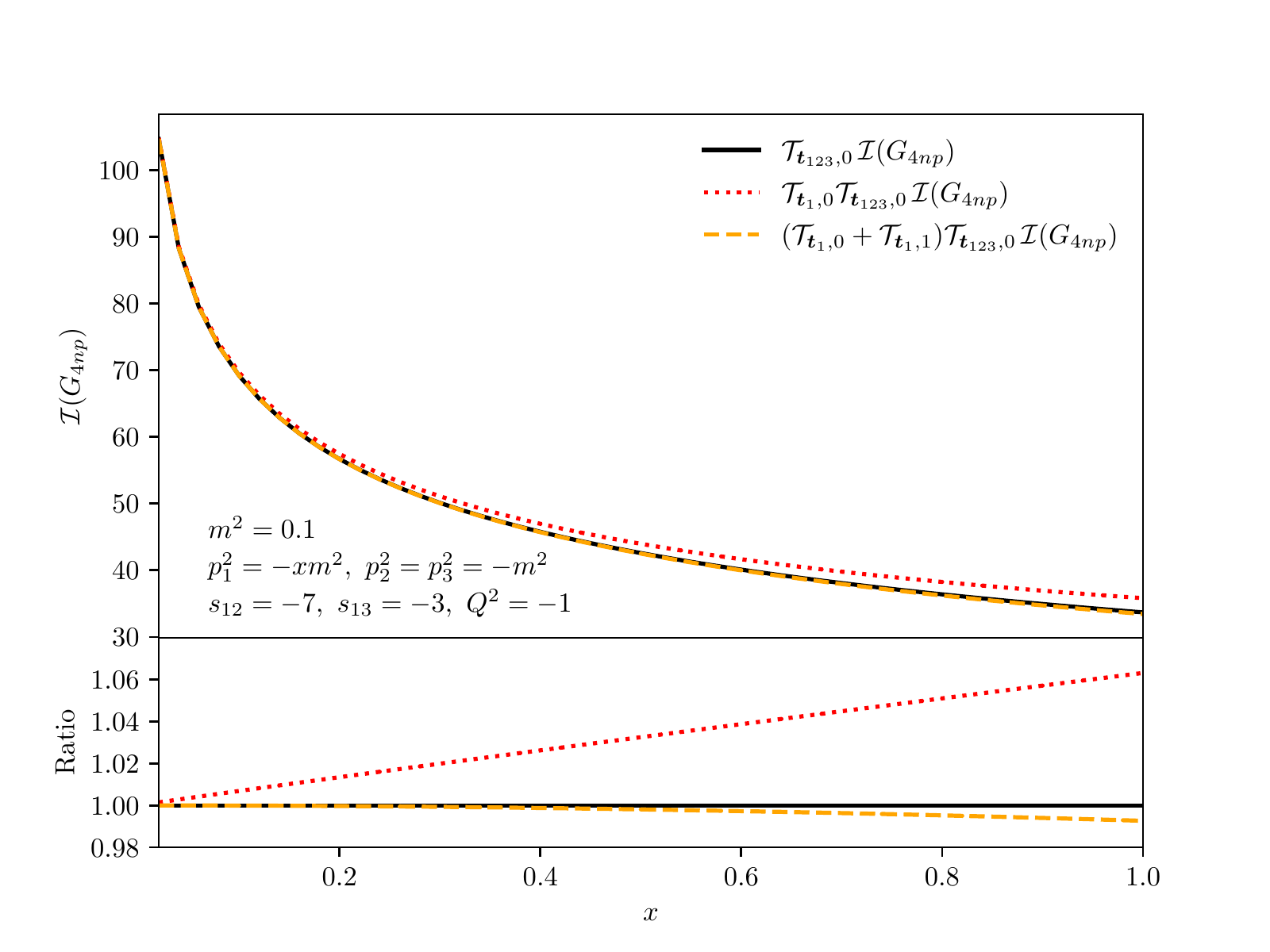}
\caption{The solid (black) curve corresponds to the right-hand side of eq.~(\ref{eq:demonstrateNC}) while the dotted (red) curve corresponds to the left-hand side. The dashed (orange) curve includes also the next-to-leading power ($x^1$) term, which is not considered in eq.~(\ref{eq:demonstrateNC}).}
\label{figure-NC}
\end{figure}
The numerical results are plotted in figure~\ref{figure-NC}, which shows that the difference between the two sides of eq.~(\ref{eq:demonstrateNC}) grows with increasing $x=p_1^2/p_2^2$. While the left-hand side of eq.~(\ref{eq:demonstrateNC}) is only expected to provide a good approximation at small $x$, we observe that the difference remains only of the order of a few percent even at $x=1$. 
The figure also shows that the next-to-leading power term in the $x$-expansion yields an improved convergence, giving further evidence that both the numerical integration and convergence of the series are reliable. We have thus obtained direct evidence that the respective region expansions do not commute for this integral, according to the definition in eq.~(\ref{commutativity_definition}).

We thus see that the simplest example that departs from the conditions of Theorem~\ref{theorem-analytic_structure_jet_pairing_soft} already  displays a non-commutative behaviour, which suggests that the theorem may well be in its strongest possible form. From an analytic perspective, it is interesting to note that the simplest non-commuting case already falls into the class of elliptic polylogarithms. It is intriguing that there seems to be a rapidly converging series approximation for this elliptic integral whose radius of convergence is surprisingly large. We will leave further exploration of these questions to future work.

\section{Conclusions and outlook}
\label{section-conclusions_outlook}

In this paper, we have studied the application of the MoR to the expansion of wide-angle scattering Feynman integrals about the on-shell limit of any subset of their external momenta.
For these integrals, we provide a complete picture of the regions that arise to any order in the power expansion, which we have described in both momentum space and in the Lee-Pomeransky parameter space. 

We began our analysis using the geometric approach to the MoR, where each region contributing to the expansion is identified as a facet of the Newton polytope in the Lee-Pomeransky parameter space. The regions consist of a hard region, where the on-shell limit is applied at the integrand level, and infrared regions, which are each characterised by a particular rescaling of the Lee-Pomeransky parameters. Each rescaling can be associated with a region vector, which is perpendicular to a given facet of the Newton polytope.
The integrals considered here do not all feature a Euclidean region, and therefore it is not obvious at the outset that all regions in the MoR are associated with facets. We nevertheless believe this to be the case for the on-shell expansion of any wide-angle scattering graph.

A key observation is that the aforementioned infrared regions stand in one-to-one correspondence with particular solutions of the Landau equations, which characterise infrared singularities of the given integral. We described the neighbourhood of the singularities in both momentum space and parameter space and formulated this relation in terms of a suitably-deformed version of the Landau equations. 
In momentum space, each infrared-sensitive region is identified by a particular partitioning of the graph into hard, jet, and soft subgraphs, as shown in figure~\ref{classical_picture}. Thus, in a given region, every propagator in the graph carries a characteristic momentum scaling (hard, jet, or soft) according to eq.~(\ref{infrared_region_momentum_scaling}). This immediately restricts the form of the region vectors that arise in parameter space. Specifically, in the Lee-Pomeransky space, entries of the region vectors are $0$, $-1$, and $-2$ for propagators that are hard, jet-like, or soft, respectively.
It is known that the second Landau equation restricts the form of potential infrared-sensitive momentum configurations, according to the Coleman-Norton analysis~\cite{ClmNtn65}. We investigated how this further constrains the possible region vectors, by imposing certain restrictions on the hard and jet subgraphs. 
The Landau-equation-based analysis of section~\ref{region_vectors_from_LE} culminated in proposition \ref{proposition-region_vectors_are_hard_and_infrared} which identifies the infrared regions at the Feynman graph level.

In section~\ref{section-regions_onshell_momentum_expansion}, we identified the \emph{complete set of requirements} for a configuration of momentum scalings (satisfying proposition~\ref{proposition-region_vectors_are_hard_and_infrared}) to correspond to a region, i.e. to be a facet of the Newton polytope. The key issue is that certain configurations of momentum scalings end up yielding scaleless integrals, which vanish in dimensional regularisation. From the geometric MoR perspective such integrals correspond to lower-dimensional faces, rather than to facets of the Newton polytope, and hence do not contribute. In order to fully characterise the facets contributing to the MoR in terms of the subgraphs ($H$, $J$ and $S$), we classified in theorem~\ref{theorem-leading_terms_form} all possible terms which can be part of the leading Lee-Pomeransky polynomial $\mathcal{P}_0^{(R)}(\x;\s)$. We showed that these fall into four classes, each admitting particular factorisation properties given in corollary~\ref{theorem-leading_terms_form-corollary}.
Using this classification, we formulated, in section~\ref{infrared_regions_in_wideangle_scattering}, the precise requirements under which the subgraphs $H$, $J$ and $S$ of proposition \ref{proposition-region_vectors_are_hard_and_infrared} qualify as a region.

In section~\ref{section-graph_finding_algorithm_regions} we translated these requirements into a purely graph-theoretical language, given in theorem~\ref{theorem-algorithm_necessary_sufficient}. On this basis, we constructed a graph-finding algorithm to obtain all the regions appearing in the on-shell expansion of generic wide-angle scattering graphs. 
This algorithm is entirely graph-theoretical and as such it does not directly involve any particular representation of the integrand. We verified that this algorithm generates precisely the set of regions that are expected based on the Newton polytope construction of ref.~\cite{HrchJnsSlk22}.

Finally, we studied the analytic structure of the on-shell expansions for a specific class of regions, which we call jet-pairing soft regions. In such regions, any connected component of the soft subgraph couples to precisely two jets, and thus it follows (theorem~\ref{theorem-analytic_structure_jet_pairing_soft}) that the non-analytic behaviour in the on-shell limit factors out and involves exclusively powers of~$(p_i^2)^{\epsilon}$, while the remaining expansion is regular in all $p_i^2$. This further implies the commutativity of multiple on-shell expansions in graphs where all the regions are jet-pairing soft. Such graphs include any Sudakov form factors and any planar wide-angle scattering processes. Conversely, the non-commutativity of multiple on-shell expansions requires nonplanar graphs with at least three external momenta that are taken on-shell, and may arise only from regions with a nontrivial soft subgraph containing at least one soft vertex. 

We emphasise that the geometric MoR is only expected to capture all the regions in cases where there exists an analytic continuation from a Euclidean domain, where all invariants have the same sign, and thus infrared singularities due to cancellations between terms in the Lee-Pomeransky polynomial are immediately excluded. In this case, all singularities correspond to endpoint singularities in parameter space. While wide-angle scattering of massless particles, as analysed here, does not always satisfy this strict requirement (e.g. $2\to2$ massless scattering, where momentum conservation precludes Mandelstam invariants with equal signs), we nevertheless expect that no new regions appear due to the aforementioned cancellations, and all the singularities are captured by the facets we have constructed. We have not provided a proof of this, but we fully expect that \emph{all} regions, whether described by facets or not, should be related to some solutions of the Landau equations.

Our work may be related to ref.~\cite{Ma20}, which constructs a forest-based infrared subtraction method (at the integrand level) for wide-angle scattering amplitudes~\cite{AntsStm18,AntsHndStmYangZeng21locally}. In this context, it is useful to recall that Smirnov~\cite{Smn90} has proven that the 
MoR remainder --- the  difference between an original Euclidean Feynman integral and the truncated $n$-th order MoR expansion around some small mass (momentum) --- can be identified with the ultraviolet- (infrared-)subtracted result, according to a variant of Zimmermann's forest formula. If a similar relation can be constructed between the on-shell expansion studied in this paper and the infrared forest formula, one may rigorously establish the completeness of the set of regions, and furthermore, derive a proof of convergence of the MoR expansion.

In this paper we focused exclusively on scalar Feynman integrals in the on-shell expansion. A natural application of this work is the study of the singularities of Feynman graphs and complete amplitudes in gauge theory. 
In particular, considering the leading power of the on-shell expansion for an ultraviolet-finite wide-angle scattering integral, we have
\begin{equation}
\label{LO_power_expansion}
\I(\s)\simeq \I^{(H)}(\s\setminus\t) + \sum_{R\in {\rm IR}}\I^{(R)}(\s)\,, 
\end{equation}
where the sum in the second term goes over all infrared regions, while the hard region contribution at leading power, $\I^{(H)}(\s\setminus\t)$, is the integral with all $p_i^2\in \t$ set strictly on shell. The latter is of course infrared divergent, and one may wish to explore its singularities, which appear as poles in the dimensional regularization parameter $\epsilon$. Since the off-shell integral $\I(\s)$ on left hand side of (\ref{LO_power_expansion}) is finite, one finds that the singularities may be expressed as
\begin{equation}
\label{LO_power_expansion_poles_in_eps}
\I^{(H)}(\s\setminus\t)= - \sum_{R\in {\rm IR}}\I^{(R)}(\s) +{\cal O}(\epsilon^0)\,,
\end{equation}
where in each infrared region $R$, the off-shell parameters $p_i^2\in \t$ act as infrared regulators.
The logarithmic dependence on these parameters cancels between the infrared regions. These region integrals do of course have ultraviolet divergences in $\epsilon$ (which were not present in $\I(\s)$), making eq.~(\ref{LO_power_expansion_poles_in_eps}) consistent up to finite ${\cal O}(\epsilon^0)$ terms. 
This decomposition of the integral, along with the detailed characterisation of the regions provided in this paper, may be useful for 
a range of applications, including the program of constructing bases of quasi-finite integrals~\cite{vonManteuffel:2014qoa,vonManteuffel:2015gxa}, studying infrared singularities in wide-angle scattering amplitudes~\cite{Almelid:2017qju,Almelid:2015jia,Becher:2019avh,Falcioni:2021buo} and in the context of infrared subtraction in cross sections. 

This work may also help with the study of factorization and its violation, in particular within Soft-Collinear Effective Theory (SCET) \cite{BurStw13lectures, BchBrgFrl15book, BurFlmLk00, BurPjlSwt02-1, BurPjlSwt02-2,RthstStw16EFT}, for which the MoR is the foundation. Processes involving QCD jets can usually be described by so-called $\text{SCET}_\text{I}$~\cite{BurStw13lectures}, where the fields are of either collinear or ultrasoft. These correspond precisely to the jet and soft momentum scalings of eq.~(\ref{infrared_region_momentum_scaling}) above. Our findings regarding the requirements the subgraphs $H$, $J_i$ and $S$ must admit are valid to all orders in the power expansion, and can therefore feed directly into the analysis of wide-angle scattering using SCET.
 
The extension of this work beyond the on-shell expansion and beyond the context of wide-angle scattering would be very interesting. Different kinematics, involving propagator masses, or forward limits, as well as different expansions, such as threshold or mass expansions, may give rise to different requirements on what constitutes a facet, as well as new types of regions, such as e.g. potential and Glauber regions. Such regions may still be captured using the Newton polytope approach~\cite{JtzSmnSmn12,AnthnrySkrRmn19,SmnvSmnSmv19}, provided all potential cancellations are identified. The analysis of the Landau equations could play a key in systematising this. Further developing purely graph-theoretical rules characterising the regions beyond the realm of the on-shell expansion in wide-angle scattering, would also be highly desirable.

\acknowledgments
We would like to thank Balasubramanian Ananthanarayan, Lorenzo Magnea, Sebastian Mizera, Ben Page, Erik Panzer, Ratan Sarkar, George Sterman and Mao Zeng for valuable discussions. This work is supported by the UKRI FLF grant ``Forest Formulas for the LHC'' (Mr/S03479x/1), the STFC Consolidated Grant ``Particle Physics at the Higgs Centre'' and the Royal Society University Research Fellowship (URF/R1/201268).

\appendix

\section{The possible regions in the nonplanar double-box graph}
\label{appendix-possible_regions_nonplanar_doublebox_graph}

In this appendix we focus on the nonplanar double-box graph shown in figure~\ref{figure-nonplanar_doublebox}. We will show that although there are positive and negative kinematic invariants in the $\mathcal{F}$ polynomial, there are no regions due to cancellation between terms. In other words, all the regions appearing in this specific graph in the on-shell expansion are of the endpoint type, which is in accordance with our more general statement regarding wide-angle scattering in section~\ref{region_vectors_from_Newton_polytope}.

From the Feynman parameterisation shown in figure~\ref{figure-nonplanar_doublebox}, the $\mathcal{F}$ polynomial is
\begin{eqnarray}
\mathcal{F}(\boldsymbol{\alpha};\s)=&& (-p_1^2) \left[ \alpha_1 \alpha_2 (\alpha_4 +\alpha_5 +\alpha_6 +\alpha_7) +\alpha_2 \alpha_4 \alpha_7 \right] \nonumber\\
+&&(-p_2^2) \left[ \alpha_2 \alpha_3 (\alpha_4 +\alpha_5 +\alpha_6 +\alpha_7) +\alpha_2 \alpha_5 \alpha_6 \right] \nonumber\\
+&&(-p_3^2) \left[ \alpha_4 \alpha_5 (\alpha_1 +\alpha_2 +\alpha_3 +\alpha_6 +\alpha_7) +\alpha_1 \alpha_5 \alpha_7 +\alpha_3 \alpha_4 \alpha_6 \right] \nonumber\\
+&&(-p_4^2) \left[ \alpha_6 \alpha_7 (\alpha_1 +\alpha_2 +\alpha_3 +\alpha_4 +\alpha_5) +\alpha_1 \alpha_4 \alpha_6 +\alpha_3 \alpha_5 \alpha_7 \right] \nonumber\\
+&&(-q_{12}^2) \left[ \alpha_1 \alpha_3 (\alpha_4 +\alpha_5 +\alpha_6 +\alpha_7) +\alpha_3 \alpha_4 \alpha_7 +\alpha_1 \alpha_5 \alpha_6 \right] \nonumber\\
+&&(-q_{13}^2) \alpha_2 \alpha_5 \alpha_7 +(-q_{14}^2) \alpha_2 \alpha_4 \alpha_6.
\end{eqnarray}
We choose a special frame of reference such that
\begin{eqnarray}
(-p_1^2),\ (-p_2^2),\ (-p_3^2),\ (-p_4^2),\ (-q_{12}^2),\ (-q_{13}^2)>0, \qquad (-q_{14}^2)<0.
\end{eqnarray}
Moreover, the on-shell limit indicates $|p_i^2|\ll |q_{jk}^2|$ for any $p_i$ and $q_{jk}$ above.
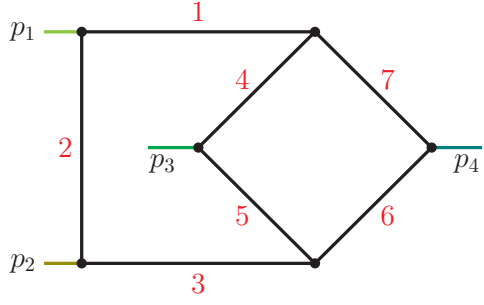
\begin{figure}[t]
\centering
\resizebox{0.45\textwidth}{!}{
\begin{tikzpicture}[decoration={markings,mark=at position 0.55 with {\arrow{latex}}}] 

\node (p1) at (1,5) {\Large $p_1$};
\node (p2) at (1,1) {\Large $p_2$};
\node (p3) at (3,3) {};
\node (p4) at (9,3) {};

\draw (p1) edge [ultra thick, color=LimeGreen] (2,5) node [right] {};
\draw (p2) edge [ultra thick, color=olive] (2,1) node [right] {};
\draw (p3) edge [ultra thick, color=Green] (4,3) node [right] {};
\draw (p4) edge [ultra thick, color=teal] (8,3) node [right] {};
\draw (2,1) edge [ultra thick] (6,1) node [right] {};
\draw (2,5) edge [ultra thick] (6,5) node [right] {};
\draw (4,3) edge [ultra thick] (6,1) node [right] {};
\draw (8,3) edge [ultra thick] (6,1) node [right] {};
\draw (4,3) edge [ultra thick] (6,5) node [right] {};
\draw (8,3) edge [ultra thick] (6,5) node [right] {};
\draw (2,5) edge [ultra thick] (2,1) node [right] {};

\path (2,5) -- (6,5) node [midway,yshift=10pt,color=Red] {\Large $1$};
\path (2,5) -- (2,1) node [midway,xshift=-8pt,color=Red] {\Large $2$};
\path (2,1) -- (6,1) node [midway,yshift=-10pt,color=Red] {\Large $3$};
\path (6,5) -- (4,3) node [midway,xshift=-7pt,yshift=5pt,color=Red] {\Large $4$};
\path (4,3) -- (6,1) node [midway,xshift=-7pt,yshift=-5pt,color=Red] {\Large $5$};
\path (6,1) -- (8,3) node [midway,xshift=7pt,yshift=-5pt,color=Red] {\Large $6$};
\path (8,3) -- (6,5) node [midway,xshift=7pt,yshift=5pt,color=Red] {\Large $7$};
\path (p3)-- (4,3) node [midway,xshift=-5pt,yshift=-8pt] {\Large $p_3$};
\path (p4)-- (8,3) node [midway,xshift=5pt,yshift=-8pt] {\Large $p_4$};

\node [draw=none,circle,minimum size=5pt,fill=Black,inner sep=0pt,outer sep=0pt] () at (2,1) {};
\node [draw=none,circle,minimum size=5pt,fill=Black,inner sep=0pt,outer sep=0pt] () at (2,5) {};
\node [draw=none,circle,minimum size=5pt,fill=Black,inner sep=0pt,outer sep=0pt] () at (4,3) {};
\node [draw=none,circle,minimum size=5pt,fill=Black,inner sep=0pt,outer sep=0pt] () at (8,3) {};
\node [draw=none,circle,minimum size=5pt,fill=Black,inner sep=0pt,outer sep=0pt] () at (6,5) {};
\node [draw=none,circle,minimum size=5pt,fill=Black,inner sep=0pt,outer sep=0pt] () at (6,1) {};
\end{tikzpicture}
}
\vspace{-3em}
\vspace{12pt}
\caption{The nonplanar double-box graph, where all the four momenta are nearly on-shell, denoted by $p_i$ ($i=1,2,3,4$). The red numbers label the Feynman parameters.}
\label{figure-nonplanar_doublebox}
\end{figure}

We further assume that all the regions (except the hard region) should be characterised by the Landau equations, which are
\begin{subequations}
\begin{align}
&\mathcal{F}(\boldsymbol{\alpha};\s)=0,
\label{Landau_equations_third_representation_I}\\ 
\forall i,\quad &\alpha_i=0\ \text{ or }\ \partial \mathcal{F}/\partial \alpha_i =0.
\label{Landau_equations_third_representation_II}
\end{align}
\label{Landau_equations_third_representation}
\end{subequations}
Note that this is the ``third representation'' of the Landau equations according to ref.~\cite{EdenLdshfOlvPkhn02book}.

Suppose there is a solution of the Landau equation above, such that there are both positive and negative terms in $\mathcal{F}$. Since the term $(-q_{14}^2) \alpha_2 \alpha_4 \alpha_6$ is the only negative one, according to eq.~(\ref{Landau_equations_third_representation_I}) we must have
\begin{eqnarray}
\alpha_i\neq 0 \qquad \text{for }i=2,4,6,
\end{eqnarray}
and hence from (\ref{Landau_equations_third_representation_II}),
\begin{eqnarray}
\partial \mathcal{F}/ \partial \alpha_i = 0 \qquad \text{for }i=2,4,6.
\end{eqnarray}
The equation $\partial \mathcal{F}/ \partial \alpha_2=0$ then yields
\begin{eqnarray}
\sum_{i=1}^4 (-p_i^2) (\cdots)_i +(-q_{13}^2) \alpha_5 \alpha_7+ (-q_{14}^2) \alpha_4 \alpha_6 =0,
\end{eqnarray}
where $(\cdots)_i$ denotes a polynomial of the Feynman parameters. In the limit $|p_i^2|/|q_{jk}^2|\to 0$, we further derive from above
\begin{eqnarray}
\alpha_5,\ \alpha_7 \neq 0.
\end{eqnarray}
The same reasoning applies to $\partial \mathcal{F}/ \partial \alpha_4 =0$ and $\partial \mathcal{F}/ \partial \alpha_6 = 0$, which gives
\begin{eqnarray}
\alpha_1,\ \alpha_3 \neq 0.
\end{eqnarray}

To summarise, for this specific example we have found that if there is a region due to cancellation between terms in $\mathcal{F}$, we must have $\alpha_i\neq 0$, and hence, $\partial \mathcal{F}/ \partial \alpha_i =0$  for any $i$. Specifically, the equation $\partial \mathcal{F} / \partial \alpha_1 =0$ yields
\begin{eqnarray}
\sum_{i=1}^4 (-p_i^2) (\cdots)_i+ (-q_{12}^2) \left[ \alpha_3 (\alpha_4 +\alpha_5 +\alpha_6 +\alpha_7) + \alpha_5 \alpha_6 \right]=0,
\end{eqnarray}
where the first term vanishes in the on-shell limit, and we observe that the second term is nonzero, and hence the condition for the Landau singularity, eq.~(\ref{Landau_equations_third_representation_II}), is not satisfied. This implies that all the regions must be of the endpoint type, and are associated to the facets of the Newton polytope.

\section{Relation between Schwinger and Lee-Pomeransky representation}
\label{appendix-schwinger-lp}

The Schwinger representation of the integral $\I(G)$ is
\begin{equation}
\label{IG_Schwinger}
\I(G) = \int_0^\infty \left( \prod_{e\in G} \frac{d\tilde{x}_e}{\tilde{x}_e} \frac{\tilde{x}_e^{\nu_e}}{\Gamma(\nu_e)}\right) \, \frac{e^{-\mathcal{F}/\mathcal{U}}}{\mathcal{U}^{D/2}}.
\end{equation}
where we use the standard definition of the Symanzik polynomials from eq.~(\ref{UFterm_general_expression}). Multiplication by
\begin{equation}
1 =  \int_0^\infty \frac{d\tilde{y}}{\tilde{y}} \frac{\tilde{y}^{(L+1)D/2-\nu}}{\Gamma((L+1)D/2-\nu)} e^{-\tilde{y}}
\end{equation}
together with a rescaling of the integration variables $\tilde{x}_e$ in eq.~(\ref{IG_Schwinger}) according to 
\begin{equation}
\tilde{x}_e = \tilde{y} x_e \quad \forall\ e \in G
\end{equation}
results in
\begin{equation}
\I(G) = \int_0^\infty  \left( \prod_{e\in G} \frac{dx_e}{x_e} \frac{x_e^{\nu_e}}{\Gamma(\nu_e)}\right)\frac{1}{\mathcal{U}^{D/2}}\, \int_0^{\infty}\frac{d\tilde{y}}{\tilde{y}}\, \frac{\tilde{y}^{D/2}}{\Gamma((L+1)D/2-\nu)} e^{-\tilde{y}\left(1+\mathcal{F}/\mathcal{U}\right)}.
\end{equation}
Here we used the homogeneity properties 
\[
\mathcal{U}(\{\tilde{x}_e\})=\tilde{y}^{L}\mathcal{U}(\{x_e\}), 
\qquad \mathcal{F}(\{\tilde{x}_e\};\s)=\tilde{y}^{L+1}\mathcal{F}(\{x_e\};\s).
\]
After rescaling the integration variable $\tilde{y}$ as
\begin{equation}
\tilde{y} = \frac{y}{1+\mathcal{F}/\mathcal{U}},
\end{equation}
we recognise the integral over $y$ as the integral representation of $\Gamma(D/2)$, leading to
\begin{equation}
\I(G) = \frac{\Gamma(D/2)}{\Gamma((L+1)D/2-\nu)}\int_0^\infty \left( \prod_{e\in G} \frac{dx_e}{x_e} \frac{x_e^{\nu_e}}{\Gamma(\nu_e)}\right) \, \frac{1}{\mathcal{U}^{D/2}} \frac{1}{(1+\mathcal{F}/\mathcal{U})^{D/2}}\,,
\end{equation}
which is the Lee-Pomeransky representation of $\I(G)$ as defined in~eqs.~(\ref{lee_pomeransky_definition}) and~(\ref{lee_pomeransky_integrand_definition}).

Consider now a specific region $R$. The Schwinger representation of the propagator (\ref{Schwinger_param}), along with the scaling of the virtualities in eq.~(\ref{virtuality_scaling}), 
imply that the Schwinger parameters scale according to
\begin{equation}
\label{Schwinger_x_scaling_rule_appendix}
\tilde{x}_e \sim \frac{1}{l_e^2(k,p,q)} \sim \lambda^{\tilde{u}_{R,e}}\,,
\end{equation}
where $\tilde{\boldsymbol{u}}_{R}$ is the scaling vector of the Schwinger parameters in region $R$.
Specifically, it follow that 
the Schwinger parameters $\tilde{x}_e$ in the integrand of 
eq.~(\ref{IG_Schwinger})
corresponding to hard, jet and soft edges  
scale as follows:
\begin{eqnarray}
\tilde{x}^{[H]}\sim \lambda^0,\ \quad \tilde{x}^{[J]}\sim \lambda^{-1},\ \quad
 \tilde{x}^{[S]}\sim \lambda^{-2}\,.
\label{Schwinger_parameter_scales}
\end{eqnarray}

Starting with the Schwinger 
representation of the rescaled integrand for region~$R$, namely 
\begin{equation}
\label{IG_Schwinger_rescaled}
\I^{(R)}(G) = \int_0^\infty \left( \prod_{e\in G} \frac{d\tilde{x}_e}{\tilde{x}_e} \frac{(\lambda^{{\tilde{u}_R,e}}\tilde{x}_e)^{\nu_e}}{\Gamma(\nu_e)}\right) \, 
\frac{\exp\bigg[{-\mathcal{F}(\lambda^{\tilde{\boldsymbol{u}}_R}\tilde{\x};\lambda^{\w}\s)
/\mathcal{U}(\lambda^{\tilde{\boldsymbol{u}}_R}\tilde{\x}})\bigg]}{\left(\mathcal{U}(\lambda^{\tilde{\boldsymbol{u}}_R}\tilde{\x})\right)^{D/2}}\,,
\end{equation}
and
repeating the derivation above, we obtain the following Lee-Pomeransky representation of region $R$, in terms of the same Symanzik polynomials, now written in terms of the rescaled Lee-Pomeransky variables $x_e$, 
\begin{equation}
\I^{(R)}(G) = \frac{\Gamma\Big(\frac{D}{2}\Big)}{\Gamma\Big((L+1)\frac{D}{2}-\nu\Big)}\int_0^\infty \left( \prod_{e\in G} \frac{dx_e}{x_e} \frac{(\lambda^{{u_R,e}}{x}_e)^{\nu_e}}{\Gamma(\nu_e)}\right) \, {\Big(\mathcal{U}(\lambda^{\boldsymbol{u}_R}{\x})+\mathcal{F}(\lambda^{\boldsymbol{u}_R}{\x};\lambda^{\w}\s) \Big)^{-\frac{D}{2}}
}\,,
\end{equation}
where in the final step we made the identification of the region vectors:
$\tilde{\boldsymbol{u}}_R ={\boldsymbol{u}}_R\,.$
The conclusion is thus, that the region vectors ${\boldsymbol{u}}_R$ of the Lee-Pomeransky representation are the same as in the Schwinger representation, as summarised by eq.~(\ref{x_scaling_rule}). In particular, given hard, jet and soft momentum modes (\ref{infrared_region_momentum_scaling}), we have established the scaling rule of the Lee-Pomeransky parameters in eq.~(\ref{LP_parameter_scales}).

\section{A proof of the requirements for an infrared region}
\label{appendix-necessity_sufficiency_requirements}

In this appendix we prove that for any configuration satisfying proposition~\ref{proposition-region_vectors_are_hard_and_infrared}, the requirements of $H,J$ and $S$ formulated in section~\ref{infrared_regions_in_wideangle_scattering} provide a \emph{necessary and sufficient} condition for the existence of a corresponding region $R$, such that $\boldsymbol{v}_R$ is normal to a lower facet of the Newton polytope $\Delta^{(N+1)}[{\cal P}(G)]$.

We first show that these requirements are sufficient, i.e. $\text{dim}(f_R)= N$ if the requirements of $H,J$ and $S$ are satisfied, where $f_R$ is the lower face that contains the points $\boldsymbol{r}$ corresponding to the minimum of $\boldsymbol{r}\cdot \boldsymbol{v}_R$. We aim to find the following vectors in $f_R$:
\begin{subequations}
\label{basis_vectors_infrared_facet}
    \begin{align}
        \Delta \boldsymbol{r}_{[H]}^{} \equiv (\underset{N(H)}{\underbrace{0,\dots,0,1,0,\dots, 0}},\underset{N(J)+N(S)}{\underbrace{0, \dots, 0}};0);
        \label{basis_vectors_infrared_facet_hard}\\
        \Delta \boldsymbol{r}_{[J]}^{} \equiv (\underset{N(H)}{\underbrace{0, \dots, 0}},\underset{N(J)}{\underbrace{0\dots,0,1,0,\dots 0}},\underset{N(S)}{\underbrace{0, \dots, 0}};-1);
        \label{basis_vectors_infrared_facet_jet}\\
        \Delta \boldsymbol{r}_{[S]}^{} \equiv (\underset{N(H)+N(J)}{\underbrace{0, \dots, 0}},\underset{N(S)}{\underbrace{0,\dots,0,1,0,\dots, 0}};-2),
        \label{basis_vectors_infrared_facet_soft}
    \end{align}
\end{subequations}
where we have parameterised the propagators of $G$ such that the first $N(H)$ parameters correspond to hard edges, the next $N(J)$ correspond to jet edges, and the final $N(S)$ correspond to soft edges. One can check that the vectors in eq.~(\ref{basis_vectors_infrared_facet}) are linearly independent of each other, and all are perpendicular to $\boldsymbol{v}_R$. Therefore, it suffices to find all the vectors above in order to justify that $f_R$ is a facet.

\bigbreak \noindent
\textbf{The existence of $N(H)$ vectors $\Delta \boldsymbol{r}_{[H]}^{}$.} Here we show that if all the propagators of $H_\text{red}$ are off-shell, then the vectors $\Delta \boldsymbol{r}_{[H]}^{}$ are in $f_R$. As introduced above, we have supposed that there are $n$ nontrivial 1VI components of $H$, denoted as $\gamma_1^H,\dots,\gamma_n^H$.

From the factorisation property of $\mathcal{U}^{(R)}(\x)$, eq.~(\ref{leading_Uterms_factorise}), any $\mathcal{U}^{(R)}$ term can be obtained by combining any chosen spanning trees of $H$, $\widetilde{J}_1,\dots,\widetilde{J}_K$ and $\widetilde{S}$, so we can consider a specific set of $\mathcal{U}^{(R)}$ terms having a common set of spanning trees of $\widetilde{J}_1,\dots,\widetilde{J}_K$ and $\widetilde{S}$, but each having a distinct spanning tree of $H$, denoted as $T^1(H)$. As in eq.~(\ref{hard_basis_vectors_U_polynomial}), it follows that the differences of these $\mathcal{U}^{(R)}$ terms take the form:
\begin{eqnarray}
\Big (\underset{N(\gamma_1^H)+\dots+N(\gamma_{j-1}^H)}{\underbrace{0,\dots,0}}, \underset{N(\gamma_j^H)}{\underbrace{1,0,\dots,0,-1,0,\dots,0}},\underset{N(\gamma_{j+1}^H)+\dots+N(\gamma_{n}^H)}{\underbrace{0,\dots,0}},\underset{\widehat{N}(H_\text{red})}{\underbrace{0,\dots,0}},\underset{N(J)+N(S)}{\underbrace{0,\dots,0}};0 \Big ).\qquad
\label{infrared_hardgraph_basis_vectors_U_polynomial}
\end{eqnarray}
For the first $N(H)$ entries of the vector above, we have chosen a specific parameterisation: the first $N(\gamma_1^H)$ entries correspond to the edges of $\gamma_1^H$, the next $N(\gamma_2^H)$ entries correspond to the edges of $\gamma_2^H$, etc.; the last $\widehat{N}(H_\text{red})$ entries are the edges of the graph $H_\text{red}$ that are also in $H$. There are two nonzero entries in eq.~(\ref{infrared_hardgraph_basis_vectors_U_polynomial}): $1$ appears in the first entry associated to the edges of $\gamma_j^H$, while $-1$ may appear in any of the other entries. It follows that there are $\sum_{j=1}^{n} (N(\gamma_j^H)-1) = \sum_{j=1}^{n} N(\gamma_j^H) -n$ vectors in total in eq.~(\ref{infrared_hardgraph_basis_vectors_U_polynomial}).

In line with eq.~(\ref{relation_propagator_number_reduced_HandJ}), we still need to find another $n+\widehat{N}(H_\text{red})$ vectors that are linear combinations of $\Delta \boldsymbol{r}_{[H]}$ and independent of eq.~(\ref{infrared_hardgraph_basis_vectors_U_polynomial}). For any of the $\mathcal{U}^{(R)}$ terms mentioned above, we denote the corresponding spanning tree by $T^1$. We now consider the spanning 2-tree $T^2$ obtained from $T^1$ by removing one off-shell edge $e\in T^1(H)$ from it. According to the table of corollary~\ref{theorem-leading_terms_form-corollary}, we know that these spanning 2-trees correspond to the $\mathcal{F}_\text{I}^{(q^2,R)}$ terms. Following the same argument in section~\ref{space_spanning_trees}, we can show that the vector $\boldsymbol{r}_2 -\boldsymbol{r}_1$, where $\boldsymbol{r_1}$ and $\boldsymbol{r}_2$ are the points associated to $T^1$ and $T^2$ respectively, may take one of the following two forms:
\begin{subequations}
\label{infrared_hardsubgraph_UandFs_difference}
    \begin{align}
        \Big (\underset{N(\gamma_1^H)+\dots+N(\gamma_{j-1}^H)}{\underbrace{0,\dots,0}}, \underset{N(\gamma_j^H)}{\underbrace{0,\dots,0,1,0,\dots,0}},\underset{N(\gamma_{j+1}^H)+\dots+N(\gamma_n^H)+\widehat{N}(H_\text{red})}{\underbrace{0,\dots,0}},\underset{N(J)+N(S)}{\underbrace{0,\dots,0}};0 \Big ),
        \label{infrared_hardsubgraph_UandFs_differenceI}
        \\
        \Big( \underset{N(\gamma_1^H)+\dots+N(\gamma_n^H)}{\underbrace{0,\dots,0}},\underset{\widehat{N}(H_\text{red})}{\underbrace{0,\dots,0,1,0,\dots,0}},\underset{N(J)+N(S)}{\underbrace{0,\dots,0}};0 \Big ).
        \label{infrared_hardsubgraph_UandFs_differenceII}
    \end{align}
\end{subequations}
In each of the vectors above, there is exactly one entry with value $1$, while all the others are~$0$. Note that in eq.~(\ref{infrared_hardsubgraph_UandFs_differenceI}), we may choose the entry $1$ at $n$ distinct positions, each corresponding to removing some off-shell propagator in a different nontrivial 1VI component $\gamma_j^H$ with $i=1,\dots,n$; these are bound to be independent. In eq.~(\ref{infrared_hardsubgraph_UandFs_differenceII}), we may choose the entry $1$ at any of the $\widehat{N}(H_\text{red})$ entries, each corresponding to removing a different internal propagator in $H_\text{red}\cap H$. Therefore, we have found another $n + \widehat{N}(H_\text{red})$ vectors that are independent of those in (\ref{infrared_hardgraph_basis_vectors_U_polynomial}), and proven that all the vectors $\Delta \boldsymbol{r}_{[H]}^{}$ shown in (\ref{basis_vectors_infrared_facet_hard}) exist.

\bigbreak \noindent
\textbf{The existence of $N(J)$ vectors $\Delta \boldsymbol{r}_{[J]}^{}$.} We now show that if all the internal propagators of $\widetilde{J}_{i,\textup{red}}$ carry exactly the momentum $p_i^\mu$, then the vectors taking the form of $\Delta \boldsymbol{r}_{[J]}^{}$ in eq.~(\ref{basis_vectors_infrared_facet_jet}) exist in~$f_R$. The requirement that all internal propagators carry exactly the jet momentum is equivalent to excluding jet subgraphs that include dead-end structures. Note that such may appear in $\widetilde{J}_{i,\text{red}}$ due to auxiliary propagators that are introduced upon separating nontrivial 1VI components.

Without loss of generality, we take $i=1$ and study the jet subgraph $J_1$ below. As above, we have assumed that $\gamma_1^{\widetilde{J}_1},\dots,\gamma_{n_1}^{\widetilde{J}_1}$ are the $n_1$ nontrivial 1VI components of $\widetilde{J}_1$.

Now consider a specific set of $\mathcal{U}^{(R)}$ terms sharing common spanning trees of $H$, $\widetilde{J}_2,\dots,\widetilde{J}_K$ and $\widetilde{S}$, but each having a distinct spanning tree of $\widetilde{J}_1$. Theorem~\ref{theorem-Uterms_space_dimensionality} then indicates that the following vectors can be found in $f_R$, which are the differences of these $\mathcal{U}^{(R)}$ terms: 
\begin{eqnarray}
&&\Big (\underset{N(H)}{\underbrace{0,\dots,0}}, \underset{N(\gamma_1^{\widetilde{J}_1})+\dots+N(\gamma_{j-1}^{\widetilde{J}_1})}{\underbrace{0,\dots,0}}, \underset{N(\gamma_j^{\widetilde{J}_1})}{\underbrace{1,0,\dots,0,-1,0,\dots,0}},\nonumber\\
&&\hspace{4cm}\underset{N(\gamma_{j+1}^{\widetilde{J}_1})+\dots+N(\gamma_{n_1}^{\widetilde{J}_1})}{\underbrace{0,\dots,0}},\underset{\widehat{N}(\widetilde{J}_{1,\text{red}})}{\underbrace{0,\dots,0}},\underset{N(J_2)+\dots+N(J_K)+N(S)}{\underbrace{0,\dots,0}};0 \Big ),
\label{pinch_jetgraph_basis_vectors_U_polynomial}
\end{eqnarray}
where $j=1,\dots,n_1$. One may verify that these vectors are linearly independent, and the total number of them is $\sum_{j=1}^{n_1} \Big( L(\gamma_j)-1 \Big)= \sum_{j=1}^{n_1} L(\gamma_j) =n_1$. From eq.~(\ref{relation_propagator_number_reduced_HandJ}), we still need to identify $n_1+\widehat{N}(\widetilde{J}_{1,\text{red}})$ additional independent vectors in $f_R$. These take one of the following two forms:
\begin{subequations}
\label{pinch_jetsubgraph_UandFs_difference}
\begin{align}
    \Big (\underset{N(H)}{\underbrace{0,\dots,0}}, \underset{N(\gamma_1)+\dots+N(\gamma_{j-1})}{\underbrace{0,\dots,0}}, \underset{N(\gamma_j)}{\underbrace{0,\dots,0,1,0,\dots,0},}\hspace{4.5cm}\nonumber\\
    \underset{N(\gamma_{j+1})+\dots+N(\gamma_{n_1})}{\underbrace{0,\dots,0}},\underset{\widehat{N}(\widetilde{J}_{1,\text{red}})}{\underbrace{0,\dots,0}},\underset{N(J_2)+\dots+N(J_K)+N(S)}{\underbrace{0,\dots,0}};-1 \Big ),
    \label{pinch_jetsubgraph_UandFs_difference1}\\
    \Big (\underset{N(H)}{\underbrace{0,\dots,0}}, \underset{N(\gamma_1)+\dots+N(\gamma_{n_1})}{\underbrace{0,\dots,0}}, \underset{\widehat{N}(\widetilde{J}_{1,\text{red}})}{\underbrace{0,\dots,0,1,0,\dots,0}},\underset{N(J_2)+\dots+N(J_K)+N(S)}{\underbrace{0,\dots,0}};-1 \Big ).
    \label{pinch_jetsubgraph_UandFs_difference2}
\end{align}
\end{subequations}
In each of the vectors above, there is exactly one entry with value $1$, while all the others associated to the Lee-Pomeransky parameters are~$0$. The last entry is $-1$. Note that in eq.~(\ref{infrared_hardsubgraph_UandFs_differenceI}), we may choose the entry $1$ at $n$ distinct positions, each corresponding to removing some jet edge in a different nontrivial 1VI component $\gamma_j^{\widetilde{J}_1}$ with $j=1,\dots,n_1$; these are bound to be independent. In (\ref{infrared_hardsubgraph_UandFs_differenceII}), we may choose the entry $1$ at any of the $\widehat{N}(\widetilde{J}_{1,\text{red}})$ entries, each corresponding to removing a different edge in $\widetilde{J}_{1,\text{red}}$.

We now explain how these vectors can be obtained under the condition that all the internal propagators of $\widetilde{J}_{1,\textup{red}}$ carry the same momentum $p_1$. In this case, the configuration of $\widetilde{J}_{1,\textup{red}}$ must be as shown in figure~\ref{reduced_form_jetsubgraph}. Recall that the contracted graph $\widetilde{J}_1$ contains a single auxiliary vertex. In any given spanning tree of $\widetilde{J}_1$ denoted by $T^1(\widetilde{J}_1)$, there is a unique path connecting this auxiliary vertex with the external momentum $p_1$. In line with figure~\ref{reduced_form_jetsubgraph}, this path must contain at least one edge from each nontrivial 1VI component of $\widetilde{J}_1$, and all the edges of $\widetilde{J}_{1,\textup{red}}$. By removing any edge from this path, one obtains a spanning 2-tree of $\widetilde{J}_1$ denoted by $T^2(\widetilde{J}_1)$.

We now take any spanning trees of the graphs $H$, $\widetilde{J}_2,\dots,\widetilde{J}_K$ and $\widetilde{S}$, and consider the union of these spanning trees with $T^1(\widetilde{J}_1)$ and $T^2(\widetilde{J}_1)$ respectively. The results are a spanning tree $T^1$ and a spanning 2-tree $T^2$ of $G$, which correspond to the points $\boldsymbol{r}_1$ and $\boldsymbol{r}_2$ of the following form:
\begin{subequations}
\begin{align}
    \boldsymbol{r}_1= \Big (\underset{N(H)}{\underbrace{\phantom{|}\dots\phantom{|}}}, \underset{N_1}{\underbrace{a_1^{(1)},\dots,a_{N_1}^{(1)}}},\dots,\underset{N_{n_1}}{\underbrace{a_1^{(n_1)},\dots,a_{N_{n_1}}^{(n_1)}}},\underset{N'}{\underbrace{0,\dots,0}},\underset{N(J_2)+\dots+N(J_K)+N(S)}{\underbrace{\phantom{b_{N_{n_1}}^{(n_1)}}\dots\phantom{b_{N_n}^{(n)}}}};0 \Big ),
    \\
    \boldsymbol{r}_2= \Big (\underset{N(H)}{\underbrace{\phantom{|}\dots\phantom{|}}}, \underset{N_1}{\underbrace{b_1^{(1)},\dots,b_{N_1}^{(1)}}},\dots,\underset{N_{n_1}}{\underbrace{b_1^{(n_1)},\dots,b_{N_{n_1}}^{(n_1)}}},\underset{N'}{\underbrace{b'_1,\dots,b'_{N'}}},\underset{N(J_2)+\dots+N(J_K)+N(S)}{\underbrace{\phantom{b_{N_{n_1}}^{(n_1)}}\dots\phantom{b_{N_{n_1}}^{(n_1)}}}};1 \Big ).
\end{align}
\end{subequations}
Note that $\boldsymbol{r}_1$ corresponds to a $\mathcal{U}^{(R)}$ term, while $\boldsymbol{r}_2$ corresponds to a $\mathcal{F}^{(p_1^2,R)}$ term. In the equations above, we have used $N_j$ ($j=1,\dots,n_1$) as the abbreviations for $N(\gamma_j^{\widetilde{J}_1})$, and $N'$ as the abbreviation for $\widehat{N}(\widetilde{J}_{1,\text{red}})$. Each of the parameter entries is either $0$ or $1$. From the analysis of the previous paragraph, the vectors $\boldsymbol{r}_1$ and $\boldsymbol{r}_2$ are identical except for two entries: one of the first $N$ entries is $0$ in $\boldsymbol{r}_1$ and $1$ in $\boldsymbol{r}_2$, and the last entry is $0$ in $\boldsymbol{r}_1$ and $1$ in $\boldsymbol{r}_2$. Clearly $\boldsymbol{r}_2 -\boldsymbol{r}_1$ exhausts all the vectors in eq.~(\ref{pinch_jetsubgraph_UandFs_difference}), thus we have obtained another $n_1 + \widehat{N}(\widetilde{J}_{1,\text{red}})$ vectors that are independent of those in (\ref{infrared_hardgraph_basis_vectors_U_polynomial}).

In summary, we proved that the requirement on $J$ is sufficient for the existence of $N(J)$ vectors $\Delta \boldsymbol{r}_{[J]}^{}$.

\bigbreak \noindent
\textbf{The existence of $N(S)$ vectors $\Delta \boldsymbol{r}_{[S]}^{}$.} Finally, we show that if every connected component of $S$ connects at least two jet subgraphs $J_i$ and $J_j$, then the all the vectors $\Delta \boldsymbol{r}_{[S]}^{}$ in eq.~(\ref{basis_vectors_infrared_facet_soft}) exist in $f_R$. Without loss of generality, we consider the graph $S_1$, which can be any connected component of $S$. By definition, its corresponding contracted graph $\widetilde{S}_1$ is a nontrivial 1VI graph.

We then consider the set of $\mathcal{U}^{(R)}$ terms sharing common spanning trees of $H$, $\widetilde{J}_1,\dots,\widetilde{J}_K$ and $\widetilde{S}_2,\dots, \widetilde{S}_n$, but each having a distinct spanning tree of $\widetilde{S}_1$. Again, from corollary~\ref{theorem-Uterms_space_dimensionality_corollary1}, the following vectors below can be found in $f_R$:
\begin{eqnarray}
\Big (\underset{N(H)+N(J)}{\underbrace{0,\dots,0}}, \underset{N(S_1)}{\underbrace{1,0,\dots,0,-1,0,\dots,0}},\underset{N(S_2)+\dots+N(S_n)}{\underbrace{0,\dots,0}};0 \Big ).
\label{pinch_softgraph_basis_vectors_U_polynomial}
\end{eqnarray}
These vectors are obtained as differences between $\mathcal{U}^{(R)}$ terms, where we have parameterised the soft propagators in a way that the first $N(S_1)$ entries of them are associated to the edges of $S_1$, the next $N(S_2)$ entries are associated to the edges of $S_2$, etc.

For any given $\mathcal{U}^{(R)}$ term, we denote the corresponding spanning tree by $T^1$. Having assumed that $S_1$ connects at least two jet subgraphs $J_i$ and $J_j$, there is an edge $e_0\in S_1\setminus T^1$, such that in the unique loop of the graph $T^1\cup e_0$, there are edges from $S_1$, $J_i$ and $J_j$. We then remove an edge $e_i\in J_i$ and $e_j\in J_j$ from it, and the result $T^2\equiv T^1\cup e_0 \setminus (e_i\cup e_j)$ is a spanning 2-tree. From corollary~\ref{theorem-leading_terms_form-corollary}, $T^2$ corresponds to an $\mathcal{F}_\text{II}^{(q^2,R)}$ term. The vector \hbox{$\Delta\boldsymbol{r}'\equiv \boldsymbol{r}_1-\boldsymbol{r}_2$}, where $\boldsymbol{r}_1$ and $\boldsymbol{r}_2$ are the points corresponding to $T^1$ and $T^2$ respectively, takes the following form
\begin{eqnarray}
\Delta \boldsymbol{r}' = \Big( \underset{N(H)}{\underbrace{0,\dots,0}}, \underset{N(J)}{\underbrace{0,\dots,0,-1,0,\dots,0,-1,0,\dots,0}},\underset{N(S)}{\underbrace{0,\dots,0,1,0,\dots,0}};0 \Big),
\end{eqnarray}
where the $-1$ values are associated to the jet propagators $e_i$ and $e_j$ respectively, and $1$ is associated to the soft propagator $e_0$, as we have explained above. First of all, the vector $\Delta \boldsymbol{r}'$ is in $f_R$ because both $\boldsymbol{r}_1$ and $\boldsymbol{r}_2$ correspond to $\mathcal{P}_0^{(R)}$ terms. Meanwhile, since we have already proved that all the vectors $\Delta \boldsymbol{r}_{[J]}$ in eq.~(\ref{basis_vectors_infrared_facet_jet}) exist in $f_R$, we can pick two other vectors in $f_R$, $\Delta \boldsymbol{r}_{i}^{}$ and $\Delta \boldsymbol{r}_{j}^{}$, such that the respective entries $e_i$ and $e_j$ are those that evaluate to~$1$. The vector $\Delta \boldsymbol{r}\equiv \Delta \boldsymbol{r}'+ \Delta \boldsymbol{r}_{i}^{} + \Delta \boldsymbol{r}_{j}^{}$, then takes the following form
\begin{eqnarray}
\Delta \boldsymbol{r}=\Big( \underset{N(H)+N(J)}{\underbrace{0, \dots, 0}},\underset{N(S_1)}{\underbrace{0,\dots,0,1,0,\dots, 0}},\underset{N(S_2)+\dots+N(S_n)}{\underbrace{0,\dots, 0}};-2 \Big),
\label{pinch_softgraph_basis_vectors_UminusF_polynomial}
\end{eqnarray}
where the only entry with the value $1$ is associated to $e_0$. Since $S_1$ can be any connected component of $S$, we can obtain further $n-1$ vectors in $f_R$ that are similar to eq.~(\ref{pinch_softgraph_basis_vectors_UminusF_polynomial}), where the value $1$ is associated to a specific propagator of $S_k$ ($k=2,\dots,n$). Therefore, we have shown that all the vectors $\Delta\boldsymbol{r}_{[S]}$ in (\ref{basis_vectors_infrared_facet_soft}) exist in $f_R$.

\bigbreak
In conclusion, we have shown that the requirements of $H$, $J$ and $S$ introduced in section~\ref{infrared_regions_in_wideangle_scattering}, are sufficient for $f_R$ to be $N$-dimensional, hence a facet. We still need to show that these requirements are necessary, which we will prove by contradiction below.

First, we suppose that the requirement of $H$ is not satisfied. In other words, there is an edge $e\in H_\textup{red}$ whose momentum is $p_i$ or $0$. We further suppose that $e\in G$, i.e. $e$ is not an auxiliary propagator of $H_\text{red}$. Then we will show that in every term of $\mathcal{P}_0^{(R)}$, the corresponding spanning (2-)tree contains $e$. To see this, recall from theorem~\ref{theorem-leading_terms_form} that, for a spanning (2-)tree $T$ corresponding to any $\mathcal{U}^{(R)}$, $\mathcal{F}^{(p_i^2,R)}$ or $\mathcal{F}_\text{II}^{(q_{ij}^2,R)}$ term, $H\cap T$ is a spanning tree of $H$. (This follows from the factorisation property in eq.~(\ref{summary_leading_terms_factorise}).) This indicates that $e\in H\cap T$ in these terms.
In the case that $T$ corresponds to an $\mathcal{F}_\text{I}^{(q^2,R)}$ term, $H\cap T$ is a spanning 2-tree, and the momentum flowing between its components must be off-shell. If $e\notin H\cap T$, then since $e\in H_\text{red}$, which has a tree structure, the momentum flowing between the components of $T$ must be the same momentum carried by $e$ itself, which is $p_i$ or $0$ based on our contradiction assumption. As a result, $e\in H\cap T$ for any $T$ that corresponds to a term of $\mathcal{P}_0^{(R)}(\x;\s)$. In other words, the parameter $x_e$ is absent in all the terms of $\mathcal{P}_0^{(R)}(\x;\s)$. This implies that all the points $\boldsymbol{r}\in f_R$, in addition to admitting the defining constraint that $\boldsymbol{r}\cdot \boldsymbol{v}_R$ is fixed at some (minimum) value, admit an \emph{extra constraint} that the entry of the vector $\boldsymbol{r}$ that corresponds to $x_e$ is always~$0$. This extra constraint makes the dimension of $f_R$ less than $N$, thus $f_R$ cannot be a lower facet. 

An argument can be made also if $e$ is an auxiliary propagator of $H_\text{red}$, carrying momentum $p_i$ or $0$. In this case, this propagator connects to a subgraph consisting of at least one loop, which is scaleless in the limit $p_i^2\to 0$, and thus cannot be part of $\mathcal{P}_0^{(R)}(\x;\s)$.

Second, suppose that the requirement of $J$ is not satisfied, i.e. there is an edge $e\in \widetilde{J}_{i,\text{red}}$ with zero momentum, then we can use the analysis above to show that $\text{dim}(f_R)<N$, which implies that $f_R$ is not a lower facet.

Third, suppose that the requirement of $S$ is not satisfied, i.e. one connected component $S_1$ is attached to only the jet $J_1$ and/or $H$. We will show below that for every spanning (2-)tree $T$ that corresponds to a $\mathcal{P}_0^{(R)}$ term, $S_1\cap T$ is a spanning tree of $S_1$. This follows directly from theorem~\ref{theorem-leading_terms_form} if $T$ corresponds to a $\mathcal{U}^{(R)}$, $\mathcal{F}^{(p_i^2,R)}$ or $\mathcal{F}_\text{I}^{(q^2,R)}$ term. For $T(=T^2)$ that corresponds to an $\mathcal{F}_\text{II}^{(q_{ij}^2,R)}$ term, we recall the discussion in section~\ref{leading_terms_in_infrared_regions}, that $T^2$ can be obtained from a $\mathcal{U}^{(R)}$ term $T^1$ by the following operations. We first add a soft edge to $T^1$, and then in the loop formed we remove two edges from different jets. However, this procedure cannot be realised if the soft edge is from $S_1$, since $S_1$ is not attached to two different jets. As a result, in any $\mathcal{F}_\text{II}^{(q_{ij}^2,R)}$ term, the aforementioned soft edge must be part of $S\setminus S_1$, and then this operation does not affect $S_1\cap T$, which must remain as a spanning tree of $S_1$ as in the original $\mathcal{U}^{(R)}$ term. This implies an extra constraint on the points $\boldsymbol{r}\in f_R$: the sum over all the entries of $\boldsymbol{r}$, which correspond to the edges of $S_1$, is fixed as $L(\widetilde{S}_1)$, as seen in corollary~\ref{theorem-leading_terms_form-corollary}. As above, this extra constraint makes the dimension of $f_R$ less than $N$.

\bigbreak
In summary, we have proved that the requirements proposed in section~\ref{infrared_regions_in_wideangle_scattering} regarding the subgraphs $H,J$ and $S$ provide a necessary and sufficient condition that $\boldsymbol{v}_R$ is normal to a lower facet of the Newton polytope $\Delta^{(N+1)}[\mathcal {P}(G)]$.

\bibliographystyle{JHEP}
\bibliography{refs}
\end{document}